\documentclass[11pt,a4paper,twoside,openany]{book}
\usepackage[top=1.1in, bottom=1.1in, left=1in, right=1in]{geometry}
\usepackage{amsmath}
\usepackage{amsthm}
\usepackage{amsfonts}
\usepackage{amssymb}
\usepackage{graphicx}
\usepackage{hyperref}
\usepackage[ruled,vlined]{algorithm2e}
\usepackage{float}
\usepackage{verbatim}
\usepackage{appendix}
\usepackage[usenames,dvipsnames]{color}
\definecolor{darkblue}{rgb}{0.0,0.0,0.3}

\hypersetup{
    colorlinks = true,
    citecolor = red,
    linkcolor = blue
}
\usepackage{subcaption}

\usepackage{multicol}

\newtheorem{theorem}{Theorem}[chapter]
\newtheorem{lemma}[theorem]{Lemma}

\newtheorem{definition}[theorem]{Definition}
\newtheorem{corollary}[theorem]{Corollary}
\newtheorem{claim}[theorem]{Claim}
\newtheorem{remark}[theorem]{Remark}

\newcommand{\E}{\mathbb{E}}
\newcommand{\Prob}{\mathbf{Pr}}
\newcommand{\NP}{\textrm{NP}}
\renewcommand{\P}{\textrm{P}}
\newcommand{\RP}{\textrm{RP}}
\newcommand{\ZPP}{\textrm{ZPP}}
\newcommand{\poly}{\texttt{poly}}
\newcommand{\calI}{\mathcal{I}}
\DeclareMathOperator {\Harm}{Harm}
\DeclareMathOperator {\Path}{Path}

\newcommand{\HLp}{HL$_p$\xspace}
\newcommand{\HL}[1]{HL$_#1$\xspace}

\newcommand{\R}{\mathbb{R}}
\newcommand{\MIS}{\texttt{MIS}}
\newcommand{\MVC}{\texttt{MVC}}

\newcommand{\X}{\mathcal{X}}
\newcommand{\F}{\mathcal{F}}
\newcommand{\C}{\mathcal{C}}
\newcommand{\cP}{\mathcal{P}}

\begin{document}
\pagenumbering{Roman}

\begin{titlepage}
\begin{center}

\LARGE \textbf{Shortest path queries, graph partitioning and covering problems in worst and beyond worst case settings}\\[20pt]
\Large by\\[10pt]
Charalampos Angelidakis\\[60pt]

\Large A thesis submitted\\[5pt]
in partial fulfillment of the requirements for\\[5pt]
the degree of\\[20pt]
Doctor of Philosophy in Computer Science\\[20pt]
at the\\[20pt]
TOYOTA TECHNOLOGICAL INSTITUTE AT CHICAGO\\[5pt]
Chicago, Illinois\\[40pt]
August, 2018\\[80pt]

Thesis Committee:\\[10pt]
Yury Makarychev (Thesis Advisor)\\[5pt]
Avrim Blum\\[5pt]
Julia Chuzhoy\\[5pt]
Aravindan Vijayaraghavan

\end{center}
\end{titlepage}

\newpage
\thispagestyle{empty}
\mbox{}
\newpage
\thispagestyle{empty}

\begin{center}

\Large\textbf{Shortest path queries, graph partitioning and covering problems in worst and beyond worst case settings}\\[30pt] \normalsize
A thesis presented\\[5pt]
by\\[25pt]
\Large Haris (Charalampos) Angelidakis\\[30pt] \normalsize
in partial fulfillment of the requirements for the degree of\\[12pt]
Doctor of Philosophy in Computer Science.\\[12pt]
Toyota Technological Institute at Chicago\\[12pt]
Chicago, Illinois\\[12pt]
August, 2018\\[45pt]

---\;Thesis Committee\;---\\[30pt]
\begin{table}[h]
\begin{tabular}{l l l}
    Julia Chuzhoy & &\\
    \rule{4.8cm}{1pt}\quad\quad\quad & \rule{5cm}{1pt}\;\;\; & \rule{3cm}{1pt}\\[5pt]
    Committee member & Signature & Date\\[40pt]

    Aravindan Vijayaraghavan & &\\
    \rule{4.8cm}{1pt}\quad\quad\quad & \rule{5cm}{1pt}\;\;\; & \rule{3cm}{1pt}\\[5pt]
    Committee member & Signature & Date\\[40pt]

    Yury Makarychev & &\\
    \rule{4.8cm}{1pt}\quad\quad\quad & \rule{5cm}{1pt}\;\;\; & \rule{3cm}{1pt}\\[5pt]
    Thesis/Research Advisor & Signature & Date\\[40pt]

    Avrim Blum & &\\
    \rule{4.8cm}{1pt}\quad\quad\quad & \rule{5cm}{1pt}\;\;\; & \rule{3cm}{1pt}\\[5pt]
    Chief Academic Officer & Signature & Date

\end{tabular}
\end{table}

\end{center}

\newpage
\thispagestyle{empty}
\mbox{}
\newpage
\thispagestyle{empty}
\begin{center}
\large \textbf{Shortest path queries, graph partitioning and covering problems in worst and beyond worst case settings}
\vspace{20pt}

by\\
Charalampos Angelidakis
\end{center}

\vspace{20pt}

\normalsize

\begin{center}
\textbf{Abstract}
\end{center}

In this thesis, we design algorithms for several \NP-hard problems in both worst and beyond worst case settings. In the first part of the thesis, we apply the traditional worst case methodology and design approximation algorithms for the Hub Labeling problem; Hub Labeling is a preprocessing technique introduced to speed up shortest path queries. Before this work, Hub Labeling had been extensively studied mainly in the beyond worst case analysis setting, and in particular on graphs with low highway dimension (a notion introduced in order to explain why certain heuristics for shortest paths are very successful in real-life road networks). In this work, we significantly improve our theoretical understanding of the problem and design (worst-case) algorithms for various classes of graphs, such as general graphs, graphs with unique shortest paths and trees, as well as provide matching inapproximability lower bounds for the problem in its most general settings. Finally, we demonstrate a connection between computing a Hub Labeling on a tree and searching for a node in a tree.

In the second part of the thesis, we turn to beyond worst case analysis and extensively study the stability model introduced by Bilu and Linial in an attempt to describe real-life instances of graph partitioning and clustering problems. Informally, an instance of a combinatorial optimization problem is stable if it has a unique optimal solution that remains the unique optimum under small (multiplicative and adversarial) perturbations of the parameters of the input (e.g.~edge or vertex weights). Utilizing the power of convex relaxations for stable instances, we obtain several results for problems such as Edge/Node Multiway Cut, Independent Set (and its equivalent, in terms of exact solvability, Vertex Cover), clustering problems such as $k$-center and $k$-median and the symmetric Traveling Salesman problem. We also provide strong lower bounds for certain families of algorithms for covering problems, thus exhibiting potential barriers towards the design of improved algorithms in this framework.

\newpage

\thispagestyle{empty}
\mbox{}
\newpage
\thispagestyle{empty}
\begin{center}
\Large
\textbf{Acknowledgements}
\end{center}
\normalsize

Concluding this 6-year journey, there are a lot of people that I would like to thank, starting with my advisor, Yury Makarychev. Yury's generosity with his time and his ideas is simply unmatched. He spent countless hours explaining to me concepts, ideas, techniques, and most of the things I learned during my PhD I learned through him. Brainstorming together and sharing ideas was really an eye-opening experience that helped me get a better grasp of how research is done. I will always be grateful to him for his time, help and kindness.

None of this would have been possible without the trust and support of Julia Chuzhoy. Julia was the reason I joined TTIC in the first place and I learned a lot working with her in the first 1.5 years. Her professionalism and work ethic are second to none and her advice and sincerity helped me a lot early on. I would like to thank her for her time during the first two years of my PhD and for happily joining my committee and helping me throughout the last stages of it.

I am grateful to Avrim Blum and Aravindan Vijayaraghavan, the other two members of my committee, I was lucky enough to overlap with Avrim for one year at TTI, and I am glad I got the chance to work with him. I would also like to thank Aravindan for his time and help in the last year of my PhD, as well as Madhur Tulsiani for all his time, support and understanding.

During my PhD, I had the chance to work with great people that I would like to thank: Pranjal Awasthi, Vaggos Chatziafratis, Chen Dan, Konstantin Makarychev, Pasin Manurangsi, Vsevolod Oparin and Colin White.

And now starts the long list of friends. From TTI, I will start with Shubhendu, my office-mate and one of the closest friends that I have had throughout. Every-day life at TTI would be much more boring without Shubhendu, his printouts, his random texts/messages, his attempts to speak and write Greek, his love for books and cigars... The list could go on forever, and I would like to thank him for making life in Chicago much more interesting. Mrinal was my other office-mate for several years, and along with Rachit, the core of Theory students at TTI. I thank them both for all the interesting conversations and all the Theory-related jokes that helped us survive the ML storm that we were caught in. Many thanks also to Somaye, Behnam, Bahador and Vikas that made adapting to the US reality much easier. 

And now come the Greeks. From the first days in Chicago, I felt that I had friends that I could trust if anything went wrong. I will start with Eleftheria, the very first person I met when I visited Chicago as a prospective student. And of course my high-schoolmate Panagiotis. Then came Dimitris (x2), Tasos, Katerina and a few years later Tony, Monika, Aristotelis, Aris, Panos and Dimitris (the new guy!); I will never forget all the fun and great moments we had together. And, of course, the great Nicholas, Pantelis and Alex, along with Elina, Iro and Maria. Valia was one of the kindest people I met in Chicago, and I was lucky to run into her the first time I visited Hyde Park! Valia also introduced me to George (epistimon), who ended up being one of my closest friends. I will never forget the never-ending discussions about mathematics, computer science, politics and all other serious issues of life that I had with George, as well as his chaotic sleep/work schedule, his love for the good old greek music, and his childlike enthusiasm towards all the things he liked. Finally, I would like to thank the friends outside of Chicago for their help, support, and for checking with me once in a while: Alex, Dimitris, Nikos, Kostas (x2), George, Andreas, Matoula and Thodoris, as well as the artists Areti and Alex. The list of people could go on and on, and I apologize to those I am forgetting.

Last but not least, I would like to thank my family; my parents George and Maria, and my sister Eirini, along with her husband George. Without their unconditional love and support, and their faith in me throughout all my life, none of this would have been possible. I am grateful to them for always being there, supporting me with no questions asked and no doubts about the choices I made. Finally, it is hard to find words of love and appreciation for Olina, who followed me patiently in this long journey and who supported and believed in me throughout the highs and (many) lows of it. I will never forget that.

\newpage
\thispagestyle{empty}
\mbox{}

\frontmatter
\tableofcontents
\newpage
\thispagestyle{empty}
\mbox{}
\newpage
\thispagestyle{empty}
\mbox{}
\listoffigures
\newpage
\thispagestyle{empty}
\mbox{}
\newpage
\listofalgorithms
\newpage
\thispagestyle{empty}
\mbox{}

\mainmatter

\chapter{Introduction}

Traditionally, the field of algorithm design has been concerned with worst-case analysis, requiring that algorithms work for every possible instance of a problem. This approach has proved very fruitful, leading to the development of an elegant theory of algorithm design and analysis. The focus on worst-case instances has also been the driving force behind the theory of \NP-completeness, a cornerstone of Computer Science. However, it has also created a significant barrier for the design of efficient (i.e.~polynomial-time) exact algorithms. Assuming $\P \neq \NP$, we know that we cannot have efficient algorithms that optimally solve \emph{every} instance of any of the so-called $\NP$-hard problems.

Ideally, the three main conditions that an algorithm should satisfy are the following: (i) it should work for every instance, i.e.~return a feasible solution for all inputs, (ii) it should always run in polynomial time, and (iii) it should return an optimal solution. \NP-hardness suggests that it is unlikely that an algorithm can satisfy all these properties for an \NP-hard problem. Thus, a natural thing to do is drop one of these conditions and aim to satisfy the remaining two. This gives rise to three predominant approaches towards handling \NP-hardness.

The first approach drops the optimality condition and suggests the design of approximation algorithms for \NP-hard problems. More precisely, one can relax the condition of optimality when designing algorithms for an \NP-hard problem and ask for algorithms that still work for every instance of the problem and return an ``approximately" good solution. The standard formalization of an $\alpha$-approximation algorithm, for some parameter $\alpha \geq 1$, is an efficient algorithm that, for a minimization problem whose optimal cost is $OPT$, returns a feasible solution whose cost is at most $\alpha \cdot OPT$. The definition for maximization problems is similar. Such approaches have led to the development of the theory of approximation algorithms, a very rich and mature field of theoretical Computer Science that has given several breakthrough results throughout the years.

The second approach drops the universality condition; it relaxes the severe restriction that the algorithm must work for every instance of a problem. In other words, instead of designing approximation algorithms that work for every instance of an \NP-hard problem, we design efficient algorithms that are optimal or near-optimal, but only work for a restricted subset of instances of an \NP-hard problem. One standard way of doing so is by looking at natural restricted classes of instances. For example, if we are dealing with an optimization problem defined on general graphs, we could first try to solve the problem on special classes of graphs such as trees, planar graphs, bounded-degree graphs, bounded-treewidth graphs, sparse/dense graphs etc. In many cases, such a restriction makes the problem much easier (e.g. Vertex Cover is easy on bipartite graphs and admits a PTAS on planar graphs~\cite{DBLP:journals/jacm/Baker94}), and also gives insights about where the difficulty of a problem stems from.

Besides these (mathematically) natural classes of instances, during the last few years, in an attempt to explain why certain heuristics seem to work in practice and why some \NP-hard problems seem to be solvable in real life, a lot of works have tried to describe classes of instances that correspond to \emph{average-case} instances or \emph{real-life} instances. This research direction is usually referred to as \emph{beyond worst-case analysis} and has gained much traction lately. There are two challenges immediately raised by such an approach. The first is to theoretically model and describe such instances (e.g.~what is a real-life instance), and the second is, given such a model of instances, to design algorithms that provably work better in this model. Many such models have been proposed, that can roughly be divided into two large classes.

\begin{enumerate}
    \item \textbf{Generative models}: In generative models, one describes a procedure that generates an average-case/real-life instance. Some examples are random instances, semi-random instances, planted random instances and others (e.g. see~\cite{DBLP:conf/focs/McSherry01, DBLP:conf/focs/DyerF86, DBLP:journals/rsa/AlonKS98, DBLP:journals/cpc/BollobasS04, DBLP:conf/focs/Boppana87, DBLP:conf/soda/Coja-Oghlan05, DBLP:conf/stoc/MakarychevMV12, DBLP:conf/stoc/MakarychevMV14}). In many cases, such approaches have led to the development of improved algorithms that work optimally or near-optimally in these models.
    \item \textbf{Descriptive models}: In descriptive models, one describes structural properties that real-life instances (seem to) satisfy. A prominent example of a family of problems for which many different descriptive models have been proposed are the various clustering problems that have been defined and proved to be \NP-hard. For example, many conditions such as approximation stability~\cite{DBLP:journals/jacm/BalcanBG13}, spectral proximity condition~\cite{DBLP:conf/focs/KumarK10} and others have been proposed for the $k$-means, $k$-median and other objectives, that aim to describe real-life ``meaningful" instances of clustering, and for which one can prove improved guarantees for various algorithms.
\end{enumerate}

Finally, the third approach allows for superpolynomial-time algorithms, such as quasi-polynomial-time algorithms, subexponential algorithms, fixed-parameter algorithms etc. In this thesis, we mainly utilize the first two approaches.

\paragraph{This Thesis.} In this thesis, we focus on the interplay between worst and beyond worst case analysis and how these two approaches have given rise to interesting problems and questions. The first part of the thesis follows the more traditional approach of designing approximation algorithms for $\NP$-hard problems, but the problems we are interested in are problems that are solved efficiently in real life and which have inspired interesting beyond worst-case notions. In particular, we study the Hub Labeling framework, a preprocessing technique aimed at speeding up shortest-path queries. Since its inception by Cohen et al.~\cite{DBLP:journals/siamcomp/CohenHKZ03} and Gavoille et al.~\cite{DBLP:journals/jal/GavoillePPR04}, the Hub Labeling framework has been very successful in practice, and is currently used in many state-of-the-art algorithms (see e.g.~\cite{DBLP:conf/wea/AbrahamDGW11}). In order to explain the success of these methods, Abraham et al.~\cite{DBLP:conf/soda/AbrahamFGW10} introduced the notion of highway dimension and claimed that road networks have small highway dimension. Moreover, they proved that small highway dimension implies the existence of efficient data structures (i.e.~hub labelings) that significantly improve the response time for both worst-case and average-case distance queries. In other words, they proposed a descriptive model in which one is able to prove good (absolute) upper bounds on the size of the data structures constructed. A natural and well-justified question to ask now is how easy it is to construct the optimal hub labelings, given that the already obtained absolute bounds demonstrate that the framework indeed works very well in practice. More formally, one can take a step back, apply the traditional theoretical methodology and ask whether computing the optimal hub labeling is $\NP$-hard, and if it is, what is the best approximation that we can get.

These questions had been posed even prior to this thesis. It is known that the most standard versions of Hub Labeling are indeed $\NP$-hard, and, moreover, there is an $O(\log n)$-approximation algorithm known for the problem. In this work, we prove strong lower bounds on the approximability of Hub Labeling, thus extending the $\NP$-hardness results to hardness of approximation results. Then, we make a structural assumption that is common in the literature, namely that in road networks shortest paths are unique. Although the problem remains \NP-hard even under such an assumption, we obtain improved approximation algorithms for graphs with unique shortest paths and shortest-path diameter $D$, as well as graphs that are trees; in particular, a structural result of ours implies that Hub Labeling on trees is equivalent to the problem of searching for a node in a tree, for which polynomial-time algorithms are known. To obtain these results, we use combinatorial techniques as well as convex relaxations and rounding techniques; to the best of our knowledge, linear/convex programming techniques had not been applied to the Hub Labeling problem prior to our work.

In the second part of the thesis, we apply the beyond worst-case methodology and study a descriptive model that was proposed in an influential paper of Bilu and Linial about a decade ago. More precisely, we are interested in the notion of stability introduced by Bilu and Linial~\cite{DBLP:journals/cpc/BiluL12} for graph partitioning and optimization problems, and its extension to clustering problems, defined by Awasthi, Blum and Sheffet~\cite{DBLP:journals/ipl/AwasthiBS12} under the term perturbation resilience. Informally, an instance is stable if there is a unique optimal solution that remains the unique optimal solution under small perturbations of the parameters of the input; the larger the perturbations that are allowed without affecting the optimal solution, the more stable the instance is. Having this definition as their starting point, Bilu and Linual explore how much stability is needed so as to be able to recover the unique optimal solution in polynomial time. Their test case was the Max Cut problem, where they gave the first upper bounds on the stability that allowed one to recover the optimal solution. Similarly, Awasthi, Blum and Sheffet studied the most common ``center-based" clustering objectives such as $k$-median, $k$-means and $k$-center and gave upper bounds on the stability that is needed in order to recover the unique optimal clustering under any such objective. Continuing the line of work inspired by these two papers, with a main focus on the work of Makarychev et al.~\cite{DBLP:conf/soda/MakarychevMV14} (that studied Max Cut and Multiway Cut), we use and extend the framework introduced in~\cite{DBLP:conf/soda/MakarychevMV14} and give improved algorithms (i.e.~algorithms that require smaller stability) for stable instances of Multiway Cut, using the CKR relaxation. Moreover, we give a tight analysis of the standard path-based LP relaxation of Edge/Node Multiway Cut, thus proving the first upper bounds for the more general Node Multiway Cut problem. Extending the notion of stability to covering problems, such as Vertex Cover and Set Cover, we give strong lower bounds on certain families of algorithms (in particular, robust algorithms, i.e.~algorithms that are not allowed to err, even if the instance is not stable) and also give several algorithmic results for stable instances of Vertex Cover. We note here that the presentation of the algorithms for Vertex Cover is in the context of Independent Set; since we are interested in exact solvability, it is easy to observe that the two problems are equivalent, and so our algorithms work for both problems. We conclude with some LP-based results for perturbation-resilient $k$-center and $k$-median, and with an analysis of the classic ``subtour-elimination" LP relaxation for stable instances of the symmetric Traveling Salesman problem.

\paragraph{Organization of material.} The thesis is organized into two parts. The first part studies the Hub Labeling problem. In Chapter~\ref{chapter:hl} we introduce the problem and present our results for graphs with unique shortest paths, as well as the hardness results for general graphs. Then, in Chapter~\ref{chapter:trees}, we focus on Hub Labeling on trees, and present several algorithms, culminating with the formalization of the equivalence between Hub Labeling and the problem of searching for a node in a tree, first observed and communicated to us by Gawrychowski et al.~\cite{GawKMW17}. We conclude the first part with some interesting open problems and directions (see Chapter~\ref{chapter:open-problems-hl}).

The second part of the thesis explores the beyond worst-case analysis framework introduced by Bilu and Linial. In Chapter~\ref{chapter:stability} we formally introduce and describe the model. In Chapter~\ref{chap:multiway-cut} we study the Multiway Cut problem and give improved algorithms for stable instances of the Edge Multiway Cut problem, as well as a tight analysis of the standard LP relaxation for stable instances of the Node Multiway Cut problem. We conclude with strong lower bounds on robust algorithms for stable instances of the Node Multiway Cut problem. In Chapter~\ref{chap:hardness} we provide strong lower bounds for robust algorithms for stable instances of covering problems such as Vertex Cover/Independent Set, Set Cover and others. In Chapter~\ref{chap:independent-set}, we study the Vertex Cover problem in its equivalent Independent Set formulation, and give algorithms for several special classes of instances, such as bounded-degree graphs, graphs with low chromatic number and planar graphs. In Chapter~\ref{chap:clustering}, we study the (equivalent) notion of perturbation resilience for clustering problems through the lens of linear programming. In particular, we give a robust LP-based algorithm for metric-perturbation-resilient instances of $k$-center, as well as lower bounds on the integrality of the standard LP relaxation for $k$-median on perturbation-resilient instances. Finally, in Chapter~\ref{chap:tsp}, we continue the exploration of the power of LP relaxations for stable instances and give a robust LP-based algorithm for stable instances of the symmetric Traveling Salesman problem. We conclude in Chapter~\ref{chapter:open-problems-stability} with some interesting open problems.

\paragraph{Notational and other conventions.} Throughout this thesis, we use the following conventions:
\begin{itemize}
    \item For any positive integer $n$, the notation $[n]$ denotes the set $\{1, ..., n\}$.
    \item Whenever not specified, $n$ denotes the number of vertices of a graph, $\Delta$ denotes the maximum degree of a graph, and $D$ denotes the shortest-path diameter of a graph, i.e.~the maximum number of vertices that appear in any shortest path.
    \item Regarding the bibliography, we always cite the journal version of a work, if any. Whenever we mention year of publication though, we write the year that the first conference version of a work appeared.
\end{itemize}

\part{Hub Labeling and related problems}\label{part:HL}

\chapter{The Hub Labeling problem}\label{chapter:hl}
\section{Introduction}

Computing shortest-path queries has become an essential part of many modern-day applications. A typical setting is a sparse input graph $G = (V,E)$ of millions of nodes (we denote $n = |V|$) with a length function $l: E \to \mathbb{R}_{>0}$, and a very large number of queries that need to be answered in (essentially) real time. Two classical approaches to such a problem are the following. One could precompute all pairwise distances and store them in an $n \times n$ matrix, and then respond to any distance query in constant time (and, by using the appropriate data structures of size $O(n^2)$, one can also recover the vertices of the shortest path in time linear in the number of vertices that the path contains). This approach, although optimal with respect to the query time, is potentially wasteful with respect to space. Moreover, in many applications quadratic space is simply prohibitive. A second approach then would be to simply store an efficient graph representation of the graph, which for sparse graphs would result in a representation of size $\widetilde{O}(n)$. In this second approach, whenever a query arrives, one can run Dijkstra's algorithm and retrieve the distance and the corresponding shortest path in linear time for undirected graphs with integer weights~\cite{DBLP:journals/jacm/Thorup99}, and, more generally, in time $O\left(|E|+ |V| \log |V| \right)$ for arbitrary weighted directed graphs~\cite{Fredman:1987:FHU:28869.28874}. Two obvious problems with this latter approach are that linear time is nowhere close to real time, and moreover, such an approach requires a network representation that is global in nature, and so it does not allow for a more distributed way of computing shortest-path queries. Thus, a natural question that arises is whether we can get a trade-off between space and query time complexity, and whether we can obtain data structures that inherently allow for distributed computations as well (the latter is a desirable property in many applications, when, ideally, one would not want a central coordination system).

Data structures that allow for responding to distance queries are usually called \emph{distance oracles}, and have been intensively studied in the last few decades, mainly focusing on the optimal trade-offs between space and query time, as well as exact/approximate recovery (e.g.~see~\cite{DBLP:journals/jacm/ThorupZ05, DBLP:journals/siamcomp/PatrascuR14, DBLP:conf/soda/KawarabayashiST13, DBLP:conf/focs/Cohen-AddadDW17, DBLP:conf/soda/GawrychowskiMWW18}). 

Here, we will mainly be interested in a slightly different approach based on vertex labelings, that allows for simple schemes that are easy to implement in a distributed setting. The starting point is the observation that in an explicit representation of (parts of) the adjacency matrix of a graph, the names of the vertices are simply place holders, not revealing any information about the structure of the graph. This motivates the search for more informative names (or labels) for each vertex, that would allow us to derive some information about the vertex.

The first such approach was introduced by Breuer and Folkman~\cite{DBLP:journals/tit/Breuer66, Breuer1967583}, and involves using more localized labeling schemes that allow one to infer the adjacency of two nodes directly from their labels, without using any additional information, while achieving sublinear space bounds. A classic work of Kannan, Naor and Rudich~\cite{DBLP:journals/siamdm/KannanNR92} further explored the feasibility of efficient adjacency labeling schemes for various families of graphs. Taking this line of research a step further, Graham and Pollak~\cite{Graham1972} were among the first to consider the problem of labeling the nodes of an unweighted graph such that the distance between two vertices can be computed using these two labels alone. They proposed to label each node with a word of $q_n$ symbols (where $n$ is the number of vertices of the graph) from the set $\{0, 1, *\}$, such that the distance between two nodes corresponds to the Hamming distance of the two words (the distance between $*$ and any symbol being zero). Referenced as the Squashed Cube Conjecture, Winkler~\cite{DBLP:journals/combinatorica/Winkler83} proved that $q_n \leq n - 1$ for every $n$ (note though that the scheme requires linear query time to decode the distance of a pair).

Moving towards the end of the 90s, Peleg~\cite{DBLP:journals/jgt/Peleg00} revisited the problem of existence of efficient labeling schemes of any kind that could answer shortest-path queries. The setting now is quite general, in that we are allowed as much preprocessing time as needed for the whole network, and the goal is to precompute labels for each vertex of the graph such that any shortest-path query between two vertices can be computed by looking only at the corresponding labels of the two vertices (and applying some efficiently computable ``decoding" function on them that actually computes the distance). If the labels are short enough on average, then the average query time can be sublinear. In~\cite{DBLP:journals/jgt/Peleg00}, Peleg did manage to give polylogarithmic upper bounds for the size of the labels needed to answer exact shortest-path queries for weighted trees and chordal graphs, and also gave some bounds for distance approximating schemes. Gavoille et al.~\cite{DBLP:journals/jal/GavoillePPR04} continued along similar lines and proved various upper and lower bounds for the label size for various classes of (undirected) graphs. They also modified the objective function and, besides getting bounds for the size of the largest label, they also obtained bounds for the average size of the labels. Shortly after that work, Cohen et al.~\cite{DBLP:journals/siamcomp/CohenHKZ03} presented their approach for the problem, proposing what is now known as the Hub Labeling framework for generating efficient labeling scheme for both undirected and directed weighted graphs.

\begin{definition}[Hub Labeling~\cite{DBLP:journals/siamcomp/CohenHKZ03, DBLP:journals/jal/GavoillePPR04}]
Consider an undirected graph $G=(V,E)$ with edge lengths $l(e) > 0$. Suppose that we are given a set system $\{H_u\}_{u\in V}$ with one set $H_u\subset V$ for every vertex $u$. We say that  $\{H_u\}_{u\in V}$ is a hub labeling if it satisfies the following covering property: for every pair of vertices $(u, v)$  ($u$ and $v$ are not necessarily distinct), there is a vertex in $H_u \cap H_v$ (a common ``hub'' for $u$ and $v$) that lies on a shortest path between $u$ and $v$. We call vertices in sets $H_u$ hubs: a vertex $v\in H_u$ is a hub for $u$.
\end{definition}
\noindent In the Hub Labeling problem (HL), our goal is to find a hub labeling with a small number of hubs; specifically, we want to minimize the $\ell_p$-cost of a hub labeling.
\begin{definition}
The $\ell_p$-cost of a hub labeling $\{H_u\}_{u\in V}$ equals $(\sum_{u \in V} |H_u|^p)^{1/p}$ for $p\in[1,\infty)$; the $\ell_\infty$-cost is $\max_{u \in V} |H_u|$. The hub labeling problem with the $\ell_p$-cost, which we denote by \HLp, asks to find a hub labeling with the minimum possible $\ell_p$-cost.
\end{definition}

We note here that, although our presentation will only involve undirected graphs, most of our results extend to the directed setting as well (see Section~\ref{sec:appendix_directed}). In the next few sections, we will study \HLp and design approximation algorithms for various classes of graphs, as well as show strong lower bounds for general graphs. But first, we will explain why we care about the Hub Labeling problem, and how it is related to the shortest-path problem.

Nowadays hundreds of millions of people worldwide use web mapping services and GPS devices to get driving directions. That creates a huge demand for fast algorithms for computing shortest paths (algorithms that are even faster than the classic Dijkstra's algorithm). Hub labelings provide a highly efficient way for computing shortest paths and is used in many state-of-the-art algorithms (see also the paper of Bast et al. \cite{DBLP:series/lncs/BastDGMPSWW16} for a review and discussion of various methods for computing shortest paths that are used in practice). 

We will now demonstrate the connection between the Hub Labeling and the problem of computing shortest paths. Consider a graph $G=(V,E)$ with edge lengths $l(e) > 0$. Let $d(u,v)$ be the shortest-path metric on $G$. Suppose that we have a hub labeling $\{H_u\}_{u \in V}$. During the preprocessing step, we compute and store the distance $d(u,w)$ between every vertex $u$ and each hub $w\in H_u$ of $u$. Observe that we can now quickly answer a distance query: to find $d(u,v)$ we compute $\min_{w\in H_u \cap H_v} \left(d(u,w) + d(v,w) \right)$. By the triangle inequality, $d(u,v) \leq \min_{w\in H_u \cap H_v} \left(d(u,w) + d(v,w) \right)$, and the covering property guarantees that there is a hub $w\in H_u \cap H_v$ on a shortest path between $u$ and $v$; so $d(u,v) = \min_{w\in H_u \cap H_v} \left(d(u,w) + d(v,w) \right)$. We can compute $\min_{w\in H_u \cap H_v} \left(d(u,w) + d(v,w) \right)$ and answer the query in time  $O(\max(|H_u|, |H_v|))$. We need to keep a lookup table of size $O(\sum_{u\in V} |H_u|)$ to store the distances between the vertices and their hubs. So, if, say, all hub sets $H_u$ are of polylogarithmic size, the algorithm answers a distance query in polylogarithmic time and requires $n \mathop{\mathrm{polylog}} n$ space. The outlined approach can be used not only for computing distances but also shortest paths between vertices. It is clear from this discussion that it is important to have a hub labeling of small size, since both the query time and storage space depend on the number of hubs.

Recently, there has been a lot of research on algorithms for computing shortest paths using the hub labeling framework (see e.g.~the following papers by Abraham et al.~\cite{DBLP:conf/soda/AbrahamFGW10, DBLP:conf/wea/AbrahamDGW11, DBLP:journals/jea/AbrahamDGW13, DBLP:conf/esa/AbrahamDGW12, DBLP:conf/icalp/AbrahamDFGW11, DBLP:conf/gis/AbrahamDFGW12}). It was noted that these algorithms perform really well in practice (see e.g.~\cite{DBLP:conf/wea/AbrahamDGW11}). A systematic attempt to explain why this is the case led to the introduction of the notion of \emph{highway dimension}~\cite{DBLP:conf/soda/AbrahamFGW10}. Highway dimension is an interesting concept that managed to explain, at least partially, the success of the above methods: it was proved that graphs with small highway dimension have hub labelings with a small number of hubs; moreover, there is evidence that most real-life road networks have low highway dimension~\cite{Bast06}. Even more recently, Kosowski and Viennot~\cite{DBLP:conf/soda/KosowskiV17}, inspired by the notion of highway dimension, introduced another related notion, the \emph{skeleton dimension}, that is a slightly more tractable and elegant notion that again explains, to some extent, why the hub labeling framework is successful for distance queries.

However, most papers on Hub Labeling offer only algorithms with absolute guarantees on the cost of the hub labeling they find (e.g.~they show that a graph with a given highway dimension has a hub labeling of a certain size and provide an algorithm that finds such a hub labeling); they do not relate the cost of the hub labeling to the cost of the optimal hub labeling. There are very few results on the approximability of the Hub Labeling problem. Only very recently, Babenko et al.~\cite {DBLP:conf/mfcs/BabenkoGKSW15} and White~\cite{DBLP:conf/esa/White15} proved respectively that \HL1 and \HL\infty are NP-hard. Cohen et al.~\cite{DBLP:journals/siamcomp/CohenHKZ03} gave an $O(\log n)$-approximation algorithm for \HL1 by reducing the problem to a Set Cover instance and using the greedy algorithm for Set Cover to solve the obtained instance (the latter step is non-trivial since the reduction gives a Set Cover instance of exponential size); later, Babenko et al.~\cite{DBLP:conf/icalp/BabenkoGGN13} gave a combinatorial $O(\log n)$-approximation algorithm for \HLp, for any $p \in [1, \infty]$.

\paragraph{Our results.} In this thesis, we will present the following results (most of which were published in 2017~\cite{DBLP:conf/soda/AngelidakisMO17}). We prove an $\Omega(\log{n})$ hardness for \HL1 and \HL\infty on graphs that have multiple shortest paths between some pairs of vertices (assuming that $\mathrm{P}\neq \mathrm{NP}$). The result (which easily extends to \HLp for $p = \Omega(\log n)$ on graphs with $n$ vertices) shows that the algorithms by Cohen et al.~and Babenko et al.~are optimal, up to constant factors. Since it is impossible to improve the approximation guarantee of $O(\log n)$ for arbitrary graphs, we focus on special families of graphs. We consider the family of graphs with \textit{unique shortest paths} --- graphs in which there is only one shortest path between every pair of vertices. This family of graphs appears in the majority of prior works on Hub Labeling (see e.g.\  \cite{DBLP:conf/icalp/AbrahamDFGW11, DBLP:conf/mfcs/BabenkoGKSW15, DBLP:conf/esa/AbrahamDGW12}) and is very natural, in our opinion, since in real life all edge lengths are somewhat random, and, therefore, any two paths between two vertices $u$ and $v$ have different lengths. For such graphs, we design an approximation algorithm with approximation guarantee $O(\log D)$, where $D$ is the shortest-path diameter of the graph (which equals the maximum hop length of a shortest path; see Section~\ref{Prelim-Def} for the definition); the algorithm works for every fixed $p\in [1,\infty)$ (the constant in the $O$-notation depends on $p$). In particular, this algorithm gives an $O(\log\log n)$  factor approximation for graphs of diameter $\mathop{\mathrm{polylog}} n$, while previously known algorithms give only an $O(\log n)$ approximation. Our algorithm crucially relies on the fact that the input graph has unique shortest paths; in fact, our lower bounds of $\Omega(\log n)$ on the approximation ratio apply to graphs of constant diameter (with non-unique shortest paths). We also extensively study HL on trees. Somewhat surprisingly, the problem is not at all trivial on trees. In particular, the standard LP relaxation for the problem is not integral. In~\cite{DBLP:conf/soda/AngelidakisMO17} we presented the following results for trees.
\begin{enumerate}
\item Design a polynomial-time approximation scheme (PTAS) for \HLp for every $p\in[1,\infty]$.
\item Design an exact quasi-polynomial time algorithm for \HLp for every $p\in[1,\infty]$, with running time $n^{O(\log^2 n)}$.
\item Analyze a simple combinatorial heuristic for trees, proposed by Peleg in 2000, and prove that it gives a 2-approximation for \HL1
(we also show that this heuristic does not work well for \HLp when $p$ is large).
\end{enumerate}

After the publication of our work~\cite{DBLP:conf/soda/AngelidakisMO17}, Gawrychowski et al.~\cite{GawKMW17} observed that an algorithm of Onak and Parys~\cite{DBLP:conf/focs/OnakP06}, combined with a structural result of ours, shows that \HL\infty can be solved exactly on trees in polynomial time. Their main observation is that the problem of computing an optimal hub labeling on trees can be cast as a problem of ``binary search" in trees; this implies that the algorithm of Onak and Parys~\cite{DBLP:conf/focs/OnakP06} solves \HL\infty optimally, and moreover, the work of Jacob et al.~\cite{DBLP:conf/icalp/JacobsCLM10} can be adapted in order to obtain a polynomial-time algorithm for \HLp on trees for fixed $p \geq 1$ and for $p \in[\varepsilon \log n, \infty]$ (for any fixed $\varepsilon > 0)$. Since we believe that our original DP approach might still be of interest, we will present it, and then we will formally state and analyze the algorithm of~\cite{DBLP:conf/icalp/JacobsCLM10} and how it can be used to solve HL on trees.

\paragraph{Organization of material.}
In Section~\ref{sec:logn-approx} we start with a simple rounding scheme that gives a relaxation-based $O(\log n)$-approximation algorithm for \HLp for every $p \in [1, \infty]$, thus matching the guarantees of the known combinatorial algorithms. In Section~\ref{sec:bounded} we present an $O(\log D)$ approximation algorithm for graphs with unique shortest paths; we first present the (slightly simpler) algorithm for \HL1, and then the algorithm for \HLp for any fixed $p \geq 1$. Then, in Section~\ref{Hardness}, we prove an $\Omega(\log n)$-hardness for \HL1 and \HL\infty by constructing a reduction from Set Cover. As mentioned, the result easily extends to \HLp for $p = \Omega(\log n)$ on graphs with $n$ vertices. Chapter~\ref{chapter:hl} concludes with a brief section that explains how our results extend to the case of directed graphs (see Section~\ref{sec:appendix_directed}). Finally, in Chapter~\ref{chapter:trees} we present several algorithms for HL on trees, and also discuss the equivalence of HL on trees with the problem of searching for a node in a tree.

\section{Preliminaries}\label{Prelim}
\subsection{Definitions}\label{Prelim-Def}
Throughout the rest of this chapter, we always assume (unless stated otherwise) that we have an undirected graph $G=(V,E)$ with positive edge lengths $l(e) > 0$, $e \in E$. We denote the number of vertices as $n = |V|$. We will say that a graph $G$ has unique shortest paths if there is a unique shortest path between every pair of vertices. We note that if the lengths of the edges are obtained by measurements, which are naturally affected by noise, the graph will satisfy the unique shortest path property with probability 1.

One parameter that our algorithms' performance will depend on is the shortest path diameter $D$ of a graph $G$, which is defined as the maximum hop length of a shortest path in $G$ (i.e.~the minimum number $D$ such that every shortest path contains at most $D$ edges). Note that $D$ is upper bounded by the aspect ratio $\rho$ of the graph:
\begin{equation*}
    D\leq \rho\equiv \frac{\max_{u,v\in V} d(u,v)}{\min_{(u,v)\in E} l(u,v)}.
\end{equation*}
Here, $d(u,v)$ is the shortest path distance in $G$ w.r.t.\ edge lengths $l(e)$. In particular, if all edges in $G$ have length at least $1$, then $D\leq \mathrm{diam}(G)$, where $\mathrm{diam}(G) = \max_{u,v\in V} d(u,v)$.

We will use the following observation about hub labelings: the covering property for the pair $(u,u)$ (technically) requires that $u\in H_u$, and from now on, we will always assume that $u \in H_u$, for every $u\in V$.

\subsection{Linear/Convex programming relaxations for HL}
In this section, we introduce a natural LP formulation for \HL1. Let $I$ be the set of all (unordered) pairs of vertices, including pairs $(u,u)$, which we also denote as $\{u,u\}$, $u \in V$. We use indicator variables $x_{uv}$, for all $(u,v) \in V \times V$, that represent whether $v \in H_u$ or not. Let $S_{uv} (\equiv S_{vu})$ be the set of all vertices that appear in any of the (possibly many) shortest paths between $u$ and $v$ (including the endpoints $u$ and $v$). We also define $S_{uu} = \{u\}$. Note that, although the number of shortest paths between $u$ and $v$ might, in general, be exponential in $n$, the set $S_{uv}$ can always be computed in polynomial time. In case there is a unique shortest path between $u$ and $v$, we use both $S_{uv}$ and $P_{uv}$ to denote the vertices of that unique shortest path. One way of expressing the covering property as a constraint is ``$\sum_{w \in S_{uv}} \min\{x_{uw}, x_{vw}\} \geq 1$, for all $\{u,v\} \in I$". The resulting LP relaxation is given in Figure~\ref{fig:lp1}.

\begin{figure}[ht]
\noindent($\mathbf{LP_1}$)
\begin{align*}
    \min: &\quad \sum_{u \in V} \sum_{v \in V} x_{uv}\\
    \text{s.t.:} &\quad \sum_{w \in S_{uv}} \min\{x_{uw}, x_{vw}\} \geq 1,  && \forall \{u, v\} \in I, \\
          &\quad x_{uv} \geq 0,  && \forall (u,v) \in V \times V.
\end{align*}
\caption{The LP relaxation for \HL1.}
\label{fig:lp1}
\end{figure}

We note that the constraint ``$\sum_{w \in S_{uv}} \min\{x_{uw}, x_{vw}\} \geq 1$" can be equivalently rewritten as follows: $\sum_{w \in S_{uv}} y_{uvw} \geq 1$, and for all $w \in S_{uv}$, $x_{uw} \geq y_{uvw}$ and $x_{vw}\geq y_{uvw}$, where we introduce variables $y_{uvw} \geq 0$ for every pair $\{u,v\} \in I$ and every $w \in S_{uv}$. Observe that these constraints are linear, and moreover, the total number of variables and constraints remains polynomial in $n$. Thus, an optimal solution can always be found efficiently.

One indication that the above LP is indeed an appropriate relaxation for HL is that we can reproduce the result of \cite{DBLP:journals/siamcomp/CohenHKZ03} and get an $O(\log n)$-approximation algorithm for \HL1 by using a very simple rounding scheme. But, we will use the above LP in more refined ways, mainly in conjunction with the notion of pre-hubs, which we introduce later on.

We also generalize the above LP to a convex relaxation for \HLp, for any $p \in [1, \infty]$. The only difference with the above relaxation is that we use a convex objective function and not a linear one. More concretely, the convex program for \HLp, for any $p \in [1, \infty)$ is given in Figure~\ref{fig:cpp}. In the case of $p = \infty$, we end up with an LP, whose objective is simply ``$\min: t$", and there are $n$ more constraints of the form ``$t \geq \sum_{v \in V} x_{uv}$", for each $u \in V$. To make our presentation more uniform, we will always refer to the convex relaxation of Figure~\ref{fig:cpp}, even when $p = \infty$.
\begin{figure}[ht]
\noindent($\mathbf{CP_p}$)
\begin{align*}
    \min: &\quad \left(\sum_{u \in V} \left(\sum_{v \in V} x_{uv}\right)^p \right)^{1/p} \\
    \text{s.t.:} &\quad \sum_{w \in S_{uv}} \min\{x_{uw}, x_{vw}\} \geq 1,  && \forall \{u, v\} \in I, \\
          &\quad x_{uv} \geq 0,  && \forall (u,v) \in V \times V.
\end{align*}
\caption{The convex relaxation for \HLp.}
\label{fig:cpp}
\end{figure}

\subsection{Hierarchical hub labeling}\label{hhl-section}
We now define and discuss the notion of hierarchical hub labeling (HHL), introduced by Abraham et al.~\cite{DBLP:conf/esa/AbrahamDGW12}. The presentation in this section follows closely the one in~\cite{DBLP:conf/esa/AbrahamDGW12}.
\begin{definition}\label{def:hhl}
Consider a set system $\{H_u\}_{u\in V}$. We say that $v\preceq u$ if $v\in H_u$. Then, the set system $\{H_u\}_{u\in V}$ is a hierarchical hub labeling if it is a hub labeling, and $\preceq$ is a partial order.
\end{definition}

We will say that $v$ is higher ranked than $u$ if $v \preceq u$. Every two vertices $u$ and $v$ have a common hub $w\in H_u \cap H_v$, and thus there is a vertex $w$ such that $w\preceq u$ and $w\preceq v$. Therefore, there is the highest ranked vertex in $G$. 

We now define a special type of hierarchical hub labelings. Given a total order $\pi: [n] \to V$, a \textit{canonical} labeling is the hub labeling $H$ that is obtained as follows: $v \in H_u$ if and only if $\pi^{-1}(v) \leq \pi^{-1}(w)$ for all $w \in S_{uv}$. It is easy to see that a canonical labeling is a feasible hierarchical hub labeling. We say that a hierarchical hub labeling $H$ respects a total order $\pi$ if the implied (by $H$) partial order is consistent with $\pi$. Observe that there might be many different total orders that $H$ respects. In~\cite{DBLP:conf/esa/AbrahamDGW12}, it is proved that all total orders that $H$ respects have the same canonical labeling $H'$, and $H'$ is a subset of $H$. Therefore, $H'$ is a minimal hierarchical hub labeling that respects the partial order that $H$ implies.

From now on, all hierarchical hub labelings we consider will be canonical hub labelings. Any canonical hub labeling can be obtained by the following process~\cite{DBLP:conf/esa/AbrahamDGW12}. Start with empty sets $H_u$, choose a vertex $u_1$ and add it to each hub set $H_u$. Then, choose another vertex $u_2$. Consider all pairs $u$ and $v$ that currently do not have a common hub, such that $u_2$ lies on a shortest path between $u$ and $v$. Add $u_2$ to $H_u$ and $H_v$. Then, choose $u_3$, \dots, $u_n$, and perform the same step. We get a hierarchical hub labeling. (The hub labeling, of course, depends on the order in which we choose vertices of $G$.)

This procedure is particularly simple if the input graph is a tree. In a tree, we choose a vertex $u_1$ and add it to each hub set $H_u$. We remove $u_1$ from the tree and recursively process each connected component of $G-u_1$. No matter how we choose vertices $u_1,\dots, u_n$, we get a canonical hierarchical hub labeling; given a hierarchical hub labeling $H$, in order to get a canonical hub labeling $H'$, we need to choose the vertex $u_i$ of highest rank in $T'$ (w.r.t. to the order $\preceq$ defined by $H$) when our recursive procedure processes subinstance $T'$. A canonical hub labeling gives a recursive decomposition of the tree to subproblems of gradually smaller size.

\section{Warm-up: a relaxation-based $O(\log n)$-approximation algorithm for \HLp}\label{sec:logn-approx}

In this section, we describe and analyze a simple rounding scheme (inspired by Set Cover) for the convex relaxation for \HLp (see Figure~\ref{fig:cpp}), that gives an $O(\log n)$-approximation for \HLp, for every $p \in [1, \infty]$, and works on all graphs (even with multiple shortest paths). This matches the approximation guarantee of the combinatorial algorithms of Cohen et al.~\cite{DBLP:journals/siamcomp/CohenHKZ03} and Babenko et al.~\cite{DBLP:conf/icalp/BabenkoGGN13}. For any graph $G = (V,E)$ with $n$ vertices, the rounding scheme is the following (see Algorithm~\ref{alg:logn-rounding}).
\begin{algorithm}[h]
\begin{enumerate}
    \item Solve $\mathbf{CP_p}$ and obtain an optimal solution $\{x_{uv}\}_{(u,v) \in V \times V}$.
    \item Pick independent uniformly random thresholds $r_w \in (0,1)$, for each $w \in V$, and set $t_w = \frac{r_w}{3 \cdot \ln n}$.
    \item Set $H_u = \{v \in V: x_{uv} \geq t_v\}$, for every $u \in V$.
    \item Return $\{H_u\}_{u \in V}$.
\end{enumerate}
\caption{A relaxation-based $O(\log n)$-approximation algorithm for \HLp on general graphs}
\label{alg:logn-rounding}
\end{algorithm}

\begin{theorem}
For every $p \in [1, \infty]$, Algorithm~\ref{alg:logn-rounding} is an $O(\log n)$-approximation algorithm for \HLp that succeeds with high probability.
\end{theorem}
\begin{proof}
First, it is easy to see that for each $u \in V$, we can write $|H_u| = \sum_{v \in V} Y_{uv}$, where $Y_{uv} = 1$ if $x_{uv} \geq t_v$, and 0 otherwise. We have $\E[Y_{uv}] = \Pr[v \in H_u] = 3 \ln n \cdot x_{uv}$, and so, by linearity of expectation, we get $\E[|H_u|] = \sum_{v \in V} \E[Y_{uv}] = 3\ln n \cdot \sum_{v \in V} x_{uv}$. We now observe that for each $u \in V$, the variables $\{Y_{uv}\}_{v \in V}$ are independent. Thus, we can use the standard Chernoff bound, which, for any $\delta > 0$, gives
\begin{equation*}
    \Pr \left[|H_u| \geq (1 + \delta) \cdot  \E[|H_u|] \right] \leq \left(\frac{e^\delta}{(1 + \delta)^{1 + \delta}}\right)^{\E[|H_u|]}.
\end{equation*}
We set $\delta = 2$ and get $\Pr \left[|H_u| \geq 3 \cdot  \E[|H_u|] \right] \leq e^{- \E[|H_u|]} \leq 1 / n^3$ (where the last inequality holds since $x_{uu} = 1$ and thus $\sum_{v \in V} x_{uv} \geq 1$). Taking a union bound, we get that with probability at least $1 - 1/n^2$, for all $u \in V$, $|H_u| \leq 9 \ln n \cdot \sum_{v \in V} x_{uv}$. We conclude that with probability at least $1 - 1/n^2$,
\begin{equation*}
    \left(\sum_{u \in V} |H_u|^p \right)^{1/p} \leq \left( (9 \ln n)^p  \sum_{u \in V} \left( \sum_{v \in V} x_{uv} \right)^p \right)^{1/p} = 9 \ln n \cdot OPT_{CP},
\end{equation*}
where $OPT_{CP} = \left(\sum_{u \in V} \left(\sum_{v \in V} x_{uv} \right)^p \right)^{1/p}$ is the optimal value of the convex program.

We will now prove that the sets $\{H_u\}_{u \in V}$ are indeed a feasible hub labeling with high probability. It is easy to verify that we always get $u \in H_u$. So, let $u \neq v$. We have
\begin{align*}
    \Pr[H_u \cap H_v \cap S_{uv} = \emptyset] &= \prod_{w \in S_{uv}} \Pr[t_w > \min\{x_{uw}, x_{vw}\}] = \prod_{w \in S_{uv}} \left(1 - 3\ln n \cdot \min\{x_{uw}, x_{vw}\} \right)\\
                                 &\leq \prod_{w \in S_{uv}} e^{-3\ln n \cdot \min\{x_{uw}, x_{vw}\}} = e^{-3\ln n \cdot \sum_{w \in S_{uv}} \min\{x_{uw}, x_{vw}\}}\\
                                 &\leq e^{-3 \ln n} = 1 / n^3.
\end{align*}
Taking a union bound over all $\binom{n}{2}$ pairs of vertices, we get that the probability that the algorithm does not return a feasible hub labeling is at most $1/n$.  Thus, we conclude that the algorithm returns a feasible solution of value at most $9 \ln n \cdot OPT_{CP}$ with probability at least $1 - 2/n$.
\end{proof}

\section{Pre-hub labeling}
We now introduce the notion of a pre-hub labeling that we will use in designing algorithms for HL. From now on, we will only consider graphs with unique shortest paths.
\begin{definition}[Pre-hub labeling]\label{def:prehub}
Consider a graph $G = (V,E)$ and a length function $l: E \to \mathbb{R}^+$; assume that all shortest paths are unique. A family of sets $\{\widehat{H}_u\}_{u \in V}$, with $\widehat{H}_u \subseteq V$, is called a pre-hub labeling, if for every pair $\{u,v\}$, there exist $u' \in \widehat{H}_u \cap P_{uv}$ and $v' \in \widehat{H}_v \cap P_{uv}$ such that $u' \in P_{v'v}$; that is, vertices $u$, $v$, $u'$, and $v'$ appear in the following order along $P_{uv}$:
$u, v', u', v$ (possibly, some of the adjacent, with respect to this order, vertices coincide).
\end{definition}

\begin{figure}[h]
\begin{center}
\scalebox{1}{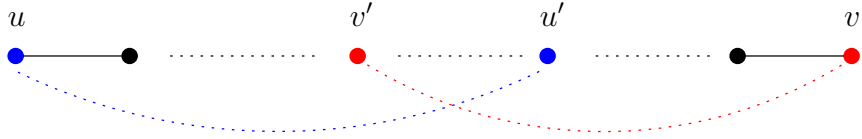}
\caption{The shortest path between $u$ and $v$ and a valid pre-hub labeling for the pair $\{u, v\}$.}
\end{center}
\end{figure}

Observe that any feasible HL is a valid pre-hub labeling. We now show how to find a pre-hub labeling given a feasible LP solution.
\begin{lemma}\label{lem:prehubs}
Consider a graph $G = (V,E)$ and a length function $l: E \to \mathbb{R}^+$; assume that all shortest paths are unique. Let $\{x_{uv}\}_{(u,v) \in V \times V}$ be a feasible solution to $\mathbf{LP_1}$ (see Figure~\ref{fig:lp1}). Then, there exists a pre-hub labeling $\{\widehat{H}_u\}_{u\in V}$ such that $|\widehat{H}_u| \leq 2\sum_{v\in V} x_{uv}$. In particular, if $\{x_{uv}\}$ is an optimal LP solution and $OPT$ is the $\ell_1$-cost of the optimal hub labeling (for \HL1), then $\sum_{u\in V}|\widehat{H}_u| \leq 2 \, OPT$. Furthermore, the pre-hub labeling $\{\widehat{H}_u\}_{u\in V}$ can be constructed efficiently given the LP solution $\{x_{uv}\}$.
\end{lemma}
\begin{proof}
Let us fix a vertex $u \in V$. We build the breadth-first search tree $T_u$ (w.r.t.~edge lengths; i.e.~the shortest path tree)  from $u$; tree $T_u$ is rooted at $u$ and contains those edges $e \in E$ that appear on a shortest path between $u$ and some vertex $v \in V$. Observe that $T_u$ is indeed a tree and is uniquely defined, since we have assumed that shortest paths in $G$ are unique. For every vertex $v$, let $T'_{uv}$ be the subtree of $T_u$ rooted at vertex $v$. Given a feasible LP solution $\{x_{uv}\}$, we define the weight of $T'_{uv}$ to be $\mathcal{W}(T'_{uv}) = \sum_{w \in T'_{uv}} x_{uw}$.

We now use the following procedure to construct set $\widehat{H}_{u}$. We process the tree $T_u$ bottom up (i.e.\ we process a vertex $v$ after we have processed all other vertices in the subtree rooted at $v$), and whenever we detect a subtree $T'_{uv}$ of $T_u$ such that $\mathcal{W}(T'_{uv}) \geq 1/2$, we add vertex $v$ to the set $\widehat{H}_u$. We then set $x_{uw} = 0$ for all $w \in T'_{uv}$, and continue (with the updated $x_{uw}$ values) until we reach the root $u$ of $T_u$. Observe that every time we add one vertex to $\widehat{H}_u$, we decrease the value of $\sum_{v \in V} x_{uv}$ by at least $1/2$. Therefore, $|\widehat{H}_u| \leq 2 \cdot \sum_{v \in V} x_{uv}$. We will now show that sets $\{\widehat{H}_u\}_{u \in V}$ form a pre-hub labeling. To this end, we prove the following two claims.
\begin{claim}\label{claim:prehub-1}
Consider a vertex $u$ and two vertices $v_1, v_2$ such that $v_1\in P_{uv_2}$. If $\widehat{H}_u \cap P_{v_1 v_2} = \emptyset$, then
$\sum_{w\in P_{v_1v_2}} x_{uw} < 1/2$.
\end{claim}
\begin{proof}
Consider the execution of the algorithm that defined $\widehat{H}_u$. Consider the moment $M$ when we processed vertex $v_1$. Since we did not add $v_1$ to $\widehat{H}_u$, we had $\mathcal{W}(T'_{uv_1}) < 1/2$. In particular, since $P_{v_1v_2}$ lies in $T'_{uv_1}$, we have $\sum_{w\in P_{v_1v_2}} x_{uw}'<1/2$, where $x_{uw}'$ is the value of $x_{uw}$ at the moment $M$. Since none of the vertices on the path $P_{v_1v_2}$ were added to $\widehat{H}_u$, none of the variables $x_{uw}$ for $w\in P_{v_1v_2}$ had been set to $0$. Therefore, $x_{uw}' = x_{uw}$ (where $x_{uw}$ is the initial value of the variable) for $w\in P_{v_1v_2}$. We conclude that $\sum_{w\in P_{v_1v_2}} x_{uw}<1/2$, as required.
\end{proof}

\begin{claim}\label{claim:prehub-2}
For any pair $\{u,v\}$, let $u' \in \widehat{H}_u \cap P_{uv}$ be the vertex closest to $v$ among all vertices in $\widehat{H}_u \cap P_{uv}$ and $v' \in \widehat{H}_v \cap P_{uv}$ be the vertex closest to $u$ among all vertices in $\widehat{H}_v \cap P_{uv}$. Then $u' \in P_{v'v}$. (Note that $\widehat{H}_u \cap P_{uv}\neq \emptyset$, since we always have $x_{uu} = 1$ and hence $u\in \widehat{H}_u \cap P_{uv}$; similarly, $\widehat{H}_v \cap P_{uv}\neq \emptyset$.)
\end{claim}
\begin{proof}
Let us assume that this is not the case; that is, $u' \notin P_{v'v}$. Then $v'\neq u$ and $u'\neq v$ (otherwise, we would trivially have $u' \in P_{v'v}$). Let $u''$ be the first vertex after $u'$ on the path $P_{u'v}$, and $v''$ be the first vertex after $v'$ on the path $P_{v'u}$. Since $u' \notin P_{v'v}$, every vertex of $P_{uv}$ lies either on $P_{v''u}$ or $P_{u''v}$, or both (i.e. $P_{v''u} \cup P_{u''v} = P_{uv}$).

By our choice of $u'$, there are no pre-hubs for $u$ on $P_{u''v}$. By Claim~\ref{claim:prehub-1}, $\sum_{w\in P_{u''v}} x_{uw} <1/2$. Similarly, $\sum_{w\in P_{v''u}} x_{vw} <1/2$. Thus,
\begin{equation*}
    1 > \sum_{w \in P_{uv''}} x_{vw} + \sum_{w \in P_{u''v}} x_{uw} \geq \sum_{w \in P_{uv}} \min\{x_{uw}, x_{vw}\}.
\end{equation*}
We get a contradiction since $\{x_{uv}\}$ is a feasible LP solution.
\end{proof}
Claim~\ref{claim:prehub-2} shows that $\{\widehat{H}_u\}$ is a valid pre-hub labeling.
\end{proof}

\section{Hub labeling on graphs with unique shortest paths}\label{sec:bounded}
In this section, we present an $O(\log D)$-approximation algorithm for \HLp on graphs with unique shortest paths, where $D$ is the shortest path diameter of the graph. The algorithm works for every fixed $p \geq 1$ (the hidden constant in the approximation factor $O(\log D)$  depends on $p$). We will first present the (slightly simpler) algorithm for \HL1, and then extend the algorithm and make it work for \HLp, for arbitrary fixed $p \geq 1$.

\subsection{An $O(\log D)$-approximation algorithm for \HL1}

Consider Algorithm~\ref{Pre-Hubs_Algorithm}. The algorithm solves the LP relaxation (see Figure~\ref{fig:lp1}) and computes a pre-hub labeling $\{\widehat{H}_u\}_{u \in V}$ as described in Lemma~\ref{lem:prehubs}. Then it chooses a random permutation $\pi$ of $V$ and goes over all vertices one-by-one in the order specified by $\pi$: $\pi_1$, $\pi_2$,\dots, $\pi_n$. It adds $\pi_i$ to $H_u$ if there is a pre-hub $u'\in \widehat{H}_u$ such that the following conditions hold: $\pi_i$ lies on the path $P_{uu'}$, there are no pre-hubs for $u$ between $\pi_i$ and $u'$ (other than $u'$), and currently there are no hubs for $u$  between  $\pi_i$ and $u'$.

\begin{algorithm}[h]
\begin{enumerate}
    \item Solve $\mathbf{LP_1}$ and get an optimal solution $\{x_{uv}\}_{(u,v) \in V \times V}$.
    \item Obtain a set of pre-hubs $\{\widehat{H}_u\}_{u \in V}$ from $x$ as described in Lemma~\ref{lem:prehubs}.
    \item Generate a random permutation $\pi : [n] \to V$ of the vertices.
    \item Set $H_u = \emptyset$, for every $u \in V$.
    \item \textbf{for} $i= 1$ \textbf{to} $n$ \textbf{do}:\\
           \hspace*{16pt}  \textbf{for} every $u \in V$ \textbf{do}:\\
           \hspace*{36pt}     \textbf{for} every $u' \in \widehat{H}_u$ such that $\pi_i \in P_{uu'}$ and $P_{\pi_i u'} \cap \widehat{H}_u = \{u'\}$ \textbf{do}:\\
           \hspace*{56pt}         \textbf{if} $P_{\pi_i u'} \cap H_u = \emptyset$ \textbf{then} $H_u := H_u \cup \{\pi_i\}$.
    \item Return $\{H_u\}_{u \in V}$.
\end{enumerate}
\caption{An $O(\log D)$-approximation algorithm for \HL1 on graphs with unique shortest paths}
\label{Pre-Hubs_Algorithm}
\end{algorithm}

\begin{theorem}
Algorithm~\ref{Pre-Hubs_Algorithm} always returns a feasible hub labeling $H$. The cost of the hub labeling  is $\E[\sum_u |H_u|] = O(\log D) \cdot OPT_{LP_1}$ in expectation, where $OPT_{LP_1}$ is the optimal value of $\mathbf{LP_1}$.
\end{theorem}
\begin{remark}\label{rem:derandomize}
Algorithm~\ref{Pre-Hubs_Algorithm} can be easily derandomized using the method of conditional expectations: instead of choosing a random permutation $\pi$, we first choose $\pi_1\in V$, then $\pi_2\in V\setminus\{\pi_1\}$ and so on; each time we choose $\pi_i \in V\setminus\{\pi_1,\dots,\pi_{i-1}\}$ so as to minimize the conditional expectation $\E \left[\sum_u |H_u| \:|\, \pi_1,\dots, \pi_i \right]$.
\end{remark}
\begin{proof}
We first show that the algorithm always finds a feasible hub labeling. Consider a pair of vertices $u$ and $v$. We need to show that they have a common hub on $P_{uv}$. The statement is true if $u=v$ since $u\in \widehat{H}_u$ and thus $u\in H_u$. So, we assume that $u \neq v$. Consider the path $P_{uv}$. Because of the pre-hub property, there exist $u' \in \widehat{H}_u$ and $v' \in \widehat{H}_v$ such that $u' \in P_{v'v}$. In fact, there may be several possible ways to choose such $u'$ and $v'$. We choose $u'$ and $v'$ so that $\widehat{H}_u \cap (P_{u'v'} \setminus \{u',v'\}) = \widehat{H}_v \cap (P_{u'v'} \setminus \{u',v'\}) = \emptyset$ (for instance, choose the closest pair of  $u'$ and $v'$ among all possible pairs). Consider the first iteration $i$ of the algorithm such that $\pi_i \in P_{u'v'}$. We claim that the algorithm adds $\pi_i$ to both $H_u$ and $H_v$. Indeed, we have: (i) $\pi_i$ lies on $P_{v'u'}\subset P_{uu'}$, (ii) there are no pre-hubs of $u$ on $P_{v'u'}\supset P_{\pi_iu'}$ other than $u'$,  (iii) $\pi_i$ is the first vertex we process on the path $P_{u'v'}$, thus currently there are no hubs on $P_{u'v'}$. Therefore, the algorithm adds $\pi_i$ to $H_u$. Similarly, the algorithm adds $\pi_i$ to $H_v$.

Now we upper bound the expected cost of the solution. We will charge every hub that we add to $H_u$ to a pre-hub in $\widehat{H}_u$; namely, when we add $\pi_i$ to $H_u$ (see line 5 of Algorithm~\ref{Pre-Hubs_Algorithm}), we charge it to pre-hub $u'$. For every vertex $u$, we have $|\widehat{H}_u| \leq 2\sum_w x_{uw}$. We are going to show that every $u' \in \widehat{H}_u$ is charged at most $O(\log D)$ times in expectation. Therefore, the expected number of hubs in $H_u$ is at most $O(2\sum_w x_{uw} \cdot \log D)$.

Consider a vertex $u$ and a pre-hub $u'\in \widehat{H}_u$ ($u'\neq u$). Let $u'' \in \widehat{H}_u$ be the closest pre-hub to $u'$ on the path $P_{u'u}$. Observe that all hubs charged to $u'$ lie on the path $P_{u''u'} \setminus \{u''\}$. Let $k = |P_{u''u'} \setminus \{u''\}|$. Note that $k \leq D$. Consider the order $\sigma : [k] \to P_{u''u'} \setminus \{u''\}$ in which the vertices of $P_{u''u'} \setminus \{u''\}$ were processed by the algorithm ($\sigma$ is a random permutation). Note that $\sigma_i$ charges $u'$ if and only if $\sigma_i$ is closer to $u'$ than $\sigma_{1},\dots,\sigma_{i-1}$. The probability of this event is $1/i$. We get that the number of hubs charged to $u'$ is $\sum_{i=1}^k \frac{1}{i} = \log k + O(1)$, in expectation.  Hence, $\E \left[\sum_{u\in V} |H_u| \right] \leq 2 \left(\log D + O(1) \right) \cdot OPT_{LP_1}$.
\end{proof}

\subsection{An $O_p(\log D)$-approximation algorithm for \texorpdfstring{\HLp}{HLp}}\label{sec:lp-norm-algorithm}

In this section, we analyze Algorithm~\ref{Pre-Hubs_Algorithm}, assuming that we solve the convex program of Figure~\ref{fig:cpp}. To analyze the performance of Algorithm \ref{Pre-Hubs_Algorithm} in this case, we need the following theorem by Berend and Tassa \cite{Berend_improvedbounds}.
\begin{theorem}[Theorem 2.4, \cite{Berend_improvedbounds}]\label{moment_theorem}
Let $X_1, ..., X_t$ be a sequence of independent random variables for which $\Prob[0 \leq X_i \leq 1] = 1$, and let $X = \sum_{i = 1}^t X_i$. Then, for all $p \geq 1$,
\begin{equation*}
    \left(\E[X^p]\right)^{1 / p} \leq 0.942 \cdot \frac{p}{\ln(p + 1)} \cdot \max\{\E[X]^{1 / p}, \E[X]\}.
\end{equation*}
\end{theorem}

In order to simplify our analysis, we slightly modify Algorithm~\ref{Pre-Hubs_Algorithm} and get Algorithm~\ref{alg:HLp-algorithm}.

\begin{algorithm}[h]
\begin{enumerate}
    \item Solve $\mathbf{CP_p}$ and get an optimal solution $\{x_{uv}\}_{(u,v) \in V \times V}$.
    \item Obtain a set of pre-hubs $\{\widehat{H}_u\}_{u \in V}$ from $x$ as described in Lemma~\ref{lem:prehubs}.
    \item For each $u \in V$, let $J_u = \bigcup_{u' \in \widehat{H}_u} P_{uu'}$ be a tree rooted at $u$, and let $F_u \subset V(J_u)$ be the set of vertices of $J_u$ whose degree (in $J_u$) is at least 3. Set $\widehat{H}_u' := \widehat{H}_u \cup F_u$.
    \item Generate a random permutation $\pi : [n] \to V$ of the vertices.
    \item Set $H_u = \emptyset$, for every $u \in V$.
    \item \textbf{for} $i= 1$ \textbf{to} $n$ \textbf{do}:\\
           \hspace*{16pt}  \textbf{for} every $u \in V$ \textbf{do}:\\
           \hspace*{36pt}     \textbf{for} every $u' \in \widehat{H}_u'$ such that $\pi_i \in P_{uu'}$ and $P_{\pi_i u'} \cap \widehat{H}_u' = \{u'\}$ \textbf{do}:\\
           \hspace*{56pt}         \textbf{if} $P_{\pi_i u'} \cap H_u = \emptyset$ \textbf{then} $H_u := H_u \cup \{\pi_i\}$.
    \item Return $\{H_u\}_{u \in V}$.
\end{enumerate}
\caption{An $O_p(\log D)$-approximation algorithm for \HLp on graphs with unique shortest paths}
\label{alg:HLp-algorithm}
\end{algorithm}

\begin{theorem}
For any $p \geq 1$, Algorithm~\ref{alg:HLp-algorithm} is an $O \left(\frac{p}{\ln(p + 1)} \cdot \log D \right)$-approximation algorithm for \HLp.
\end{theorem}
\begin{proof}

First, it is easy to see that, since all leaves of $J_u$ are pre-hubs of the set $\widehat{H}_u$, we have $|F_u| \leq |\widehat{H}_u|$, and so $|\widehat{H}_u'| \leq 2 \cdot |\widehat{H}_u|$. 

Let $\mathcal{P}_u$ be the collection of subpaths of $J_u$ defined as follows: $P$ belongs to $\mathcal{P}_u$ if $P$ is a path between consecutive pre-hubs $u''$ and $u'$ of $\widehat{H}_u'$, with $u''$ being an ancestor of $u'$ in $J_u$, and no other pre-hub $u''' \in \widehat{H}_u'$ appears in $P$. For convenience, we exclude the endpoint $u''$ that is closer to $u$: $P = P_{u''u'} - u''$. Note that any such path $P$ is uniquely defined by the pre-hub $u'$ of $u$, and so we will denote $P$ as $P_{(uu')}$. The modification we made in the algorithm allows us now to observe that $P \cap P' = \emptyset$, for $P, P' \in \mathcal{P}_u$, $P \neq P'$. 

Let $ALG'$ be the cost of the solution $\{H_u\}_{u \in V}$ that the modified algorithm (i.e. Algorithm~\ref{alg:HLp-algorithm}) returns. We have $\E[ALG'] = \E \left[ \left(\sum_{u \in V} |H_u|^p \right)^{1 / p} \right] \leq \left(\sum \E[|H_u|^p] \right)^{1 / p}$ (by Jensen's inequality).

We can write $|H_u| \leq \sum_{v \in \widehat{H}_u'} X_v^{u}$, where $X_v^u$ is the random variable indicating how many vertices are added to $H_u$ ``because of" the pre-hub $v \in \widehat{H}_u'$ (see line 6 of the algorithm). Observe that we can write $X_v^u$ as follows: $X_v^u = \sum_{w \in P_{(uv)}} Y_w^{uv}$, with $Y_w^{uv}$ being 1 if $w$ is added in $H_u$, and 0 otherwise. The modification that we made in the algorithm implies, as already observed, that any variable $Y_w^{uv}$, $w \in P_{(uv)}$, is independent from $Y_{w'}^{uv'}$, $w' \in P_{(uv')}$, for $v \neq v'$, as the corresponding paths $P_{(uv)}$ and $P_{(uv')}$ are disjoint.

Let $u \in \widehat{H}_u'$, and let $\pi_{uv}: [|P_{(uv)}|] \to P_{(uv)}$ be the induced permutation when we restrict $\pi$ (see line 4 of the algorithm) to the vertices of $P_{(uv)}$. We can then write $\sum_{w \in P_{(uv)}} Y_w^{uv} = \sum_{i = 1}^l Z_i^{uv}$, $l = |P_{(uv)}|$, where $Z_i^{uv}$ is 1 if the $i^{\textrm{th}}$ vertex considered by the algorithm that belongs to $P_{(uv)}$ (i.e. the $i^{\textrm{th}}$ vertex of permutation $\pi_{uv}$) is added to $H_u$ and 0 otherwise. It is easy to see that $\Pr[Z_i^{uv} = 1] = 1/ i$. We now need one last observation. We have $\Pr[Z_i^{uv} = 1 \;|\; Z_1^{uv}, ..., Z_{i - 1}^{uv}] = 1/i$. To see this, note that the variables $Z_i^{uv}$ do not reveal which particular vertex is picked from the permutation at each step, but only the relative order of the current draw (i.e. $i^{\textrm{th}}$ random choice) with respect to the current best draw (where best here means the closest vertex to $v$ that we have seen so far, i.e. in positions $\pi_{uv}(1), ..., \pi_{uv}(i -1)$). Thus, regardless of the relative order of $\pi_{uv}(1), ..., \pi_{uv}(i -1)$, there are exactly $i$ possibilities to extend that order when the permutation picks $\pi_{uv}(i)$, each with probability $1/i$. This shows that the variables $\{Z_i^{uv}\}_i$ are independent, and thus all variables $\{Z_i^{uv}\}_{v \in \widehat{H}_v', \;i \in [|P_{(uv)}|]}$ are independent.

We can now apply Theorem \ref{moment_theorem}. This gives
\begin{equation*}
\begin{split}
    \E[|H_u|^p] \leq \E \left[ \left(\sum_{v \in \widehat{H}_u'} \sum_{i  = 1}^{|P_{(uv)}|} Z_i^{uv} \right)^p \right] \leq \left(0.942 \cdot \frac{p}{\ln(p + 1)} \right)^p \cdot \Harm_D^p \cdot |\widehat{H}_u'|^p.
\end{split}
\end{equation*}
Here, $\Harm_D = \sum_{i=1}^D \frac{1}{i} = \log D + O(1)$ is the $D$-th harmonic number. Thus,
\begin{equation*}
\begin{split}
    \E[ALG'] &\leq 0.942 \cdot \frac{p}{\ln(p + 1)} \cdot \Harm_D \cdot \left(\sum_{u \in V} |\widehat{H}_u'|^p \right)^{1 / p} \\
            &\leq 0.942  \cdot \frac{p}{\ln(p + 1)} \cdot \Harm_D \cdot \left(\sum_{u \in V} 4^p \cdot \left(\sum_{v \in V} x_{uv} \right)^p \right)^{1 / p} \\
            &\leq 3.768 \cdot \frac{p}{\ln(p + 1)} \cdot \Harm_D \cdot OPT_{CP},
\end{split}
\end{equation*}
where $OPT_{CP}$ is the optimal value of the convex relaxation.
\end{proof}

\subsection{Any ``natural" rounding scheme cannot break the $\widetilde{O}(\log n)$ barrier for $HL_1$ on graphs with unique shortest paths and diameter $D$}\label{lower_bound_rounding}

In this section, we show that any rounding scheme that may assign $v \in H_u$ only if $x_{uv} > 0$ gives $\Omega(\log n / \log \log n)$ approximation, even on graphs with shortest-path diameter $D = O(\log n)$. For that, consider the following tree $T$, which consists of a path $P = \{1, ..., k\}$ of length $k = 3t$, $t \in \mathbb{N} \setminus \{0\}$, and two stars $\mathcal{A}$ and $\mathcal{B}$, with $N = \binom{k}{2t}$ leaves each (each leaf corresponding to a subset of $[k]$ of size exactly $2t$). The center $a$ of $\mathcal{A}$ is connected to vertex ``1" of $P$ and the center $b$ of $\mathcal{B}$ is connected to vertex ``$k$" of $P$. The total number of vertices of $T$ is $n = 2N + 2 + k$, which implies that $t = \Omega(\log n / \log \log n)$.

\begin{figure}[h]
\begin{center}
\scalebox{1.1}{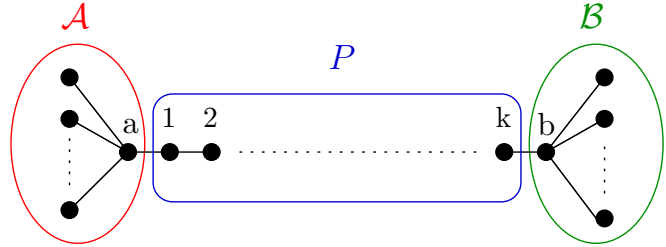}
\caption{An instance that cannot be rounded well with any ``natural" rounding scheme.}
\end{center}
\end{figure}

Consider now the following LP solution for the LP of Figure~\ref{fig:lp1} (all variables not assigned below are set to zero):
\begin{itemize}
    \item $x_{uu} = 1$, for all $u \in T$.
    \item $x_{Sa} = 1$, for all $S \in \mathcal{A}$.
    \item $x_{Wb} = 1$, for all $W \in \mathcal{B}$.
    \item $x_{Si} = 1/t$, for all $S \in \mathcal{A}$, $i \in S \subseteq P$.
    \item $x_{Wi} = 1/t$, for all $W \in \mathcal{B}$, $i \in W \subseteq P$.
    \item $x_{ab} = x_{ba} = 1$.
    \item $x_{ia} = x_{ib} = 1$, for all $i \in [k]$.
    \item $\{x_{ij}\}_{i,j \in [k]}$ is an optimal solution for $P$.
\end{itemize}
Observe that the above solution is indeed a feasible fractional solution. Its cost is at most $n + 3(|\mathcal{A}| + |\mathcal{B}|) + 2 + 2k + c \cdot k \cdot \log k = \Theta(n)$, for some constant $c$. Suppose now that we are looking for a rounding scheme that assigns $v \in H_u$ only if $x_{uv} > 0$, and let's assume that there exists a vertex $S \in \mathcal{A}$ whose resulting hub set satisfies $|H_S \cap P| < t$. We must also have $H_S \cap \mathcal{B} = \emptyset$, since $x_{Su} = 0$ for all $u \in \mathcal{B}$. This implies that there exists a $W \in \mathcal{B}$ such that $W \cap H_S = \emptyset$. Since the above fractional solution assigns non-zero values only to $x_{Wi}$ with $i \in W$ and $x_{Wb}$, this means that $x_{Wi} = 0$ for all $i \in H_S$. Thus, the resulting hub set cannot be feasible, which implies that any rounding that satisfies the aforementioned property and returns a feasible solution must satisfy $|H_S \cap P| \geq t$ for all $S \in \mathcal{A}$ (similarly, the same holds for all $W \in \mathcal{B}$). This means that the returned solution has cost $\Omega(n \cdot t) = \Omega(n \cdot \log n / \log \log n)$, and so the approximation factor must be at least $\Omega(\log n / \log \log n)$.

\section{Hardness of approximating hub labeling on general graphs}\label{Hardness}

In this section, we prove that \HL1 and \HL\infty are \NP-hard to approximate on general graphs with $n$ vertices and multiple shortest paths within a factor better than $\Omega(\log n)$, by using the $\Omega(\log n)$-hardness results for Set Cover. This implies that the current known algorithms for \HL1 and \HL\infty are optimal (up to constant factors). The result for \HL\infty also almost immediately implies the same hardness for \HLp, when $p = \Omega(\log n)$.

\subsection{$\Omega(\log n)$-hardness for \HL1}\label{sec:l1-hardness}

In this section, we show that it is \NP-hard to approximate \HL1 on general graphs with multiple shortest paths within a factor better than $\Omega(\log n)$. We will use the hardness results for Set Cover, that, through a series of works spanning more than 20 years \cite{DBLP:journals/jacm/LundY94, DBLP:journals/jacm/Feige98, DBLP:conf/stoc/RazS97, DBLP:journals/talg/AlonMS06}, culminated in the following theorem.
\begin{theorem}[Dinur \& Steurer \cite{DBLP:conf/stoc/DinurS14}]\label{Set-Cover-Hardness}
For every $\alpha > 0$,  it is \NP-hard to approximate Set Cover to within a factor $(1 - \alpha) \cdot \ln n$, where $n$ is the size of the universe.
\end{theorem}

We start with an arbitrary unweighted instance of Set Cover. Let $\mathcal{X} = \{x_1, ..., x_n\}$ be the universe and $\mathcal{S} = \{S_1, ..., S_m\}$ be the family of subsets of $\mathcal{X}$, with $m = \poly(n)$. Our goal is to pick the smallest set of indices $I \subseteq [m]$ (i.e. minimize $|I|$) such that $\bigcup_{i \in I} S_i = \mathcal{X}$.

The high-level idea of our argument is the following: we define a weighted variant of \HLp, and we show that an $\alpha$-approximation for the standard \HLp can be used to obtain an $O(\alpha)$-approximation for the weighted \HLp. We then proceed to construct a weighted instance of \HL1 such that, given an $f(n)$-approximation algorithm for the weighted \HL1, we can use it to construct a solution for the original Set Cover instance of cost $O(f(\poly(n))) \cdot OPT_{SC}$, where $OPT_{SC}$ is the cost of the optimal Set Cover solution. Formally, we prove the following theorem.

\begin{theorem}\label{HL1_hardness_theorem}
Given an arbitrary unweighted Set Cover instance ($\mathcal{X}, \mathcal{S}$), $|\mathcal{X}|=n$, $|\mathcal{S}| = m$, with optimal value $OPT_{SC}$, and an $f(n)$-approximation algorithm for weighted \HL1, there is an algorithm that returns a solution for the Set Cover instance of cost $O(f(\poly(n))) \cdot OPT_{SC}$.
\end{theorem}

Using the above theorem, if we assume that $f(n) = o(\log n)$, then we have $O(f(\poly(n))) = o(\log \mathtt{poly} (n)) = o(\log n)$, and so this would imply that we can get a $o(\log n)$-approximation algorithm for Set Cover. By Theorem \ref{Set-Cover-Hardness}, this is \NP-hard, and so we must have $f(n) = \Omega(\log n)$.

\begin{corollary}
It is \NP-hard to approximate \HL1 to within a factor $c \cdot \log n$, for some constant $c$, on general graphs with $n$ vertices (and multiple shortest paths).
\end{corollary}

Before proving Theorem \ref{HL1_hardness_theorem}, we need a few lemmas.

\begin{lemma}\label{degree_one_vertices_lemma}
Let $G = (V, E)$ and $l: E \to \mathbb{R}^+$  be an instance of \HLp, for any $p \geq 1$, and let $Z = \{u \in V: deg(u) = 1\}$ be the set of vertices of $G$ of degree 1. Suppose that $|Z| > 0$, and let $n(u)$ denote the unique neighbor of a vertex $u \in A$. Then, any feasible solution $\{H_v\}_{v \in V}$  can be converted to a solution $H'$ of at most twice the $\ell_p$-cost, with the property that $H_u' = H_{n(u)}' \cup \{u\}$ and $u \notin H_v'$, for every vertex $u \in Z$ and $v \neq u$.
\end{lemma}
\begin{proof}
Let $\{H_v\}$ be any feasible hub labeling. If the desired property already holds for every vertex of degree 1, then we are done. So let us assume that the property does not hold for some vertex $u \in Z$. Let $w = n(u)$ be its unique neighbor and let
$$B = \begin{cases}
         H_w \setminus \{u\}, & \mathrm{ if } \;\;|H_w \setminus \{u\}| \leq |H_u \setminus \{u\}|,\\
         H_u \setminus \{u\}, & \textrm{otherwise}.
   \end{cases}$$
We now set
\begin{itemize}
    \item $H_u' = B \cup \{u,w\}$.
    \item $H_w' = B \cup \{w\}$.
    \item $\forall v \in V \setminus \{u,w\}, \;H_v' = \begin{cases}
            H_v, & \mathrm{ if } \;u \notin H_v,\\
            (H_v \setminus \{u\}) \cup \{w\}, & \textrm{otherwise}.
   \end{cases}$
\end{itemize}
We first check the feasibility of $H'$. The pairs $\{u,w\}$, and $\{v,v\}$, for all $v \in V$, are clearly satisfied. Also, every pair $\{v,v'\}$ with $v,v' \notin \{u,w\}$ is satisfied, since $u \notin S_{vv'}$. Consider now a pair $\{u,v\}$, with $v \in V \setminus \{u,w\}$. If $u \in H_u \cap H_v$, we have $w \in H_u' \cap H_v'$. Otherwise, $\{u,v\}$ is covered with some vertex $z\in S_{uv} \setminus \{u\}$, and since $S_{uv} \setminus \{u\} = S_{wv}$, we have that $z \in H_u \cap H_v$ and $z \in H_w \cap H_v$.
It follows that $z\in H'_u \cap H'_v$. Now, consider a pair $\{w,v\}$, $v \in V \setminus \{u,w\}$. We have either $H_w' = (H_w \setminus \{u\}) \cup \{w\}$, which gives $H_w' \cap S_{wv} = H_w \cap S_{vw}$, or $H_w' = (H_u \setminus \{u\}) \cup \{w\}$. In the latter case, either $u \in H_v$ and so $w \in H_v'$, or $H_u \cap S_{uv} = H_u \cap S_{wv}$. It is easy to see that in all cases the covering property is satisfied.

We now argue about the cost of $H'$. We distinguish between the two possible values of $B$:
\begin{itemize}
    \setlength{\parskip}{0pt}
    \item $B = H_w \setminus \{u\}$: In this case, $|H_w'| \leq |H_w|$, since $w \in B$. If $u \in H_w$, then $|H_u'| = |H_w| \leq |H_u|$. Otherwise, it holds that $|H_w| \leq |H_u| - 1$, and so $|H_u'| = |H_w| + 1 \leq |H_u|$. For all $v \in V \setminus \{u,w\}$, it is obvious that $|H_v'| \leq |H_v|$.
    \item $B = H_u \setminus \{u\}$: If $w \in H_u$, then $|H_w'| = |H_u| - 1 < |H_w \setminus \{u\}| \leq |H_w|$, and $|H_u'| = |H_u|$. Otherwise, we must have $u \in H_w$, which means that $|H_u| < |H_w|$. Thus, $|H_w'| = |H_u| <  |H_w|$, and $|H_u'| = |H_u| + 1 \leq 2|H_u|$. Again, it is obvious that $|H_v'| \leq |H_v|$, for all $v \in V \setminus \{u,w\}$.
\end{itemize}
By the above case analysis, it is easy to see that we can apply the above argument to every vertex $u \in Z$, one by one, and in the end we will obtain a feasible hub labeling $H'$ that satisfies the desired properties, such that $|H'_u| \leq 2 \cdot |H_u|$ for every $u \in V$. Thus, for every $p \in [1, \infty]$, we have $\|H'\|_p \leq 2 \|H\|_p$.
\end{proof}

To make our construction slightly simpler, we now introduce a weighted variant of \HLp.
\begin{definition}[Weighted \HLp]
Let $G = (V,E, l)$ be an edge-weighted graph, $n = |V|$, and let $w: V \to \mathbb{R}_{>0}$ be a weight function that assigns positive weights to every vertex. Let $H$ be a feasible hub labeling for $G$. Then, its weighted $\ell_p$-cost is defined as
\begin{equation*}
    \left(\sum_{u \in V} \left(w_u \cdot |H_u|\right)^p \right)^{1/p}.
\end{equation*}
The $\ell_\infty$-cost is defined as $\max_{u \in V} (w(u) \cdot |H_u|)$.
\end{definition}

It is immediate that the weighted \HLp is a generalization of \HLp. We will now show that any $\alpha$-approximation algorithm for \HLp can give an $O(\alpha)$-approximation for the weighted \HLp when $p$ is fixed and the vertex weights are polynomially bounded.
\begin{lemma}
Let $G = (V,E,l)$ and suppose that we have an $\alpha$-approximation algorithm for the unweighted \HLp. Let $w: V \to \{1, ..., \mathrm{poly}(|V|\}$. Then, there exists a $(2e \cdot \alpha)$-approximation algorithm for the weighted \HLp, for any fixed $p \geq 1$.
\end{lemma}
\begin{proof}
Let $n = |V|$. We create a new graph as follows. We first create a copy of $G$, and for each vertex $u \in V$, with $w(u) > 1$, we attach $w(u)^p - 1$ vertices $b(u,1), ..., b(u, w(u)^p - 1)$ to it, each having degree exactly 1 with $u$ being its unique neighbor. The length of the new edges added is set to 1 (although that length is not important). Let $G' = (V', E', l')$ be the resulting graph. Since we have assumed that $w(u) \leq \mathrm{poly}(n)$, for every $u \in V$, and $p$ is fixed (and not part of the input), the size of $G'$ is polynomial in $n$.

We denote the unweighted $\ell_p$-cost of a hub labeling $H$ for a graph $G$ as $cost_p(G, H)$ and the weighted cost as $cost_p(G, H, w)$. Let $H$ be an optimal solution for $G$ for the weighted \HLp and $H'$ be an optimal solution for $G'$ for the unweighted \HLp. We will need the simple fact that for every $x \geq 1$ we have $\left(\frac{x + 1}{x} \right)^p \leq e^p$.

First, we define $H_u'' = H_u$ for every $u \in V$ and $H_{b(u,i)}'' = H_u \cup \{b(u,i)\}$ for every $u \in V$ and $i \in [w(u)^p - 1]$. It is easy to see that $H''$ is a feasible hub labeling for $G'$. We have
\begin{align*}
    cost_p(G', H'')^p &= \sum_{u \in V} |H_u|^p + \sum_{u \in V}\sum_{i = 1}^{w(u)^p - 1} (|H_u| + 1)^p \leq \sum_{u \in V} |H_u|^p + \sum_{u \in V}\sum_{i = 1}^{w(u)^p - 1} e^p |H_u|^p\\
                      &= \sum_{u \in V} |H_u|^p + e^p \sum_{u \in V} (w(u)^p - 1) |H_u|^p\\
                      &\leq e^p \sum_{u \in V} w(u)^p |H_u|^p = e^p \cdot cost_p(G, H, w)^p.
\end{align*}
Thus, $cost_p(G', H') \leq cost_p(G', H'') \leq e \cdot cost_p(G, H, w)$.

Suppose now that we have an $\alpha$-approximation algorithm for the unweighted \HLp. Then, given an instance $G = (V,E, l)$ with polynomially bounded integer weights $w$, we construct the graph $G'$ and run the algorithm on this graph, thus obtaining a hub labeling $H'''$ that is an $\alpha$-approximate solution for the unweighted \HLp for $G'$. Note that since $p$ is fixed and $w$ is polynomially bounded, the resulting graph $G'$ has polynomially many vertices. Using Lemma~\ref{degree_one_vertices_lemma}, we get a solution $\widetilde{H}$ such that $\widetilde{H}_{b(u,i)} = \widetilde{H}_u \cup \{b(u,i)\}$ for every $u \in V$ and $i \in [w(u)^p - 1]$, such that $cost_p(G', \widetilde{H}) \leq 2 \cdot cost_p(G', H''')$.  We observe that $\{\widetilde{H}_u\}_{u \in V}$ is a feasible hub labeling for $G$, and we have
\begin{align*}
    cost_p(G, \widetilde{H}, w)^p &= \sum_{u \in V} w(u) |\widetilde{H}_u|^p = \sum_{u \in V} |\widetilde{H}_u|^p + \sum_{u \in V} \sum_{i = 1}^{w(u)^p - 1} |\widetilde{H}_u|^p \\
                                  &\leq \sum_{u \in V} |\widetilde{H}_u|^p + \sum_{u \in V} \sum_{i = 1}^{w(u)^p - 1} (|\widetilde{H}_u| + 1)^p \\
                                  &= \sum_{u \in V} |\widetilde{H}_u|^p + \sum_{u \in V} \sum_{i = 1}^{w(u)^p - 1} |\widetilde{H}_{b(u,i)}|^p \\
                                  &=cost_p(G', \widetilde{H})^p \\
                                  &\leq 2^p \cdot cost_p(G', H''')^p \\
                                  &\leq (2\alpha)^p cost_p(G', H')^p.
\end{align*}
We conclude that $cost_p(G, \widetilde{H}, w)^p \leq (2e \cdot \alpha)^p \cdot cost_p(G, H, w)^p$, which implies that $cost_p(G, \widetilde{H}, w) \leq (2e \cdot \alpha) \cdot cost_p(G, H, w)$. Thus, we obtain a $(2e \cdot \alpha)$-approximation for the weighted \HLp.
\end{proof}

The above lemma will allow us to reduce Set Cover to the weighted \HL1. Thus, if we assume that we have a $o(\log n)$-approximation algorithm for the unweighted \HL1, this would imply an $o(\log n)$-approximation algorithm for the weighted \HL1, which would further imply an $o(\log n)$-approximation for Set Cover. And this will give a contradiction. We are now ready to prove Theorem \ref{HL1_hardness_theorem}.

\begin{proof}[Proof of Theorem \ref{HL1_hardness_theorem}]
Given an unweighted Set Cover instance, we create a graph $G = (V, E)$ and a corresponding weighted \HL1 instance. We fix two integer parameters $A$ and $B$ (whose values we specify later) and do the following (see Figure \ref{HL_1 fig}):
\begin{itemize}
	\item The two layers directly corresponding to the Set Cover instance are the $2^{nd}$ and the $3^{rd}$ layer. In the $2^{nd}$ layer we introduce one vertex for each set $S_i \in \mathcal{S}$, whose weight is 1, and in the $3^{rd}$ layer we introduce one vertex for each element $x_j \in \mathcal{X}$, whose weight is $B$. We then connect $x_j$ to $S_i$ if and only if $x_j \in S_i$. 
    \item The $1^{st}$ layer contains $A$ vertices $\{r_1, ..., r_A\}$, each of weight $B$. Each vertex $r_i$ is connected to all vertices $\{S_1, ..., S_m\}$.
    \item Finally, we introduce a single vertex $q$ of weight 1 in the $4^{th}$ layer, which is connected to every vertex $x_i$ of the $3^{th}$ layer.
\end{itemize}
We also assign lengths to the edges. The (black) edges $(q, x_i)$ have length $\varepsilon < 1 / 2$ for every $x_i \in \mathcal{X}$, while all other (brown) edges have length 1. We will show that by picking the parameters $A$ and $B$ appropriately, we can get an $O(f(\poly(n))$-approximation for the Set Cover instance, given an $f(n)$-approximation for the weighted \HL1.

\begin{figure}[h]
\begin{center}
\scalebox{1.30}{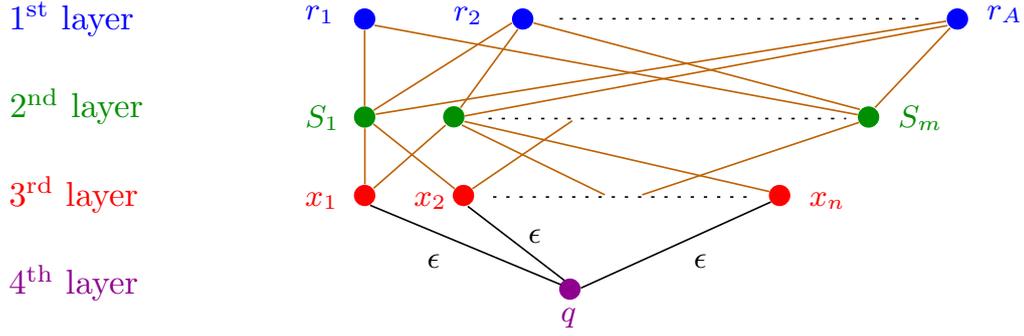}
\caption{The \HL1 instance corresponding to an arbitrary unweighted Set Cover instance \label{HL_1 fig}.}
\end{center}
\end{figure}
We will now define a solution for this HL instance, whose cost depends on the cost of the optimal Set Cover. Let $I \subseteq [m]$ be the set of indices of an optimal Set Cover solution. We define the following HL solution, given in the table below. We use the notation $S(x_i)$ to denote an arbitrarily chosen set of the optimal Set Cover solution that covers $x_i$.

\begin{equation*}
\begin{array}{|c|c|}
\hline
    \mathbf{Layer} & \mathbf{Hubs}\\
    \hline
    1^{st} & \textrm{For every }i \in [A], H_{r_i} = \{r_i\} \cup \{S_j: j \in I\}\\ \hline
    2^{nd} & \textrm{For every }i \in [m], \;\; H_{S_i} = \{S_i, q\} \cup \{r_1, ..., r_A\} \cup \{x_j: x_j \in S_i\} \\ \hline
    3^{rd} & \textrm{For every }i \in [n], \;\; H_{x_i} = \{x_i, q, S(x_i)\} \\ \hline
    4^{th} & H_q = \{q\} \cup \{r_1, ..., r_A\}\\
    \hline
\end{array}
\end{equation*}
\\\\
We argue that the above solution is a feasible solution for the constructed instance. To this end, we consider all possible pairs of vertices for all layers. The notation ``$i$ - $j$" means that we check a pair with one vertex at layer $i$ and the other at layer $j$. The pairs $\{u,u\}$ are trivially satisfied, so we will not consider them below:
\begin{itemize}
    \setlength{\parskip}{0pt}
    \item 1 - 1: $\{r_i, r_j\}$. The common hub is any $S_j$ with $j \in I$.
    \item 1 - 2: $\{r_i, S_j\}$. The common hub is $r_i$.
    \item 1 - 3: $\{r_i, x_j\}$. The common hub is $S(x_j)$.
    \item 1 - 4: $\{r_i, q\}$. The common hub is $r_i$.
    \item 2 - 2: $\{S_i, S_j\}$. The common hub is any $r_t$.
    \item 2 - 3: $\{S_i, x_j\}$. If $x_j \in S_i$, then $x_j \in H_{S_{i}}$. If $x_j \notin S_i$, then $q \in H_{S_i} \cap H_{x_j}$.
    \item 2 - 4: $\{S_i, q\}$. The common hub is $q$.
    \item 3 - 3: $\{x_i, x_j\}$. The common hub is $q$.
    \item 3 - 4: $\{x_j, q\}$. The common hub is $q$.
\end{itemize}

Thus, the above solution is indeed a feasible one. We compute its weighted $\ell_1$-cost, which we denote as $\textrm{COST}_1$ (each term from left to right corresponds to the total cost of the vertices of the corresponding layer):
\begin{align*}
	\textrm{COST}_1 &\leq AB \cdot (OPT_{SC} + 1)+  m \cdot (A + n + 2) + B \cdot n \cdot 3 + (A + m + 1) \\
                    &= O(AB \cdot OPT_{SC})  + O(Am + mn) + O(Bn) + O(m + A).
\end{align*}
We set $A = B = \lceil\max\{m, n\}^{3/2} \rceil$. Then, the total cost is dominated by the term $AB \cdot OPT_{SC}$, and so we get that the cost $OPT$ of the optimal weighted $HL_1$ solution is at most $c \cdot AB \cdot OPT_{SC}$, for some constant $c$. It is also easy to see that $OPT \geq A \cdot B$.

Let $N = A + m + n + 1 = O(\max\{m, n\}^{3/2})$ denote the number of vertices of the constructed graph. Assuming that we have an $f(n)$-approximation for the weighted \HL1, we can get a solution $H'$ of cost $cost_1(G, H', w) \leq c \cdot f(N) \cdot AB \cdot OPT_{SC}$. We will show that we can extract a feasible Set Cover solution of cost at most $O\left(\frac{cost_1(G, H', w)}{AB} \right)$.

To extract a feasible Set Cover, we first modify $H'$. We add $\{r_1, ..., r_A\}$ to the hub set of $q$, and $S_1$ to the hub set of every $r_i$, $i \in [A]$, thus increasing the weighted cost by at most $AB + A$. Thus, we end up with a solution $H''$ whose weighted $\ell_1$-cost is at most $c \cdot f(N) \cdot AB \cdot OPT_{SC} + AB + A \leq c' \cdot f(N) \cdot AB \cdot OPT_{SC}$.
We now look at every vertex $r_i$ of the $1^{st}$ layer for which we have $x_j \in H_{r_i}''$, for some $j \in [n]$. The hub $x_j$ can only be used for the pair $\{r_i, x_j\}$. In that case, we can remove $x_j$ from $H_{r_i}''$ and add $r_i$ to $H_{x_j}''$. The cost of the solution cannot increase, and we again call this new solution $H''$.

We are ready to define our Set Cover solution. For each $i \in [A]$, we define $F_i = H_{r_i}'' \cap \{S_1, ..., S_m\}$. Let $Z_i = |\mathcal{X} \setminus \bigcup_{S \in F_i} S|$ be the number of uncovered elements. If $Z_i = 0$, then $F_i$ is a valid Set Cover. If $Z_i > 0$, then we cover the remaining elements using some extra sets (at most $Z_i$ such sets). At the end, we return $\min_{i \in [A]} \{|F_i| + Z_i\} $.

In order to analyze the cost of the above algorithm, we need the following observation. Let us look at $H_{r_i}''$, and an element $x_j$ that is not covered. By the structure of $H''$, this means that $r_i \in H_{x_j}''$. Thus, the number of uncovered elements $Z_i$ contributes a term $B \cdot Z_i$ to the weighted cost of $H''$.  For each $i$, the number of uncovered elements $Z_i$ implies an increase in the cost  that is ``disjoint" from the increase implied by $Z_j$, for $j \neq i$, and so the total weighted cost is at least $\sum_{i = 1}^A B \cdot (|H_{r_i}''| + Z_i)$. This means that there must exist a $j \in [A]$ such that
\begin{equation*}
	|H_{r_j}''| + Z_j \leq \frac{cost_1(G, H'', w)}{AB}.
\end{equation*}
We pick the Set Cover with cost at most $\min_{i \in [A]} \{|F_i| + Z_i\} \leq \min_{i \in [A]} \{|H_{r_i}''| + Z_i\}$, and so we end up with a feasible Set Cover solution of cost at most
\begin{equation*}
	\frac{c' \cdot f(N) \cdot AB \cdot OPT_{SC}}{AB} = O(f(\mathrm{poly}(n))) \cdot OPT_{SC}.
\end{equation*}
\end{proof}

\subsection{$\Omega(\log n)$-hardness for \HL\infty}\label{sec:infinity-hardness}

In this section, we will show that it is \NP-hard to approximate \HL\infty to within a factor better than $\Omega(\log n)$. We will again use the hardness results for Set Cover.

\begin{theorem}
Given an arbitrary unweighted Set Cover instance $(\mathcal{X}, \mathcal{S})$, $|\mathcal{X}| = n$, $|\mathcal{S}| = m = \poly(n)$, with optimal value $OPT_{SC}$, and an $f(n)$-approximation algorithm for \HL\infty, there is an algorithm that returns a solution for the Set Cover instance with cost at most $O \left(f(\mathrm{poly}(n)) \right) \cdot OPT_{SC}$.
\end{theorem}
\begin{proof}
Let $\mathcal{X} = \{x_1, ..., x_n\}$ and $\mathcal{S} = \{S_1, ..., S_m\}$, $S_i \subseteq \mathcal{X}$, be a Set Cover instance. We will construct an instance of \HL\infty, such that, given an $f(n)$-approximation algorithm for it, we will be able to solve the Set Cover instance within a factor of $O \left(f(O(n^4m)) \right)$. We now describe our construction:
\begin{itemize}
    \setlength{\parskip}{0pt}
    \item We introduce a complete bipartite graph $(A,B,E)$. By slightly abusing notation, we denote $|A| = A$ and $|B| = B$, where $A$ and $B$ are two parameters to be set later on.
    \item Each vertex $u \in A$ ``contains'' $K$ vertices $\{r_{u,1}, ..., r_{u,K}\}$.
    \item Each vertex vertex $v \in B$ ``contains" a copy of the universe $\{x_{v,1}, ..., x_{v,n}\}$.
    \item Each edge $(u,v)$ is replaced by an intermediate layer of vertices $\mathcal{S}_{uv} = \{S_{uv,1}, ..., S_{uv,m}\}$, which is essentially one copy of $\mathcal{S}$. We then connect every vertex $r_{u,i}$, $i \in [K]$, to every vertex $S_{uv,j}$, $j \in [m]$, and we also connect each $S_{uv,j}$ to $x_{v,t}$, if $x_t \in S_j$. All these edges (colored red in the figure) have length 1.
    \item Finally, we introduce three extra vertices $q_A$ and $q_B$ and $q_S$, and the edges $(q_A, r_{u,i})$, for all $u\in A$ and $i \in [K]$, the edges $(q_B, x_{v,j})$, for all $v \in B$ and $j \in [n]$, and the edges $(q_S, S_{uv,j})$, for all $u \in A$, $v \in B$ and $j \in [m]$. All these edges (colored black in the figure) have length $\varepsilon < 1$.
\end{itemize}
The construction is summarized in Figures \ref{HL infty fig1}, \ref{HL infty fig2}.

\begin{figure}[ht!]
\centering
\begin{subfigure}{.5\textwidth}
  \centering
\scalebox{0.5}{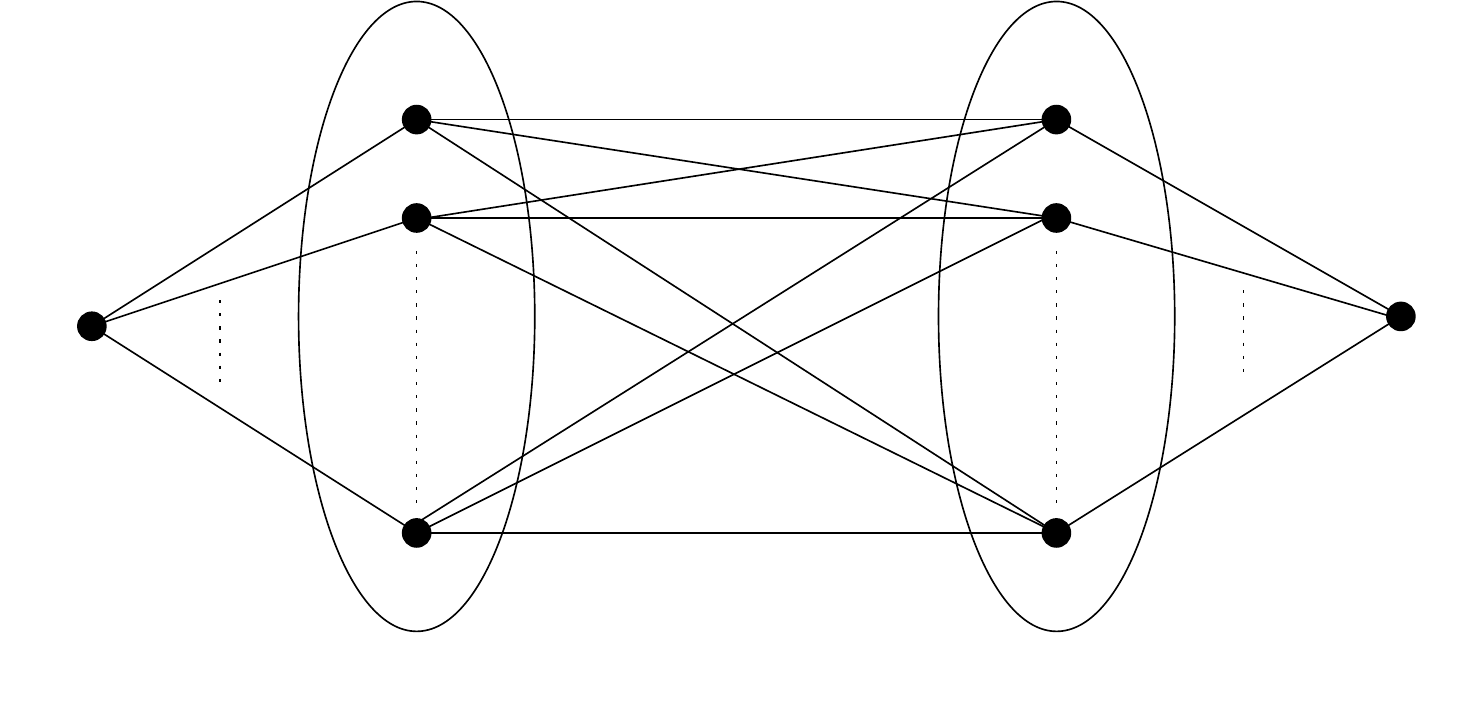}
  \caption{The general structure of the graph.}
\label{HL infty fig1}
\end{subfigure}%
\begin{subfigure}{.5\textwidth}
  \centering
  \scalebox{0.4}{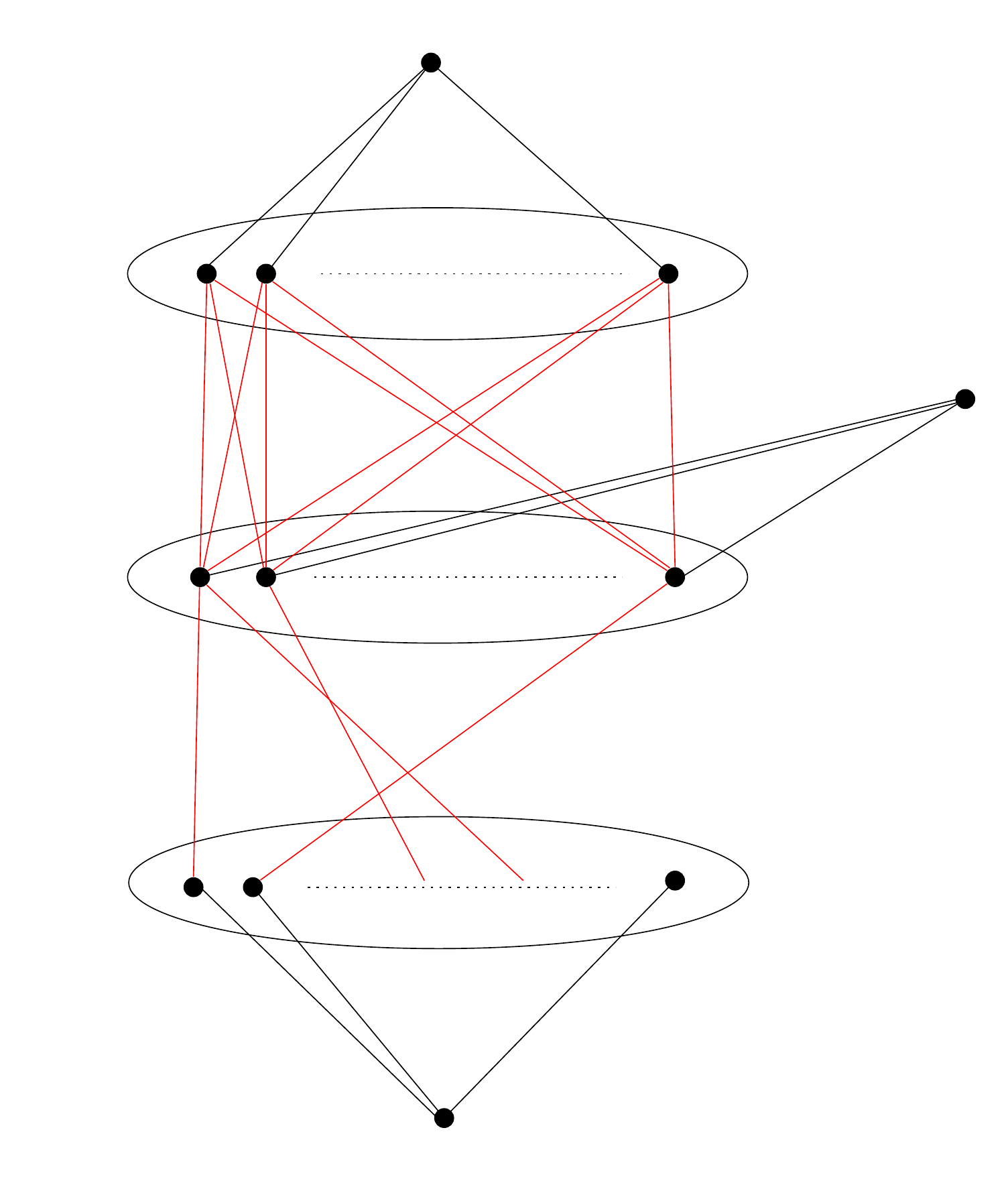}
  \caption{A closer look at an edge $(u,v) \in E$.}
\label{HL infty fig2}
\end{subfigure}
\end{figure}

In the resulting construction, the number of vertices, denoted by $N$, is $N = AK + Bn + ABm + 3$, and the number of edges, denoted by $M$, is at least $M \geq AB(Km + m) + AK + ABm + Bn$. Let $OPT$ denote the cost of an optimal \HL\infty solution $H$ for this instance. Then, by a standard pigeonhole principle argument, and since every edge is a unique shortest path, we get that $OPT \geq \frac{M}{N}$. We now set the parameters, as follows: $A = B = K = n^2$. With these values, we have $N = \Theta(n^4 \cdot m)$, $M = \Omega(n^6 \cdot m)$ and $OPT = \Omega(n^2)$.

We will describe an intended feasible solution for this instance, that will give an upper bound on OPT. Let $I \subseteq[m]$ denote an optimal Set Cover of our original Set Cover instance, and let $I_j \in I$ denote the index of an arbitrarily chosen set of the optimal solution that covers $x_j$. The HL solution is the following:
\begin{itemize}
    \setlength{\parskip}{0pt}
	\item $H_{r_{u,i}} = \{r_{u,i}\} \cup \left( \bigcup_{{v \in B}} \{ S_{uv,j}: j \in I \} \right) \cup \{q_A, q_B, q_S\}$, for $u \in A$ and $i \in [K]$.
    \item $H_{x_{v,j}} = \{x_{v,j}\} \cup \left( \bigcup_{u \in A} \{S_{uv,I_j}\} \right) \cup \{q_A, q_B, q_S\}$, for $v \in B$ and $j \in [n]$.
    \item $H_{S_{uv,t}} = \{ S_{uv,t} \} \cup \{r_{u,1}, ..., r_{u,K} \} \cup \{ x_{v,1}, ..., x_{v,n} \} \cup \{q_A, q_B, q_S\}$, for $u\in A$, $v \in B$ and $t \in [m]$.
    \item $H_{q_t} = \{q_A, q_B, q_S\}$, for $t \in \{A,B,S\}$.
\end{itemize}
We now compute the sizes of these hub sets. We have:
\begin{itemize}
    \setlength{\parskip}{0pt}
	\item $|H_{r_{u,i}}| = B |I| + 4 = \Theta(n^2 \cdot |I|)$, for $u \in A$ and $i \in [K]$.
    \item $|H_{x_{v,j}}| = A + 4 = \Theta(n^2)$, for $v \in B$ and $j \in [n]$.
    \item $|H_{S_{uv,t}}| = K + n + 4 = \Theta(n^2)$, for $u\in A$, $v \in B$ and $t \in [m]$.
    \item $|H_{q_t}| = 3$, for $t \in \{A,B,S\}$.
\end{itemize}
Thus, we get that the value of the above solution is $\|H\|_\infty = Val = \Theta(n^2 \cdot |I|)$. We now show that the above is indeed a feasible solution. For that, we consider all possible pairs of vertices:
\begin{itemize}
    \setlength{\parskip}{0pt}
	\item $r_{u,i}$ - $r_{v,j}$: The common hub is $q_A$.
    \item $r_{u,i}$ - $S_{uv,j}$: The common hub is $r_{u,i}$.
    \item $r_{u,i}$ - $S_{wv,j}$, $w \neq u$, $v \neq u$: The common hub is $q_S$.
    \item $r_{u,i}$ - $x_{v,j}$: The common hub is $S_{uv,I_j}$.
    \item $r_{u,i}$ - $q_t$, for $t \in \{A,B,S\}$: The common hub is  $q_t$.
    \item $S_{uv,i}$ - $S_{u'v',j}$: The common hub is $q_S$.
    \item $S_{uv,i}$ - $x_{v,j}$:  The common hub is $x_{v,j}$.
    \item $S_{uv,i}$ - $x_{v',j}$, $u \neq v'$, $v \neq v'$: The common hub is $q_B$ (or $q_S$).
    \item $S_{uv,i}$ - $q_t$, for $t \in \{A,B,S\}$: The common hub is $q_t$.
    \item $x_{v,i}$ - $x_{v',j}$: The common hub is $q_B$.
    \item $x_{v,j}$ - $q_t$, for $t \in \{A,B,S\}$: The common hub is $q_t$
    \item $q_t$ - $q_{t'}$, for $t, t' \in \{A,B,S\}$: The common hub is $q_{t'}$.
\end{itemize}
Thus, the proposed solution is indeed a feasible solution. Assuming now that we have an $f(n)$-approximation algorithm for \HL\infty, we can obtain a solution $H'$ of cost $\|H'\|_\infty \leq f(N) \cdot OPT \leq c \cdot f(N) \cdot n^2 \cdot |I|$. We will now show that we can extract a feasible solution for the original Set Cover instance, of cost $O(f(N)) \cdot |I|$. As a reminder, we have already proved that $\|H'\|_\infty = \Omega(n^2)$. We first transform $H'$ to a solution $H''$ that will look more like our intended solution, as follows:
\begin{itemize}
    \setlength{\parskip}{0pt}
	\item $H''_{S_{uv},t} := H_{S_{uv},t}' \cup \{r_{u,1}, ..., r_{u,K} \} \cup \{ x_{v,1}, ..., x_{v,n} \} \cup \{q_A, q_B, q_S\}$, for $u\in A$, $v \in B$ and $t \in [m]$. We have $|H_{S_{uv},t}''| \leq |H_{S_{uv},t}'| + K + n + 3 \leq \|H'\|_\infty + n^2 + n + 3 = O(\|H'\|_\infty)$.
    \item $H_{q_t}'' := H_{q_t}' \cup \{q_A, q_B, q_S\}$, for $t \in \{A,B,S\}$. We have $|H_{q_t}''| \leq |H_{q_t}'| + 3 = O(|H_{q_t}'|) = O(\|H'\|_\infty)$.
    \item We now look at $H_{r_{u,i}}'$. For every $x_j \in \mathcal{X}$, we (arbitrarily) pick a set $S(x_j) \in \mathcal{S}$ with $x_j \in S(x_j)$, that we will use to cover it. Now, if $x_{v,j} \in H_{r_{u,i}}'$, we remove $x_{v,j}$ from $H_{r_{u,i}}'$ and add $S_{uv}(x_j)$ to $H_{r_{u,i}}'$ (here we slightly abuse notation; the vertex $S_{uv}(x_j)$ corresponds to the vertex $S_{uv, t}$ where $t$ is the index of the set $S(x_j)$). This doesn't change the size of $H_{r_{u,i}}'$. We also add $S_{uv}(x_j)$ to the hub set of $x_{v,j}$. This increases the size of $H_{x_{v,j}}'$ by 1. The crucial observation here is that since we have decided in advance which set we will use to cover $x_j$, then $|H_{x_{v,j}}'|$ can only increase by 1, for every edge $(u,v)$. Thus, the total increase in $|H_{x_{v,j}}'|$ is at most $A$, i.e.~$|H_{x_{v,j}}''| \leq |H_{x_{v,j}}'| + n^2 = O(\|H'\|_\infty)$.
\end{itemize}
The above transformed solution, as shown, has the same (up to constant factors) cost as the solution that the algorithm returns, i.e. $\|H''\|_\infty = O(\|H'\|_\infty) = O(f(N)) \cdot n^2 \cdot |I|$, and is clearly feasible.

In order to recover a good Set Cover solution, we look at the sets $H_{r_{u,i}}'' \cap \mathcal{S}_{uv}$. Each such intersection can be viewed as a subset $C_{u,v,i}$ of $\mathcal{S}$. Let $Z_{u,v,i}$ denote the number of elements that are not covered by $C_{u,v,i}$, i.e. $Z_{u,v,i} = |\mathcal{X} \setminus (\bigcup_{S \in C_{u,v,i}} S)|$. Our goal is to show that there exists a $\{u,v,i\}$ such that $|C_{u,v,i}| + Z_{u,v,i} = O(\|H''\|_\infty / n^2)$. Since there is a polynomial number of choices of $\{u,v,i\}$, we can then enumerate over all choices and find a Set Cover with cost $O(f(N)) \cdot |I|$.

To prove that such a good choice exists, we will make a uniformly random choice over $\{u,v,i\}$, and look at the expected value $\E \left[|C_{u,v,i}| + Z_{u,v,i} \right]$. We have $\E\left[|C_{u,v,i}| + Z_{u,v,i}\right] = \E\left[|C_{u,v,i}|\right] + \E[Z_{u,v,i}]$. We look separately at the two terms. We make the following 2 observations:
\begin{equation*}
	\sum_{v \in B} |C_{u,v,i}| \leq |H_{r_{u,i}}''| = O(\|H'\|_\infty),
\end{equation*}
and
\begin{equation*}
	\sum_{u \in A} \sum_{i \in K} Z_{u,v,i} \leq \sum_{j \in [n]} |H_{x_{v,j}}''| = n \cdot O(\|H'\|_\infty).
\end{equation*}
The second observation follows from the fact that for any given $r_{u,i}$ and edge $(u,v)$, the uncovered elements $x_{v,j}$ must have $r_{u,i} \in H_{x_{v,j}}''$. With these, we have
\begin{equation*}
	\E[|C_{u,v,i}|] = \frac{1}{ABK} \sum_{u \in A, i \in [K]} \sum_{v \in B} |C_{u,v,i}| \leq \frac{1}{ABK} \cdot A K \cdot O(\|H'\|_\infty) = O(\|H'\|_\infty / B).
\end{equation*}
Similarly,
\begin{equation*}
	\E[Z_{u,v,i}] = \frac{1}{ABK} \sum_{v \in B} \sum_{u \in A} \sum_{i \in K} Z_{u,v,i} \leq \frac{1}{ABK} \cdot B \cdot n \cdot O(\|H'\|_\infty) = \frac{n}{AK} \cdot O(\|H'\|_\infty).
\end{equation*}
Thus, we get that $\E\left[|C_{u,v,i}| + Z_{u,v,i}\right] = \left(\frac{1}{B} + \frac{n}{AK} \right) \cdot O(\|H'\|_\infty) = O(\|H'\|_\infty / n^2) = O(f(N)) \cdot |I|$. This means that there exists a choice of $\{u,v,i\}$ such that the corresponding Set Cover has size $O(f(N)) \cdot |I|$. As already mentioned, there are polynomially many choices, so we can enumerate them and find the appropriate $\{u,v,i\}$, and, thus, recover a Set Cover solution for our original Set Cover instance of cost $O(f(N)) \cdot |I|$, where, as already stated, $N = \Theta(n^4 \cdot m)$.
\end{proof}

\begin{corollary}\label{cor:infinity-hardness}
It is \NP-hard to approximate \HL\infty to within a factor better than $\Omega(\log n)$.
\end{corollary}
\begin{proof}
 	The previous theorem gives an $O(f(O(n^4m))$-approximation algorithm for Set Cover, given that an $f(n)$-approximation algorithm for \HL\infty exists. If we assume that there exists such an algorithm with $f(n) = o(\log n)$, then we could use it to approximate Set Cover within a factor $o(\log O(n^4m)) = o(\log \poly(n)) = o(\log n)$, and, by Theorem~\ref{Set-Cover-Hardness}, this is impossible, assuming that $\P \neq \NP$.
\end{proof}

\subsection{$\Omega(\log n)$-hardness for \HLp, for $p = \Omega(\log n)$}\label{sec:lp-hardness}

In this section, we use the well-known fact that the $\ell_{\log n}$-norm of an $n$-dimensional vector is within a constant factor from its $\ell_\infty$-norm to conclude that \HLp is $\Omega(\log n)$-hard to approximate on graphs with $n$ vertices, for $p = \Omega(\log n)$.

\begin{theorem}
For any fixed $\varepsilon > 0$, it is \NP-hard to approximate \HLp to within a factor better than $O(\log n)$, for every $p \geq \varepsilon \log n$.
\end{theorem}
\begin{proof}
Let $G = (V, E, l)$ be a hub labeling instance, and let $H^{(p)}$ denote an optimal \HLp solution, and $H^{(\infty)}$ denote an optimal \HL\infty solution. We have
\begin{align*}
    \left(\sum_{u \in V} |H_u^{(p)}|^p \right)^{1/p} &\leq \left(\sum_{u \in V} |H_u^{(\infty)}|^p \right)^{1/p} \leq n^{1/p} \|H^{(\infty)}\|_\infty = 2^{\log n / p} \|H^{(\infty)}\|_\infty \\
                                                     &\leq 2^{1 / \varepsilon} \|H^{(\infty)}\|_\infty.
\end{align*}
If we have an $f(n)$-approximation for \HLp, then this means that we can get a solution $H'$ such that $\left(\sum_{u \in V} |H_u'|^p \right)^{1/p} \leq f(n) \cdot \left(\sum_{u \in V} |H_u^{(p)}|^p \right)^{1/p}$. From the previous discussion, this implies
\begin{equation*}
    \|H'\|_\infty \leq \left(\sum_{u \in V} |H_u'|^p \right)^{1/p} \leq f(n) \cdot \left(\sum_{u \in V} |H_u^{(p)}|^p \right)^{1/p} \leq (f(n) \cdot 2^{1 / \varepsilon}) \cdot \|H^{(\infty)}\|_\infty.
\end{equation*}
By Corollary~\ref{cor:infinity-hardness}, we now conclude that we must have $f(n) \geq \frac{c \cdot \log n}{2^{1 / \varepsilon}} = \Omega(\log n)$, where $c$ is some universal constant.
\end{proof}

\section{Hub labeling on directed graphs}\label{sec:appendix_directed}

In this section, we sketch how some of the presented techniques can be used for the case of directed graphs. Let $G = (V, E)$ be a directed graph with edge lengths $l(e) > 0$. Instead of having one set of hubs, each vertex $u$ has two sets of hubs, the forward hubs $H_u^{(f)}$ and the backward hubs $H_u^{(b)}$. The covering property is now stated as follows: for every (directed) pair $(u,v)$ and some directed shortest path $P$ from $u$ to $v$, we must have $H_u^{(f)} \cap H_v^{(b)} \cap P \neq \emptyset$. The \HLp objective function can be written as $\left(\sum_{u \in V} L_u^p \right)^{1 / p}$, where $L_u = |H_u^{(f)}| + |H_u^{(b)}|$.

The Set Cover based approach of Cohen et al.~\cite{DBLP:journals/siamcomp/CohenHKZ03} and Babenko et al.~\cite{DBLP:conf/icalp/BabenkoGGN13} can be used in this setting in order to obtain an $O(\log n)$-approximation for \HLp, $p \in [1, \infty]$. It is also straightforward to see that there is a very simple 2-approximation preserving reduction from undirected \HLp to directed \HLp, implying that an $\alpha$-approximation for directed \HLp would give a $2\alpha$-approximation for undirected \HLp. Thus, the hardness results of Section~\ref{Hardness} can be applied to the directed case as well, and so we end up with the following theorem.
\begin{theorem}
\HLp is $\Omega(\log n)$-hard to approximate in directed graphs with $n$ vertices and multiple shortest paths, for $p \in \{1\} \cup [\varepsilon \log n, \infty]$, unless \P = \NP.
\end{theorem}

Having matching lower and upper bounds (up to constant factors) for the general case, we turn again to graphs with unique (directed) shortest paths. The notion of pre-hubs can be extended to the directed case as follows: a family of sets $\{(\widehat{H}_u^{(f)}, \widehat{H}_u^{(b)})\}_{u \in V}$ is a family of pre-hubs if for every pair $(u,v)$ there exist $u' \in \widehat{H}_u^{(f)} \cap P_{uv}$ and $v' \in \widehat{H}_v^{(b)} \cap P_{uv}$ such that $u' \in P_{v'v}$.

We now present the LP relaxation for the $\ell_1$ case (see Figure~\ref{fig:directed-LP}). In order to obtain a feasible set of pre-hubs, for each vertex $u \in V$ we construct two trees (both rooted at $u$): $T_u^{(f)}$ is the union of all directed paths from $u$ to all other vertices, and $T_u^{(b)}$ is the union of all directed paths to $u$ from all other vertices. We drop the orientation on the edges, and we note that these are indeed trees. We proceed as in the undirected case and obtain a set $\widehat{H}_u^{(f)}$ from $T_u^{(f)}$ (using the variables $x_{uv}^{(f)}$) of size at most $2 \sum_{v \in V} x_{uv}^{(f)}$, and a set $\widehat{H}_u^{(b)}$ from $T_u^{(b)}$ (using the variables $x_{uv}^{(b)}$) of size at most $2 \sum_{v \in V} x_{uv}^{(b)}$. It is not hard to see that the obtained sets $(\widehat{H}_u^{(f)}, \widehat{H}_u^{(b)})_{u \in V}$ are indeed pre-hubs.

\begin{figure}[h]
\noindent($\mathbf{DIR-LP_1}$)
\begin{align*}
    \min:          &\quad \sum_{u \in V} \sum_{v \in V} \left(x_{uv}^{(f)} + x_{uv}^{(b)}\right) \\
    \textrm{s.t.:} &\quad \sum_{w \in P_{uv}} \min\{x_{uw}^{(f)}, x_{vw}^{(b)}\} \geq 1,  && \forall (u, v) \in V \times V, \\
                   &\quad x_{uv}^{(f)} \geq 0,  && \forall (u,v) \in V \times V, \\
                   &\quad x_{uv}^{(b)} \geq 0,  && \forall (u,v) \in V \times V.
\end{align*}
\caption{The LP relaxation for \HL1 on directed graphs.}
\label{fig:directed-LP}
\end{figure}

We can now use a modified version of Algorithm \ref{Pre-Hubs_Algorithm}; see Algorithm~\ref{DIR-Pre-Hubs_Algorithm}. It is easy to see that the obtained solution is always feasible, and, with similar analysis as before, we prove that $\E[|\widehat{H}_u^{(f)}|] \leq 2(\log D + O(1)) \cdot \sum_{v} x_{uv}^{(f)}$ and $\E[|\widehat{H}_u^{(b)}|] \leq 2(\log D + O(1)) \cdot \sum_{v} x_{uv}^{(b)}$. Thus, in expectation, we obtain a solution of cost $O(\log D) \cdot OPT_{DIR-LP_1}$.

\begin{algorithm}[h]
\begin{enumerate}
    \item Solve $\mathbf{DIR-LP_1}$ and get an optimal solution $\{(x_{uv}^{(f)}, x_{uv}^{(b)})\}_{(u,v) \in V \times V}$.
    \item Obtain a set of pre-hubs $\{(\widehat{H}_u^{(f)}, \widehat{H}_u^{(b)})\}_{u \in V}$ from $x$. 
    \item Generate a random permutation $\pi : [n] \to V$ of the vertices. 
    \item Set $(H_u^{(f)}, H_u^{(b)}) = (\varnothing, \varnothing)$, for every $u \in V$.
    \item \textbf{for} $i= 1$ \textbf{to} $n$ \textbf{do}:\\
          \hspace*{16pt}  \textbf{for} every $u \in V$ \textbf{do}:\\
          \hspace*{36pt}     \textbf{for} every $u' \in \widehat{H}_u^{(f)}$ such that $\pi_i \in P_{uu'}$ and $P_{\pi_i u'} \cap \widehat{H}_u^{(f)} = \{u'\}$ \textbf{do}:\\
          \hspace*{56pt}         \textbf{if} $P_{\pi_i u'} \cap H_u^{(f)} = \varnothing$ \textbf{then} $H_u^{(f)} := H_u^{(f)} \cup \{\pi_i\}$\\
          \hspace*{36pt}     \textbf{for} every $u' \in \widehat{H}_u^{(b)}$ such that $\pi_i \in P_{u'u}$ and $P_{u' \pi_i} \cap \widehat{H}_u^{(b)} = \{u'\}$ \textbf{do}:\\
          \hspace*{56pt}         \textbf{if} $P_{u' \pi_i} \cap H_u^{(b)} = \varnothing$ \textbf{then} $H_u^{(b)} := H_u^{(b)} \cup \{\pi_i\}$.
    \item Return $\{(H_u^{(f)}, H_u^{(b)})\}_{u \in V}$.
\end{enumerate}
\caption{Algorithm for \HL1 on directed graphs with unique shortest paths}
\label{DIR-Pre-Hubs_Algorithm}
\end{algorithm}

The analysis can also be generalized for arbitrary fixed $p \geq 1$, similar to the analysis in Section~\ref{sec:lp-norm-algorithm}. The algorithm is modified in the same way, and using the fact that for $x,y,p \geq 1$, we have $x^p + y^p \leq (x + y)^p \leq 2^p (x^p + y^p)$, we can again obtain a solution of cost at most $O_p(\Harm_D) \cdot OPT_{REL}$, where $OPT_{REL}$ is the optimal value of the corresponding convex relaxation. Thus, we obtain the following theorem.
\begin{theorem}
There is an $O(\log D)$-approximation algorithm for \HLp, for any fixed $p \geq 1$, on directed graphs with unique shortest paths.
\end{theorem}

\chapter{Hub Labeling on trees}\label{chapter:trees}

\newcommand{\tc}{\tilde{c}}

In this chapter, we study the Hub Labeling problem on trees. Although trees might seem a very simple class of graphs, the problem is not at all trivial even on trees. We will present several algorithms and results, culminating in the equivalence of HL on trees with a seemingly unrelated problem, namely the problem of searching for a node in a tree.  We first observe that when the graph is a tree, the length function $l$ does not play any role in the task of choosing the optimal hubs (it only affects the actual distances between the vertices), and so we assume that we are simply given an unweighted tree $T = (V, E)$, $|V| = n$. We start with proving a structural result about optimal solutions in trees; we show that there always exists a \textit{hierarchical hub labeling} that is also an optimal hub labeling. We then analyze a simple and fast heuristic for HL on trees proposed by Peleg \cite{DBLP:journals/jgt/Peleg00}, and prove that it gives a 2-approximation for \HL1. We do not know if our analysis is tight, but we prove that there are instances where the heuristic finds a suboptimal solution of cost at least $\left(\frac{3}{2}-\varepsilon\right) \cdot OPT$ (for every $\varepsilon > 0$). We then refine the approximation factor by presenting a DP-based polynomial-time approximation scheme (PTAS) and a quasi-polynomial-time exact algorithms for \HLp on trees, for every $p \in [1, \infty]$.

As mentioned in the introduction, after the publication of our work~\cite{DBLP:conf/soda/AngelidakisMO17}, it was pointed out to us~\cite{GawKMW17} that our structural result that there always exists a hierarchical hub labeling that is optimal allows one to cast the Hub Labeling problem on trees as a problem of vertex search in trees. We discuss this connection after the presentation of our results, and give a complete description and analysis of how previous results imply exact polynomial-time algorithms for HL on trees.


\section{Optimal solutions for trees are hierarchical}
Let $T=(V,E)$ be a tree. In this section, we show that any feasible hub labeling $H$ for $T$ can be converted to a hierarchical hub labeling $H'$ of at most the same $\ell_p$-cost (for every $p\in[1,\infty]$). Therefore, there always exists an optimal solution that is hierarchical.

\begin{theorem}\label{HHL_optimal_for_trees}
For every tree $T = (V,E)$, there always exists an optimal \HLp solution that is hierarchical, for every $p \in [1, \infty]$.
\end{theorem}
\begin{proof}
To prove this, we consider a feasible solution $H$ and convert it to a hierarchical solution $H'$ such that $|H'_u|\leq |H_u|$ for every $u \in V$. In particular, the $\ell_p$-cost of $H'$ is at most the $\ell_p$-cost of $H$ for every $p$.

The construction is recursive (the underlying inductive hypothesis for smaller subinstances being that a feasible HL $H$ can be converted to a hierarchical solution $H'$ such that $|H'_u|\leq |H_u|$ for every $u$.) First, for each $u \in V$, define an induced subtree $T_u \subseteq T$ as follows: $T_u$ is the union of paths $P_{uv}$ over all $v\in H_u$. In other words, a vertex $w$ belongs to $H_u$ if there is a hub $v\in H_u$ such that $w\in P_{uv}$. Note that $T_u$ is a (non-empty) connected subtree of $T$.

The crucial property that we need is that $T_u \cap T_v \neq \emptyset$, for every $u,v \in V$. To see this, consider any pair $\{u,v\}$, $u \neq v$.  We know that $H_u \cap H_v \cap P_{uv} \neq \emptyset$. Let $w \in H_u \cap H_v \cap P_{uv}$. By construction, $w \in T_u$ and $w \in T_v$, and so $T_u \cap T_v \neq \emptyset$. We now use the fact that a family of subtrees of a tree satisfies the \textit{Helly property} (which first appeared as a consequence of the work of Gilmore \cite{Gilmore}, and more explicitly a few years later in \cite{gyarfas1970helly}) that can be stated as follows. If we are given a family of subtrees of $T$ such that every two subtrees in the family intersect, then all subtrees in the family intersect (i.e. they share a common vertex).

Let $r \in \bigcap_{u \in V} T_u$. We remove $r$ from $T$. Consider the connected components $Q_1, ...,Q_c$ of $T-r$. Denote the connected component that contains vertex $u$ by $Q^u$. Let $\widetilde H_u = H_u\cap Q^u$. Note that $|\widetilde H_u|\leq|H_u| -1$, since $r \in T_u$, which, by the definition of $T_u$, implies that there exists some $w \notin Q^u$ with $w \in H_u$. Consider now two vertices $u,v\in Q_i$. They have a common hub $w\in H_u \cap H_v\cap P_{uv}$. Since $P_{uv} \subset Q^u = Q^v = Q_i$, we have $w\in \widetilde H_u \cap \widetilde H_v \cap P_{uv}$. Therefore, $\{\widetilde H_u:u\in Q_i\}$ is a feasible hub labeling for $Q_i$. Now, we recursively find hierarchical hub labelings for the subtrees $Q_1,\dots, Q_c$. Denote the hierarchical hub labeling for $u$ in $Q^u$ by $H''_u$. The inductive hypothesis ensures that for every $u \in Q^u$, $|H''_u| \leq |\widetilde H_u|\leq |H_u| -1$.

Finally, define $H_u' = H''_u \cup \{r\}$, for $u\neq r$, and $H_r' = \{r\}$. We show that $H_u'$ is a hub labeling. Consider $u,v\in V$. If $u,v\in Q_i$ for some $i$, then $H'_u \cap H'_v\cap P_{uv} \supset H''_u \cap H''_v\cap P_{uv}\neq \emptyset$ since $H''$ is a hub labeling for $Q_i$. If $u\in Q_i$ and $v \in Q_j$ ($i\neq j$), then $r\in  H'_u \cap H'_v\cap P_{uv}$. Also, if either $u=r$ or $v=r$, then again $r\in H'_u \cap H'_v\cap P_{uv}$. We conclude that $H'$ is a feasible hub labeling. Furthermore, $H'$ is a \textit{hierarchical} hub labeling: $r\preceq u$ for every $u$ and $\preceq$ is a partial order on every set $Q_i$; elements from different sets $Q_i$ and $Q_j$ are not comparable w.r.t.\ $\preceq$.

We have $|H_u'| = |H''_u| + 1 \leq |H_u|$ for $u\neq r$ and $|H'_r| = 1\leq |H_r|$, as required.
\end{proof}

This theorem allows us to restrict our attention only to hierarchical hub labelings, which have a much simpler structure than arbitrary hub labelings, when we design algorithms for HL on trees.

\section{An analysis of Peleg's heuristic for \HL1 on trees} \label{trees_2_approx}

In this section, we  analyze a purely combinatorial algorithm for HL proposed by Peleg in \cite{DBLP:journals/jgt/Peleg00} and show that it returns a hierarchical 2-approximate hub labeling on trees (see Algorithm~\ref{Tree-Algorithm}). In Peleg's paper~\cite{DBLP:journals/jgt/Peleg00}, it is only proved that the algorithm returns a feasible hub labeling $H$ with $\max_{u \in V} |H_u| = O(\log n)$ for a tree on $n$ vertices.

\begin{definition}\label{def:bal-sep-vertex}
Consider a tree $T$ on $n$ vertices. We say that a vertex $u$ is a balanced separator vertex if every connected component of $T - u$  has at most $n/2$ vertices. The weighted balanced separator vertex for a vertex-weighted tree is defined analogously.
\end{definition}
It is well known that every tree $T$ has a balanced separator vertex (in fact, a tree may have either exactly one or exactly two balanced separator vertices) and such a separator vertex can be found efficiently (i.e.~in linear time) given $T$. The algorithm by Peleg, named here \emph{Tree Algorithm}, is described in the figure below (Algorithm~\ref{Tree-Algorithm}).

\begin{algorithm}[h]
\begin{enumerate}
    \item Find a balanced separator vertex $r$ in $T'$.
    \item Remove  $r$ and recursively find a HL in each subtree $T_i$ of $T'-r$. Let $H'$ be the \\labeling obtained by the recursive procedure. \\(If $T'$ consists of a single vertex and, therefore, $T'-r$ is empty, the algorithm does\\ not make any recursive calls.)
    \item Return $H_u := H_u' \cup \{r\}$, for every vertex $u$ in $ T'- \{r\}$, and $H_r = \{r\}$.
\end{enumerate}
\caption{Tree Algorithm}
\label{Tree-Algorithm}
\end{algorithm}

It is easy to see that the algorithm always returns a feasible hierarchical hub labeling, in total time $O(n \log n)$. To bound its cost, we use the primal-dual approach. We consider the dual of $\mathbf{LP_1}$ (which was presented in Figure~\ref{fig:lp1}). Then, we define a dual feasible solution whose cost is at least half of the cost of the solution that the algorithm returns. We formally prove the following theorem.
\begin{theorem}\label{peleg_theorem}
The Tree Algorithm (Algorithm~\ref{Tree-Algorithm}) is a 2-approximation algorithm for \HL1 on trees.
\end{theorem}
\begin{proof}
The primal and dual linear programs for \HL1 on trees are given in Figure~\ref{primal-dual}. We note that the dual variables $\{a_{uv}\}_{u,v}$ correspond to unordered pairs $\{u,v\} \in I$, while the variables $\{\beta_{uvw}\}_{u,v,w}$ correspond to ordered pairs $(u,v) \in V \times V$, i.e. $\beta_{uvw}$ and $\beta_{vuw}$ are different variables.

\begin{figure}
\begin{minipage}{0.4\textwidth}
\footnotesize
\vspace{-6pt}
(\textbf{PRIMAL-LP})
\vspace{12pt}
\begin{align*}
    \min: &\quad \sum_{u \in V} \sum_{v \in V} x_{uv}\\
    \text{s.t.:}  &\quad \sum_{w \in P_{uv}} y_{uvw} \geq 1, &&\forall\, \{u,v\} \in I, \\
          &\quad x_{uw} \geq y_{uvw}, &&\forall\, \{u,v\} \in I, \; \forall\, w \in P_{uv},\\
          &\quad x_{vw} \geq y_{uvw}, &&\forall\, \{u,v\} \in I, \; \forall\, w \in P_{uv},\\
          &\quad x_{uv} \geq 0, &&\forall\, \{u,v\} \in V \times V,\\
          &\quad y_{uvw} \geq 0, &&\forall\, \{u,v\} \in I, \;\forall\, w \in P_{uv}.
\end{align*}
\vspace{8pt}
\end{minipage}%
\hspace{4mm}\vline \hspace{4mm}
\begin{minipage}{0.4\textwidth}
\footnotesize
(\textbf{DUAL-LP})\\
variables: $\alpha_{uv}$ and $\beta_{uvw}$ for $w\in P_{uv}$
\begin{align*}
    \max: &\quad \sum_{\{u,v\} \in I} \alpha_{uv}\\
    \text{s.t.:} &\quad \alpha_{uv} \leq \beta_{uvw} + \beta_{vuw}, && \forall\, \{u,v\} \in I \;, u \neq v,\\
    & &&\forall\, w \in P_{uv},\\
    &\quad \alpha_{uu} \leq \beta_{uuu}, &&\forall\, u \in V,\\
    &\quad \sum_{v: w \in P_{uv}} \beta_{uvw} \leq 1, &&\forall\, (u,w) \in V \times V,\\
    &\quad  \alpha_{uv} \geq 0, && \forall\, \{u,v\} \in I,\\
    &\quad  \beta_{uvw} \geq 0, && \forall\, \{u,v\} \in I, \forall\, w \in P_{uv},\\
    &\quad  \beta_{vuw} \geq 0,&& \forall\, \{u,v\} \in I, \forall\, w \in P_{uv}. 
\end{align*}
\end{minipage}
\caption{Primal and Dual LPs for \HL1 on trees.}
\label{primal-dual}
\end{figure}
\normalsize

As already mentioned, it is straightforward to prove that the algorithm finds a feasible hierarchical hub labeling. We now bound the cost of the solution by  constructing a fractional solution for the DUAL-LP. To this end, we track the execution of the algorithm and gradually define the fractional solution. Consider one iteration (i.e. one level of the recursion) of the algorithm in which the algorithm processes a tree $T'$ ($T'$ is a subtree of $T$). Let $r$ be the balanced separator vertex that the algorithm finds in line 1. At this iteration, we assign dual variables $a_{uv}$ and $\beta_{uvw}$ for those pairs $u$ and $v$ in $T'$ for which $P_{uv}$ contains vertex $r$. Let $n'$ be the size of $T'$, $A = 2/n'$ and $B=1/n'$. Denote the connected components of $T'-r$ by $T_1, ..., T_t$; each $T_i$ is a subtree of $T'$.

Observe that we assign a value to each $a_{uv}$ and $\beta_{uvw}$ exactly once. Indeed, since we split $u$ and $v$ at some iteration, we will assign a value to $a_{uv}$ and $\beta_{uvw}$ at least once. Consider the first iteration in which we assign a value to $a_{uv}$ and $\beta_{uvw}$. At this iteration, vertices $u$ and $v$ lie in different subtrees $T_i$ and $T_j$ of $T'$ (or $r\in\{u,v\}$). Therefore, vertices $u$ and $v$ do not lie in the same subtree $T''$ in the consecutive iterations; consequently, we will not assign new values to $a_{uv}$ and $\beta_{uvw}$ later.

For $u\in T_i$ and $v\in T_j$ (with $i\neq j$), we define $\alpha_{uv}$, $\beta_{uvw}$ and $\beta_{vuw}$ as follows
\begin{itemize}
	\item $\alpha_{uv} = A$,
    \item For $w \in P_{ur} \setminus \{r\}$: $\beta_{uvw} = 0$ and $\beta_{vuw} = A$.
    \item For $w \in P_{rv} \setminus \{r\}$: $\beta_{uvw} = A$ and $\beta_{vuw} = 0$.
    \item For $w = r$: $\beta_{uvr} =\beta_{vur} = B$.
\end{itemize}
For $u\in T_i$ and $v=r$, we define $\alpha_{ur}$, $\beta_{urw}$ and $\beta_{ruw}$ as follows
\begin{itemize}
    \item $\alpha_{ur} = A$.
    \item For $w \in P_{ur} \setminus \{r\}$: $\beta_{urw} = 0$ and $\beta_{ruw} = A$.
    \item For $w = r$: $\beta_{urr} = \beta_{rur} = B$.
\end{itemize}
Finally, we set $\alpha_{rr} = \beta_{rrr} = B$.

We now show that the obtained solution $\{\alpha, \beta\}$ is a feasible solution for DUAL-LP. Consider the first constraint: $\alpha_{uv}\leq \beta_{uvw}+\beta_{vuw}$. If $u\neq r$ or $v\neq r$, $A=\alpha_{uv} = \beta_{uvw} + \beta_{vuw}  = 2B$. The second constraint is satisfied since $\alpha_{rr} =\beta_{rrr}$.

We now verify that the third constraint, $\sum_{v:w\in P_{uv}} \beta_{uvw} \leq 1$, is satisfied. Consider a \textit{non-zero} variable $\beta_{uvw}$ appearing in the sum. Consider the iteration of the algorithm in which we assign $\beta_{uvw}$ a value. Let $r$ be the balanced separator vertex during this iteration. Then, $r\in P_{uv}$ (otherwise, we would not assign any value to $\beta_{uvw}$) and $w\in P_{rv}$. Therefore,  $r\in P_{uw}$; that is, we assign the value to $\beta_{uvw}$ in the iteration when the algorithm splits $u$ and $w$ (the only iteration when $r\in P_{uw}$). In particular, we assign a value to all non-zero variables $\beta_{uvw}$ appearing in the constraint in the same iteration of the algorithm. Let us consider this iteration.

If $u\in T_i\cup \{r\}$ and $w\in T_j$, then every $v$ satisfying $w\in P_{uv}$ lies in $T_j$. For every such $v$, we have $\beta_{uvw} = A$. Therefore, $\sum_{v:w\in P_{uv}} \beta_{uvw} \leq |T_j| \cdot A \leq \frac{n'}{2} \cdot \frac{2}{n'} = 1$, as required. If $u\in T_i\cup\{r\}$ and $w=r$, then we have $\sum_{v: w\in P_{uv}} \beta_{uvw} = \sum_{v: r\in P_{uv}} \beta_{uvr} = \sum_{v: r \in P_{uv}} B \leq n' B   =  1$, as required. We have showed that $\{\alpha,\beta\}$ is a feasible solution. Now we prove that its value is at least half of the value of the hub labeling found by the algorithm. Since the value of any feasible solution of DUAL-LP is at most the cost of the optimal hub labeling, this will prove that the algorithm gives a 2-approximation.

We consider one iteration of the algorithm. In this iteration, we add $r$ to the hub set $H_u$ of every vertex $u\in T'$. Thus, we increase the cost of the hub labeling by $n'$. We are going to show that the dual variables that we set during this iteration contribute at least $n'/2$ to the value of DUAL-LP.

Let $k_i = |T_i| \leq  n' / 2$, for all $i \in \{1,\dots, t\}$. We have $\sum_i k_i  = n'-1$. The contribution $C$ of the variables $\alpha_{uv}$ that we set during this iteration to the objective function equals
\begin{align*}
    C &= \sum_{i<j} \sum_{u\in T_i, v\in T_j} \alpha_{uv} + \sum_{i} \sum_{u\in T_i} \alpha_{ur} + \alpha_{rr} = A\sum_{i<j} k_i k_j + A (n'-1) + B \\
      &= \frac{2}{n'} \sum_{i<j} k_i k_j + \frac{2n'-1}{n'}.
\end{align*}
Now, since $\sum_{j:j\neq i} k_j = (n'-1 - k_i)\geq (n'-2)/2$, we have
\begin{equation*}
    \frac{2}{n'} \sum_{i<j} k_i k_j =\frac{1}{n'} \sum_{i\neq j} k_i k_j = \frac{1}{n'}\sum_{i}k_i \left(\sum_{j:j\neq i} k_j\right) \geq \frac{n'-2}{2n'} \sum_{i} k_i = \frac{(n'-1)(n'-2)}{2n'}.
\end{equation*}
Thus,
\begin{equation*}
    C \geq \frac{(4n'-2) + ((n')^2 -3 n' + 2)}{2n'} = \frac{n'+1}{2}.
\end{equation*}
We proved that $C \geq n'/2$. This concludes the proof.
\end{proof}

Given the simplicity of the Tree Algorithm, it is interesting to understand whether the 2 approximation factor is tight or not. We do not have a matching lower bound, but we show an asymptotic lower bound of $3/2$. The instances that give this $3/2$ lower bound are, somewhat surprisingly, the (very symmetric!) complete binary trees.
\begin{lemma}\label{tree-alg-lower-bound}
The approximation factor of the Tree Algorithm for \HL1 is at least $3/2 - \varepsilon$, for every fixed $\varepsilon > 0$.
\end{lemma}
\begin{proof}
We consider the complete binary tree of height $h$, whose size is $n_h = 2^{h + 1} - 1$ (a single vertex is considered to have height 0). The cost of the Tree Algorithm on a complete binary tree of height $h$, denoted by $ALG(h)$, can be written as
\begin{equation*}
ALG(h) =
\begin{cases}
      (2^{h + 1} - 1) + 2 \cdot ALG(h - 1), & h \geq 1, \\
      1 & h = 0.
   \end{cases}
\end{equation*}
It is easy to see that the above implies that $ALG(h) = 2 \cdot h \cdot 2^h + 1$, for all $h \geq 0$. To obtain a $3/2$ gap, we now present an algorithm that gives a hub labeling of size $(1 + o_h(1)) \cdot \frac{4}{3} \cdot h \cdot 2^h$ on complete binary trees (where the $o_h(1)$ term goes to $0$ as $h\to\infty$).

It will again be a recursive algorithm (i.e.~a hierarchical labeling), only this time the recursion handles complete binary trees that may have some ``tail" at the root. More formally, the algorithm operates on graphs that can be decomposed into two disjoint parts, a complete binary tree of height $h$, and a path of length $t$. The two components are connected with an edge between the root of the binary tree and an endpoint of the path. Such a graph can be fully described by a pair $(h, t)$, where $h$ is the height of the tree and $t$ is the length of the path attached to the root of the tree.

The proposed algorithm for complete binary trees works as follows. Let $p$ be the root of the tree. Assuming $h \geq 2$, let $l$ be the left child of $p$, and $r$ be the right child of $p$. The algorithm picks $l$ as the vertex with the highest rank, and then recurses on the children of $l$ and on $r$. Observe that on the recursive step for $r$, we have a rooted tree on $r$, and the original root $p$ is now considered part of the tail. For $h = 0$, we have a path of length $t + 1$, and for $h = 1$, we simply remove $p$ and then end up with a path of length $t$ and two single vertices.

Let $\Path(t)$ denote the optimal \HL1 cost for a path of length $t$. It is not hard to show that the Tree Algorithm performs optimally for paths, and a closed formula for $\Path$ is
\begin{equation*}
    \Path(t) = (t + 1) \lceil \log (t + 1) \rceil - 2^{\lceil \log (t + 1) \rceil} + 1, \quad t \geq 0.
\end{equation*}
So, at the base cases, the proposed algorithm uses the Tree Algorithm on paths. Let $P(h, t)$ be the cost of this algorithm. Putting everything together, we obtain the recursive formula
\begin{equation*}
P(h,t) =
\begin{cases}
      (2^{h + 1} - 1) + t + 2 \cdot P(h - 2, 0) + P(h - 1, t + 1), & h \geq 2, \;t \geq 0, \\
      5 + t + \Path(t) & h = 1, \;t \geq 0, \\
      \Path(t + 1), & h = 0, \;t \geq 0.
   \end{cases}
\end{equation*}

The cost of the solution we obtain from this algorithm is $P(h, 0)$. Let $f(n) = \Path(n + 1) + 5$, $n \geq 0$, and $g(h) = C / \sqrt{h}$, $h \geq 0$, for some appropriate constant $C$. We prove by induction on $h$ that
\begin{equation*}
P(h, t) \leq \frac{4}{3} \cdot h \cdot 2^h + g(h) \cdot h \cdot 2^h + f(h + t) + h \cdot t, \quad\forall h \geq 0, t \geq 0.
\end{equation*}

The cases with $h = 0$ and $h = 1$ are obvious. If $h \geq 2$, then
\begin{equation*}
P(h, t) = 2^{h + 1} - 1 + t + 2 \cdot P(h - 2, 0) + P(h - 1, t + 1).
\end{equation*}
By the induction hypothesis, we have that
\begin{align*}
2 \cdot P(h - 2, 0) & \leq \frac{1}{2} \cdot \frac{4}{3} h \cdot 2^h - \frac{4}{3} 2^h + \frac{1}{2} g(h - 2) \cdot (h - 2) 2^h + 2 \cdot f(h - 2), \quad \textrm{and}\\
P(h - 1, t + 1) & \leq \frac{1}{2} \cdot \frac{4}{3} h \cdot 2^h - \frac{2}{3} 2^h + \frac{1}{2} g(h-1) \cdot (h - 1) 2^h + f(h + t) + (h - 1)\cdot(t + 1).
\end{align*}

\noindent Thus, we obtain
\begin{multline*}
    P(h, t) \leq \frac{4}{3} \cdot h \cdot 2^h +\left(\frac{g(h-2) \cdot(h-2) + g(h - 1) \cdot (h - 1)}{2} + \frac{h + 2f(h-2) - 2}{2^h} \right) 2^h + \\ + f(h + t) + h \cdot t.
\end{multline*}

\noindent Choosing the right constant $C$, we can show that for all $h \geq 2$, we have
\begin{equation*}
\frac{g(h-2) \cdot(h-2) + g(h - 1) \cdot (h - 1)}{2} + \frac{h + 2f(h-2) - 2}{2^h} \leq h \cdot g(h),
\end{equation*}
and so the inductive step is true. This means that
\begin{equation*}
P(h,0) \leq \frac{4}{3} \cdot h \cdot 2^h \cdot \left(1 + \frac{3}{4} g(h) + \frac{3 f(h)}{4h \cdot 2^h} \right) = \left(1 + o_h(1)\right) \cdot  \frac{4}{3} \cdot h \cdot 2^h,
\end{equation*}
(where the $o_h(1)$ term goes to $0$ as $h\to\infty$)
and so, for any $\varepsilon > 0$ there are instances where $\frac{ALG}{OPT} \geq \frac{3}{2} - \varepsilon$.
\end{proof}

\paragraph{Performance of the Tree Algorithm for \HLp on trees.} The Tree Algorithm does not find a good approximation for the $\ell_p$-cost, when $p$ is large. Let $k>1$ be an integer. Consider a tree $T$ defined as follows: it consists of a path $a_1,\dots, a_k$ and leaf vertices connected to the vertices of the path; vertex $a_i$ is connected to $2^{k-i}-1$ leaves. The tree has $n=2^k-1$ vertices. It is easy to see that the Tree Algorithm will first choose vertex $a_1$, then it will process the subtree of $T$ that contains $a_k$ and will choose $a_2$, then $a_3$ and so on. Consequently, the hub set $H_{a_k}$ equals $\{a_1,\dots, a_k\}$ in the obtained hub labeling. The $\ell_p$-cost of this hub labeling is greater than $k$. However, there is a hub labeling $\widetilde H$ for the path $a_1,\dots, a_k$ with $|\widetilde H_{a_i}| \leq O(\log k)$, for all $i \in [k]$. This hub labeling can be extended to a hub labeling of $T$, by letting $\widetilde H_{l} = \widetilde H_{a_i} \cup \{l\}$ for each leaf $l$ adjacent to a vertex $a_i$. Then we still have $|\widetilde H_u| \leq O(\log k)$, for any vertex $u \in T$. The $\ell_p$-cost of this solution is $O(n^{1/p} \log k)$. Thus, for $k=p$, the gap between the solution $H$ and the optimal solution is at least $\Omega(p/\log p)$. For $p=\infty$, the gap is at least $\Omega(\log n/\log\log n)$, asymptotically.


\section{A PTAS for HL on trees}\label{sec:ptas-trees}

\subsection{A PTAS for \HL1 on trees}\label{sec:ptas-l1-trees}

We now present a polynomial-time approximation scheme (PTAS) for HL on trees. We first present the algorithm for \HL1, based on dynamic programming (DP), and then we slightly modify the DP and show that it can work for \HLp, for any $p \in [1, \infty]$.

Let $T = (V,E)$ be any tree. The starting point is Theorem~\ref{HHL_optimal_for_trees}, which shows that we can restrict our attention to hierarchical hub labelings. That is, we can find an optimal solution by choosing an appropriate vertex $r$, adding $r$ to every $H_u$, and then recursively solving HL on each connected component of $T - r$ (see Section~\ref{hhl-section}). Of course, we do not know what vertex $r$ we should choose, so to implement this approach, we use dynamic programming (DP). Let us first consider a very basic dynamic programming solution. We store a table $B$ with an entry $B[T']$ for every subtree $T'$ of $T$. Entry $B[T']$ equals the cost of the optimal hub labeling for tree $T'$. Now if $r$ is the common  hub of all vertices in $T'$, we have
\begin{equation*}
    B[T'] = |T'|+ \sum_{T'' \text{ is c.c. of } T'-r} B[T'']
\end{equation*}
(the term $|T'|$ captures the cost of adding $r$ to each $H_u$). Here, ``c.c." is an abbreviation for ``connected component". We obtain the following recurrence formula for the DP:
\begin{equation}\label{eq:DP-main}
B[T'] = |T'|+ \min_{r \in T'} \sum_{T'' \text{ is c.c. of } T'-r} B[T''].
\end{equation}
The problem with this approach, however, is that a tree may have exponentially many subtrees, which means that the size of the dynamic program and the running time may be exponential.

To work around this, we will instead store $B[T']$ only for some subtrees $T'$, specifically for subtrees with a ``small boundary''. For each subtree $T'$ of $T$, we define its boundary $\partial(T')$ as $\partial(T'): = \{v \notin T': \exists u \in T' \textrm{ with } (u,v) \in E\}$. Consider now a subtree $T'$ of $T$ and its boundary $S= \partial(T')$. Observe that if $|S| \geq 2$, then the set $S$ uniquely identifies the subtree $T'$: $T'$ is the unique connected component of $T-S$ that has all vertices from $S$ on its boundary (every other connected component of  $T-S$ has only one vertex from $S$ on its boundary). If $|S|=1$, that is, $S = \{u\}$ for some $u \in V$, then it is easy to see that $u$ can serve as a boundary point for $\deg(u)$ different subtrees.

Fix $\varepsilon < 1$. Let $k = 4 \cdot \lceil 1 / \varepsilon \rceil$. In our dynamic program, we only consider subtrees $T'$ with $|\partial(T')| \leq k$. Then, the total number of entries is upper bounded by $\sum_{i = 2}^k \binom{n}{i} + \sum_{u\in V} \deg(u) = O(n^k)$. Note that now we cannot simply use formula~(\ref{eq:DP-main}). In fact, if $|\partial(T')| < k$, formula~(\ref{eq:DP-main}) is well defined since each connected component $T''$ of $T' -r$ has boundary of size at most $|\partial(T')| + 1 \leq k$ for any choice of $r$ (since $\partial(T'') \subseteq \partial(T') \cup \{r\}$). However, if $|\partial(T')| =k$, it is possible that $|\partial(T'')| = k+1$, and formula~(\ref{eq:DP-main}) cannot be used. Accordingly, there is no clear way to find the optimal vertex $r$. Instead, we choose a vertex $r_0$ such that for every connected component $T''$ of $T'-r_0$, we have $|\partial(T'')|\leq k/2+1$. To prove that such a vertex exists, we consider the tree $T'$ with vertex weights $w(u) = \left|\{v\in \partial(T'): (u,v)\in E\} \right|$ and find a balanced separator vertex $r_0$ of $T'$ w.r.t.~weights $w(u)$ (see Definition~\ref{def:bal-sep-vertex}). Then, the weight $w$ of every connected component $T''$ of $T'-r_0$ is at most $k/2$. Thus, $|\partial(T'')| \leq k/2+1 < 3k/4 <k$ (we add 1 because $r_0\in \partial(T'')$).

The above description implies that the only cases where our algorithm does not perform ``optimally" are the subproblems $T'$ with $|\partial(T')| = k$. It is also clear that these subproblems cannot occur too often, and more precisely, we can have at most 1 every $k / 2$ steps during the recursive decomposition into subproblems. Thus, we will distribute the cost (amortization) of each such non-optimal step that the algorithm makes over the previous $k/4$ steps before it occurs, whenever it occurs, and then show that all subproblems with boundary of size at most $3k/4$ are solved ``almost" optimally (more precisely, the solution to such a subproblem is $(1 + 4 / k)$-approximately optimal). This implies that the final solution will also be $(1 + 4/k)$-approximately optimal, since its boundary size is 0.

We now describe our algorithm in more detail. We keep two tables $B[T']$ and $C[T']$. We will define their values so that we can find, using dynamic programming, a hub labeling for $T'$ of cost at most $B[T'] + C[T']$. Informally, the table $C$ can be viewed as some extra budget that we use in order to pay for all the recursive steps with $|\partial(T')| = k$. For every $T'$ with $|\partial(T')| \leq k$, we define $C[T']$ as follows:
\begin{equation*}
    C[T'] = \max \left\{0, \left(|\partial(T')| - 3k/4 \right) \cdot 4|T'| / k \right\}.
\end{equation*}
We define $B$ (for $|T'|\geq 3$) by the following recurrence (where $r_0$ is a balanced separator):
\begin{equation*}
    B[T'] = \begin{cases}
             (1 + 4 / k) \cdot|T'|+ \min_{r \in T'} \sum_{T'' \text{ is c.c. of } T'-r} B[T''], & \text{ if } |\partial(T')| < k,\\
            \sum_{T'' \text{ is c.c. of } T'-r_0} B[T''], & \text{ if } |\partial(T')| = k.
\end{cases}
\end{equation*}
The base cases of our recursive formulas are when the subtree $T'$ is of size 1 or 2. In this case, we simply set $B[T'] = 1$, if $|T'| = 1$, and $B[T'] = 3$, if $|T'| = 2$.

In order to fill in the table, we generate all possible subsets of size at most $k$ that are candidate boundary sets, and for each such set we look at the resulting subtree, if any, and compute its size. We process subtrees in increasing order of size, which can be easily done if the generated subtrees are kept in buckets according to their size. Overall, the running time will be $n^{O(k)}$.

We will now show that the algorithm has approximation factor $(1 + 4/k)$ for any $k = 4t$, $t \geq 1$.

\begin{theorem}
The algorithm is a polynomial-time approximation scheme (PTAS) for \HL1 on trees.
\end{theorem}
\begin{proof}
We first argue about the approximation guarantee. The argument consists of an induction that incorporates the amortized analysis that was described above. More specifically, we will show that for any subtree $T'$, with $|\partial(T')| \leq k$, the total cost of the algorithm's solution is at most $B[T'] + C[T']$, and $B[T'] \leq \left(1 + \frac{4}{k} \right) \cdot OPT_{T'}$. Then, the total cost of the solution to the original HL instance is at most $B[T] + C[T]$, and, since $C[T] = 0$, we get that the cost is at most $(1 + 4/k) \cdot OPT$.

The induction is on the size of the subtree $T'$. For $|T'| = 1$ or $|T'| = 2$, the hypothesis holds. Let's assume now that it holds for all trees of size at most $t \geq 2$. We will argue that it then holds for trees $T'$ of size $t + 1$. We distinguish between the cases where $|\partial(T')| < k$ and $|\partial(T')| = k$.

\medskip
\noindent\textbf{Case ${|\partial(T')| < k}$:} Let $u_0 \in T'$ be the vertex that the algorithm picks and removes. The vertex $u_0$ is the minimizer of the expression $\min_{r' \in T'} \sum_{T'' \textrm{ is c.c of } T'-r'} B[T'']$, and thus, using the induction hypothesis, we get that the total cost of the solution returned by the algorithm is at most:

\begin{equation*}
\begin{split}
    ALG(T') &\leq |T'| + \sum_{\substack{T'' \textrm{ is c.c.}\\\textrm{of }T' - u_0}} \Big( B[T''] + C[T''] \Big) \\
            &\leq |T'| + \sum_{\substack{T'' \textrm{ is c.c.}\\\textrm{of }T' - u_0}} B[T''] + \sum_{\substack{T'' \textrm{ is c.c.}\\\textrm{of }T' - u_0}} \max \left\{0, (|\partial(T')| + 1 - 3k/4) \cdot 4|T''|/k \right\} \\
            &= |T'| + \sum_{\substack{T'' \textrm{ is c.c.}\\\textrm{of }T' - u_0}} B[T''] + \max\{0, |\partial(T')| + 1 - 3k/4\} \cdot (4/k) \cdot \sum_{\substack{T'' \textrm{ is c.c.}\\\textrm{of }T' - u_0}}|T''| \\
            &\leq |T'| + \sum_{\substack{T'' \textrm{ is c.c.}\\\textrm{of }T' - u_0}} B[T''] + 4|T'|/k + \max\{0, |\partial(T')| - 3k/4 \} \cdot 4|T'| / k\\
            &\leq (1 + 4/k) \cdot |T'| + \left(\sum_{T'' \textrm{ is c.c.~of }T' - u_0} B[T''] \right) + C[T'] \\
            &= B[T'] + C[T'].
\end{split}
\end{equation*}

We proved the first part. We now have to show that $B[T'] \leq (1 + 4/k) OPT_{T'}$. Consider an optimal HL for $T'$. By Theorem~\ref{HHL_optimal_for_trees}, we may assume that it is a hierarchical labeling. Let $r \in T'$ be the vertex with the highest rank in this optimal solution. We have
\begin{equation*}
    OPT_{T'} = |T'| + \sum_{T'' \textrm{ is comp. of }T' - r} OPT_{T''}.
\end{equation*}
By definition, we have that
\begin{equation*}
\begin{split}
    B[T'] &= (1 + 4/k) \cdot |T'| + \min_{u \in T'} \sum_{T'' \textrm{ is c.c. of }T' -u} B[T''] \leq (1 + 4/k) \cdot |T'| + \sum_{T'' \textrm{ is c.c. of }T' -r} B[T''] \\
          &\stackrel{(ind.hyp.)}{\leq} (1 + 4/k) \cdot |T'| + (1 + 4/k) \cdot \sum_{T'' \textrm{ is c.c. of }T' -r} OPT_{T''} = (1 + 4/k) \cdot OPT_{T'}.
\end{split}
\end{equation*}

\smallskip
\noindent\textbf{Case $|\partial(T')| = k$:}
 Using the induction hypothesis, we get that the total cost of the solution returned by the algorithm is at most:
\begin{equation*}
    ALG(T') \leq |T'| + \sum_{T'' \textrm{ is c.c. of }T' -r_0} B[T''] + \sum_{T'' \textrm{ is c.c. of }T' - r_0} C[T''].
\end{equation*}
By our choice of $r_0$, we have $|\partial(T'')| \leq 3k/4$, and so $C[T''] = 0$, for all trees $T''$ of the forest $T' - {r_0}$. Thus,
\begin{equation*}
    ALG(T') \leq |T'| + \sum_{T'' \textrm{ is c.c. of }T' - {r_0}} B[T''] = C[T'] + B[T'].
\end{equation*}
We now need to prove that $B[T'] \leq (1 + 4/k) \cdot OPT_{T'}$. We have
\begin{equation*}
    B[T'] = \sum_{T'' \textrm{ is c.c. of }T'- {r_0}} B [T''] \stackrel{(ind.hyp.)}{\leq} \sum_{T'' \textrm{ is c.c. of }T' -r_0} \left(1 + \frac{4}{k} \right) OPT_{T''} \leq \left(1 + \frac{4}{k} \right) OPT_{T'},
\end{equation*}
where in the last inequality we use that $\sum_{T''} OPT_{T''}\leq OPT_{T'}$, which can be proved as follows. We convert the optimal hub labeling $H'$ for $T'$ to a set of hub labelings for all subtrees $T''$ of $T' - r_0$: the hub labeling $H''$ for $T''$ is the restriction of $H'$ to $T''$; namely, $H''_v = H'_v \cap V(T'')$ for every vertex $v\in T''$; it is clear that the total number of hubs in labelings $H''$ for all subtrees $T''$ is at most the cost of $H'$. Also, the cost of each hub labeling $H''$ is at least $OPT_{T''}$. The inequality follows.

We have considered both cases, $|S|<k$ and $|S|=k$, and thus shown that the hypothesis holds for any subtree $T'$ of $T$. In particular, it holds for $T$. Therefore, the algorithm finds a solution of cost at most $B[T] + C[T] = B[T] \leq \left(1 + \frac{4}{k} \right) OPT$.

Setting $k = 4 \cdot \lceil 1 / \varepsilon \rceil$, as already mentioned, we get a $(1 + \varepsilon)$-approximation, for any fixed $\varepsilon \in (0,1)$, and the running time of the algorithm is $n^{O(1 / \varepsilon)}$.
\end{proof}

\subsection{A PTAS for \HLp on trees}\label{sec:ptas_lp}

In this section, we describe a polynomial-time approximation scheme (PTAS) for \HLp for arbitrary $p \in [1,\infty)$. Our algorithm is a modification of the dynamic programming algorithm for \HL1. The main difficulty that we have to deal with is that the $\ell_p$-cost of an instance cannot be expressed in terms of the $\ell_p$-cost of the subproblems, since it might happen that suboptimal solutions for its subproblems give an optimal solution for the instance itself. Thus, it is not enough to store only the cost of the ``optimal" solution for each subproblem.

Let
\begin{equation*}
    OPT[T',t]^p = \min_{H\textrm{ is an HHL for } T'} \sum_{u \in T'} (|H_u| + t)^p.
\end{equation*}
 Clearly, $OPT[T,0]^p$ is the cost of an optimal \HLp solution for $T$, raised to the power $p$. Observe that $OPT[T', t]^p$ satisfies the following recurrence relation:
\begin{equation}\label{eq:OPT-lp-shifted}
OPT[T',t]^p = (1 + t)^p + \min_{r \in T'} \sum_{T'' \textrm{ is a c.c.~of }T' - r} OPT[T'', t + 1]^p.
\end{equation}
Indeed, let $\widetilde H$ be an HHL for $T'$ that minimizes $\sum_{u \in T'} (|\widetilde H_u| + t)^p$. Let $r'$ be the highest ranked vertex in $T'$ w.r.t.~the ordering defined by $\widetilde H$. For each tree $T''$ in the forest $T'-r'$, consider the hub labeling $\{\widetilde H_u \cap T''\}_{u\in T''} = \widetilde H_u - r'$. Since $|\widetilde H_u| = |\widetilde H_u \cap T''| +1$, we have
\begin{equation*}
    \sum_{u \in T''} (|\widetilde H_u| + t)^p = \sum_{u \in T''} (|\widetilde H_u \cap T''| + t+1)^p \geq OPT[T'', t+1]^p.
\end{equation*}
Also, $|\widetilde H_{r'} | =1$. Therefore,
\begin{equation*}
    OPT[T',t]^p = \sum_{u\in T'} (|\widetilde H_u|+t)^p = (|\widetilde H_{r'}|+t)^p + \sum_{T''}\sum_{u\in T''} (|\widetilde H_u|+t)^p \geq (1 + t)^p + \sum_{T''} OPT[T'', t + 1]^p.
\end{equation*}
The proof of the inequality in the other direction is similar. Consider $r$ that minimizes the expression on the right hand side of (\ref{eq:OPT-lp-shifted}) and optimal HHLs for subtrees $T''$ of $T'-r$. We combine these HHLs and obtain a feasible HHL $\widetilde H$. We get
\begin{equation*}
    OPT[T',t]^p \leq (1+t)^p + \sum_{u \in T'} (|\widetilde H_u| + t + 1)^p = (1 + t)^p + \sum_{T''} OPT[T'', t + 1]^p.
\end{equation*}
This concludes the proof of the recurrence.

If we were not concerned about the running time of the algorithm, we could have used this recursive formula for $OPT[T',t]^p$ to find the exact solution (the running time would be exponential).  In order to get a polynomial-time algorithm, we again consider only subtrees $T'$ with boundary of size  at most $k$. We consider the cases when $|\partial(T')| < k$ and when $|\partial(T')| = k$. In the former case, we use formula (\ref{eq:OPT-lp-shifted}). In the latter case, when $|\partial(T')| = k$, we perform the same step as the one performed in the algorithm for \HL1: we pick a weighted balanced separator vertex $r_0$ of $T'$ such that $|\partial(T'')| \leq k/2 +1$ for every subtree $T''$ of $T'-r_0$. Formally, we define a dynamic programming table $B[T',t]$ as follows:
\begin{equation*}
    B[T',t] = 
        \begin{cases}
            (1 + t)^p + \min_{r \in T'} \sum_{T'' \text{ is a c.c.~of } T'-r} B[T'',t+1], &\text{ if } |\partial(T')| < k,\\
            (1 + t)^p + \sum_{T'' \text{ is a c.c.~of } T'-r_0} B[T'',t], &\text{ if } |\partial(T')| = k.
        \end{cases}
\end{equation*}
The base cases of our recursive formulas are when the subtree $T'$ is of size 1 or 2. In this case, we simply set $B[T',t] = (1 + t)^p$, if $|T'| = 1$, and $B[T',t] = (1 + t)^p + (2 + t)^p$, if $|T'| = 2$. We will need the following two claims.

\begin{claim}\label{superadditive}
For any tree $T$ and a partition of $T$ into disjoint subtrees $\{T_1, ..., T_j\}$ such that $\bigcup_{i = 1}^j T_i = T$, we have
\begin{equation*}
\sum_{i = 1}^j OPT[T_i,t]^p \leq OPT[T,t]^p.
\end{equation*}
\end{claim}
\begin{proof}
Consider an optimal hierarchical solution $H$ for the \HLp problem defined by $(T,t)$. Define $H^{(i)} = \{H_u \cap T_i: u \in T_i\}$. Observe that $\{H_u^{(i)}\}_{u \in T_i}$ is a feasible hub labeling for $T_i$, since the original instance is a tree. We have 
\begin{equation*}
    OPT[T,t]^p = \sum_{i = 1}^j \sum_{u \in T_i} (|H_u| + t)^p \geq \sum_{i = 1}^j \sum_{u \in T_i} (|H_u^{(i)}| + t)^p \geq \sum_{i = 1}^j OPT[T_i,t]^p.
\end{equation*}
\end{proof}

\begin{claim}\label{lower_bound_claim_ptas_p_norm}
For any $T'$ and $t \geq 0$, $B[T',t] \leq OPT[T',t]^p$.
\end{claim}
\begin{proof}
We do induction on the size of $T'$. For $|T'| \in \{1,2\}$, the claim holds trivially for all $t \geq 0$. Let us assume that it holds for all subtrees of size at most $s$ and for all $t \geq 0$. We will prove that it holds for all subtrees of size $s+ 1$ and for all $t \geq 0$. We  again consider two cases.

\medskip
\noindent\textbf{Case ${|\partial(T')| < k}$:}
\begin{equation*}
\begin{split}
    B[T', t] &= (1 + t)^p + \min_{r \in T'} \sum_{T'' \text{ is c.c.~of } T'-r} B[T'',t+1] \\
             &\leq (1 + t)^p + \min_{r \in T'} \sum_{T'' \text{ is c.c.~of } T'-r} OPT[T'', t + 1]^p \quad\quad (\textrm{by ind.~hyp.})\\
             &= OPT[T', t]^p.
\end{split}
\end{equation*}

\smallskip
\noindent\textbf{Case $|\partial(T')| = k$:}
\begin{equation*}
\begin{split}
    B[T', t] &= (1 + t)^p + \sum_{T'' \text{ is c.c.~of } T'-r_0} B[T'',t] \\
             &\leq (1 + t)^p + \sum_{T'' \text{ is c.c.~of } T'- r_0} OPT[T'', t]^p \quad\quad (\textrm{by ind.~hyp.})\\
             &= OPT[\{r_0\},t]^p + \sum_{T'' \text{ is c.c.~of } T'- r_0} OPT[T'', t]^p \\
             &\leq OPT[T', t]^p,
\end{split}
\end{equation*}
where the last inequality follows from Claim \ref{superadditive} and the fact that the connected components of $T'-r_0$ together with $\{r_0\}$ form a partition of $T'$.
\end{proof}

\begin{theorem}
There is a polynomial-time approximation scheme (PTAS) for \HLp for every $p \in [1,\infty)$. The algorithm finds a $(1+\varepsilon)$ approximate solution in time $n^{O(1/\varepsilon)}$; the running time does not depend on $p$.
\end{theorem}
\begin{proof}
Fix $\varepsilon < 1$, and set $k = 2 \cdot \lceil 4/\varepsilon \rceil$. Let $H$ be the solution returned by the dynamic programming algorithm presented in this section. Consider the set $X$ of all weighted balanced separators that the algorithm uses during its execution; that is, $X$ is the set of hubs $r_0$ that the algorithm adds when it processes trees $T'$ with $|\partial(T')| = k$.
 
Let $\widetilde H_u = (H_u \setminus X) \cup \{u\}$; the set $\widetilde H_u$ consists of the hubs added to $H_u$ during the steps when $\partial(T') < k$, with the exception that we include $u$ in $\widetilde H_u$ even if $u\in X$. 
It is easy to prove by induction (along the lines of the previous inductive proofs) that 
\begin{equation*}
    B[T',t] = \sum_{u \in T'} \left(|\widetilde H_u \cap T'| + t \right)^p.
\end{equation*}
Therefore, $B[T,0] = \sum_{u \in V} |\widetilde H_u |^p$.

Now, consider a vertex $u$ and its hub set $H_u$. We want to estimate the ratio $|H_u\cap X| / |H_u|$. We look at the decomposition tree implied by the algorithm and find the subinstance $T'$ in which the algorithm picked $u$ as the highest ranked vertex in $T'$. The path from the root of the decomposition tree to that particular subinstance $T'$ contains exactly $|H_u|$ nodes. Observe that in any such path, the nodes of the path that correspond to subinstances with boundary size exactly $k$ are at distance at least $k / 2$ from each other (since the size of the boundary increases by at most 1 when we move from one node to the consecutive node along the path). Thus, there can be at most $2|H_u|/k$ such nodes. This means that $|H_u \cap X| \leq 2|H_u| / k$, which gives  $|H_u| \leq (1 + \frac{2}{k - 2}) \cdot |H_u \setminus X|\leq (1 + \frac{2}{k - 2}) \cdot |\widetilde H_u|$. So, the $\ell_p$-cost of the hub labeling is
\begin{align*}
    \|H\|_p &= \Bigl(\sum_{u \in V} |H_u|^p \Bigr)^{1/p}\leq \Bigl(1 + \frac{2}{k - 2} \Bigr) \cdot \Bigl(\sum_{u \in V} |\widetilde H_u|^p \Bigr)^{1/p}\\
            &= \Bigl(1 + \frac{2}{k - 2} \Bigr) \cdot B[T,0]^{1/p} \leq \Bigl(1 + \frac{2}{k - 2} \Bigr) \cdot OPT[T,0],
\end{align*}
where the last inequality follows from Claim \ref{lower_bound_claim_ptas_p_norm}. We get that the algorithm finds a hub labeling of $\ell_p$-cost at most $(1 + \frac{2}{k - 2}) \cdot OPT$. The running time is $n^{O(k)} \cdot n = n^{O(k)}$.
\end{proof}

\subsection{A PTAS for \HL\infty on trees}\label{appendix_ptas_infty}

Our approach for \HL1 (see Section~\ref{sec:ptas-l1-trees}) works almost as is for \HL\infty as well. The only modifications that we need to make are the following:
\begin{itemize}
    \item $B[T']$ is now defined as
\begin{equation*}
    B[T'] =
        \begin{cases}
            (1 + 4/k) + \min_{r' \in T'} \;\;\max_{T'' \textrm{ is c.c.~of } T' - r'} B[T''],  &\textrm{if }|\partial(T')| < k, \\
            \max_{T'' \textrm{ is c.c.~of } T' - r_0} B[T''],  &\textrm{if }|\partial(T')| = k,
        \end{cases}
\end{equation*}
where $r_0$ is the weighted balanced separator vertex of $T'$, as defined in the description of the algorithm for \HL1.
    \item $C[T']$ is now equal to $C[T'] = \max \left\{ 0, \left(|\partial(T')| - 3k/4\right) \cdot \frac{4}{k} \right\}$.
\end{itemize}
We can again prove using induction (along the same lines as the proof for \HL1) that the total cost of the solution that the algorithm returns at any subinstance $T'$ is at most $B[T'] + C[T']$, and that it always holds that $B[T'] \leq (1 + 4/k) \cdot OPT_{T'}$. Thus, for $T' = T$ we have $C[T] = 0$, and so we obtain a solution of cost at most $(1 + 4/k) \cdot OPT$ in time $n^{O(k)}$.

\section{Bounds on the size of the largest hub set in optimal solutions}\label{sec:bounds-on-size}

In this section, we give upper bounds on the size of the largest hub set in an optimal HHL solution in a tree. As Theorem~\ref{HHL_optimal_for_trees} guarantees that for any \HLp there is always an optimal solution that is hierarchical, such bounds translate to bounds on the $\ell_\infty$-norm of an optimal solution for \HLp. These will prove very useful for the design of exact algorithms for \HLp.

We start by making a very simple observation, namely that the Tree Algorithm (see Algorithm~\ref{Tree-Algorithm}) gives a feasible hub labeling $H$ that always satisfies $\|H\|_\infty = O(\log n)$, where $n$ is the size of the tree. More precisely, we get the following theorem.
\begin{theorem}\label{thm:hl-infinity-max-bound}
Let $T = (V,E)$, $|V| = n$, be an instance of \HL\infty, and let $H$ denote an optimal (w.r.t.~the $\ell_\infty$-cost) solution. Then, $\|H\|_\infty = \max_{u \in V} |H_u| \leq \log n + 1 \leq 2\log n = O(\log n)$.
\end{theorem}
The proof is very straighforward, and thus omitted. Since the $\ell_\infty$-norm of an $n$-dimensional vector is within a constant factor to the $\ell_{\log n}$-norm, we also immediately conclude the following.
\begin{theorem}\label{thm:hl-logn-max-bound}
Fix some constant $\varepsilon > 0$. Let $T = (V,E)$, $|V| = n$, be an instance of \HLp, for $p \geq \varepsilon \log n$, and let $H$ denote an optimal (w.r.t.~the $\ell_p$-cost) solution. Then, $\|H\|_\infty = \max_{u \in V} |H_u| = O_{\varepsilon}(\log n)$.
\end{theorem}
\begin{proof}
We again consider the solution $H'$ that the Tree Algorithm (see Algorithm~\ref{Tree-Algorithm}) produces. We have $\|H'\|_p \leq n^{1/p} \cdot 2 \log n \leq 2^{1 + 1/\varepsilon} \cdot \log n$. Let $H$ denote an optimal (w.r.t.~the $\ell_p$-cost) solution for $T$. We have $\|H\|_\infty \leq \|H\|_p \leq \|H'\|_p \leq 2^{1 + 1/\varepsilon} \cdot \log n$. Thus, we conclude that for constant $\varepsilon > 0$, we always have $\|H\|_\infty = O(\log n)$.
\end{proof}

We now turn to the case of $p \in [1, \varepsilon \log n)$ and prove the following theorem.
\begin{theorem}\label{thm:hub-set-bound-small-p}
Fix some constant $\varepsilon < 0$. Let $T = (V,E)$, $|V| = n$, be an instance of \HLp, for some $p < \varepsilon \log n$. Let $H$ denote an optimal (w.r.t.~the $\ell_p$-cost) HHL solution. Then, $\|H\|_\infty = \max_{u \in V} |H_u| = O(p \cdot \log n)$. In particular, if $p$ is constant, then we have $\|H\|_\infty = O(\log n)$.
\end{theorem}

Before proving the theorem, we first prove some useful intermediate lemmas.

\begin{lemma}\label{intermediate-lemma}
Let $T= (V, E)$, $|V| = n$, be an instance of \HLp, for some $p \geq 1$, and let $H$ be an optimal HHL solution. Let $T' = (V', E')$, $n' = |V'|$, be a subproblem occuring after $p \cdot \log n$ recursive steps, with $n' > 1$ (if any such problem exists). Then, $\sum_{u \in V'} |H_u|^p \leq n' \cdot (p + 2)^p \log^{p} n$.
\end{lemma}
\begin{proof}
If we modify $H$, and after the first $p \cdot \log n$ recursive steps we switch and use the Tree Algorithm (see Algorithm~\ref{Tree-Algorithm}) on the tree $T'$, we get a hub labeling where the total contribution of the vertices of $T'$ to the objective value (raised to the power $p$) is at most $n' \cdot (p  \log n + 2\log n')^p \leq n' \cdot (p + 2)^p \cdot \log^p n$. Thus, $\sum_{u \in V'} |H_u|^p \leq n' \cdot (p + 2)^p \log^p n$.
\end{proof}

\begin{lemma}\label{lemma:increase-by-k}
Let $T= (V, E)$, $|V| = n$, be an instance of \HLp, for some $p \geq 1$, and let $H$ be an optimal HHL solution. Let $T' = (V', E')$, $n' = |V'|$, be a subproblem occuring after $p \cdot \log n$ recursive steps, with $n' > 1$ (if any such problem exists). Let $S \subseteq V'$. Then, for every positive integer $k \leq \log n$, we have
\begin{equation*}
    \sum_{u \in S} (|H_u| + k)^p  - \sum_{u \in S} |H_u|^p \leq e \cdot (p + 2)^p \cdot n' \cdot k \cdot \log^{p-1} n.
\end{equation*}
\end{lemma}
\begin{proof}
Let $x, t \geq 1$. By the mean value theorem, for some $y \in (x, x + t)$ we have $\frac{(x+t)^p - x^p}{t} = p \cdot y^{p-1} \leq p (x + t)^{p-1}$. Moreover, we have $\frac{(x + t)^{p - 1}}{x^{p - 1}} = \left(1 + \frac{t}{x} \right)^{p - 1} \leq e^{t(p-1)/x}$. Thus, we conclude that $(x + t)^p - x^p \leq e^{t(p - 1) / x} \cdot tp \cdot x^{p - 1}$. This means that for every $u \in V'$ we have
\begin{align*}
    (|H_u| + k)^p - |H_u|^p &\leq e^{k(p - 1)/|H_u|} \cdot kp \cdot |H_u|^{p - 1} \leq e^{k(p - 1)/(p \cdot \log n)} \cdot kp \cdot |H_u|^{p - 1}\\
                           &\leq e^{k/\log n} \cdot kp \cdot |H_u|^{p - 1} \leq e \cdot kp \cdot |H_u|^{p - 1},
\end{align*}
since $|H_u| \geq p \cdot \log n$ and $k \leq \log n$. The above inequality and Lemma~\ref{intermediate-lemma} now imply that
\begin{align*}
    \sum_{u \in S} (|H_u| + k)^p  - \sum_{u \in S} |H_u|^p &\leq e k p \cdot \sum_{u \in S} |H_u|^{p-1} \leq e k p \cdot \sum_{u \in V'} \frac{|H_u|^p}{|H_u|} \\
                                                           &\leq e k p \cdot \frac{n' (p + 2)^p \log^p n}{\min_{u \in V'} |H_u|} \leq e (p + 2)^p \cdot n' k \cdot \log^{p-1} n.
\end{align*}
\end{proof}

\begin{lemma}\label{lemma:reduce-by-l}
Let $T= (V, E)$, $|V| = n$, be an instance of \HLp, for some $p \geq 1$, and let $H$ be an optimal HHL solution. Let $T' = (V', E')$, $n' = |V'|$, be a subproblem occuring after $p \cdot \log n$ recursive steps, with $n' > 1$ (if any such problem exists). Let $S \subseteq V'$. Then, for every positive integer $l \leq \log n$, we have
\begin{equation*}
    \sum_{u \in S} |H_u|^p  - \sum_{u \in S} (|H_u| - l)^p \geq \frac{l}{e} \cdot |S| \cdot p^p \log^{p - 1}n.
\end{equation*}
\end{lemma}
\begin{proof}
For $x > l \geq 1$, by the mean value theorem, we get that $x^p - (x - l)^p \geq lp (x - l)^{p - 1}$. We also have $\frac{x^{p - 1}}{(x - l)^{p - 1}} = \left(1 + \frac{l}{x - l}\right)^{p - 1} \leq e^{(p - 1)l/(x - l)}$. Thus, for $u \in V'$, we have
\begin{equation*}
    \frac{|H_u|^{p - 1}}{(|H_u| - l)^{p - 1}} \leq e^{(p - 1)l/(|H_u| - l)} \leq e^{(p - 1)l / (p \log n - l)} \leq e^{l / \log n} \leq e,
\end{equation*}
since $l \leq \log n$. This implies that $|H_u|^p - (|H_u| - l)^p \geq \frac{lp}{e} \cdot |H_u|^{p - 1}$, and so we get
\begin{equation*}
    \sum_{u \in S} |H_u|^p  - \sum_{u \in S} (|H_u| - l)^p \geq \frac{lp}{e} \cdot \sum_{u \in S} |H_u|^{p - 1} \geq \frac{l}{e} \cdot |S| \cdot p^p \log^{p - 1}n.
\end{equation*}
\end{proof}

\begin{proof}[Proof of Theorem~\ref{thm:hub-set-bound-small-p}]
Let $H$ denote an optimal HHL (w.r.t.~the $\ell_p$-cost), and let's assume that there exists $u \in V$ such that $|H_u| > (p + l) \cdot \log n$, for some constant $l$ that will be specified later. We will transform $H$ into an HHL $H'$ such that $|H_u'| \leq (p + l) \cdot \log n$ for every $u \in V$ and $\|H'\|_p < \|H\|_p$.

Let $h_u = |H_u|$ for every $u \in V$. Let's consider the $(p \cdot \log n)$-th level of the decomposition tree, and let $T'$ be a connected subtree (i.e.~a subproblem) of $T$ at this level, with $|T'| > 1$ (by our assumption, such a tree exists). Each vertex of $T'$ has more than $p \cdot \log n$ hubs. We now consider the induced ordering of vertices in $T'$, and in particular, w.l.o.g.~we consider the ordering that assigns higher rank to the highest rank vertex of the largest connected component (at each level of the recursion). More precisely, let $q_1 \in T'$ be the first vertex of this ordering (according to $H$), $q_2$ be the first vertex of the largest connected component of $T' \setminus \{q_1\}$, and so on. We will prove that after a constant number of ``iterations", and in particular after $l$ iterations, the optimal solution will have split $T'$ intro subtrees of size at most $n'/2$ (where $n' = |T'|$), otherwise it would not be optimal. So, let's assume that this is not the case, i.e.~let's assume that after $l$ iterations (starting with the tree $T'$), the largest connected component has size strictly larger than $n'/2$. This means that the hub set of more than $n'/2$ vertices has size strictly larger than $p \log n + l$. We now intervene, modify the solution, and we will show that the resulting solution is strictly better than $H$, which will give a contradiction. Let $A = \{q_1, ..., q_l\}$ be the set of vertices that are picked in $l$ consecutive steps from the largest component (of each round), and let $S$ be the largest connected component of $T' \setminus A$. Note that the hub sets of the vertices of $S$ contain all vertices of the set $A$. Our assumption implies that $|S| > n' / 2$. We now modify $H$ as follows. For $k = \lceil\log l \rceil$ rounds, we pick balanced separators $s_1, ..., s_k$ and add them as hubs, one by one, in the corresponding subproblems, where the balancing is with respect to the set of vertices $A = \{q_1, ..., q_l\}$. More precisely, in each subproblem, we assign weight 0 to all vertices not in $A$, and weight 1 to the vertices of $A$. Then, $s_1$ is the weighted balanced separator of $T'$, $s_2$ is the weighted balanced separator of a connected component of $T' \setminus \{s_1\}$ and so on (the order in which the resulting connected components are processed does not matter). Observe that after these $k$ steps, no two vertices of $A$ belong to the same subproblem. After these $k$ steps, we can resume selecting hubs in the order induced by $H$, and observe now that after $k + l$ steps, for each vertex $u \in S$, this process will have added at most $k + 1$ hubs, and no more than that (since the vertices of $A$ have been distributed into different connected components).

Our goal now is to prove that this modified solution is strictly better than the original one, thus contradicting the optimality of $H$. We consider the difference (which we denote as $\Delta$) of the contribution (raised to the power $p$) of the vertices of $T'$ in the original solution and the contribution of the vertices of $T'$ in this modified solution. We will prove that $\Delta > 0$. We have
\begin{align*}
    \Delta &\geq \left( \sum_{u \in T' \setminus S} h_u^p + \sum_{u \in S} h_u^p \right) - \left( \sum_{u \in T' \setminus S} (h_u + k + 1)^p + \sum_{u \in S} (h_u + k + 1 - l)^p \right) \\
           &= \left(\sum_{u \in S} h_u^p - \sum_{u \in S} (h_u - (l - k - 1))^p \right)  - \left(\sum_{u \in T' \setminus S} (h_u + k + 1)^p - \sum_{u \in T' \setminus S} h_u^p \right).
\end{align*}
By Lemmas~\ref{lemma:reduce-by-l} and~\ref{lemma:increase-by-k} we get
\begin{align*}
    \Delta &\geq \frac{(l - k - 1)\cdot |S| \cdot p^p \log^{p-1} n}{e} - e (p + 2)^p \cdot n' \cdot (k + 1) \cdot \log^{p-1} n \\
                  &> \left(\frac{(l - k - 1)}{2e} - e \left(1 + \frac{2}{p} \right)^p \cdot (k + 1) \right) \cdot  p^p \cdot n' \cdot \log^{p-1} n \\
                  &\geq \left(\frac{(l - k - 1)}{2e} - e^3 \cdot (k + 1) \right) \cdot  p^p \cdot n' \cdot \log^{p-1} n.
\end{align*}
We now claim that for appropriately chosen constant $l$, we have $\frac{(l - k - 1)}{2e} - e^3 \cdot (k + 1) > 0$. For this to hold, it is sufficient to have $l > (2e^4 + 1) \cdot (k + 1)$. We remind the reader that $k = \lceil \log l \rceil$. Thus, it is sufficient to have $l > (2e^4 + 1) \cdot (\log l + 2)$, which holds for every $l \geq 980$. We conclude that if we have $|S| > n' / 2$, then we can get an improved solution, thus contradicting the optimality of $H$. This implies that $|S| \leq n' / 2$.

So far, we have proved that once the recursion reaches the $(p \cdot \log n)$-th level, the size of the subproblems after that level reduces by a constant factor every $l$ iterations. This implies that after $l \log n$ iterations, the size of the resulting subproblems will be at most 1. Thus, the total depth of the decomposition tree, or in other words, the $\ell_\infty$-cost of $H$ is at most $p \cdot \log n + l \cdot \log n = O(p \cdot \log n)$.
\end{proof}

Theorems~\ref{thm:hl-infinity-max-bound}, \ref{thm:hl-logn-max-bound} and \ref{thm:hub-set-bound-small-p} now imply the following general theorem for the size of the hub sets for any \HLp.
\begin{theorem}\label{thm:general-bound-for-hub-size}
Let $T = (V, E)$ be a tree with $n$ vertices, and let $H$ be an optimal HHL solution for \HLp, for any $p \in [1, \infty]$. Then, $\|H\|_\infty = O(\log^2 n)$.
\end{theorem}

We now present an alternative proof of the above statement, thus circumventing the more involved analysis of the proof of Theorem~\ref{thm:hub-set-bound-small-p}.
\begin{proof}[Alternative proof of Theorem~\ref{thm:general-bound-for-hub-size}]
The interesting case is when $p \leq \varepsilon \log n$. Let $H$ be an optimal HHL solution for \HLp, and let's assume that $\|H\|_\infty > \log^2 n$. We will show that after the $(\log^2 n)$-th level of the recursive decomposition, the size of the resulting subproblems reduces by a factor of 2 every $4\log n$ levels. Thus, the total depth, or in other words, the size of the largest hub set, will always be at most $5 \log ^2 n$. So, let $T'$ be a subproblem at the $(\log^2 n)$-th level of the decomposition, and let's assume that after $4\log n$ levels, the size of the largest resulting subproblem (coming from $T'$) is $x > n' / 2$, where $n' = |T'|$. This means that the contribution of the vertices of $T'$ to the objective value (raised to the power $p$) is at least $x \cdot (\log ^2 n + 4\log n)^p + (n' - x) (\log^2 n)^p$. 

We now modify the solution and run the Tree Algorithm (see Algorithm~\ref{Tree-Algorithm}) on $T'$. The contribution of the vertices of $T'$ to the objective value (raised to the power $p$) in this modified solution is at most $n' (\log^2 n + 2\log n)^p$. Since $H$ is optimal, we must have
\begin{equation*}
    n' (\log^2 n + 2\log n)^p \geq x \cdot (\log ^2 n + 4\log n)^p + (n' - x) (\log^2 n)^p.
\end{equation*}
This is equivalent to
\begin{equation*}
    (1 + 2/\log n)^p \geq \frac{x}{n'} \cdot (1 + 4/\log n)^p + \left(1 - \frac{x}{n'} \right).
\end{equation*}
We have $\frac{x}{n'} \cdot (1 + 4/\log n)^p + \left(1 - \frac{x}{n'} \right) > \frac{1}{2} \cdot (1 + 4/\log n)^p + \frac{1}{2}$. We will now prove that we always have $\frac{1}{2} \cdot (1 + 4/\log n)^p + \frac{1}{2} \geq (1 + 2/\log n)^p$. Let $f(x) = \frac{1}{2} \cdot (1 + 2x)^p + \frac{1}{2} - (1 + x)^p$. We have
\begin{equation*}
    f'(x) = p (1 + 2x)^{p - 1} - p (1 + x)^{p - 1}.
\end{equation*}
Note that $f'(x) > 0$ for every $x > 0$ (and $f'(0) = 0$). Since $f$ is continuous, this implies that $f(x) \geq f(0) = 0$ for every $x \geq 0$. We conclude that $f(2/\log n) \geq 0$, which implies a contradiction. This proves that after $4\log n$ iterations the size of the subproblems has reduced by a factor of 2. It is easy to see that this argument can be applied at any level of the recursion at depth at least $\log^2 n$, and so we conclude that indeed, after at most $5\log^2 n$ levels, the size of the subproblems will be at most 1. Thus, $\|H\|_\infty = O(\log^2 n)$.
\end{proof}

\section{Quasi-polynomial time algorithms for HL on trees}\label{quasi_sec}

The results of the previous section now imply that the DP techniques presented in Section~\ref{sec:ptas-trees} can be used to obtain quasi-polynomial-time algorithms for \HLp. This is based on the observation that the set of boundary vertices of a subtree is a subset of the hub set of every vertex in that subtree. Thus, Section~\ref{sec:bounds-on-size} suggests that by restricting our DP to subinstances with polylogarithmic boundary size, we obtain exact quasi-polynomial time-algorithms. In particular, we can easily now get the following theorem, as an immediate corollary of the results of the previous sections.
\begin{theorem}
There exist quasi-polynomial time exact algorithms for \HLp on trees, for every $p \in [1, \infty]$. For any fixed $\varepsilon > 0$, the corresponding running times of the algorithms on trees with $n$ vertices are:
\begin{enumerate}
    \item $n^{O(\log n)}$, when $p$ is either a constant or at least as large as $\varepsilon \cdot \log n$,
    \item $n^{O(\log^2 n)}$, when $p$ is superconstant smaller than $\varepsilon \log n$.
\end{enumerate}
\end{theorem}

\newcommand{\PP}{\mathcal{P}}

\section{Hub Labeling on trees and the problem of searching on trees}

In this section, we discuss the equivalence of Hub Labeling on trees and the problem of searching for a node in trees, as first observed and communicated to us by Gawrychowski et al.~\cite{GawKMW17}. The uniform (i.e.~unweighted) vertex-query version of the problem can be described as follows. Let $T = (V, E)$ be a tree, and let $t \in V$ be a hidden marked vertex of the tree. The goal is to detect vertex $t$ by quering vertices of the tree. A query asks whether a vertex $u$ of the graph is the target vertex $t$ and if not, the response is the subtree of $T-\{u\}$ that contains $t$. A vertex $t$ is found when the algorithm queries a vertex $u$, and the response is that $u \equiv t$. Let $Q_t \subseteq V$ be the set of vertices queried until we find vertex $t$ (note that $t \in Q_t$). Our goal is to find a deterministic strategy that minimizes the following quantity:
\begin{equation*}
    \left(\sum_{t \in V} |Q_t|^p \right)^{1/p}.    
\end{equation*}

The above problem and related questions (such as variants with edge queries, weighted versions of the problem and generalizations to arbitrary graphs) have been posed in various works, as possible generalizations of the standard binary search (see e.g.~\cite{Iyer:1988:ONR:49320.49322, Schaffer:1989:ONR:71361.71368, Mozes:2008:FOT:1347082.1347202, DBLP:journals/siamcomp/Ben-AsherFN99, DBLP:conf/icalp/JacobsCLM10, DBLP:conf/stoc/Emamjomeh-Zadeh16, DBLP:conf/icalp/DereniowskiKUZ17}.

We will now formally define the problem and show its equivalence to \HLp on trees. A deterministic strategy can be defined (recursively) as follows.
\begin{definition}\label{def:deterministic-strategy}
Let $T=(V,E)$ be a tree. A deterministic search strategy $D$ for the tree $T$ is an ordering $\langle \pi_1, ..., \pi_n \rangle$ of the vertices of $T$ such that:
\begin{enumerate}
    \item If $V = \{u\}$, then $D = \langle u \rangle$.
    \item If $|V| > 1$, then $D_{T'} = \langle \pi_2, ..., \pi_n \rangle_{T'}$ is a deterministic strategy for $T'$, for every subtree $T'$ of $T - \{\pi_1\}$,
\end{enumerate}
(the notation $\langle \pi_2, ..., \pi_n \rangle_{T'}$ denotes the restriction of the ordering to the vertices contained in $T'$).

The cost $C(T,D,t)$ of a strategy $D = \langle \pi_1, ..., \pi_n \rangle$ for detecting a vertex $t$ in a tree $T$ is defined as
\begin{equation*}
    C(T,D,t) =  \begin{cases}
                    1, & \textrm{if }t = \pi_1,\\   
                    1 + C(T', D_{T'}, t), &\textrm{if }t \neq \pi_1 \textrm{ and } T' \textrm{ is the subtree of }T - \{\pi_1\} \textrm{ that contains }t.
                \end{cases}
\end{equation*}
Finally, the $\ell_p$-cost of strategy $D$ for the tree $T$ is defined as
\begin{equation*}
    C_p(T,D) = \left(\sum_{t \in V} C(T,D,t)^p \right)^{1/p}.
\end{equation*}
The $\ell_\infty$-cost of strategy $D$ for the tree $T$ is defined as $C_\infty(T,D) = \max_{t \in V} C(T,D,t)$.
\end{definition}

\begin{definition}[$\ell_p$-searching in trees]
Let $T = (V, E)$ be a tree and let $p \geq 1$. The $\ell_p$-searching in trees problem asks to compute a deterministic search strategy $D$ for $T$ that minimizes the cost $C_p(T, D) = \left(\sum_{t \in V} C(T,D,t)^p \right)^{1/p}$. The $\ell_\infty$-searching in trees problem asks to compute a deterministic search strategy $D$ for $T$ that minimizes the cost $C_\infty(T, D) = \max_{t \in V} C(T,D,t)$.
\end{definition}

By Theorem~\ref{HHL_optimal_for_trees}, we know that solving \HLp on trees is equivalent to finding the optimal hierarchical hub labeling. Using this result and Definitions~\ref{def:hhl} and \ref{def:deterministic-strategy}, it is now straighforward to prove the following theorem (and thus the proof is omitted).
\begin{theorem}
Let $T=(V,E)$ be a tree. Then, for every $p \in [1, \infty]$, the $\ell_p$-searching in trees problem for $T$ is equivalent to optimally solving \HLp for $T$.
\end{theorem}

In~\cite{DBLP:conf/focs/OnakP06} a linear-time exact algorithm is given for the $\ell_\infty$-searching in trees problem. Thus, this results translates to an exact linear-time algorithm for \HL\infty on trees, improving upon the quasi-polynomial-time algorithm of the previous section. This is exactly the observation that was communicated to us by Gawrychowski et al.~\cite{GawKMW17}. After realizing that the two problems are intimately connected, we contacted Eduardo Laber and Marco Molinaro~\cite{LM17} and they suggested that a slight modification of the algorithm given in the work of Jacobs et al.~\cite{DBLP:conf/icalp/JacobsCLM10} ``should" work for the $\ell_1$-searching in trees problem (or equivalently, for \HL1 on trees), as the DP approach of~\cite{DBLP:conf/icalp/JacobsCLM10} gives an exact algorithm (for the edge-query variant of the problem) whose running time is $2^{O(h)} \poly(n)$, where $n$ is the number of vertices in the tree and $h$ is the height (depth) of an optimal search tree for the problem. Thus, if one establishes an $O(\log n)$ upper bound on the height of the search tree, then the algorithm runs in polynomial time. We remind the reader that we proved this very fact in Section~\ref{sec:bounds-on-size}, when $p$ is either a constant or $p = \Omega(\log n)$. Since the algorithm is not written down for the problem as defined above, we give a complete presentation and proof of correctness of the algorithm, thus establishing polynomial-time exact algorithms for \HLp on trees, when $p$ is either a constant or $p = \Omega(\log n)$.

\paragraph{Adapting the DP approach of Jacobs et al.~\cite{DBLP:conf/icalp/JacobsCLM10}.} In this section, we present an algorithm that optimally solves the $\ell_p$-searching in trees problem and runs in time $2^{O(h)} \poly(n)$ on a tree with $n$ vertices. Here, $h$ is the height of an optimal search tree for $T$ (or equivalently, the size of the largest hub set in an optimal solution). The bounds established in Section~\ref{sec:bounds-on-size} imply that for fixed $p$ and for $p \in [\varepsilon \log n, \infty]$ (for any fixed $\varepsilon > 0$), the algorithm runs in polynomial time.

Following the presentation of~\cite{DBLP:conf/icalp/JacobsCLM10}, it will be more convenient to describe a deterministic search strategy with a tree. To distinguish between the original tree $T$ and a search tree for $T$, we will use the term ``vertex" for a vertex of the tree $T$ and the term ``node" for a vertex of the search tree. We will also assume that the input tree $T=(V, E)$ is rooted at some arbitrary vertex $r$. Then, a deterministic search strategy can be respresented as a rooted tree $D$, where each node corresponds to a query (i.e.~a vertex) of the original tree $T$. 

A search tree for a rooted tree $T = (V,E)$ is a rooted tree $D = (N, E', A)$ where $N$ and $E'$ are the nodes and edges of the tree and $A: N \to V$ is an assignment. The nodes of the search tree correspond to queries. More precisely, a path from the root of $D$ to a node $u$ of $D$ indicates which queries should be made at each step to discover a particular vertex $A(u)$. The assignment $A$ describes exactly this correspondence between vertices of the original tree and nodes of the search tree. For each vertex $u \in V$, there is exactly one vertex $l$ of $D$ such that $A(l) = u$; in particular, we require a certificate that a vertex was found, and thus, even when we are left with one vertex in the resulting subtree, we still require that we make the query so as to discover the vertex (in that case, the corresponding node in the search tree will be a leaf). It is clear now that $|V| = |N|$. Moreover, we require the following property. For each inner (non-leaf) node $x \in N$, its children are partitioned into two classes, left and right. A node has at most one left child, but might have several right children. If node $x \in N$ has a left child, which we will denote as $y_0$, then for every $z \in N$ that is in the subtree hanging from $y_0$, $A(z)$ is not in any of the subtrees rooted at the children of $A(x)$. All other children of $x$ (if any) are right children and are denoted as $y_1, ..., y_k$. A node $z \in N$ is in the subtree hanging from $y_i$, $1 \leq i \leq k$, if $A(z)$ is in the subtree of $T$ rooted at $A(y_i)$.

Given a search tree $D$ for $T$, let $d(x,y)$ be the number of vertices from $x$ to $y$ in $D$ (note that we count the number of vertices instead of edges, since, technically, in our definition we ask to certify that a vertex is found, and thus, we will always query a vertex even if it is the only one remaining in the resulting tree). Then the $\ell_p$-cost of the search tree $D$ for tree $T$ is defined as
\begin{equation*}
    C_p(T, D) = \left(\sum_{z \in N} d\left(root(D), z \right)^p \right)^{1/p}.
\end{equation*}
It is easy to see that the above definition is equivalent to the definition of cost as defined previously in terms of the ordering of vertices that a strategy induces. So, from now on, we will use this tree description of a search strategy, and our goal is to compute the search tree of minimum cost. For that, we use a slightly modified version of the dynamic programming approach of~\cite{DBLP:conf/icalp/JacobsCLM10}. Before describing the algorithm, we introduce a variant of a search tree, similar to~\cite{DBLP:conf/icalp/JacobsCLM10}. Since we always consider rooted trees, we use the notation $T_u$ to denote the subtree hanging from a vertex $u$ of a rooted tree $T$.
\begin{definition}
Let $T = (V, E)$ be a rooted tree. An extended search tree (EST) for the tree $T$ is a triple $D = (N, E', A)$ where $N$ and $E'$ are the nodes and edges of a rooted tree and the assignment $A: N \to V \cup \{blocked, unassigned\}$ satisfies the following properties:
\begin{enumerate}
    \item For every vertex $u$ of $T$, $D$ contains exactly one vertex $x$ such that $A(x) = u$.
    \item Every non-leaf node $x \in N$ has at most one left child, but it might have several right children.
    \item For every $x,z \in N$ with $A(x), A(z) \in V$ and $x$ being a non-leaf node, the following holds: if $z$ is in the subtree of $D$ rooted at a right child $y_i$ of $x$ (for some $1 \leq i \leq k$), then $A(z) \in T_{A(y_i)}$, while if it is in the subtree rooted at the left child $y_0$, then $A(z)$ is not in any of the subtrees hanging from the children of $A(x)$.
    \item If $x \in N$ with $A(x) \in \{blocked, unassigned\}$, then $x$ has exactly one left child and no right children.
\end{enumerate}
The $\ell_p$-cost of an EST $D = (N, E', A)$ for a tree $T = (V,E)$ is defined as 
\begin{equation*}
    C_p(T, D) = \left(\sum_{z \in N: A(z) \in V} d(root(D), z)^p \right)^{1/p},
\end{equation*}
where again $d(u,v)$ denotes the number of nodes betwee $u$ and $v$ in $D$.
\end{definition}

It is straightforward to see that a search tree is also an EST, and that the cost of an optimal EST is at least as much as the cost of an optimal search tree. For the latter, we can convert any EST to a search tree by deleting any node $x$ (whose parent is $p$) with $A(x) \in \{block, unassigned\}$ and make the (unique) left child $y$ of $x$ the left child of $p$; we do this one node at a time, and the resulting tree is a search tree of cost at most the cost of the original EST (w.r.t.~any $\ell_p$-cost).

Before describing the algorithm, we need one more concept. A left path of a rooted search tree is the path obtained when we traverse the tree (starting from the root) by only going to the left child, until we reach a node that does not have a left child. A \emph{partial left path} (PLP) is a left path where every node is assigned (via a function $g$) to either \emph{blocked} or \emph{unassigned}. Let $D = (N, E', A)$ be an EST, and let $L = \{x_1, .., x_k\}$ be its left path. We say that $D$ is compatible with a PLP $P = \{p_1, ..., p_q\}$ if $k = q$ and $g(p_i) = blocked$ implies that $A(x_i) = blocked$.

We now introduce the subproblems that the DP will be solving. Let $T=(V,E)$ be the (rooted) input tree with root $r$. For a vertex $u \in V$, we denote as $c_1(u), ..., c_k(u)$ the children of $u$ (we denote the number of children of $u$ as $\delta(u)$; here we arbitrarily order the children of $u$). Let $T_{u, i}$, $1 \leq i \leq \delta(u)$, denote the subtree of $T$ containing $u$ and the subtrees hanging from its first $i$ children, i.e.~$T_{u,i} = \{u\} \cup \{T_{c_1(u)}\} \cup ... \cup \{T_{c_i(u)}\}$. A problem $\PP^B(T_{u,i}, P)$ consists of finding an EST for the tree $T_{u,i}$ with minimum $\ell_p$-cost among all EST's for $T_{u,i}$ that are compatible with $P$ and have height at most $B$. For simplicity of notation, the subproblem $\PP^B(T_{u,\delta(u)}, P)$ will also be denoted as $\PP^B(T_u, P)$. In order to recover the optimal search tree for $T$, in the end we will return $\PP^B(T_r, P)$ for sufficiently large $B$ (in particular, for $B$ equal to the bounds guaranteed by Theorems~\ref{thm:hl-infinity-max-bound}, \ref{thm:hl-logn-max-bound} and \ref{thm:hub-set-bound-small-p}), such that $P$ is of length $B$ and consists only of unassigned nodes.

The subproblems are computed bottom-up, and from left to right. We now describe how to optimally compute $\PP^B(T_{u,i}, P)$ for given $B$, $u$, $i$ and $P$. We will always denote as $g$ the function that assigns \emph{blocked} or \emph{unassigned} to the path $P$, and $A$ the assignment computed for the EST. The path $P$ will always be assumed to have $k$ vertices, i.e. $P = \{p_1, ..., p_k\}$, such that $k \leq |B|$; if $k > |B|$ we simply declare the subproblem ``not feasible". We will also store the cost of the optimal solution (w.r.t.~the $\ell_p$-cost) raised to the power $p$, when $p < \infty$, as this will turn out to be more convenient. If $p = \infty$, we will simply store the standard $\ell_\infty$-cost.

\paragraph{Base case: $T_u$ has only one vertex $u$.} Let $j$ be the smallest index, if any, such that $g(p_j) = unassigned$. If there is no such $j$, then the subproblem is ``not feasible". Otherwise, the EST is simply the path $P$ where we set $A(p_j) = u$. Its $\ell_p$-cost, for any $p \geq 1$, is equal to $j^p$. For $p = \infty$, the cost is $j$.

\paragraph{Case 1: $T_{u,1}$ (where $u$ is a non-leaf vertex with at least one child).} For simplicity of notation, let $v = c_1(u)$. We assume that we have already solved all subproblems of $T_v$. In order to compute an optimal solution for $\PP^B(T_{u,1}, P)$, for every $t \in [k]$ with $g(p_t) = unassigned$, we define $P_t = \{p_1, ..., p_k\}$ to be the path with assignment $g_t(p_j) = g(p_j)$ for $j < t$, $g_t(p_t) = blocked$ and $g_t(p_j) = unassigned$ for $j > t$. For each $t$, we then construct the EST $D_t$ as follows: we consider the optimal EST $D'$ (with corresponding assignment $A'$) for $\PP^B(T_{v}, P_t)$, and then set $A(p_t) = u$. We then look at the left child of $p_t$ in $D'$ and the tree hanging from it, we remove it and ``rehang" it as the unique right child of $p_t$. We finish by adding enough \emph{blocked} nodes in the left path of this modified EST so as to make it compatible with $P$. Let $D_t$ be the resulting EST (if $\PP^B(T_{v}, P_t)$ is ``not feasible", then we simply cannot construct $D_t$ and the corresponding value is ``not feasible"). It is easy to see that this is indeed an EST for $T_{u,1}$ whose depth is at most $B$ and is compatible with $P$. Its cost is $OPT^p(\PP^B(T_{v}, P_t)) + t^p$. Finally, among all (feasible) choices of $t$, we pick the one that minimizes $OPT^p(\PP^B(T_{v}, P_t)) + t^p$.

We claim now that this is indeed an optimal EST for $\PP^B(T_{u,1}, P)$. Let $D$ be an optimal EST for $\PP^B(T_{u,1}, P)$, whose cost (raised to the power $p$) is $OPT^p(\PP^B(T_{u,1}, P))$. We first observe that, by the definition of an EST, the node $x$ with $A(x) = u$ must be in the left path of $D$, i.e.~$A(p_t) = u$ for some $t \in [k]$ with $g(p_t) = unassigned$. We will now perform the ``reverse" operation compared to what we did in the previous paragraph. We define $P_t = \{p_1, ..., p_k\}$ to be the path with assignment $g_t(p_j) = g(p_j)$ for $j < t$, $g_t(p_t) = blocked$ and $g_t(p_j) = unassigned$ for $j > t$. We then modify $D$ by looking at the (unique) right child of the node $p_t$, and making the whole subtree hanging from that child the unique left child of the (blocked) node $p_t$. Let $D_t$ be the resulting EST. First, we observe that $D_t$ does not violate the assignment $g_t$, since every node after $p_t$ is an \emph{unassigned} node. It is also easy to see that the depth of $D_t$ is the same as the depth of $D$. The only property that is, potentially, violated, is  the length of the left path of $D_t$. As stated, we will eventually set $k = B$, and so, w.l.o.g. we will assume that this is the case, i.e.~$k = B$. We now look at the resulting left path of $D_t$. If it is at most $k$, then we can fill the left path of $D_t$ with blocked nodes so as to make it compatible with $P_t$. If its length is strictly larger than $k$, then, since $k = B$, this would imply that $D$'s height is strictly larger than $B$, which is impossible. Thus, we can always construct an EST $D_t$ such that $OPT^p(\PP^B(T_{u,1}, P)) = cost(D_t)^p + t^p$. It is now easy to see that, since our algorithm considers all possible values of $t$, it will return a solution of cost at most $OPT^p(\PP^B(T_{v}, P_t)) + t^p \leq cost(D_t)^p + t^p$, and thus, it will compute an optimal solution.

\paragraph{Case 2: $T_{u, i + 1}$ for some $i \geq 1$ (where $u$ is a non-leaf vertex with at least two children).} Again, for simplicity of notation, let $T_1 = T_{c_1(u)} \cup ... \cup T_{c_i(u)} \cup \{u\}$ and $T_2 = T_{c_{i + 1}(u)}$. Let $I$ be the set of indices corresponding to the unassigned nodes of $P$, i.e.~$I = \{j: g(p_j) = unassigned\}$. We consider all possible bipartitions of these nodes. For a bipartition  $(I_1, I_2)$ of $I$, let $P_1 = \{p_1, ..., p_k\}$ be the path with assignment $g_1$ such that $g_1(p_i) = unassigned$ for every $i \in I_1$ and $g_1(p_i) = blocked$, otherwise. We consider the EST $D_1$ of the problem $\PP^B(T_1, P_1)$ with corresponding assignment $A_1$. Let $p_t$ be the node of the left path such that $A(p_t) = u$. Let now $P_2 = \{p_1, ..., p_k\}$ be the path with assignment $g_2$ such that $g_2(p_i) = unassigned$ if $i \in I_2$ or $i > t$, and $g_2(p_i) = blocked$, otherwise. Let $D_2$ with assignment $A_2$ be the EST of the problem $\PP^B(T_2, P_2)$. We first take the ``union" $D''$ of the two trees, i.e.~we align their left paths. By construction, there is no conflict with the two assignments, since there is no assigned node on the left path after $p_t$ for $A_1$. In order to obtain a valid EST $D'$ now for $\PP^B(T_{u, i + 1}, P)$, we consider the left child of $p_t$ (which is part of $D_2$, as $D_1$ only has unassigned/blocked nodes at that part of the tree), remove the subtree hanging from there and rehang it as a right child of $p_t$. Note that the depth of the tree does not change, and, by adding blocked vertices on the left path of the tree after $p_t$ so as to make it compatible with $P$, we now obtain a valid EST $D'$ for $\PP^B(T_{u, i + 1}, P)$. Its cost is equal to $OPT^p(\PP^B(T_1, P_1)) + OPT^p(\PP^B(T_2, P_2))$. Finally, we pick the bipartition $(I_1, I_2)$ that minimizes this sum.

We will now show that this is indeed an optimal solution for $\PP^B(T_{u, i + 1}, P)$. Let $D$ be an optimal EST for $\PP^B(T_{u, i + 1}, P)$ with corresponding assignment $A$, and let $p_t$ be the node of the left path such that $A(p_t) = u$. Let $I_1 = \{i \leq t: A(p_i) \in T_1\}$ and $I_2 = \{i < t: A(p_i) \in T_2\}$. Clearly, $g(p_i) = unassigned$ for every $i \in I_1 \cup I_2$. It is easy to see that $A(p_i) \in \{unassigned, blocked\}$ for every $i > t$. Let $P_1 = \{p_1, ..., p_k\}$ with assignment $g_1$ such that $g_1(p_i) = unassigned$ if $i \in I_1$, and $g_1(p_i) = blocked$, otherwise. We now construct the EST $D_1$ (with assignment $A_1$) for subproblem $\PP^B(T_1, P_1)$ from $D$ as follows: we look at $p_i$ for $i \in I_2$ and we delete all the right children of $p_i$ and set $A_1(p_i) = blocked$. We also look at $p_t$ and delete its right child that corresponds to a subtree of $T_2$. The resulting tree $D_1$ is indeed an EST for $\PP^B(T_1, P_1)$ of depth at most $B$. Similarly, let $P_2 = \{p_1, ..., p_k\}$ with assignment $g_2$ such that $g_2(p_i) = unassigned$ if $i \in I_2$ or $i > t$, and $g_2(p_i) = blocked$, otherwise. We now construct an EST $D_2$ with assignment $A_2$ for $\PP^B(T_2, P_2)$ as follows: we consider $D$, and for every $i < t$ with $i \in I_1$, we delete the right children of $p_i$ and set $A_2(p_i) = blocked$. For $p_t$, we delete all the right children corresponding to subtrees of $T_1$, and we move the unique subtree corresponding to a subtree of $T_2$ and make it the left child of $p_t$. Note that this indeed results to a feasible solution of height at most $B$ for $\PP^B(T_2, P_2)$ (after, potentially, adding some blocked nodes in the left path so as to make it compatible with the length of $P_2$, which, as already mentioned, can be assumed to be $B$). It is easy to see that $OPT^p(\PP^B(T_{u, i + 1}, P)) = cost^p(D_1) + cost^p(D_2)$. By construction, our algorithm will consider the set $I_1$ and the corresponding path $P_1$, and will optimally solve the problem $\PP^B(T_1, P_1)$. Let $D'$ be an optimal EST for $\PP^B(T_1, P_1)$, with assignment $A'$. We must have $A'(p_i) = u$ for some $i \in I_1$. Thus, by construction, the path $P_2'$ that the algorithm will consider in this case is a path whose blocked vertices are a subset of the blocked vertices of $P_2$. Let $D''$ be an optimal solution for $\PP^B(T_2, P_2')$. The previous observation implies that $OPT^p(\PP^B(T_2, P_2')) \leq cost^p(D_2)$, since $D_2$ is a feasible solution for $\PP^B(T_2, P_2')$. Thus, we conclude that the algorithm returns a solution of cost at most $cost^p(D') + cost^p(D'') = OPT^p(\PP^B(T_1, P_1)) + OPT^p(\PP^B(T_2, P_2')) \leq cost^p(D_1) + cost^p(D_2) = OPT^p(\PP^B(T_{u, i + 1}, P))$. This shows that the algorithm indeed computes the optimal solution.

\begin{theorem}
For trees of size $n$, the above algorithm optimally solves the $\ell_p$-searching in trees problem, or equivalently the \HLp problem on trees, in time $2^{O(h)} \cdot \poly(n)$, where $h$ is the maximum number of queries in an optimal deterministic strategy, or equivalently, the size of the largest hub set in an optimal HHL solution. In particular, when $p$ is a constant or $p \in [\varepsilon \log n, \infty]$, the running time is $\poly(n)$ (since $h = O(\log n)$ in such cases); in all other cases, the running time is $n^{O(\log n)}$.
\end{theorem}

\chapter{Open problems from Part I}\label{chapter:open-problems-hl}

In this concluding chapter of Part I, we will state a few open problems that we believe are of interest. Regarding Hub Labeling, there are still quite a few open problems, such as the following:
\begin{enumerate}
    \item Is there a constant factor approximation algorithm for Hub Labeling on graphs with unique shortest paths? As of now, the hardness of approximation results seem to require that graphs have multiple shortest paths. A first attempt towards answering this question would be to construct integrality gap examples for the LP for \HL1, introduced in Figure~\ref{fig:lp1}.
    \item Is it $\Omega(\log n)$-hard to approximate \HLp on general graphs for the whole range of parameter $p$? One should expect this result to hold, so it would be nice to obtain a uniform $\Omega(\log n)$-hardness result for every $p \in [1, \infty]$ (we remind the reader that in Section~\ref{Hardness} we proved $\Omega(\log n)$-hardness for \HL1 and for \HLp when $p = [\log n, \infty]$).
    \item Can the $O_p(\log D)$-approximation algorithm for \HLp (presented in Section~\ref{sec:bounded}) be extended to work for every $p$, and in particular for $p = \infty$?
    \item A related question that seems to be of interest to the community is whether there exist $\poly(\log n)$-approximation algorithms for the Hierarchical Hub Labeling problem. The only known approximation algorithms give a polynomial approximation, since the standard approach is to compare their performance with the cost of the optimal hub labeling (and not just the optimal hierarchical one), and it is known that there are cases where the gap between the two optimums is polynomial~\cite{DBLP:conf/mfcs/GoldbergRS13}.
\end{enumerate}

As explained in Chapter~\ref{chapter:trees}, Hub Labeling on trees is equivalent with the problem of searching for a node in a tree. Using this equivalence, we showed how one can obtain exact polynomial-time algorithm for \HLp, when $p$ is a fixed constant or when $p \in [\log n, \infty]$ (with $n$ being the number of vertices on the tree). Some open problems about Hub Labeling and related search problem in trees are the following:
\begin{enumerate}
    \item Is the size of the largest hub set in an optimal solution for \HLp on trees always $O(\log n)$? In Section~\ref{sec:bounds-on-size} we proved that this is indeed the case for constant $p$ and for $p \in [\log n, \infty]$. It would be nice to obtain this upper bound for every $p \in [1, \infty]$. Such a result would imply polynomial-time algorithms for \HLp on trees for every $p \in [1, \infty]$.
    \item There are several generalizations of the problem of searching for a node in a tree, many of which are \NP-hard. A particular one that seems interesting is the generalization introduced by Dereniowski et al.~\cite{DBLP:conf/icalp/DereniowskiKUZ17}, in which they introduce vertex weights that correspond to the query time when a vertex is queried. In other words, if $T = (V,E)$ and a strategy makes queries $Q_t \subseteq V$ in order to discover vertex $t$, it pays $w(Q_t) = \sum_{v \in Q_t} w_v$ for that vertex. The objective they consider is the $\ell_\infty$-cost, i.e.~they minimize the worst-case query time, which corresponds to minimizing $\max_{t \in V} w(Q_t)$. They obtain a QPTAS and an $O(\sqrt{\log n})$-approximation algorithm (that runs in polynomial time). A natural question is whether one can obtain a constant factor approximation algorithm for the problem, and, ideally, a PTAS.
\end{enumerate}


\part{Stability and perturbation resilience}\label{part:stability}

\chapter{Bilu-Linial stability and perturbation resilience} \label{chapter:stability}

\section{Introduction and definitions}

The notion of stability that we are interested in is the one defined by Bilu and Linial in 2010~\cite{DBLP:journals/cpc/BiluL12}. Informally, an instance of an optimization problem is \textit{stable} if it has a unique optimal solution, and this solution remains the unique optimal solution under small perturbations of the parameters of the input. As Max Cut was the first problem studied in this framework, we will use it as an example to illustrate the definition.

A Max Cut instance is defined by an edge-weighted undirected graph $G = (V,E,w)$, where $w: E \to \mathbb{R}_{>0}$, and the goal is to find a partition $(X, V\setminus X)$ of the vertex set such that the weight of the edges cut (i.e.~the edges whose endpoints end up in different sets of the partition) is maximized. In such an instance, the parameters are simply the edge weights. An instance of Max Cut is called $\gamma$-stable, for some $\gamma \geq 1$, if there is a unique optimal partition $(X^*, V \setminus X^*)$, and this partition remains the unique optimal solution for every instance $G' = (V,E, w')$ that satisfies $w_e \leq w_e' \leq \gamma \cdot w_e$ for every $e \in E$.

Given such an instance, the goal is to design an exact polynomial-time algorithm that recovers this unique optimal partition $(X^*, V \setminus X^*)$. Observe that as $\gamma$ grows, the restrictions imposed on the instance are stronger, and fewer instances satisfy the definition. For $\gamma = 1$, the definition is equivalent to the statement that the instance has a unique optimal solution. From these observations, it follows that the main goal in such a framework is to design exact algorithms that work for $\gamma$-stable instances, for as small value of $\gamma \geq 1$ as possible.

Before giving the formal definition of stability and an overview of previous results, we would like to give the motivation behind such a notion. It is a well-observed fact that in many real-life instances, the parameter values are merely approximations to the actual parameters, since they are obtained from inherently noisy measurements. Thus, it is reasonable to believe, given that the optimization problem we are solving is meaningful to begin with, that the objective function is not sensitive to small perturbations of these parameters. Moreover, in many cases, the objective simply serves as a proxy towards recovering an intended underlying solution, and, so, small perturbations should not really affect the ground truth. The case of clustering problems exemplify this belief, as in such problems the objective function is commonly used to guide us to the ground-truth solution, meaning that we do not really care about computing the exact optimal value per se, but rather we are interested in recovering the underlying ground-truth clustering, and we choose the right objective function to help us discover this clustering.

We will now give the formal definition of stability/perturbation resilience for optimization and clustering problems. The definition of \emph{stability} was first given in the context of graph optimization problems by Bilu and Linial~\cite{DBLP:journals/cpc/BiluL12}, was later extended to clustering problems by Awasthi et al.~\cite{DBLP:journals/ipl/AwasthiBS12} under the name \emph{perturbation resilience}, and finally, the \emph{metric} version of perturbation resilience was introduced by Makarychev and Makarychev~\cite{DBLP:journals/corr/MakarychevM16} and published, along with several other results, in a joint work with these two authors~\cite{DBLP:conf/stoc/AngelidakisMM17}.

\begin{definition}[stability and perturbation resilience~\cite{DBLP:journals/cpc/BiluL12, DBLP:journals/ipl/AwasthiBS12, DBLP:journals/corr/MakarychevM16, DBLP:conf/stoc/AngelidakisMM17}]\label{def:perturbation-stability}
Consider an instance ${\calI} = (G, w)$ of a graph optimization problem with a set of vertex or edge weights $w_i$. An instance $(G, w')$, with weights $w_i'$, is a \emph{$\gamma$-perturbation} ($\gamma \geq 1$) of $(G, w)$ if $w_i \leq w'_i \leq \gamma \cdot w_i$ for every vertex/edge $i$; that is, a $\gamma$-perturbation is an instance obtained from the original one by multiplying each weight by a number from $1$ to $\gamma$ (the number may depend on $i$).

Now, consider an instance ${\calI} = (\X, d)$ of a clustering problem, where $\X$ is a set of points and $d: \X \times \X \to \R_{\geq 0}$ is a metric on $\X$. An instance $(\X, d')$ is a \emph{$\gamma$-perturbation} of $(\X, d)$ if $d(u,v) \leq d'(u,v) \leq \gamma \cdot d(u,v)$ for every $u,v \in \X$; here, $d'$ does not have to be a metric. If, in addition, $d'$ is a metric, then $d'$ is a \emph{$\gamma$-metric perturbation} of $(\X,d)$.

An instance $\calI$ of a graph optimization or clustering problem is \emph{$\gamma$-stable} or \emph{$\gamma$-perturbation-resilient} if it has a unique optimal solution and every $\gamma$-perturbation of $\calI$ has the same unique optimal solution/clustering as $\calI$. We will refer to $\gamma$ as the stability or perturbation resilience parameter.

Adhering to the literature, we call $\gamma$-stable instances of graph partitioning problems ``$\gamma$-Bilu--Linial stable'' or simply ``$\gamma$-stable'' and $\gamma$-stable instances of clustering problems   ``$\gamma$-perturbation-resilient''.
\end{definition}

Note that, in principle, the problem of designing algorithms for stable/perturbation-resilient instances is a promise problem, meaning that a correct algorithm must solve every $\gamma$-stable instance, but, potentially, might return a suboptimal solution, in the case where an instance turns out not to be stable. To address this, Makarychev et al.~\cite{DBLP:conf/soda/MakarychevMV14} introduced the notion of \emph{robust} algorithms for stable instances.
\begin{definition}[robust algorithm~\cite{DBLP:conf/soda/MakarychevMV14}]\label{def:robust}
A robust algorithm for a $\gamma$-stable (or $\gamma$-perturbation-resilient) instance $\calI$  is a polynomial-time algorithm that behaves as follows:
\begin{itemize}
    \item if $\calI$ is $\gamma$-stable, then the algorithm always returns the unique optimal solution.
    \item if $\calI$ is not $\gamma$-stable, then the algorithm either returns an optimal solution or reports that the instance is not $\gamma$-stable.
\end{itemize}
\end{definition}
Observe that, in particular, a robust algorithm is not allowed to err. The robustness property is a very useful property to have, especially when using such algorithms for solving real-life instances, since we do not know whether they are indeed stable or not.

We will now describe the results of previous works in this framework. Bilu and Linial studied Max Cut in their original paper~\cite{DBLP:journals/cpc/BiluL12} and showed that one can solve $O(n)$-stable instances of Max Cut. The \emph{stability threshold}, as we call the (current best) upper bound on the stability parameter $\gamma$, for Max Cut on general graphs was later improved to $O(\sqrt{n})$ by Bilu et al.~\cite{DBLP:conf/stacs/BiluDLS13}, where they also showed that one can optimally solve $(1+\varepsilon)$-stable instances of (everywhere) dense Max Cut. The stability threshold for Max Cut was further improved to $O(\sqrt{\log n} \cdot \log \log n)$ by Makarychev, Makarychev \& Vijayaraghavan in 2014~\cite{DBLP:conf/soda/MakarychevMV14}, and moreover, their algorithm is robust. The latter paper also gave some indications that $\Omega(\sqrt{\log n})$ might be the right answer for the stability threshold for Max Cut. From now on, we will refer to the work of Makarychev et al.~\cite{DBLP:conf/soda/MakarychevMV14} as [MMV14].

One of the main contributions of [MMV14], apart from introducing the notion of robust algorithms, was the introduction of a general technique for designing algorithms for stable instances of optimization problems that we heavily rely on in this thesis. Roughly speaking, [MMV14] introduced some sufficient conditions under which convex relaxations of stable instances are integral. Their result is strong, because it allows for the design of robust algorithms for $\gamma$-stable instances that can be simply stated as follows: solve the convex relaxation, and if it is integral then report solution, otherwise report that the instance is not stable. Using this technique, they proved that the CKR linear programming relaxation (\cite{DBLP:journals/jcss/CalinescuKR00}) for another classic graph partitioning/optimization problem, the (Edge) Multiway Cut problem, is integral for 4-stable instances.

A major contribution towards extending the research agenda proposed by Bilu and Linial was done by Awasthi, Blum \& Sheffet in 2012~\cite{DBLP:journals/ipl/AwasthiBS12}. In that work, the authors extend the definition of stability to clustering problems, and, as mentioned above, in order to make the distinction between standard optimization problems and clustering problems, they use the term \textit{perturbation resilience} to refer to essentially the same notion of stability. The authors then proceed to show that one can solve $3$-perturbation-resilient instances of so-called ``separable center-based" objectives, such as $k$-median, $k$-means and $k$-center. A bit later, Balcan and Liang~\cite{DBLP:journals/siamcomp/BalcanL16} showed that the stability threshold for these problems can be improved to $1+\sqrt{2} \approx 2.414$, and in 2016, Balcan, Haghtalab and White~\cite{DBLP:conf/icalp/BalcanHW16} showed that one can solve $2$-perturbation-resilient instances of both symmetric and asymmetric $k$-center. Moreover, they showed that this threshold of 2 is tight for $k$-center, unless $\NP = \textrm{RP}$. Finally, Makarychev and Makarychev~\cite{DBLP:journals/corr/MakarychevM16} showed that one can also solve 2-perturbation-resilient instances of $k$-median, $k$-means and other ``natural center-based" objectives, a result that was then merged with our results for Multiway Cut and covering problems in a single paper~\cite{DBLP:conf/stoc/AngelidakisMM17}.

Before proceeding to describe our results, we also introduce here one slightly relaxed notion of stability. The definition of Bilu-Linial stability, as given in Definition~\ref{def:perturbation-stability}, is quite strong, in that it imposes a lot of contraints in an instance. For that, [MMV14] also introduced a relaxed notion of stability, that allows the optimal solution to slightly change in a $\gamma$-perturbation. More concretely, they introduced the notion of \emph{weak stability}. The optimal solution of every perturbed instance of a weakly stable instance is close to the optimal solution of the original instance but may not be exactly the same. This is arguably a more realistic assumption than $\gamma$-stability in practice, and, following the techniques of [MMV14], our results in many cases extend to this setting as well. We now give the formal definition of weak stability in the context of graph optimization problems.

\begin{definition}[weak stability~\textrm{[MMV14]}] \label{def:weak-stability}
Let $\mathcal{I} = (G, w)$ be an instance of a graph optimization problem with a set of vertex or edge weights $w_i$, and suppose that it has a unique optimal solution $X^*$. Let $\mathcal{N}$ be a set of feasible solutions that contains $X^*$. We say that the instance $\mathcal{I}$ is $(\gamma, \mathcal{N})$-weakly-stable if for every $\gamma$-perturbation $(G, w')$ and every solution $X' \notin \mathcal{N}$, the solution $X^*$ has a strictly better cost than $X'$ w.r.t.~$w'$.
\end{definition}

Given the above definition, it is clear now that the notion of weak stability indeed generalizes the notion of stability: an instance is $\gamma$-stable if and only if it is $(\gamma, \{X^*\})$-weakly-stable, where $X^*$ is the unique optimal solution. We can think of the set $\mathcal{N}$ in the definition as a neighborhood of the optimal solution $X^*$, i.e.~it contains feasible solutions that are ``close enough" to the optimal one. Intuitively, the definition requires that every solution that is sufficiently different from the optimal solution be significantly worse compared to the optimal solution, but does not impose any restrictions on the solutions that are close enough to the optimal one.

We note here that, when given a $(\gamma, \mathcal{N})$-stable instance, the main task is to recover a solution $X \in \mathcal{N}$ in polynomial time. An interesting fact about the algorithms of [MMV14] (and our algorithms as well) is that the algorithm does not need to know anything about $\mathcal{N}$.

\section{Our results}

\textit{Several results mentioned in this section are based on the following works:
\begin{itemize}
    \item {[AMM17]}: Haris Angelidakis, Konstantin Makarychev, and Yury Makarychev. Algorithms for stable and perturbation-resilient problems. Appeared in STOC 2017~(\cite{DBLP:conf/stoc/AngelidakisMM17}).
    \item {[AMMW18]}: Haris Angelidakis, Konstantin Makarychev, Yury Makarychev, and Colin White. Work in progress~(\cite{AMMW18}).
    \item {[AABCD18]}: Haris Angelidakis, Pranjal Awasthi, Avrim Blum, Vaggos Chatziafratis, and Chen Dan. Bilu-Linial stability and the Independent Set problem. Preprint~(\cite{AABCD18}).
\end{itemize}
}

\vspace{10pt}
In this section, we describe the results that we will present in the next few chapters of this thesis. Starting with the Edge Multiway Cut problem, in [AMM17] we improve the stability threshold of Multiway Cut to $2 - 2/k$, where $k$ is the number of terminals, and we also give a polynomial-time algorithm that, given a $(2 - 2/k + \delta , \mathcal{N})$-weakly-stable instance of Minimum Multiway Cut with integer weights, finds a solution $E' \in \mathcal{N}$ (for every $\delta \geq 1 / \poly(n) > 0$). Moreover, we show a lower bound of $\frac{4}{3 + \frac{1}{k-1}} - \varepsilon$ for the stability threshold for which our current approach fails. Finally, we give the first results for the Node Multiway Cut problem, a strict generalization of the Edge Multiway Cut problem. In particular, we give a robust algorithm for $(k-1)$-stable instances of Node Multiway Cut (and an algorithm for $(k - 1 + \delta, \mathcal{N})$-weakly-stable instances with integer weights). We also utilize a well-known approximation-preserving reduction from Vertex Cover to Node Multiway Cut that, combined with the results of the following chapters, implies strong lower bounds on the existence of robust algorithms for Node Multiway Cut. Detailed presentation and proofs of the results for Multiway Cut can be found in Chapter~\ref{chap:multiway-cut}.

We then turn to standard covering problems. In all the results that follow, $n$ denotes the number of vertices in the graph. In [AMM17], we prove that there are no robust algorithms for $n^{1-\varepsilon}$-stable instances of Vertex Cover (and Independent Set), Set Cover, Min 2-Horn Deletion and Multicut on Trees, unless $\P = \NP$. These hardness results can be found in Chapter~\ref{chap:hardness}. On the positive side, in [AABCD18] we give robust algorithms for $(k - 1)$-stable instances of Independent Set on $k$-colorable graphs, for $(\Delta - 1)$-stable instances of Independent Set on graphs of maximum degree $\Delta$ and for $(1 + \varepsilon)$-stable instances of Independent Set on planar graphs. The algorithm for planar graphs can also be extended to work for $(1 + \varepsilon, \mathcal{N})$-weakly-stable instances with integer vertex weights. We also give a non-robust algorithm for $(\varepsilon n)$-stable instances of Independent Set on general graphs that runs in time $n^{O(1 / \varepsilon)}$. We note here that all results for Independent Set can be applied to the Vertex Cover problem as well, since the two problems are equivalent with respect to exact solvability and the notion of Bilu-Linial stability. All aforementioned results can be found in Chapter~\ref{chap:independent-set}.

In Chapter~\ref{chap:clustering} we initiate the study of convex relaxations for perturbation-resilient clustering. We present a robust algorithm for 2-metric-perturbation-resilient instances of symmetric $k$-center, and also give some non-integrality examples of perturbation-resilient instances for the standard $k$-median LP relaxation.

Finally, in Chapter~\ref{chap:tsp}, inspired by the work of Mihal{\'a}k et al.~\cite{DBLP:conf/sofsem/MihalakSSW11}, we study stable instances of the symmetric Traveling Salesman problem (TSP). In particular, we analyze the ``subtour-elimination" relaxation of TSP and prove that its integrality gap is 1 for 1.8-stable instances of TSP, thus giving a robust analog of the (non-robust) greedy approach of~\cite{DBLP:conf/sofsem/MihalakSSW11} for $1.8$-stable instances of TSP. These results can be found in [AMMW18].

\chapter{Stability and the Multiway Cut problem}\label{chap:multiway-cut}

In this chapter, we present a robust algorithm for $(2 - 2/k)$-stable instances of the Edge Multiway Cut problem and a robust algorithm for $(k - 1)$-stable instances of the Node Multiway Cut problem, where $k$ is the number of terminals. Moreover, following the [MMV14] framework, we give algorithms for weakly stable instances of these problems, with similar guarantees.

\section{The Edge Multiway Cut problem}

We first define the problem.
\begin{definition}[Edge Multiway Cut]
Let $G = (V, E)$ be a connected undirected graph and let $T = \{s_1, ..., s_k\} \subseteq V$ be a set of terminals. In the Edge Multiway Cut problem, we are given a function $w: E \to \R_{>0}$ and the goal is to remove the minimum weight set of edges $E' \subseteq E$ such that in the graph $G' = (V, E \setminus E')$ there is no path between any of the terminals.

Equivalently, the goal is to compute a partition $P_1, ..., P_k$ of the set $V$ such that $s_i \in P_i$, for each $i \in [k]$, $P_i \cap P_j = \emptyset$ for $i \neq j$ and $\bigcup P_i = V$, so as to minimize the weight of cut edges (i.e.~the edges whose endpoints are in different sets of the partition).
\end{definition}

The Multiway Cut problem is one of the very well studied graph partitioning problems. For $k = 2$, the problem is simply the Minimum $s-t$ cut problem, which is solvable in polynomial time. For $k \geq 3$, Dahlhaus et al.~\cite{DBLP:journals/siamcomp/DahlhausJPSY94} showed that the problem is APX-hard and gave a combinatorial $(2 - 2/k)$-approximation algorithm. Major progress in terms of the approximability of the problem was made with the introduction of the so-called CKR linear programming relaxation by C{\u{a}}linescu, Karloff and Rabani~\cite{DBLP:journals/jcss/CalinescuKR00} (see Figure~\ref{fig:CKR-LP}).

\begin{figure}
\begin{align*}
    \text{min}:  \quad & \sum_{e \in E} w_e \cdot d(e)\\
    \text{s.t.}: \quad & d(u,v) = \frac{1}{2} \cdot \|\bar u - \bar v\|_1, &&\textrm{for all } u,v \in V,\\
                       & \sum_{i=1}^k u_i = 1, &&\text{for every } u \in V.\\
                       & \bar{s}_j = e_j, &&\text{for every } j\in [k],\\
                       & u_j \geq 0, &&\text{for every } u \in V \text{ and } j\in [k].
\end{align*}
\caption{The CKR relaxation for Multiway Cut.}
\label{fig:CKR-LP}
\end{figure}

C{\u{a}}linescu, Karloff and Rabani gave a rounding scheme for this LP that yields a $(3/2 - 1/k)$-approximation algorithm for Multiway Cut. Karger et al.~\cite{DBLP:journals/mor/KargerKSTY04} gave improved rounding schemes for the relaxation; for general $k$, they gave a $1.3438$-approximation algorithm, and also pinpointed the integrality gap when $k = 3$; in particular, they gave a $12/11$-approximation algorithm and proved that this is tight by constructing an integrality gap example of ratio $12/11 - \varepsilon$, for every $\varepsilon > 0$. The same result was also independently discovered by Cunningham and Tang~\cite{DBLP:conf/ipco/CunninghamT99}. More recently, Buchbinder et al.~\cite{DBLP:conf/stoc/BuchbinderNS13} gave an elegant $4/3$-approximation algorithm for general $k$ and additionally showed how to push the ratio down to $1.3239$. Their algorithm was later improved by Sharma and Vondr\'{a}k~\cite{DBLP:conf/stoc/SharmaV14} to get an approximation ratio of $1.2965$. This remains the state-of-the-art approximation for sufficiently large $k$. Since the Sharma-Vondr\'{a}k algorithm is quite complicated and requires a computer-assisted proof, Buchbinder et al.~\cite{DBLP:conf/soda/BuchbinderSW17} recently came up with a simplified algorithm and analytically showed that it yielded roughly the same approximation ratio as Sharma and Vondr\'{a}k's.

The CKR relaxation also has a remarkable consequence on the approximability of the problem. Manokaran et al.~\cite{DBLP:conf/stoc/ManokaranNRS08} proved that, assuming the Unique Games Conjecture (UGC), if there exists an instance of Multiway Cut with integrality gap $\tau$ for the CKR relaxation, then it is \NP-hard to approximate Multiway Cut to within a factor of $\tau - \varepsilon$ of the optimum, for every constant $\varepsilon > 0$. Roughly speaking, Manokaran et al.~'s result means that, if one believes in the UGC, the CKR relaxation achieves essentially the best approximation ratio one can hope to get in polynomial time for Multiway Cut. Despite this strong connection, few lower bounds for the CKR relaxation are known. Apart from the aforementioned $(12 / 11 - \varepsilon)$ integrality gap for $k = 3$ by Karger et al.~\cite{DBLP:journals/mor/KargerKSTY04} and Cunningham and Tang~\cite{DBLP:conf/ipco/CunninghamT99}, the only other known lower bound until recently was an $8/\left(7 + \frac{1}{k - 1} \right)$-integrality gap which was constructed by Freund and Karloff~\cite{DBLP:journals/ipl/FreundK00} not long after the introduction of the CKR relaxation. Recently, in a joint work with Yury Makarychev and Pasin Manurangsi~\cite{DBLP:conf/ipco/AngelidakisMM17}, we gave an improved lower bound of $6/ \left(5 + \frac{1}{k - 1} \right) - \varepsilon$, for every constant $\varepsilon > 0$.

In the Bilu-Linial stability framework, an instance $G = (V, E, w)$ of Multiway Cut with terminal set $T \subseteq V$ is $\gamma$-stable if the instance has a unique optimal partition $(P_1, ..., P_k)$ and every $\gamma$-perturbation has the same unique optimal partition. Equivalently, the instance is $\gamma$-stable if it has a unique optimal solution $E^* \subseteq E$ and every $\gamma$-perturbation has the same unique optimal solution $E^*$. As was proved in the original paper of Bilu and Linial for Max Cut, one can equivalently write the definition of stability as follows.
\begin{definition}[stability~\cite{DBLP:journals/cpc/BiluL12}]\label{def:equiv-stab}
Let $G = (V, E, w)$ be an instance of Multiway Cut with terminal set $T \subseteq V$. The instance is $\gamma$-stable if and only if it has a unique optimal solution $E^* \subseteq E$ and for any feasible solution $E' \neq E^*$ we have
\begin{equation*}
    \gamma \cdot w(E^* \setminus E') < w(E' \setminus E^*).
\end{equation*}
\end{definition}
It is easy to see that the above definition of stability is equivalent to the original definition, as it considers the ``worst-case" perturbation that increases the edges cut by the optimal solution by a factor of $\gamma$.

In [MMV14], the authors prove a very interesting structural result about convex relaxations and their performance on stable instances, one of them being the CKR relaxation, as given in Figure~\ref{fig:CKR-LP}. In particular, they prove the following theorem, which we present here in the context of Multiway Cut.
\begin{theorem}[\textrm{[MMV14]}]\label{thm:MMV-original}
Let $G = (V, E, w)$ be an instance of Multiway Cut with terminal set $T \subseteq V$. Suppose that we are given a convex relaxation for Multiway Cut that assigns length $d(e) \in [0,1]$ to every edge $e \in E$ and its objective function is $\sum_{e \in E} w_e \cdot d(e)$. Let $d$ an be optimal fractional solution, and suppose that there exists a randomized rounding scheme that, for some $\alpha, \beta \geq 1$, always returns a feasible solution $E' \subseteq E$ such that for each edge $e \in E$, the following two conditions hold:
\begin{enumerate}
    \item $\Pr[e \textrm{ is cut}] \leq \alpha \cdot d(e)$ \quad\quad\quad\quad\;\;\;\;\quad\quad(approximation condition)
    \item $\Pr[e \textrm{ is not cut}] \geq 1/\beta \cdot (1 - d(e))$ \quad\quad(co-approximation condition)
\end{enumerate}
Then, the relaxation is integral for $(\alpha\beta)$-stable instances; in particular the relaxation has a unique optimal solution that assigns length 1 to every edge that is cut in the unique optimal integral solution, and 0, otherwise.
\end{theorem}

In [MMV14] it is observed  that it is highly non-trivial to satisfy both properties, and most rounding schemes for the CKR relaxation indeed do not satisfy both properties. However, they prove that the Kleinberg-Tardos rounding scheme~\cite{DBLP:journals/jacm/KleinbergT02} does satisfy both properties with $\alpha\beta = 4$. Thus, the CKR relaxation is integral for $4$-stable instances. We note here that the above theorem suggests a very simple algorithm: solve the relaxation, and if it is integral, then report the solution, otherwise report that the instance is not stable. In particular, the rounding scheme is not part of the algorithm but is only used in the analysis.

From now one, we call a rounding scheme that satisfies the above properties an \emph{$(\alpha, \beta)$-rounding}. Since the rounding scheme required by the above theorem is only needed in the analysis and is not part of the actual algorithm, in [AAM17] we observe that one can reprove the theorem by designing a rounding scheme only for fractional solutions that satisfy certain properties. In particular, we will consider ``almost integral" fractional solutions and we will design improved rounding schemes for such solutions. We reprove the theorem of [MMV14] here by incorporating the above observation. To do that, we first need two definitions.

\begin{definition}
Let $G = (V,E, w)$, $T \subseteq V$, be an instance of Multiway Cut. Fix $\varepsilon > 0$. We say that a feasible solution $\{(u_1, ..., u_k)\}_{u \in V}$ of the CKR LP relaxation for the instance $(G, T)$ is $\varepsilon$-close to an integral solution if $u_i \in [0, \varepsilon] \cup [1 - \varepsilon, 1]$ for every $u \in V$, $i \in [k]$.
\end{definition}

\begin{definition}[$\varepsilon$-local $(\alpha,\beta)$-rounding]
A randomized rounding scheme for $\varepsilon$-close solutions of the CKR relaxation is an $\varepsilon$-local $(\alpha,\beta)$-rounding, for some $\alpha,\beta \geq 1$, if, given a feasible solution $\{(u_1, ..., u_k)\}_{u \in V}$ with $u_i \in [0, \varepsilon] \cup [1 - \varepsilon, 1]$ for every $u \in V$, $i \in [k]$, it always returns a feasible solution $E' \subseteq E$ such that for each edge $e \in E$, the following two conditions hold:
\begin{enumerate}
    \item $\Pr[e \textrm{ is cut}] \leq \alpha \cdot d(e)$,
    \item $\Pr[e \textrm{ is not cut}] \geq 1/\beta \cdot (1 - d(e))$
\end{enumerate}
\end{definition}

\begin{theorem}\label{thm:stable-almost-integral}
Suppose that there exists an $\varepsilon$-local $(\alpha,\beta)$-rounding for the CKR LP relaxation, for some $\varepsilon > 0$. Then, the relaxation is integral for $(\alpha\beta)$-stable instances; in particular the relaxation has a unique optimal solution that assigns length 1 to every edge that is cut in the unique optimal integral solution, and 0, otherwise.
\end{theorem}
\begin{proof}
Let $G = (V, E, w)$ be an $(\alpha\beta)$-stable instance of Multiway Cut with terminals $T \subseteq V$, and let $E^* \subseteq E$ be its unique optimal solution. We denote $w(E^*) = OPT$. Let $\{\bar{u}^{OPT}\}_{u \in V}$ be the CKR solution corresponding to $E^*$ and $d^{OPT}$ be the resulting distance function. Let's assume now that the relaxation is not integral for $(G, T)$, which means that there exists a non-integral optimal solution $\{\bar{u}^{FRAC}\}_{u \in V}$ with corresponding distance function $d^{FRAC}$ such that $OPT_{LP} = \sum_{e \in E} w_e \cdot d^{FRAC}(e) \leq w(E^*)$. We now define $\bar{u}^{(\varepsilon)} = (1 - \varepsilon) \cdot \bar{u}^{OPT} + \varepsilon \cdot \bar{u}^{FRAC}$, for each $u \in V$. Clearly, $\bar{u}^{(\varepsilon)}$ is also non-integral, and by convexity, is a feasible solution. Let $d^{(\varepsilon)}(u,v) = \frac{1}{2} \sum_{i = 1}^k \left\|\bar{u}^{(\varepsilon)} - \bar{v}^{(\varepsilon)} \right\|_1$, for every $u, v \in V$.

We first prove that $\sum_{e \in E} w_e \cdot d^{(\varepsilon)}(e) \leq OPT$. From the subbaditivity of the $\ell_1$ norm, we get
\begin{align*}
    d^{(\varepsilon)}(u,v) &= \frac{1}{2} \left\|\bar{u}^{(\varepsilon)} - \bar{v}^{(\varepsilon)} \right\|_1 \leq \frac{1 - \varepsilon}{2} \cdot \left\|\bar{u}^{OPT} - \bar{v}^{OPT} \right\|_1 + \frac{\varepsilon}{2} \cdot \left\|\bar{u}^{FRAC} - \bar{v}^{FRAC} \right\|_1 \\
                           &= (1 - \varepsilon) \cdot d^{OPT}(u,v) + \varepsilon \cdot d^{FRAC}(u,v).
\end{align*}
Thus, we get that
\begin{equation*}
    \sum_{e \in E} w_e \cdot d^{(\varepsilon)}(e) \leq (1 - \varepsilon) \cdot OPT + \varepsilon \cdot OPT_{LP} \leq OPT.
\end{equation*}

We now apply the $\varepsilon$-local $(\alpha, \beta)$-rounding to $\{\bar{u}^{(\varepsilon)}\}_{u \in V}$ and $d^{(\varepsilon)}$ and obtain a feasible solution $E' \subseteq E$. Observe that $d^{(\varepsilon)}$ is non-integral, and thus there exists at least one edge $e \in E^*$ such that $d'(e) < 1$ and an edge $e' \in E \setminus E^*$ such that $d'(e') > 0$. This implies that $\Pr[E' \neq E^*] > 0$.

By the definition of Bilu-Linial stability (see Definition~\ref{def:equiv-stab}), in the case where $E' \neq E^*$, we have $(\alpha \beta) \cdot w(E^* \setminus E') < w(E' \setminus E^*)$. By monotonicity of expectation, this implies that
\begin{equation*}
    \E[(\alpha\beta) \cdot w(E^* \setminus E')] < \E[w(E' \setminus E^*)].
\end{equation*}
We now expand each term. We have
\begin{equation*}
    \E[w(E' \setminus E^*)] = \sum_{e \in E \setminus E^*} w_e \Pr[e \in E'] \leq \alpha \sum_{e \in E \setminus E^*} w_e \cdot d^{(\varepsilon)}(e),
\end{equation*}
and
\begin{equation*}
    \E[(\alpha\beta) \cdot w(E^* \setminus E')] = (\alpha\beta) \cdot \sum_{e \in E^*} w_e \Pr[u \notin E'] \geq \alpha \cdot \sum_{u \in E^*} w_e (1 - d^{(\varepsilon)}(e)).
\end{equation*}
Putting things together, we get that
\begin{equation*}
    \sum_{u \in E^*} w_e (1 - d^{(\varepsilon)}(e)) < \sum_{e \in E \setminus E^*} w_e d^{(\varepsilon)}(e),
\end{equation*}
which gives $w(E^*) < \sum_{e \in E} w_e d^{(\varepsilon)}(e)$. Thus, we get a contradiction.
\end{proof}

The above theorem implies that it is sufficient to design a rounding scheme that satisfies the desired properties and works only for ``almost-integral" fractional solutions. We do this in the next section.

\subsection{An improved analysis of the CKR relaxation on stable instances}

Let $G = (V, E, w)$ be an instance of Multiway Cut with terminal set $T = \{s_1, ..., s_k\} \subseteq V$. We now present an $\varepsilon$-local $(\alpha, \beta)$-rounding with $\alpha \beta = 2 - 2 / k$ and $\varepsilon = 1 / (10 k)$. Since the LP solutions we consider are $\varepsilon$-close to an integral one, for every vertex $u \in V$ there exists a unique $j\in [k]$ such that $d(u, s_j) \leq \varepsilon$. We denote this $j$ by $j(u)$. Note that, in particular, $u_{j(u)} \geq 1- \varepsilon$ and $u_{j'} \leq \varepsilon$ for $j' \neq j(u)$. We now present our rounding scheme (see Algorithm~\ref{alg:epsilon-rounding-mc}).

\begin{algorithm}[h]
\begin{enumerate}
    \item Let $p = 1/k$, $\theta = 6/(5k) $ (note that $\theta > \varepsilon$).
    \item Choose $r\in (0,\theta)$ uniformly at random.\\
    \item Choose $i\in \{1,\dots, k\}$ uniformly at random.\\
    \item With probability $p$ apply rule \textbf{A} to every $u \in V$; with probability~$1-p$  apply rule \textbf{B}\\ to every $u \in V$:\\
            \quad\quad rule \textbf{A}: if $u_{j(u)} \geq 1 - r$, add $u$ to $P_{j(u)}$; otherwise, add $u$ to $P_i$\\
            \quad\quad rule \textbf{B}: if $u_{i} < r$, add $u$ to $P_{j(u)}$; otherwise, add $u$ to $P_i$
    \item Return partition $P=(P_1, \dots, P_k)$.
\end{enumerate}
\caption{An $(\alpha, \beta)$-rounding for ``almost integral" fractional solutions of Edge Multiway Cut}
\label{alg:epsilon-rounding-mc}
\end{algorithm}

\begin{theorem}\label{thm:rounding}
Algorithm~\ref{alg:epsilon-rounding-mc} is an $\varepsilon$-local $(\alpha, \beta)$-rounding for the CKR relaxation of Multiway Cut with $\alpha \beta = 2 -2/k$ and $\varepsilon = 1/(10 k)$. Given a solution $\varepsilon$-close to an integral one, the algorithm runs in polynomial time and generates a distribution of multiway cuts with a domain of polynomial size.
\end{theorem}
\begin{proof}
First, we show that the algorithm returns a feasible solution. To this end, we prove that the algorithm always adds $u = s_t$ to $P_t$. Note that $j(u) = t$. If the algorithm uses rule \textbf{A}, then $u_{j(u)} = 1 > 1 - r$, and thus it adds $u$ to $P_{j(u)} = P_t$. If the algorithm uses rule \textbf{B}, then $u_{i} \geq r$ only when $i = j(u)$; thus the algorithm adds $u$ to $P_{j(u)}= P_t$, as required.

Let
\begin{equation*}
    \alpha = \frac{2 (k-1)}{k^2 \theta} = \frac{5}{3} \Bigl(1 - \frac{1}{k}\Bigr) \qquad  \text{and} \qquad \beta = k\theta = \frac{6}{5}.
\end{equation*}
We will show now that the rounding scheme satisfies the approximation and co-approximation conditions with parameters $\alpha$ and $\beta$. Consider two vertices $u$ and $v$. Let $\Delta = d(u,v)$. We verify that the approximation condition holds for $u$ and $v$. There are two possible cases: $j(u) = j(v)$ or $j(u) \neq j(v)$. Consider the former case first. Denote $j = j(u) = j(v)$. Note that $P(u) \neq P(v)$ if and only if one of the vertices is added to $P_{i}$ and the other to $P_{j}$, and $i\neq j$. Suppose first that rule \textbf{A} is applied. Then, $P(u) \neq P(v)$ exactly when $1-r \in (\min(u_j, v_j),\max(u_j, v_j)]$ and $i\neq j$.
The probability of this event (conditioned on the event that rule \textbf{A} is applied) is
\begin{align*}
    \Pr[i\neq j] \cdot \Pr \left[1 - r \in (\min(u_j, v_j),\max(u_j, v_j)] \right] &= \frac{k-1}{k} \cdot\frac{\max(u_j, v_j) -\min(u_j, v_j)}{\theta}\\
                                                                                   &= \frac{k-1}{k} \cdot\frac{|u_j - v_j|}{\theta}.
\end{align*}
(here we used that $\max(u_j, v_j) \geq 1 - \varepsilon > 1- \theta$). Now suppose that rule \textbf{B} is applied. Then, we have $P(u) \neq P(v)$ exactly when
$r\in (\min(u_i, v_i),\max(u_i, v_i)]$ and $i \neq j$. The probability of this event (conditioned on the event that rule \textbf{B} is used) is
\begin{equation*}
\frac{1}{k}\sum_{i:i\neq j} \Pr[r\in (\min(u_i, v_i),\max(u_i, v_i)]] = \frac{1}{k}\sum_{i:i\neq j} \frac{|u_i - v_i|}{\theta},
\end{equation*}
where again we use the fact that $\varepsilon < \theta$. Thus,
\begin{align*}
    \Pr[P(u) \neq P(v)] &= p \cdot \frac{k-1}{k} \cdot \frac{|u_j - v_j|}{\theta}  + (1-p) \cdot \frac{1}{k} \cdot \sum_{i:i\neq j} \frac{|u_i - v_i|}{\theta}\\
                        &= \frac{k-1}{k^2 \theta} \sum_{i \in [k]} |u_i - v_i| = \frac{2(k-1)}{k^2 \theta} \Delta = \alpha \Delta.
\end{align*}

Now consider the case when $j(u) \neq j(v)$. Then, the approximation condition holds simply because $\Pr[P(u) \neq P(v)] \leq 1$ and $\alpha \Delta \geq 1$. Namely, we have
\begin{equation*}
    \Delta = d(u, v) \geq d(s_{j(u)}, s_{j(v)}) - d(u, s_{j(u)}) - d(v, s_{j(v)}) \geq 1 - 2\varepsilon \geq 1 - 2/30 = 14/15,
\end{equation*}
and $\alpha \geq \frac{5}{3}\left(1 - \frac{1}{3}\right) = 10/9$; thus, $\alpha \Delta \geq (10/9) \times (14/15) > 1$.

Let us verify that the co-approximation condition holds for $u$ and $v$. Assume first that $j(u) = j(v)$. Let $j= j(u) = j(v)$. Then, $\Delta= d(u,v) \leq d(u, s_j) + d(v, s_j) \leq 2 \varepsilon \leq 1/15$. As we showed, $\Pr[P(u) \neq P(v)] \leq \alpha \Delta$. This implies that $\Pr[P(u) = P(v)] \geq 1 - \alpha \Delta \geq \beta^{-1} (1 - \Delta)$, where the last bound follows from  the following inequality $\frac{1-\beta^{-1}}{\alpha -\beta^{-1}} \geq \frac{1/6}{5/3 - 5/6} = \frac{1}{5} \geq \Delta$.

Assume now that $j(u) \neq j(v)$. Without loss of generality, we assume that $u_{j(u)} \leq v_{j(v)}$. Suppose that rule \textbf{A} is applied. Event $P(u) = P(v)$ happens in the following disjoint cases:
\begin{enumerate}
    \item $u_{j(u)} \leq v_{j(v)} < 1 - r$ (then both $u$ and $v$ are added to $P_i$);
    \item $u_{j(u)} < 1 - r \leq v_{j(v)}$ and $i = j(v)$.
\end{enumerate}
The probabilities that the above happen are $(1 - v_{j(v)})/\theta$ and $(v_{j(v)} - u_{j(u)})/\theta \times (1/k)$, respectively. Note that $d_u \equiv d(u, s_{j(u)}) = \frac{1}{2}\left(1 - u_{j(u)} + \sum_{t:t\neq j(u)} u_t\right) = 1 - u_{j(u)}$, since we have $\sum_{t:t\neq j(u)} u_t = 1 - u_{j(u)}$. Similarly, $d_v \equiv d(v, s_{j(v)}) =  1 - v_{j(v)}$.
We express the total probability that one of the two cases happens in terms of $d_u$ and $d_v$ (using that $\Delta \geq d(s_{j(u)}, s_{j(v)}) - d_u - d_v = 1 - d_u - d_v$):
\begin{equation*}
    \left(d_v + \frac{d_u - d_v}{k}\right) \cdot \frac{1}{\theta} = \frac{(k-1)d_v + d_u}{\theta k} \geq \frac{d_u + d_v}{\theta k} \geq \frac{1- \Delta}{\theta k} = \beta^{-1} (1-\Delta).
\end{equation*}

Now, suppose that rule \textbf{B} is applied. Note that if $u_i \geq r$ and $v_i \geq r$, then both $u$ and $v$ are added to $P_i$, and thus $P(u) = P(v)$. Therefore,
\begin{align*}
    \Pr[P(u) = P(v) | \text{ rule \textbf{B}}] &\geq \Pr[u_i \geq r ,\  v_i \geq r] =\frac{1}{k} \sum_{i=1}^k \frac{\min(u_i, v_i)}{\theta} = \frac{1}{k\theta} \sum_{i=1}^k \frac{u_i+v_i - |u_i - v_i|}{2} \\
                                               &= \frac{1}{k\theta} (1 - \Delta) = \beta^{-1} (1 - \Delta).
\end{align*}
We conclude that
\begin{equation*}
    \Pr[P(u) = P(v)] \geq p \cdot \beta^{-1} (1 - \Delta) + (1-p) \cdot \beta^{-1} (1 - \Delta) = \beta^{-1} (1 - \Delta).
\end{equation*}
We have verified that both conditions hold for $\alpha = 2 (k-1)/(k^2 \theta)$ and $\beta = k\theta$. As required, $\alpha \beta = 2 - 2/k$.

The algorithm clearly runs in polynomial-time. Since the algorithm generates only two random variables $i$ and $r$, and additionally makes only one random decision, the size of the distribution of $P$ is at most $2\times k \times (nk) = 2k^2 n$.
\end{proof}

From Theorems~\ref{thm:stable-almost-integral} and~\ref{thm:rounding} we get the main theorem of this section.
\begin{theorem}
The optimal LP solution for a $(2-2/k)$-stable instance of Minimum Multiway Cut is integral. Consequently, there is a robust polynomial-time algorithm for solving $(2 - 2 / k)$-stable instances.
\end{theorem}


\subsection{An improved algorithm for weakly stable instances of Edge Multiway Cut}\label{sec:weakly-stable-MC}

In this section, we show that the improved analysis of the previous section can be applied to weakly stable instances, following closely the techniques of [MMV14]. As a reminder, [MMV14] gives an algorithm for $(4 + \delta, \mathcal{N})$-weakly-stable instances of Multiway Cut. In this section, we present an algorithm for $(2 - 2/k + \delta, \mathcal{N})$-weakly-stable instances of Multiway Cut. To do so, we prove the following theorem.

\begin{theorem}\label{thm:multiway-alphabeta-weak}
Assume that there is a polynomial-time $\varepsilon$-local $(\alpha, \beta)$-rounding for the CKR relaxation, for some $\varepsilon = \varepsilon(n, k) > 1/\poly(n)$; here $n$ is the number of vertices and $k$ is the number of terminals in the instance. Moreover, assume that the support of the distribution of multiway cuts generated by the rounding has polynomial size\footnote{If we do not make this assumption, we can still get a \textit{randomized} algorithm for $(\alpha \beta + \delta, N)$-weakly stable instances.}. Let $\delta > 1/\poly(n) > 0$. Then, there is a polynomial-time algorithm for $(\alpha \beta + \delta, \mathcal{N})$-weakly-stable instances of Minimum Multiway Cut with integer weights. Given an  $(\alpha \beta + \delta, \mathcal{N})$-weakly-stable instance, the algorithm finds a solution $E' \in \mathcal{N}$ (the algorithm does not know the set $\mathcal{N}$).
\end{theorem}

For the proof of the above theorem, we will need two lemmas. We will use the following lemma from [MMV14].
\begin{lemma}[\textrm{[MMV14]}]\label{lem:alt-def-weak-stability-multiway}
Consider a $(\gamma, \mathcal{N})$-weakly-stable instance of Minimum Multiway Cut. Let $E^*$ be its unique optimal solution. Then for every multiway cut $E' \notin \mathcal{N}$, we have
\begin{equation*}
    \gamma \cdot w(E^* \setminus E') < w(E' \setminus E^*).
\end{equation*}
\end{lemma}

We also prove the following lemma, similar in spirit to [MMV14].
\begin{lemma}\label{lem:improve-multiway}
Suppose that there is a polynomial-time $\varepsilon$-local $(\alpha,\beta)$-rounding for the CKR relaxation, where $\varepsilon \geq 1 / \poly(n) > 0$. Let $\delta \geq 1/\mathrm{poly}(n)>0$. Then there is a polynomial-time algorithm that, given an $(\alpha \beta + \delta, \mathcal{N})$-weakly-stable instance of Minimum Multiway Cut and a feasible multiway cut $E^{\circ}$, does the following:
\begin{itemize}
    \item if $E^{\circ}\notin \mathcal{N}$, it finds a multiway cut  $E'$ such that
        \begin{equation*}
            w(E') - w(E^*) \leq (1 - \tau) \left(w(E^{\circ}) - w(E^*)\right),
        \end{equation*}
        where $E^*$ is the minimum multiway cut, and $\tau = \frac{\varepsilon \delta}{\beta(\alpha\beta + \delta)}\geq \frac{1}{\poly(n)} > 0$.
    \item if $E^{\circ}\in \mathcal{N}$, it either returns a multiway cut $E'$ better than $E^{\circ}$ or certifies that $E^{\circ} \in \mathcal{N}$.
\end{itemize}
\end{lemma}
\begin{proof}
We define edge weights $w_e'$ by
\begin{equation*}
w_e' =
    \begin{cases}
        w_e, & \textrm{if } e \in E^\circ,\\
        (\alpha \beta) \cdot w_e, & \text{otherwise}.
    \end{cases}
\end{equation*}
We solve the CKR LP relaxation for Minimum Multiway Cut with weights $\{w_e'\}_{e \in E}$. If we get an integral solution, then we report the solution and we are done (since in this case we have a solution $E'$ that satisfies $w'(E') \leq w'(E^*)$ for some $\gamma$-perturbation $w'$, which implies that $E' \in \mathcal{N}$). So, suppose that we get a non-integral optimal solution $\{\bar{u}\}_{u \in V}$ with corresponding distance function $d$. Let $OPT'$ denote the optimal integral value of the instance $G' = (V,E, w')$. We have $\sum_{e \in E} w_e' d(e) \leq OPT'$. Let $\{\bar{u}^\circ\}_{u \in V}$, $d^\circ$, denote the CKR solution corresponding to solution $E^\circ$. For each $u \in V$, we define $\bar{u}^{(\varepsilon)} = (1 - \varepsilon) \bar{u}^\circ + \varepsilon \bar{u}$. From the subadditivity of the $\ell_1$ norm we again get that for every $u, v \in V$
\begin{equation*}
    d^{(\varepsilon)}(u,v) \leq (1 - \varepsilon) d^\circ(u,v) + \varepsilon d(u,v).
\end{equation*}

We now apply the $\varepsilon$-local $(\alpha, \beta)$-rounding to this solution and get a feasible solution $E'$. We have

\begin{align*}
    \E[w(E^{\circ}) - w(E')] &= \E[w(E^{\circ} \setminus E')] - \E[w(E' \setminus E^{\circ})] \\
                             &= \sum_{e \in E^{\circ}} w_e \Pr[e \textrm{ is not cut}] - \sum_{e \in E \setminus E^{\circ}} w_e \Pr[e \textrm{ is cut}]\\
                             &\geq \frac{1}{\beta} \sum_{e \in E^{\circ}} w_e \cdot \left(1 - d^{(\varepsilon)}(e) \right) - \alpha \sum_{e \in E \setminus E^{\circ}} w_e \cdot d^{(\varepsilon)}(e)\\
                             &= \frac{1}{\beta} \left( w(E^{\circ}) - \sum_{e \in E} w_e' \cdot d^{(\varepsilon)}(e) \right)\\
                             &\geq \frac{\varepsilon}{\beta} \left( w(E^{\circ}) - OPT' \right)\\
                             &\geq \frac{\varepsilon}{\beta} \left( w(E^{\circ}) - w'(E^*) \right)\\
                             &= \frac{\varepsilon}{\beta} \left( w(E^{\circ} \setminus E^*) - \alpha\beta \cdot w(E^* \setminus E^{\circ}) \right).
\end{align*}
Let's assume now that $E^{\circ} \notin \mathcal{N}$. Using Lemma~\ref{lem:alt-def-weak-stability-multiway}, we get that $w(E^* \setminus E^{\circ}) < \frac{1}{\alpha\beta + \delta} \cdot w(E^{\circ} \setminus E^*)$. Thus, we conclude that
\begin{equation*}
    \E[w(E^{\circ}) - w(E')] > \frac{\varepsilon\delta}{\beta(\alpha\beta + \delta)} \cdot w(E^{\circ} \setminus E^*) \geq \frac{\varepsilon\delta}{\beta(\alpha\beta + \delta)} \cdot \left(w(E^{\circ}) - w(E^*) \right),
\end{equation*}
which implies that
\begin{equation*}
    \E[w(E') - w(E^*)] < (1 - \tau) \cdot (w(E^\circ) - w(E^*)).
\end{equation*}
This further implies that there exists at least one multiway cut $E'$ in the distribution that the rounding scheme produces such that $w(E') - w(E^*) < (1 - \tau) \cdot (w(E^\circ) - w(E^*))$. Since the distribution has polynomial-sized support,  we can efficiently identify this $E'$.

Note that the algorithm does not know whether $E^\circ \in \mathcal{N}$ or not; it tries all multiway cuts $E'$ and finds the best one $E''$. If $E''$ is better than $E^\circ$, the algorithm returns $E''$; otherwise, it certifies that $E^\circ \in \mathcal{N}$.
\end{proof}

\begin{proof}[Proof of Theorem~\ref{thm:multiway-alphabeta-weak}]
We assume that all edge costs are integers between $1$ and some $W$. Let $C^*$ be the cost of the optimal solution. We start with an arbitrary feasible multiway cut $E^{(0)}$. Denote its cost by $C^{(0)}$. Let $T  = \lceil \log_{1/(1-\tau)} C^{(0)} \rceil + 2= O(n^2 \tau \log W)$ (note that $T$ is polynomial in the size of the input). We iteratively apply the algorithm from Lemma~\ref{lem:improve-multiway} $T$ times: first we get a multiway cut $E^{(1)}$ from $E^{(0)}$, then $E^{(2)}$ from $E^{(1)}$, and so on. Finally, we get a multiway cut $E^{(T)}$. If at some point the algorithm does not return a multiway cut, but certifies that the current
multiway cut $E^{(i)}$ is in $\mathcal{N}$, we output $E^{(i)}$ and terminate the algorithm.

So, let's assume now that the algorithm does $T$ iterations, and we get multiway cuts $E^{(0)}, ..., E^{(T)}$. Denote the cost of $E^{(i)}$ by $C^{(i)}$. Note that $C^{(0)} > C^{(1)} > ... > C^{(T)}\geq C^*$. Further, if $E^{(i)} \notin \mathcal{N}$ then $C^{(i+1)} - C^* \leq (1-\tau) (C^{(i)} - C^*)$ and thus $C^{(i+1)} - C^{(T)} \leq (1-\tau) (C^{(i)} - C^{(T)})$. Observe that we cannot have $C^{(i+1)} - C^{(T)} \leq (1-\tau) (C^{(i)} - C^{(T)})$ for every $i$, because then we would have
\begin{equation*}
    C^{(T-1)} - C^{(T)} \leq (1 - \tau)^{T-1} (C^{(0)} - C^{(T)}) \leq (1-\tau)^{T-1} C^{(0)} < 1,
\end{equation*}
which contradicts to our assumption that all edge weights are integral (and, consequently, that $C^{(T-1)} - C^{(T)}$ is a positive integer number). We find an $i$ such that $C^{(i+1)} - C^{(T)} > (1-\tau) (C^{(i)} - C^{(T)})$ and output $E^{(i)}$. We are guaranteed that $P^{(i)} \in N$.
\end{proof}

Combining Theorems~\ref{thm:rounding} and~\ref{thm:multiway-alphabeta-weak}, we obtain the main result of this section.
\begin{theorem}
There is a polynomial-time algorithm that given a $(2 - 2/k + \delta, \mathcal{N})$-weakly-stable instance of Minimum Multiway Cut with integer weights, finds a solution $E' \in \mathcal{N}$ (for every $\delta \geq 1/\poly(n) > 0$).
\end{theorem}


\subsection{Lower bounds for the CKR relaxation on stable instances}

In this section, we present a lower bound for integrality of stable instances for the CKR relaxation for Minimum Multiway Cut. For that, we first make two claims regarding the construction of stable instances and the use of integrality gap examples as lower bounds for integrality of stable instances. We state both claims in the setting of Minimum Multiway Cut, but they can be easily applied to other partitioning problems as well.

\begin{claim}\label{MC_claim1}
Given an instance $G = (V, E, w)$, $w: E \to \mathbb{R}_{\geq 0}$, of Minimum Multiway Cut with terminals $T = \{s_1, ..., s_k\}$, and an optimal solution $E^* \subseteq E$, for every $\gamma > 1$ and every $\varepsilon \in (0, \gamma - 1)$, the instance $G^{(E^*,\gamma)} = (V,E, w^{(E^*, \gamma)})$, where $w_e^{(E^*, \gamma)} = w_e / \gamma$ for $e \in E^*$, and $w_e^{(E^*, \gamma)} = w_e$ for $e \in E \setminus E^*$, is a $(\gamma - \varepsilon)$-stable instance (whose unique optimal solution is $E^*$).
\end{claim}
\begin{proof}
First, it is easy to see that for every $\gamma > 1$, $E^*$ is the unique optimal solution for $G^{(E^*, \gamma)}$. We will now prove that $G^{(E^*, \gamma)}$ is ($\gamma - \varepsilon)$-stable, for every $\varepsilon \in (0, \gamma - 1)$. For that, we consider any $(\gamma - \varepsilon)$-perturbation of $G^{(E^*, \gamma)}$. More formally, this is a graph $G' = (V, E, w')$, where $w_e' = f(e) \cdot w_e^{(E^*, \gamma)}$, and $f(e) \in [1, \gamma - \varepsilon]$ for all $e \in E$. Let $\bar{E} \neq E^*$ be any feasible solution of $(G', T)$. We have
\begin{align*}
    w'(\bar{E}) &= \sum_{e \in E^*} w_e' - \sum_{e \in E^* \setminus \bar{E}} w_e' + \sum_{e \in \bar{E} \setminus E^*} w_e'\\
                &= w'(E^*) - \sum_{e \in E^* \setminus \bar{E}} f(e) w_e^{(E^*, \gamma)} + \sum_{e \in \bar{E} \setminus E^*} f(e) w_e^{(E^*, \gamma)} \\
                &\geq w'(E^*) - (\gamma - \varepsilon) \sum_{e \in E^* \setminus \bar{E}} w_e^{(E^*, \gamma)} + \sum_{e \in \bar{E} \setminus E^*} w_e^{(E^*, \gamma)} \\
                &= w'(E^*) - \frac{\gamma - \varepsilon}{\gamma} \sum_{e \in E^* \setminus \bar{E}} w_e + \sum_{e \in \bar{E} \setminus E^*} w_e \\
                &> w'(E^*) - \sum_{e \in E^* \setminus \bar{E}} w_e + \sum_{e \in \bar{E} \setminus E^*} w_e \\
                &= w'(E^*) - \sum_{e \in E^*} w_e + \sum_{e \in \bar{E}} w_e \\
                &\geq w'(E^*),
\end{align*}
where the last inequality holds because $\bar{E}$ is a feasible solution for the original instance $(G, T)$ while $E^*$ is an optimal solution for $(G, T)$. Thus, $w'(\bar{E}) > w'(E^*)$, and so $E^*$ is the unique optimal solution for every $(\gamma - \varepsilon)$-perturbation of $G^{(E^*,\gamma)}$. We conclude that $G^{(E^*,\gamma)}$ is $(\gamma - \varepsilon)$-stable.
\end{proof}

We will now use the above claim to show how an integrality gap example for Minimum Multiway Cut can be converted to a certificate of non-integrality of stable instances.
\begin{claim}\label{IG_stability_claim}
Let $(G, T)$ be an instance of Minimum Multiway Cut, such that $OPT/OPT_{LP} = \alpha > 1$, where $OPT$ is the value of an optimal integral Multiway Cut, and $OPT_{LP}$ is the value of an optimal fractional solution for the CKR relaxation. Then, for every $\varepsilon \in (0, \alpha - 1)$, we can construct an $(\alpha - \varepsilon)$-stable instance such that the CKR relaxation is not integral for that instance.
\end{claim}
\begin{proof}
Let $G = (V, E, w)$, $T \subseteq V$, be an instance of Minimum Multiway Cut such that $OPT/OPT_{LP} = \alpha > 1$. Let $\gamma = \alpha - \delta$, for any fixed $\delta \in (0, \alpha - 1)$.  Let $E^*$ be an optimal integral solution, i.e.~$OPT = \sum_{e \in E^*} w_e$. By Claim \ref{MC_claim1}, for every $\varepsilon' \in (0, \gamma - 1)$, $G^{(E^*, \gamma)}$ is a $(\gamma - \varepsilon')$-stable instance whose unique optimal solution is $E^*$. Let $\{\bar{u}\}_{u \in V}$ be an optimal LP solution for $G$. We define $d(u,v) = \frac{1}{2} \|\bar{u} - \bar{v}\|_1$, for every $u, v \in V$, and we have $OPT_{LP} = \sum_{e \in E} w_e d(e)$. Note that $\{\bar{u}\}_{u \in V}$ is a feasible fractional solution for $G^{(E^*, \gamma)}$, and we claim that its cost for $G^{(E^*, \gamma)}$ is strictly smaller than the (integral) cost of the optimal solution $E^*$ of $G^{(E^*, \gamma)}$. For that, we have
\begin{equation*}
\begin{split}
    w^{(E^*, \gamma)}(E^*) &= \sum_{e \in E^*} w_e^{(E^*, \gamma)} = \frac{1}{\gamma} \sum_{e \in E^*} w_e = \frac{\alpha}{\alpha - \delta} \sum_{e \in E} w_e d(e) \\
                           &> \sum_{e \in E} w_e d(e) \geq \sum_{e \in E} w_e^{(E^*, \gamma)} d(e),
\end{split}
\end{equation*}
which implies that the LP is not integral for the instance $G^{(E^*, \gamma)}$. Setting $\delta = \varepsilon' = \varepsilon / 2$ finishes the proof.
\end{proof}
Claim \ref{IG_stability_claim} allows us to convert any integrality gap result for the CKR relaxation into a lower bound for non-integrality. Thus, by using the Freund-Karloff integrality gap construction \cite{DBLP:journals/ipl/FreundK00}, we can deduce that there are $\left(\frac{8}{7 + \frac{1}{k-1}} - \varepsilon \right)$-stable instances of Minimum Multiway Cut for which the CKR relaxation is not integral. An improved integrality gap construction given in a joint work with Yury Makarychev and Pasin Manurangsi~\cite{DBLP:conf/ipco/AngelidakisMM17} (as mentioned in the beginning of the chapter) also implies that there are $\left(\frac{6}{5 + \frac{1}{k-1}} - \varepsilon \right)$-stable instances of Minimum Multiway Cut for which the CKR relaxation is not integral. But, with a more careful analysis, we can obtain a stronger lower bound. More formally, we prove the following theorem.

\begin{theorem}
For every $\varepsilon > 0$ and $k \geq 3$, there exist $\left(\frac{4}{3 + \frac{1}{k-1}} - \varepsilon \right)$-stable instances of Minimum Multiway Cut with $k$ terminals for which the CKR relaxation is not integral.
\end{theorem}
\begin{proof}
We use the Freund-Karloff construction \cite{DBLP:journals/ipl/FreundK00}, that is, for any $k$, we construct the graph $G = (V, E, w)$, where the set of vertices is $V = \{1, ..., k\} \cup \{(i,j): 1 \leq i < j \leq k\}$, and the set of edges is $E = E_1 \cup E_2$, $E_1 = \left\{ [i, (i,j)], [j, (i,j)]: 1 \leq i < j \leq k \right\}$ and $E_2 = \left\{ [(i,j), (i',j')]: i < j, i' < j', |\{i,i',j,j'\}| = 3 \right\}$. Here, we use the notation $[u,v]$ to denote an edge, instead of the standard $(u,v)$, so as to avoid confusion with the tuples used to describe the vertices. The set of terminals is $T = \{1, .., k\} \subset V$. The weights are set in the same way as in the Freund and Karloff construction, i.e.~the edges in $E_1$ all have weight 1 and the edges in $E_2$ all have weight $w = \frac{3}{2k}$. Freund and Karloff proved that by setting the weights in this way, the graph has an optimal solution that assigns every vertex $(i,j), i < j$, to terminal $i$. Let $E^* \subseteq E$ be the edges cut by this solution. We have $OPT = w(E^*) = \binom{k}{2} + \frac{3}{2k} \cdot 2 \binom{k}{3} = (k - 1)^2$. They also proved that an optimal fractional solution assigns each vertex $(i, j)$ to the vector $(e_i + e_j) / 2$, and, thus, the (fractional) length of each edge $e \in E$ is $d(e) = \frac{1}{2}$. This implies that $OPT_{LP} = \frac{1}{2} \sum_{e \in E} w_e = \frac{1}{2} \cdot \left(2 \binom{k}{2} + \frac{3}{2k} \cdot 3 \binom{k}{3} \right) = OPT / \frac{8}{7 + \frac{1}{k-1}}$.

We now scale the weights of all edges in $E^*$ down by a factor $\gamma > 1$, and, by Claim \ref{MC_claim1}, obtain a $(\gamma - \varepsilon)$-stable instance $G^{(E^*, \gamma)}$, whose unique optimal solution is $E^*$. The cost of this optimal solution is $OPT_\gamma = \frac{1}{\gamma} \cdot OPT$. We consider the same fractional solution that assigns every node $(i,j)$ to the vector $(e_i + e_j) / 2$. The fractional cost now is:
\begin{equation*}
    X^{(E^*,\gamma)} = \frac{1}{2} \left[\frac{1}{\gamma} \cdot \binom{k}{2} + \frac{3}{2\gamma k} \cdot 2\binom{k}{3} \right] + \frac{1}{2}\left[ \binom{k}{2} + \frac{3}{2k} \binom{k}{3} \right].
\end{equation*}
We want to maintain non-integrality, i.e.~we want $OPT_\gamma > X^{(E^*, \gamma)}$. Thus, we must have
\begin{equation*}
\begin{split}
    \frac{1}{2\gamma} (k - 1)^2 > \frac{1}{8} (k - 1) (3k - 2), \;\;\textrm{ which gives}\;\; \gamma < \frac{4(k - 1)}{3k - 2}.
\end{split}
\end{equation*}
This implies that, for every $\varepsilon > 0$, there exist $\left(\frac{4}{3 + \frac{1}{k-1}} - \varepsilon \right)$-stable instances of Minimum Multiway Cut with $k$ terminals that are not integral with respect to the CKR relaxation.
\end{proof}


\section{The Node Multiway Cut problem}\label{sec:node-mc}

We first define the problem.
\begin{definition}[Node Multiway Cut]
Let $G = (V, E)$ be a connected undirected graph and let $T = \{s_1, ..., s_k\} \subseteq V$ be a set of terminals such that for every $i \neq j$, $(s_i, s_j) \notin E$. In the Node Multiway Cut problem, we are given a function $w: V \to \R_{>0}$ and the goal is to remove the minimum weight set of vertices $V' \subseteq V \setminus T$ such that in the induced graph $G' = G[V \setminus V']$, there is no path between any of the terminals.
\end{definition}

The Node Multiway Cut problem is a harder problem than the Edge Multiway Cut problem. In particular, the Edge Multiway Cut problem reduces in an approximation preserving fashion to the Node Multiway Cut problem~\cite{DBLP:journals/jal/GargVY04}. The problem is polynomially solvable for $k = 2$ and APX-hard for $k \geq 3$. For every $k \geq 3$, for Node Multiway Cut, a $2(1 - 1/k)$-approximation algorithm is known~\cite{DBLP:journals/jal/GargVY04}, and the same work also proves that the standard LP relaxation (see Figure~\ref{fig:MC-LP}) always has a half-integral optimal solution. Finally, in~\cite{DBLP:journals/jal/GargVY04} it is shown that there is an approximation-preserving reduction from Minimum Vertex Cover to the Minimum Node Multiway Cut problem, which implies that, assuming $\P \neq \NP$, there is no $(\sqrt{2} - \varepsilon)$-approximation algorithm for Node Multiway Cut~\cite{DBLP:journals/eccc/KhotMS18}, and assuming UGC, there is no $(2 - \varepsilon)$-approximation algorithm~\cite{DBLP:journals/jcss/KhotR08}.

Regarding stability, we first observe that it is straightforward to reprove the theorem of [MMV14] (see Theorem~\ref{thm:MMV-original}) in the setting of Node Multiway Cut, and in particular, one can easily prove that it suffices to obtain an $(\alpha, \beta)$-rounding for a half-integral optimal solution, since such a solution always exists. We now give such a rounding for the standard LP relaxation for Node Multiway Cut (see Figure~\ref{fig:MC-LP}) that satisfies $\alpha \beta = k - 1$, where $k$ is the number of terminals.

Let $G = (V, E, w)$, $T  = \{s_1, ..., s_k\} \subseteq V$, be an instance of Node Multiway Cut. The standard LP relaxation is given in Figure~\ref{fig:MC-LP}. The LP has one indicator variable for each vertex $u \in V$. For each pair of terminals $s_i$ and $s_j$, $i < j$, let $\mathcal{P}_{ij}$ denote the set of all paths between $s_i$ and $s_j$. Let $\mathcal{P} = \bigcup_{i < j} \mathcal{P}_{ij}$.

\begin{figure}
\begin{align*}
    \min: \quad          & \sum_{u \in V \setminus T} w_u x_u \\
    \textrm{s.t.:} \quad & \sum_{u \in P} x_u \geq 1, && \textrm{for all } P \in \mathcal{P},\\
                         & x_{s_{i}} = 0,                         && \textrm{for all } i \in [k],\\
                         & x_u \in [0,1],                         && \textrm{for all } u \in V.
\end{align*}
\caption{The standard LP relaxation for Node Multiway Cut.}
\label{fig:MC-LP}
\end{figure}

We now present a rounding scheme for the LP (Algorithm~\ref{alg:half-integral-rounding-mc}) that only works for half-integral solutions. Let $\{x_u\}_{u \in V}$ be a half-integral optimal solution for the LP of Figure~\ref{fig:MC-LP}. Let $V_0 = \{u \in V: x_u = 0\}$, $V_{1/2} = \{u \in V: x_u = 1/2\}$ and $V_1 = \{u \in V: x_u = 1\}$. Since $x$ is half-integral, we have $V = V_0 \cup V_{1/2} \cup V_1$. For a path $P$, let $len(P) = \sum_{u \in P} x_u$. Let $\mathcal{P}_{uv}$ denote the set of all paths between two vertices $u$ and $v$. We define $d(u,v) = \min_{P \in \mathcal{P}_{uv}} len(P)$; we note that this function is not an actual metric, since we always have some $u \in V$ with $d(u,u) > 0$. We consider the following rounding scheme (see Algorithm~\ref{alg:half-integral-rounding-mc}).

\begin{algorithm}[h]
\begin{enumerate}
    \item Let $G' = G[V_0 \cup V_{1/2}]$  (if graph $G'$ has more than one connected component, we\\ apply the rounding scheme on each connected component, separately).
    \item For each $i \in [k]$, let $B_i = \{u \in V_0: d(s_i, u) = 0\}$ and $\delta(B_i) = \{u \in V_{1/2}: \exists v \in B_i \textrm{ such that }(u,v) \in E\}$ \\(we note that the function $d$ is computed separately in each connected component\\ of $G'$).
    \item Pick uniformly random $j^* \in [k]$.
    \item Return $X: = V_1 \cup (\bigcup_{i \neq j^*} \delta(B_i))$.
\end{enumerate}
\caption{An $(\alpha, \beta)$-rounding for half-integral solutions for Node Multiway Cut.}
\label{alg:half-integral-rounding-mc}
\end{algorithm}

\begin{theorem}
Algorithm~\ref{alg:half-integral-rounding-mc} is an $(\alpha, \beta)$-rounding for Minimum Node Multiway Cut for half-integral optimal solutions, for some $\alpha$ and $\beta$, with $\alpha \beta = k - 1$. More precisely, given an optimal half-integral solution $\{x_u\}_{u \in V}$, it always returns a feasible solution $X \subseteq V \setminus T$ such that for each vertex $u \in V \setminus T$, the following two conditions are satisfied:
\begin{enumerate}
    \item $\Pr[u \in X] \leq \alpha \cdot x_u$,
    \item $\Pr[u \notin X] \geq \frac{1}{\beta} \cdot (1 - x_u)$,
\end{enumerate}
with $\alpha = \frac{2(k - 1)}{k}$ and $\beta = \frac{k}{2}$.
\end{theorem}
\begin{proof}
We first show that $X$ is always a feasible Multiway Cut. It is easy to see that $s_i \notin X$ for every $i \in [k]$. Let's fix now a path $P$ between $s_i$ and $s_j$. If there exists a vertex $u \in P$ such that $x_u = 1$, then clearly the algorithm ``cuts" this path, since $X$ contains all vertices whose LP value is 1. So, let's assume that for every $u \in P$ we have $x_u \in \{0,1/2\}$. Observe that the whole path $P$ is contained in the graph $G'$. Since $x_{s_t} = 0$ for every $t \in [k]$, we have $s_i \in B_i$ and $s_j \in B_j$ and we know that at least one of the sets $\delta(B_i)$ or $\delta(B_j)$ will be included in the solution. The LP constraints imply that $\sum_{q \in P} x_q \geq 1$. Thus, there are at least 2 vertices in $P$ whose LP value is exactly $1/2$. So, we start moving along the path $P$ from $s_i$ to $s_j$, and let $q_1 \in P$ be the first vertex with $x_{q_1} = 1 / 2$. Similarly, we start moving along the path from $s_j$ to $s_i$, and let $q_2 \in P$ be the first vertex with $x_{q_2} = 1 / 2$. Our assumption implies that $q_1 \neq q_2$. Clearly, $d(s_i, q_1) = d(s_j, q_2) = 1/2$, and it is easy to see that $q_1 \in \delta(B_i)$ and $q_2 \in \delta(B_j)$. Thus, at least one of the vertices $q_1$ or $q_2$ will be included in the final solution $X$. We conclude that the algorithm always returns a feasible solution.  

We will now show that the desired properties of the rounding scheme are satisfied with $\alpha\beta = k - 1$. For that, we first prove that $\bigcup_{i \in [k]} \delta(B_i) = V_{1/2}$, and moreover, each $u \in V_{1/2}$ belongs to exactly one set $\delta(B_i)$. By definition $\bigcup_{i \in [k]} \delta(B_i) \subseteq V_{1/2}$. Let $u \in V_{1/2}$. It is easy to see that there must exist at least one path $P$ between two terminals such that $u \in P$ and $x_v < 1$ for every $v \in P$, since otherwise we could simply set $x_u = 0$ and still get a feasible solution with lower cost. Let's assume now that $u \notin \bigcup_{i \in [k]} \delta(B_i)$. This means that for any path $P$ between two terminals $s_i$ and $s_j$ such that $u \in P$ and $x_v < 1$ for every $v \in P$, if we start moving from $s_i$ to $s_j$, we will encounter at least one vertex $q_1 \neq u$ with $x_{q_1} = 1/2$, and similarly, if we start moving from $s_j$ to $s_i$, we will encounter at least one vertex $q_2 \neq u$ with $x_{q_2} = 1/2$. Since this holds for any two terminals $s_i$ and $s_j$, it is easy to see that we can set $x_u = 0$ and get a feasible solution with a smaller cost. Thus, we get a contradiction. This shows that $\bigcup_{i \in [k]} \delta(B_i) = V_{1/2}$. We will now prove that for every $u \in V_{1/2}$ there exists a unique $i \in [k]$ such that $u \in \delta(B_i)$. Suppose that $u \in \delta(B_i) \cap \delta(B_j)$, for some $i \neq j$. Let $q_1 \in B_i$ such that $(u, q_1) \in E$, and let $q_2 \in B_j$ such that $(u, q_j) \in E$. Let $P_1$ be a shortest path between $s_i$ and $q_1$, and let $P_2$ be a shortest path between $s_j$ and $q_2$. We now consider the path $P' = P_1 \cup \{u\} \cup P_2$. This is indeed a valid path in $G'$ between $s_i$ and $s_j$. It is easy to see that $\sum_{v \in P'} x_v = 1/2$, and so an LP constraint is violated. Again, we get a contradiction, and thus, we conclude that for each $u \in V_{1/2}$ there exists exactly one $i \in [k]$ such that  $u \in \delta(B_i)$.

We are almost done. We will now verify that the two conditions of the rounding scheme are satisfied. Let $u \in V \setminus T$. If $x_u = 1$, then $u$ is always picked and we have $\Pr[u \textrm{ is picked}] = 1 = x_u$ and $\Pr[u \textrm{ is not picked}] = 0 = 1 - x_u$. If $x_u = 0$, then the vertex $u$ will never be picked, and so $\Pr[u \textrm{ is picked}] = 0 = x_u$ and $\Pr[u \textrm{ is not picked}] = 1 = 1 - x_u$. So, let's assume now that $x_u = 1/2$. By the previous discussion, $u \in \delta(B_i)$ for some unique $i \in [k]$. Since each set $\delta(B_i)$ is not included in the solution with probability $1/k$, we get that
\begin{equation*}
    \Pr[u \textrm{ is not picked}] = \frac{1}{k} = \frac{2}{k} \cdot \frac{1}{2} = \frac{2}{k} \cdot (1 - x_u),
\end{equation*}
and
\begin{equation*}
    \Pr[u \textrm{ is picked}] = \frac{k - 1}{k} = \frac{2(k-1)}{k} \cdot \frac{1}{2} = \frac{2(k-1)}{k} \cdot x_u.
\end{equation*}
Thus, the rounding scheme satisfies the desired properties with $\alpha\beta = \frac{2(k-1)}{k} \cdot \frac{k}{2} = k-1$.
\end{proof}

The above theorem, combined with the adaptation of Theorem~\ref{thm:MMV-original} for the problem directly gives the following result.
\begin{theorem}
The standard LP relaxation for Node Multiway Cut is integral for $(k - 1)$-stable instances, where $k$ is the number of terminals.
\end{theorem}

Mimicking the techniques of [MMV14] (or Section~\ref{sec:weakly-stable-MC}), we can also prove the following theorem about weakly stable instances.
\begin{theorem}
There is a polynomial-time algorithm that, given a $(k - 1 + \delta, \mathcal{N})$-weakly-stable instance of Minimum Node Multiway Cut with $n$ vertices, $k$ terminals and integer weights, finds a solution $X' \in \mathcal{N}$ (for every $\delta \geq 1/\poly(n) > 0$).
\end{theorem}

We now prove that the above analysis is tight, i.e.~there are $(k - 1 - \varepsilon)$-stable instances for which the LP is not integral.
\begin{theorem}
For every $\varepsilon > 0$, there exist $(k - 1-\varepsilon)$-stable instances of the Node Multiway Cut problem with $k$ terminals for which the LP of Figure~\ref{fig:MC-LP} is not integral.
\end{theorem}
\begin{proof}
We consider a variation of the star graph, as shown in Figure~\ref{fig:mc-bad-example}. The graph $G = (V, E, w)$ is defined as follows:
\begin{enumerate}
    \item $V = \{s_1, ..., s_k\} \cup \{u_1, ..., u_k\} \cup \{c\}$, with $T = \{s_1, ..., s_k\}$ being the set of terminals. Observe that $|V| = 2k + 1$.
    \item $E = \{(c,u_i): i \in [k]\} \cup \{(s_i, u_i): i\in [k]\}$.
    \item For each $i \in \{1, ..., k - 1\}$, we have $w_{u_i} = 1$. We also have $w_{u_k} = k - 1 - \frac{\varepsilon}{2}$ and $w_c = k^3$.
\end{enumerate}
\begin{figure}[h]
\begin{center}
\scalebox{0.7}{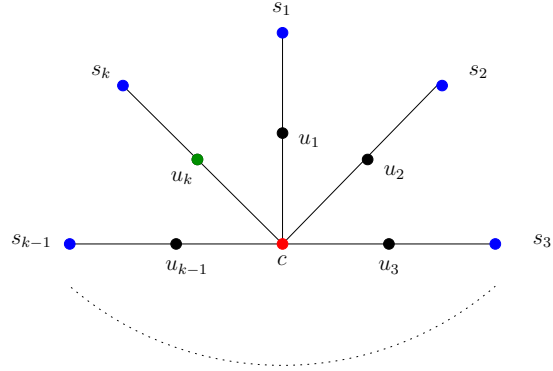}
\end{center}
\caption{An integrality gap example of a stable instance of Node Multiway Cut.}
\label{fig:mc-bad-example}
\end{figure}

It is easy to see that there is unique optimal integral solution $X^* = \{u_i : 1 \leq i \leq k - 1\}$ of cost $OPT = k - 1$. It is also clear that any feasible solution must either remove vertex $c$ or must remove at least $k - 1$ vertices from the set $\{u_1, ..., u_k\}$. A minimal solution that contains $c$ is $X_c = \{c\}$. We have $(k - 1 - \varepsilon) w(X^* \setminus X_c) < (k - 1)^2$ and $w(X_c \setminus X^*) = k^3$. Let's consider now a solution that does not contain $c$. By the previous observations, we only have to consider the solutions $Y_i = \{u_1, ..., u_k\} \setminus \{u_i\}$, $1 \leq i \leq k - 1$, and $Y_0 = \{u_1, ..., u_k\}$. For any $Y_i$, $1 \leq i \leq k - 1$, we have $(k - 1 - \varepsilon) \cdot w(X^* \setminus X_i) = (k - 1 - \varepsilon) \cdot w_{u_i} = k - 1 - \varepsilon$ and $w(Y_i \setminus X^*) = w_{u_k} = k - 1 - \varepsilon / 2$. For $Y_0$ we have $(k - 1 - \varepsilon) \cdot w(X^* \setminus Y_0) = 0$ and $w(Y_i \setminus X^*) = w(u_k) = k - 1 - \varepsilon/2$. Thus, in all cases, the stability condition is satisfied with $\gamma = k - 1 - \varepsilon$.

We now look at the LP. Let $x_{u_i} = 1/2$ for every $i \in [k]$ and let $x_c = 0$. We also set $x_{s_i} = 0$ for every $i \in [k]$. Observe that this is a feasible solution. The objective function is equal to
\begin{equation*}
    \frac{k-1}{2} + \frac{k -1 - (\varepsilon/2)}{2} = k - 1 - (\varepsilon/4) < k - 1 = OPT.
\end{equation*}
Thus, the integrality gap is strictly greater than 1, and thus, the LP is not integral.
\end{proof}

Finally, we show that if there exists an algorithm for $\gamma$-stable instances of Node Multiway Cut, then there exists an algorithm for $\gamma$-stable instances of Vertex Cover. This reduction, combined with the results of the next chapter, implies very strong lower bounds on the existence of \textit{robust} algorithms for Node Multiway Cut.
\begin{theorem}
Let $\mathcal{A}$ be an algorithm for $\gamma$-stable instances of Minimum Node Multiway Cut. Then, there exists an algorithm $\mathcal{B}$ for $\gamma$-stable instances of Minimum Vertex Cover. Moreover, if $\mathcal{A}$ is robust, then $\mathcal{B}$ is robust.
\end{theorem}
\begin{proof}
We use the straightforward approximation-preserving reduction of Garg et al.~\cite{DBLP:journals/jal/GargVY04}. Let $G = (V, E, w)$ be a $\gamma$-stable instance of Minimum Vertex Cover, with $V = \{u_1, ..., u_n\}$. We construct $G' = (V', E', w')$, where $G'$ contains the whole graph $G$, and moreover, for each vertex $u_i \in V$, we create a terminal vertex $s_i$ and we connect it to $u_i$ with an edge $(s_i, u_i) \in E'$. As implied, the set of terminals is $T = \{s_1, ..., s_n\}$. The weights of non-terminal vertices remain unchanged. This is clearly a polynomial-time reduction. We will now prove that each feasible vertex cover $X$ of $G$ corresponds to a feasible Mulitway Cut of $G'$ of the same cost, and vice versa. To see this, let $X$ be a feasible vertex cover of $G$, and let's assume that there is a path between two terminals $s_i$ and $s_j$ in $G'[V' \setminus X]$. By construction, this means that there is a path between $u_i$ and $u_j$ in $G'[V' \setminus X]$, which implies that there is at least one edge in this path that is not covered. Thus, we get a contradiction. Since the weight function is unchanged, we also conclude that $w(X) = w'(X)$. Let now $X'$ be a feasible Multiway Cut for $G'$, and let's assume that $X'$ is not a vertex cover in $G$. This means that there is an edge $(u_i, u_j) \in E$ such that $\{u_i, u_j\} \cap X' = \emptyset$. This means that the induced graph $G'[V' \setminus X']$ contains the path $s_i - u_i - u_j - s_j$, and so we get a contradiction, since we assumed that $X'$ is a feasible Node Multiway Cut. Again, the cost is clearly the same, and thus, we conclude that there is a one-to-one correspondence between vertex covers of $G$ and multiway cuts of $G'$.

Since the cost function is exactly the same, it is now easy to prove that a $\gamma$-stable instance $G$ of Vertex Cover implies that $G'$ is a $\gamma$-stable instance of Multiway Cut, and moreover, if $G'$ is not $\gamma$-stable, then $G$ cannot be $\gamma$-stable to begin with. Thus, we can run algorithm $\mathcal{A}$ on instance $G'$, and return its output as the output of algorithm $\mathcal{B}$. By the previous discussion, this is a $\gamma$-stable algorithm for Vertex Cover, and, if $\mathcal{A}$ is robust, then so is $\mathcal{B}$.
\end{proof}

The above result, combined with the result of the next chapter (see Theorem~\ref{VC_stable_lower_bound_theorem}), implies the following theorem.
\begin{theorem}
\hfill
\begin{enumerate}
\itemsep0em
    \item For every constant $\varepsilon > 0$, there is no robust algorithm for $\gamma$-stable instances of Minimum Node Multiway Cut, for $\gamma = n^{1 - \varepsilon}$, assuming that $P \neq NP$.
    \item For every constant $\varepsilon > 0$, there is no robust algorithm for $\gamma$-stable instances of Minimum Node Multiway Cut, for $\gamma = \frac{n}{2^{(\log n)^{3/4+\varepsilon}}}$, assuming that $NP \not \subseteq \mathtt{BPTIME}\left(2^{(\log n)^{O(1)}}\right)$.
\end{enumerate}
\end{theorem}


\chapter{Hardness results for robust algorithms}\label{chap:hardness}

In this chapter, we prove a very strong negative result about robust algorithms for the Minimum Vertex Cover problem, and its equivalent, in terms of exact solvability, problem, namely the Maximum Independent Set problem. We then extend the result to the Min 2-Horn Deletion problem and the Minimum Multicut on Trees.

\section{Lower bounds for Vertex Cover}\label{sec:hardness-VC}

We prove that, under standard complexity assumptions, no robust algorithms (as defined in Definition \ref{def:robust}) exist for $\gamma$-stable instances of Minimum Vertex Cover, even when $\gamma$ is very large (we precisely quantify this later in this section). Before presenting our results, it is worth noting that robustness is a very desirable property of algorithms, since it guarantees that the output is always correct, even when the instance is not stable (and it is usually the case that we do not know whether the input is stable or not). Furthermore, proving that no robust algorithm exists for $\gamma$-stable instances of a given problem implies that no LP/SDP or other convex relaxation that is solvable in polynomial time can be integral for $\gamma$-stable instances of the problem, thus ruling out the possibility of having an algorithm that solves $\gamma$-stable instances by solving the corresponding relaxation. We now turn our attention to Minimum Vertex Cover, which from now on we sometimes denote as \MVC{}.

An \MVC{} instance $G = (V, E, w)$, $w: V \to \mathbb{R}_{\geq 0}$, is called $\gamma$-stable, for $\gamma \geq 1$, if it has a unique optimal solution $X^* \subseteq V$, and for every $\gamma$-perturbation (i.e.~for every instance $G' = (V, E, w')$ that satisfies $w_u \leq w_u' \leq \gamma \cdot w_u$ for every $u \in V$), the solution $X^*$ remains the unique optimal solution. In order to prove our impossibility result for Vertex Cover, we need the following definition.

\begin{definition}[GAP-IS]
For any $0 < \alpha < \beta$, the $(\alpha, \beta)$-GAP-IS problem is a promise problem that takes as input a (vertex-weighted) graph $G$ whose independent set is either strictly larger than $\beta$ or at most $\alpha$ and asks to distinguish between the two cases, i.e.~decide whether $G$ has an independent set of weight
\begin{itemize}
    \item strictly larger than $\beta$ (i.e.~$OPT > \beta$; YES instance)
    \item at most $\alpha$ (i.e.~$OPT \leq \alpha$; NO instance)
\end{itemize}
\end{definition}

We will prove that the existence of a robust algorithm for $\gamma$-stable instances of \MVC{} would allow us to solve $(\beta / \gamma - \delta, \beta)$-GAP-IS, for every $\beta > 0$ and arbitrarily small $\delta > 0$.
\begin{lemma}\label{stable_VC_gap_IS_lemma}
Given a robust algorithm for $\gamma$-stable instances of Minimum Vertex Cover, for some $\gamma > 1$, there exists an algorithm that can be used to efficiently solve $(\beta / \gamma - \delta, \beta)$-GAP-IS, for every $\beta > 0$ and every $\delta \in (0, \beta/ \gamma)$.
\end{lemma}
\begin{proof}
Given a $(\beta / \gamma - \delta, \beta)$-GAP-IS instance $G = (V,E,w)$, $w: V \to \mathbb{R}_{\geq 0}$, we construct the graph $G' = (V', E', w')$, where $V' = V \cup \{s\}$, $E' = E \cup \{(v,s): v \in V\}$, $w_u' = w_u$ for all $u \in V$ and $w_s' = \beta$. Every vertex cover $X \subseteq V'$ of $G'$ is of one of the following forms:
\begin{itemize}
    \item $X = V$, with cost $w'(X) = w(V)$.
    \item $X = (V \setminus I) \cup \{s\}$, where $I$ is an independent set of the original graph $G$. The cost of $X$ in this case is $w'(X) = w(V) - w(I) + \beta$.
\end{itemize}
Let $I^* \subseteq V$ denote a maximum independent set of $G$ and $OPT_{IS(G)} = w(I^*)$ denote its cost. Then, an optimal vertex cover is either $V$ or $(V \setminus I^*) \cup \{s\}$. Observe that we can never have $w(V) = w((V \setminus I^*) \cup \{s\})$, since this would imply that $OPT_{IS(G)} = \beta$, and this is impossible, given that $G$ is a $(\beta / \gamma - \delta, \beta)$-GAP-IS instance.

We now run the robust algorithm for $\gamma$-stable instances of \MVC{} on $G'$, and depending on the output $Y$, we make the following decision:
\begin{itemize}
    \item $Y = V$: $V$ is the optimal vertex cover of $G'$, and so $w(V) \leq w(V) - w(I) + \beta$ for all independent sets $I$ of $G$. This implies that $w(I^*) \leq \beta$, and, since the instance is a $(\beta / \gamma - \delta, \beta)$-GAP-IS instance, we must have $w(I^*) \leq \beta / \gamma - \delta$. We output \textbf{NO}.
    \item $Y = (V \setminus I^*) \cup \{s\}$ for some (maximum) independent set $I^*$: In this case, we have $w(V) \geq w(V) - w(I^*) + \beta$, and so $w(I^*) \geq \beta$. From the above discussion, this implies that $w(I^*) > \beta$, and so we output \textbf{YES}.
    \item $Y = \mathbf{not\; stable}$: Since the instance is not $\gamma$-stable, it is not hard to see that there must exist an independent set $I$ of $G$, such that $w(V \setminus I) + \gamma w(I) \geq w(V \setminus I) + \beta$ (since otherwise the instance would be $\gamma$-stable with $V$ being the optimal vertex cover), which implies that $w(I) \geq \beta / \gamma$. Thus, $w(I^*) > \beta$, and so we output \textbf{YES}.
\end{itemize}
We designed an algorithm that uses a robust algorithm for $\gamma$-stable instances of \MVC{} as a black-box and solves the $(\beta / \gamma - \delta, \beta)$-GAP-IS problem, for every $\beta > 0$ and arbitrarily small $\delta > 0$.
\end{proof}

We now use the known inapproximability results for Maximum Independent Set in conjunction with Lemma \ref{stable_VC_gap_IS_lemma}. In particular, we need the following two theorems, the first proved by Zuckerman~\cite{DBLP:journals/toc/Zuckerman07} (also proved earlier by H\aa stad in~\cite{DBLP:conf/focs/Hastad96} under the complexity assumption that $\NP \not\subseteq \ZPP$), and the second by  Khot and Ponnuswami \cite{KP06}.

\begin{theorem}[Zuckerman~\cite{DBLP:journals/toc/Zuckerman07}]\label{IS_hardness}
It is $\NP$-hard to approximate the Maximum Independent Set to within $n^{1 - \varepsilon}$, for every constant $\varepsilon > 0$. Equivalently, it is $\NP$-hard to solve $(\alpha, \beta)$-GAP-IS, for $\beta / \alpha = n^{1 - \varepsilon}$, for every constant $\varepsilon > 0$.
\end{theorem}
\begin{theorem}[Khot and Ponnuswami \cite{KP06}]
For every constant $\varepsilon > 0$, there is no polynomial time algorithm that approximates the Maximum Independent Set to within $n / 2^{(\log n)^{3/4+\varepsilon}}$, assuming that $\NP \not\subseteq BPTIME\left(2^{(\log n)^{O(1)}}\right)$.
\end{theorem}

Combining Lemma \ref{stable_VC_gap_IS_lemma} with the above two theorems, we obtain the following theorem.
\begin{theorem}\label{VC_stable_lower_bound_theorem}
\hfill
\begin{enumerate}
\itemsep0em
    \item For every constant $\varepsilon > 0$, there is no robust algorithm for $\gamma$-stable instances of Minimum Vertex Cover (and Maximum Independent Set), for $\gamma = n^{1 - \varepsilon}$, assuming that $\P \neq \NP$.
    \item For every constant $\varepsilon > 0$, there is no robust algorithm for $\gamma$-stable instances of Minimum Vertex Cover (and Maximum Independent Set), for $\gamma = \frac{n}{2^{(\log n)^{3/4+\varepsilon}}}$, assuming that $\NP \not \subseteq \mathtt{BPTIME}\left(2^{(\log n)^{O(1)}}\right)$.
\end{enumerate}
\end{theorem}

As an immediate corollary, we get the same lower bounds for stability for Set Cover, since Minimum Vertex Cover can be formulated as a Set Cover instance.
\begin{corollary}
\hfill
\begin{enumerate}
\itemsep0em
    \item For every constant $\varepsilon > 0$, there is no robust algorithm for $\gamma$-stable instances of Set Cover, for $\gamma = n^{1 - \varepsilon}$, assuming that $\P \neq \NP$.
    \item For every constant $\varepsilon > 0$, there is no robust algorithm for $\gamma$-stable instances of Set Cover, for $\gamma = \frac{n}{2^{(\log n)^{3/4+\varepsilon}}}$, assuming that $\NP \not \subseteq \mathtt{BPTIME}\left(2^{(\log n)^{O(1)}}\right)$.
\end{enumerate}
\end{corollary}

\section{Lower bounds for Min 2-Horn Deletion}

In this section, we focus on Min 2-Horn Deletion, and prove that the lower bound for robust algorithms for \MVC{} can be extended to this problem as well, since \MVC{} can be formulated as a Min 2-Horn Deletion problem in a convenient way. We start with the definition of Min 2-Horn Deletion and then state and prove the main theorem of this section.

\begin{definition}[Min 2-Horn Deletion]
Let $\{x_i\}_{i \in [n]}$ be a set of boolean variables and let $\mathcal{F} = \{C_j\}_{j \in [m]}$ be a set of clauses on these variables, where each $C \in \mathcal{F}$ has one of the following forms: $x_i$, $\bar{x}_i$, $\bar{x}_i \vee x_j $, or $\bar{x}_i \vee \bar{x}_j $. In words, each clause has at most two literals and is allowed to have at most one positive literal. We are also given a weight function $w : \mathcal{F} \to \mathbb{R}_{\geq 0}$, and the goal is to find an assignment $f: \{x_1, ..., x_n\} \to \{true, false\}$ such that the weight of the unsatisfied clauses is minimized.
\end{definition}

It will be convenient to work with the dual Min 2-Horn Deletion, in which each clause contains at most one negated literal. Observe that the two problems are equivalent, since, given a Min 2-Horn Deletion instance with variables $\{x_i\}_{i \in [n]}$, we can define the variables $y_i = \bar{x}_i$, $i \in [n]$, and substitute them in $\mathcal{F}$, thus obtaining a dual Min 2-Horn Deletion with the exact same value. We now prove the following theorem.

\begin{theorem}
\hfill
\begin{enumerate}
\itemsep0em
    \item For every constant $\varepsilon > 0$, there is no robust algorithm for $\gamma$-stable instances of Min 2-Horn Deletion, for $\gamma = n^{1 - \varepsilon}$, assuming that $P \neq NP$.
    \item For every constant $\varepsilon > 0$, there is no robust algorithm for $\gamma$-stable instances of Min 2-Horn Deletion for $\gamma = \frac{n}{2^{(\log n)^{3/4+\varepsilon}}}$, assuming that $NP \not \subseteq \mathtt{BPTIME}\left(2^{(\log n)^{O(1)}}\right)$.
\end{enumerate}
\end{theorem}
\begin{proof}

Let us assume that there exists a robust algorithm for $\gamma$-stable instances of Min 2-Horn Deletion, for some $\gamma > 1$. We will prove that this would give a robust algorithm for $\gamma$-stable instances of \MVC{}. For that, we consider any \MVC{} instance $G = (V, E, w)$, $w: V \to \mathbb{R}_{\geq 0}$, and construct an instance $F(G)$ of Min 2-Horn Deletion as follows (for convenience, as explained above, we assume that each clause contains at most one negation, i.e.~we construct a dual Min 2-Horn Deletion formula). We introduce variables $\{x_u\}_{u \in V}$ and $|V| + |E|$ clauses, with $C_u := \bar{x}_u$, for every $u \in V$, and $C_{(u,v)} := x_u \vee x_v$, for every $(u,v) \in E$. We also assign weights $w'$, with $w'(C_u) = w_u$, $u \in V$, and $w'(C_{(u,v)}) = 1 + \gamma \cdot \sum_{q \in V} w_q$, for every $(u,v) \in E$.

Observe that an immediate upper bound for the cost of the optimal assignment of $F(G)$ is $\sum_{u \in V} w_u$, since we can always delete all the clauses $C_u$ and set all variables to $true$. Thus, an optimal assignment never violates a clause $C_{(u,v)}$, $(u,v) \in E$. This means that in an optimal assignment $f^*$, for every $(u,v) \in E$, either $f^*(x_u) = true$ or $f^*(x_v) = true$. This implies that the set $X(f^*) = \{u \in V: f^*(x_u) = true\}$ is a feasible vertex cover of $G$. It also means that the cost of an optimal assignment is $\sum_{u \in V: f^*(x_u) = true} w_u = w(X(f^*))$.  We will now show that $X(f^*)$ is in fact an optimal vertex cover of $G$. First, note that the cost of any assignment $g$ (not necessarily optimal) that does not violate any of the clauses $C_{(u,v)}$, $(u,v) \in E$, is $\sum_{u \in V: g(x_u) = true} w_u$. Suppose now that there exists a vertex cover $X' \neq X(f^*)$ with cost $w(X') < w(X(f^*))$. Let $g(x_u) = true$ if $u \in X'$, and $g(x_u) = false$ if $u \notin X'$. It is easy to see that $g$ does not violate any of the clauses $C_{(u,v)}$, $(u,v) \in E$. Thus, the cost of the assignment $g$ is equal to $\sum_{u \in V: g(x_u) = true} w_u = w(X') < w(X(f^*))$, which contradicts the optimality of $f^*$, and, so, we conclude that the set $X(f^*)$ is an optimal vertex cover of $G$.

We will now show that if $F(G)$ is not $\gamma$-stable, then $G$ cannot be $\gamma$-stable. First, observe that any $\gamma$-perturbation of $F(G)$ has an optimal solution of cost at most $\gamma \cdot \sum_{u \in V} w_u$, implying that in every $\gamma$-perturbation of $F(G)$, an optimal solution only deletes clauses of the form $C_u = x_u$, for $u \in V$. In other words, in every $\gamma$-perturbation of $F(G)$, an optimal assignment $g$ defines a feasible vertex cover $X = \{u \in V: g(u) = true\}$. This also implies that the perturbation of the weights $w(C_{(u,v)})$ cannot change the optimal assignment, and so, the weights of the clauses $C_u$, $u\in V$, completely specify the optimal value. Moreover, if $\tilde{w}$ is the weight function for a $\gamma$-perturbation of $F(G)$ (whose optimal assignment defines the set $X$ as before), we can use the observation of the previous paragraph to conclude that the vertex cover X is optimal for the instance $G' = (V, E, w')$, in which $w_u' = \tilde{w}(C_u)$ for all $u \in V$. Note that $G'$ is a $\gamma$-perturbation of $G$. Suppose now that $F(G)$ is not $\gamma$-stable. Thus, there exists a subset $X \subseteq V$ such that an optimal assignment for $F(G)$ deletes the clauses $\{C_u: u \in X\}$ (i.e.~$f(x_u) = true$ iff $u \in X$) while there exists a $\gamma$-perturbation $F'(G)$ of $F(G)$ such that an optimal assignment for $F'(G)$ deletes the clauses $\{C_u: u \in X'\}$ for some $X' \neq X$. As argued, $X$ is an optimal vertex cover for $G$ and $X'$ is an optimal vertex cover for some $\gamma$-perturbation of $G$. Since $X \neq X'$, the instance $G$ is not $\gamma$-stable.

We are ready to present our robust algorithm for $\gamma$-stable instances of \MVC{}. We use the robust algorithm for $\gamma$-stable instances of Min 2-Horn Deletion on $F(G)$. Let $Y$ be the output of the algorithm, when ran on the instance $F(G)$:
\begin{itemize}
    \item $Y = f$, where $f: \{x_u\}_{u \in V} \to \{true, false\}$: As discussed previously, the set $X = \{u \in V: f(x_u) = true\}$ is an optimal vertex cover for $G$, and so we output $X$.
    \item $Y = \mathbf{not\;stable}$: We output \textbf{``not stable"}, since, by the previous discussion, the \MVC{} instance cannot be $\gamma$-stable.
\end{itemize}
Plugging in the bounds of Theorem~\ref{VC_stable_lower_bound_theorem}, we obtain the desired lower bounds.
\end{proof}



\section{Lower bounds for Multicut on Trees}

In this section, we combine the result of Section~\ref{sec:hardness-VC} with the straightforward approximation-preserving reduction from Minimum Vertex Cover to Minimum Multicut on Trees, introduced by Garg et al.~\cite{DBLP:journals/algorithmica/GargVY97}. More precisely, we prove the following theorem.

\begin{theorem}\label{thm:multicut}
Assuming $\P \neq \NP$, there are no robust algorithms for $n^{1 - \varepsilon}$-stable instances of Multicut on Trees.
\end{theorem}
Before we give the proof of the theorem, we formally define the problem.
\begin{definition}[Multicut on Trees]
Let $T = (V, E)$ be an edge-weighted tree, with non-negative weights $c: E \to \R_{\geq 0}$. Let $\{(s_1, t_1 ), ..., (s_k , t_k)\}$ be a specified set of pairs of vertices, where each pair is distinct, but vertices in different pairs are not required to be distinct. A multicut is a set of edges whose removal separates each of the pairs. The problem is to find a minimum weight multicut in $T$.
\end{definition}
The Minimum Multicut problem is \NP-hard even when restricted to trees of height 1 and unit weight edges. We are now ready to prove Theorem~\ref{thm:multicut}.

\begin{proof}[Proof of Theorem~\ref{thm:multicut}].
Let $G = (V,E,w)$, $w: V \to \R_{>0}$, be a $\gamma$-stable instance of \MVC{}, with $V = \{u_1, ..., u_n\}$. We consider the following Multicut instance. We construct the star graph $T = (V_T, E_T, w_T)$, where $V_T = \{c\} \cup V$, $E_T = \{(u_i,c)\}_{i = 1, ..., n}$ and $w_T: E_T \to \R_{>0}$ with $w_T(u_i,c) = w_{u_i}$, for every $i \in [n]$. The set of demand pairs is $D = \{(u_i,u_j): (u_i, u_j) \in E\}$.

It is easy to see that any feasible vertex cover $X \subseteq V$ of $G$ corresponds to a feasible multicut $X_T = \{(u,c): u \in X\}$ and vice versa, and moreover, $w(X) = w_T(X_T)$. We will now show that $T$ is also a $\gamma$-stable instance, given that $G$ is a $\gamma$-stable instance, and vice versa. Let $X^* \subseteq V$ be the unique optimal Vertex Cover for $G$. By the previous observations, it is easy to see that $X_T^* = \{(u,c): u \in X^*\}$ is the unique optimal multicut for $T$. Let $T'=(V_T,E_T, w_T')$ be any $\gamma$-perturbation of $T$, i.e.~for every $(u,c) \in E_T$, $w_T(u,c) \leq w_T'(u,c) \leq \gamma \cdot w_T(u,c)$. We want to prove that $T'$ has the same optimal solution as $T$. Suppose that $T'$ has an optimal solution $X_{T'} \neq X_T^*$. This immediately implies that the graph $G'=(V,E,w')$, with $w_u' = w_T'(u,c)$ for every $u \in V$, has an optimal vertex cover $X' = \{u \in V: (u,c) \in X_T'\}$ that is not equal to $X^*$. Note that $G'$ is a $\gamma$-perturbation of $G$, and thus we get a contradiction. So, $T$ is indeed a $\gamma$-stable instance of Multicut. The other direction is proved similarly, namely, if the constructed $T$ is $\gamma$-stable, then $G$ has to be $\gamma$-stable.

Suppose now that there exists a robust algorithm for $n^{1 - \varepsilon}$-stable instances of Multicut on Trees. Then, by using the above reduction, one can use such an algorithm to construct a robust algorithm for \MVC{}, and this is impossible, assuming $\P \neq \NP$, as shown in Theorem~\ref{VC_stable_lower_bound_theorem}.
\end{proof}

\chapter{Stability and the Independent Set problem}\label{chap:independent-set}

\section{Introduction}
The Maximum Independent Set problem is a central problem in theoretical computer science and has been the subject of numerous works over the last few decades. As a result we now have a thorough understanding of the \textit{worst-case} behavior of the problem. In general graphs, the problem is $n^{1-\varepsilon}$-hard to approximate, assuming that $\P \neq \NP$~\cite{DBLP:conf/focs/Hastad96, DBLP:journals/toc/Zuckerman07}, and $n/2^{(\log n)^{3/4 + \varepsilon}}$-hard to approximate, assuming that $\NP \not\subseteq \texttt{BPTIME}(2^{ (\log n)^{O(1)} })$. On the positive side, the current best algorithm is due to Feige~\cite{DBLP:journals/siamdm/Feige04} achieving a $\widetilde{O}(n / \log^3 n)$-approximation\footnote{The notation $\widetilde{O}$ hides some $\mathrm{poly}(\log \log n)$ factors.}. In order to circumvent the strong lower bounds, many works have focused on special classes of graphs, such as bounded-degree graphs (see e.g.~\cite{DBLP:journals/njc/HalldorssonR94, DBLP:journals/mp/AlonK98, DBLP:journals/jgaa/Halldorsson00, DBLP:journals/siamcomp/Halperin02, DBLP:conf/soda/Bansal15, DBLP:conf/stoc/BansalGG15, DBLP:journals/toc/AustrinKS11, DBLP:journals/jacm/Chan16}), planar graphs~\cite{DBLP:journals/jacm/Baker94} etc. In this chapter, we continue this long line of research and study the Maximum Independent Set problem  (which, from now on, we denote as \MIS{}) within the beyond worst-case analysis framework introduced by Bilu and Linial.

In this chapter, our focus is on understanding the complexity of stable instances of \MIS{}, with an emphasis on designing robust algorithms. From a practical point of view, designing algorithms for $\gamma$-stable instances for small values of $\gamma$ is highly desirable. Unfortunately, for general stable instances of \MIS{} this is not always possible. In Chapter~\ref{chap:hardness}, we proved that there is no robust algorithm for $n^{1 - \varepsilon}$-stable instances of \MIS{} on general graphs (unbounded degree), assuming that $\P \neq \NP$ (see Theorem~\ref{VC_stable_lower_bound_theorem}). As a result, our focus is on special classes of graphs, such as bounded-degree graphs and planar graphs,  where we prove that one can indeed handle small values of the stability parameter. Nevertheless, we do provide an algorithm for stable instances of \MIS{} on general graphs as well.

We now restate the definition of stability in the context of the Independent Set problem. As in the case of Multiway Cut, we use the following equivalent (and more convenient) definition.
\begin{definition}[$\gamma$-stable instance of \MIS{}]
Let $G = (V,E,w)$, $w: V \to \R_{>0}$, and let $I^*$ be a maximum independent set of $G$. The instance $G$ is $\gamma$-stable, for some parameter $\gamma \geq 1$, iff $w(I^* \setminus S) > \gamma \cdot w(S \setminus I^*)$ for every feasible independent set $S \neq I^*$.
\end{definition}

\paragraph{Related Work.} As mentioned, there have been many works about the worst-case complexity of \MIS{} and the best known approximation algorithm due to Feige~\cite{DBLP:journals/siamdm/Feige04} achieves a factor of $\widetilde{O}(n / \log^3 n)$. For degree-$\Delta$ graphs, Halperin~\cite{DBLP:journals/siamcomp/Halperin02} designed an ${O}(\frac{\Delta \log \log \Delta}{\log \Delta})$-approximation algorithm.  The \MIS{} problem has also been studied from the lens of beyond worst-case analysis. In the case of random graphs with a planted independent set, the problem is equivalent to the classic planted clique problem. Inspired by semi-random models of~\cite{blum1995coloring}, Feige and Killian~\cite{feige2001heuristics} designed SDP-based algorithms for computing large independent sets in semi-random graphs.

In the Bilu-Linial stability framework, Bilu~\cite{Bilu} analyzed the greedy algorithm and showed that it recovers the optimal solution for $\Delta$-stable instances of graphs of maximum degree $\Delta$. The same result is also a corollary of a general theorem about the greedy algorithm and $p$-extendible independence systems proved by Chatziafratis et al.~\cite{DBLP:conf/esa/ChatziafratisRV17}. Finally, we would like to mention that there has also been work on studying \MIS{} under adversarial perturbations to the graph~\cite{DBLP:conf/approx/MagenM09, chan2012approximation, bansal2017lp}.

\paragraph{Our results.} In this chapter, we explore the notion of stability in the context of \MIS{} and significantly improve our understanding of the problem's behavior on stable instances. In particular, using both combinatorial and LP-based methods, we design algorithms for stable instances of \MIS{} for different classes of graphs. More concretely, we obtain the following results.
\begin{itemize}
    \item \textbf{Independent set on planar graphs:} We show that on planar graphs, any constant stability suffices to solve the problem exactly in polynomial time. More precisely, we provide robust algorithms for $(1+\varepsilon)$-stable instances of planar \MIS{}, for any fixed $\varepsilon > 0$.

\item \textbf{Independent set on graphs of bounded degree or low chromatic number:} We provide a robust algorithm for solving $(k-1)$-stable instances of \MIS{} on $k$-colorable graphs, and $(\Delta - 1)$-stable instances of \MIS{} on graphs of maximum degree $\Delta$.

\item \textbf{Independent set on general graphs:} For general graphs, we present an algorithm for $(\varepsilon n)$-stable instances of \MIS{} on $n$ vertices whose running time is $n^{O(1/\varepsilon)}$.

\item \textbf{Convex relaxations and stability:} We present a structural result about the integrality gap of convex relaxations of several maximization problems on stable instances: if the integrality gap is at most $\alpha$, then it is at most $\min \left\{\alpha, 1 + \frac{1}{\beta - 1} \right\}$ for $(\alpha \beta)$-stable instances, for any $\beta > 1$.
 \end{itemize}

\paragraph{Organization of material.} In Section~\ref{sec:stable}, we present robust algorithms for stable instances of \MIS{} on special classes of graphs, such as bounded-degree graphs, planar graphs, and graphs with small chromatic number. For general, graphs we give a (non-robust) algorithm for stable instances in Section~\ref{sec:upper-bound-general}.

\section{Robust algorithms for stable instances of Independent Set}\label{sec:stable}

In the next few sections, we obtain \emph{robust} algorithms for stable instances of \MIS{} by using the standard LP relaxation and the Sherali-Adams hierarchy. Since there are strong lower bounds for robust algorithms on general graphs (see Chapter~\ref{chap:hardness}), we focus on special classes of graphs, such as bounded-degree graphs and planar graphs.


\subsection{Convex relaxations and robust algorithms}

In order to design robust algorithms, we use convex relaxations of \MIS{}. An important component is the structural result of [MMV14] that we extensively used in Chapter~\ref{chap:multiway-cut}. We restate the theorem here, in the context of \MIS{}, for the readers that skipped Chapter~\ref{chap:multiway-cut}. We first introduce a definition and then restate their theorem in the setting of \MIS{}.

\begin{definition}[$(\alpha,\beta)$-rounding]
Let $x: V \to [0,1]$ be a feasible fractional solution of a convex relaxation of \MIS{} whose objective value for an instance $G = (V, E,w)$ is $\sum_{u \in V} w_u x_u$. A randomized rounding scheme for $x$ is an $(\alpha,\beta)$-rounding, for some parameters $\alpha, \beta \geq 1$, if it always returns a feasible independent set $S$, such that for every vertex $u \in V$,
\begin{enumerate}
    \item $\Pr[u \in S] \geq \frac{1}{\alpha} \cdot x_u$,
    \item $\Pr[u \notin S] \leq \beta \cdot (1 - x_u)$.
\end{enumerate}
\end{definition}

\begin{theorem}[\textrm{[MMV14]}]\label{thm:MMV}
Let $x: V \to [0,1]$ be an optimal (fractional) solution of a convex relaxation of \MIS{} whose objective value for an instance $G = (V, E,w)$ is $\sum_{u \in V} w_u x_u$. Suppose that there exists an $(\alpha,\beta)$-rounding for $x$, for some $\alpha,\beta \geq 1$. Then, $x$ is integral for $(\alpha \beta)$-stable instances; in particular, $x_u \in \{0,1\}$ for every $u \in V$.
\end{theorem}
The proof is identical to the proof given in~\cite{DBLP:conf/soda/MakarychevMV14}, and thus, is not repeated here. As already explained, the theorem suggests a simple robust algorithm: given a convex relaxation for which we have shown the existence of such a scheme, we solve it, and if the solution is integral, we report it, otherwise we report that the instance is not stable (observe that the rounding scheme is used only in the analysis).

In the next section, we study a rounding scheme for the standard LP for \MIS{}, and prove that it satisfies the properties of the theorem. The standard LP for \MIS{} for a graph $G = (V,E,w)$ has an indicator variable $x_u$ for every vertex $u\in V$, and is given in Figure~1.
\begin{figure}
\centering
\begin{align*}
    \max: &\quad \sum_{u \in V} w_u x_u \\
    \textrm{s.t.:}  &\quad x_u + x_v \leq 1, \quad\; \forall (u,v) \in E,\\
                    &\quad x_u \in [0,1], \,\quad\quad \forall u \in V.
\end{align*}
\caption{The standard LP relaxation for Independent Set.}
\label{fig:lp-mis}
\end{figure}
It is a well-known fact that the vertices of this polytope are half-integral~\cite{Nemhauser1975}, and thus there always exists an optimal solution $x$ that satisfies $x_u \in \left\{0,\frac{1}{2},1 \right\}$ for every vertex $u \in V$; moreover, such solution can be computed in polynomial time. This fact will prove very useful in the design of $(\alpha,\beta)$-rounding schemes (as was already shown in Section~\ref{sec:node-mc} for Node Multiway Cut), since it essentially allows us to consider randomized combinatorial algorithms and present them as rounding schemes, as long as they ``preserve" the integral part of the LP (i.e. they never pick a vertex $u$ if $x_u = 0$ and they always pick a vertex if $x_u = 1$).

\subsection{A robust algorithm for $(k-1)$-stable instances of Independent Set on $k$-colorable graphs}

In this section, we give a robust algorithm for $(k-1)$-stable instances of \MIS{} on $k$-colorable graphs. The crucial observation that we make is that, since the rounding scheme in Theorem~\ref{thm:MMV} is only used in the analysis and not in the algorithm, it can be an exponential-time scheme.

Let $G = (V,E,w)$ be a $k$-colorable graph, and let $x$ be an optimal half-integral solution. Let $V_0 = \{u \in V: x_u = 0\}$, $V_{1/2} = \{u \in V: x_u = 1/2\}$ and $V_1 = \{u \in V: x_u = 1\}$. We consider the following rounding scheme of Hochbaum~\cite{HOCHBAUM1983243} (see Algorithm~\ref{alg:k-colorable-rounding}).

\begin{algorithm}[h]
\begin{enumerate}
    \item Let $G_{1/2} = G[V_{1/2}]$ be the induced graph on the set $V_{1/2}$.
    \item Compute a $k$-coloring $f: V_{1/2} \to [k]$ of $G_{1/2}$.
    \item Pick $j$ uniformly at random from the set $[k]$, and set $V_{1/2}^{(j)} := \{u \in V_{1/2}: f(u) = j\}$.
    \item Return $S:= V_{1/2}^{(j)} \cup V_1$.
\end{enumerate}
\caption{Hochbaum's $k$-colorable rounding scheme}
\label{alg:k-colorable-rounding}
\end{algorithm}

\begin{theorem}\label{thm:color-rounding}
Let $G = (V,E,w)$ be a $k$-colorable graph. Given an optimal half-integral solution $x$, the above rounding scheme is a $\left(\frac{k}{2}, \frac{2(k-1)}{k}\right)$-rounding for $x$.
\end{theorem}
\begin{proof}
It is easy to see that the rounding scheme always returns a feasible solution. For $u \in V_0 \cup V_1$, the properties are trivially satisfied. Let $u \in V_{1/2}$. We have $\Pr[u \in S] \geq \frac{1}{k} = \frac{2}{k} \cdot \frac{1}{2} = \frac{2}{k} \cdot x_u$. We also have $\Pr[u \notin S] \leq 1 - \frac{1}{k} = \frac{k-1}{k} = \frac{2(k-1)}{k} \cdot \frac{1}{2} = \frac{2(k-1)}{k} \cdot (1-x_u)$. 
\end{proof}

Combining Theorems~\ref{thm:MMV} and~\ref{thm:color-rounding}, we get the following result.
\begin{theorem}\label{thm:colorable-robust}
The standard LP for \MIS{} is integral for $(k-1)$-stable instances of $k$-colorable graphs.
\end{theorem}
It is easy to see that the above result is tight. For that, we fix some small $\varepsilon > 0$. For any $k \leq n$, we consider a clique of $k$ vertices $\{u_1, ..., u_k\}$, and $n - k$ vertices $\{q_1, ..., q_{n - k}\}$ that are of degree 1, and whose only neighbor is $u_k$. We set $w_{u_i} = 1$, for every $i \in \{1, ..., k-1\}$, $w_{u_k} = k - 1 - \varepsilon/2$ and $w_{q_i} = \frac{\varepsilon}{4(n - k)}$ for every $i \in [n - k]$. It is easy to see that unique optimal solution is $X^* = \{u_k\}$ with cost $w(X^*) = k - 1 - \varepsilon/2$. Let $X_i = \{u_i\} \cup \{q_1, ..., q_{n - k}\}$, for $i \in [k - 1]$. We have
\begin{align*}
    (k - 1 - \varepsilon) \cdot w(X_i \setminus X^*) &= (k - 1 - \varepsilon) \cdot w(X_i) = (k - 1 - \varepsilon) \cdot 1 + (\varepsilon/4)\\
                                                     &= k - 1 - 3\varepsilon/4 < k - 1 - \varepsilon / 2 = w(X^*)\\
                                                     &= w(X^* \setminus X_i).
\end{align*}
It is easy to verify now that this covers all interesting solutions (i.e. maximal), and so the instance is indeed $(k - 1 - \varepsilon)$-stable. If we now consider the fractional solution that assigns $1/2$ to every vertex, we get a solution of cost $k - 1 - \varepsilon / 4 + \varepsilon / 8 > w(X^*)$, and thus, the integrality gap of the LP is strictly larger than 1.

It is a well-known fact that the chromatic number of graph of maximum degree $\Delta$ is at most $\Delta + 1$. Thus, the above result implies a robust algorithm for $\Delta$-stable instances of graphs of maximum degree $\Delta$. This gives a robust analog of the result of Bilu~\cite{Bilu}. And, although the above example seems to suggest that the result is tight, we will now see how we can slightly improve upon it by using Theorem~\ref{thm:colorable-robust} and Brook's theorem~\cite{brooks_1941}.
\begin{theorem}[Brook's theorem~\cite{brooks_1941}]
The chromatic number of a graph is at most the maximum degree $\Delta$, unless the graph is complete or an odd cycle, in which case it is $\Delta + 1$.
\end{theorem}

\MIS{} is easy to compute on cliques and cycles. Thus, by Brook's theorem, every interesting instance of maximum degree $\Delta$ is $\Delta$-colorable. More formally, we obtain the following theorem.
\begin{theorem}\label{thm:delta_minus_one_robust}
There exists a robust algorithm for $(\Delta - 1)$-stable instances of \MIS{}, where $\Delta$ is the maximum degree.
\end{theorem}
\begin{proof}
The algorithm is very simple. If $\Delta \leq 2$, then the graph is a collection of paths and cycles, and we can find the optimal solution in polynomial time. So, let's assume that $\Delta > 2$. In that case, we first separately solve all $K_{\Delta + 1}$ disjoint components (we pick the heaviest vertex of each $K_{\Delta + 1}$), if any, and then we solve the standard LP on the remaining graph (whose stability is the same as the stability of the whole graph). The remaining graph, as implied by Brook's theorem, is $\Delta$-colorable. If the LP is integral, we return the solution (for the whole graph), otherwise we report that the instance is not stable.
\end{proof}

\subsection{Robust algorithms for $(1+\varepsilon)$-stable instances of Independent Set on planar graphs}

In this section, we design a robust algorithm for $(1+\varepsilon)$-stable instances of \MIS{} on planar graphs. We note that Theorem~\ref{thm:colorable-robust} already implies a robust algorithm for $3$-stable instances of planar \MIS{}, but we will use the Sherali-Adams hierarchy (which we denote as SA from now on) to reduce this threshold down to $1 + \varepsilon$, for any fixed $\varepsilon > 0$. In particular, we show that $O(1/\varepsilon)$ rounds of SA suffice to optimally solve $(1+\varepsilon)$-stable instances of \MIS{} on planar graphs. We will not introduce the SA hierarchy formally, and we refer the reader to the many available surveys about LP/SDP hierarchies (see e.g.~\cite{Chlamtac2012}). The $t$-th level of the SA relaxation for \MIS{} has a variable $Y_S$ for every subset $S \subseteq V$ of vertices of size at most $|S| \leq t + 1$, whose intended value is $Y_S = \prod_{u \in S} x_u$, where $x_u$ is the indicator variable of whether $u$ belongs to the independent set. The relaxation has size $n^{O(t)}$, and thus can be solved in time $n^{O(t)}$. For completeness, we give the relaxation in Figure~\ref{fig:SA-LP}.
\begin{figure}[ht]
\begin{align*}
    \max: &\quad \sum_{u \in V} w_u Y_{\{u\}} \\
    \textrm{s.t.:}  &\quad \sum_{T' \subseteq T} (-1)^{|T'|} \cdot \left(Y_{S \cup T' \cup \{u\}} + Y_{S \cup T' \cup \{v\}} - Y_{S \cup T'} \right) \leq 0, \:\:\:\:\: \forall (u,v) \in E, |S| + |T| \leq t,\\
                    &\quad 0 \leq \sum_{T' \subseteq T} (-1)^{|T'|} \cdot Y_{S \cup T' \cup \{u\}} \leq \sum_{T' \subseteq T} (-1)^{|T'|} \cdot Y_{S \cup T'}, \:\:\quad \forall u \in V, |S| + |T| \leq t,\\
                    &\quad Y_{\emptyset} = 1,\\
                    &\quad Y_{S} \in [0,1], \quad \quad\quad \quad\quad \quad\quad \quad\quad \quad\quad \quad\quad \quad\quad \quad\quad \quad\quad\;\; \forall S \subseteq V, |S| \leq t + 1.
\end{align*}
\caption{The Sherali-Adams relaxation for Independent Set.}
\label{fig:SA-LP}
\end{figure}

Our starting point is the work of Magen and Moharrami~\cite{DBLP:conf/approx/MagenM09}, which gives a SA-based PTAS for \MIS{} on planar graphs, inspired by Baker's technique~\cite{DBLP:journals/jacm/Baker94}. In particular, \cite{DBLP:conf/approx/MagenM09} gives a rounding scheme for the $O(t)$-th round of SA that returns a $(1 + O(1/t))$-approximation. In this section, we slightly modify and analyze their rounding scheme, and prove that it satisfies the conditions of Theorem~\ref{thm:MMV}. For that, we need a theorem of Bienstock and Ozbay~\cite{DBLP:journals/disopt/BienstockO04}. For any subgraph $H$ of a graph $G = (V,E)$, let $V(H)$ denote the set of vertices contained in $H$.
\begin{theorem}[\cite{DBLP:journals/disopt/BienstockO04}]\label{thm:SA-treewidth}
Let $t \geq 1$ and $Y$ be a feasible vector for the $t$-th level SA relaxation of the standard Independent Set LP for a graph $G$. Then, for any subgraph $H$ of $G$ of treewidth at most $t$, the vector $(Y_{\{u\}})_{u \in V(H)}$ is a convex combination of independent sets of $H$.
\end{theorem}
The above theorem implies that the $t$-th level SA polytope is equal to the convex hull of all independent sets of the graph, when the graph has treewidth at most $t$.

\paragraph{The rounding scheme of Magen and Moharrami~\cite{DBLP:conf/approx/MagenM09}.}
Let $G = (V,E,w)$ be a planar graph and $\{Y_S\}_{S \subseteq V: |S| \leq t + 1}$ be an optimal $t$-th level solution of SA. We denote $Y_{\{u\}}$ as $y_u$, for any $u \in V$. We first fix a planar embedding of $G$. The vertex set $V$ can then be naturally partitioned into sets $V_0, V_1, ..., V_L$, for some $L \in \{0, ..., n-1\}$, where $V_0$ is the set of vertices in the boundary of the outer face, $V_1$ is the set of vertices in the boundary of the outer face after $V_0$ is removed, and so on. Note that for any edge $(u,v) \in E$, we have $u \in V_i$ and $v \in V_j$ with $|i - j| \leq 1$. We will assume that $L \geq 4$, since, otherwise, the input graph is at most $4$-outerplanar and, in such cases the problem can be solved optimally~\cite{DBLP:journals/jacm/Baker94}.

Following~\cite{DBLP:journals/jacm/Baker94}, we fix a parameter $k \in \{1, ...,L\}$, and for every $i \in \{0, ..., k-1\}$, we define $B(i) = \bigcup_{j \equiv i (\textrm{mod } k)} V_j$. We now pick an index $j \in \{0, ..., k-1\}$ uniformly at random. Let $G_0 = G[V_0 \cup V_1 ... \cup V_j]$, and for $i \geq 1$, $G_i = G[\bigcup_{q = (i-1)k + j}^{ik + j} V_q]$, where for a subset $X \subseteq V$, $G[X]$ is the induced subgraph on $X$. Observe that every edge and vertex of $G$ appears in one or two of the subgraphs $\{G_i\}$, and every vertex $u \in V \setminus B(j)$ appears in exactly one $G_i$.

Magen and Moharrami observe that for every subgraph $G_i = (V(G_i), E(G_i))$, the set of vectors $\{Y_{S}\}_{S \subseteq V(G_i): |S| \leq t + 1}$ is a feasible solution for the $t$-th level SA relaxation of the graph $G_i$. This is easy to see, as the Independent Set LP associated with $G_i$ is weaker than the LP associated with $G$ (on all common variables), since $G_i$ is a subgraph of $G$, and this extends to SA as well. We need one more observation. In~\cite{BODLAENDER19981}, it is proved that a $k$-outerplanar graph has treewidth at most $3k-1$. By construction, each graph $G_i$ is a $(k+1)$-outerplanar graph. Thus, by setting $t = 3k+2$, Theorem~\ref{thm:SA-treewidth} implies that the vector $\{y_u\}_{u \in V(G_i)}$  (we remind the reader that $y_u = Y_{\{u\}}$) can be written as a convex combination of independent sets of $G_i$.

Let $p_i$ be the corresponding distribution of independent sets of $G_i$, implied by the fractional solution $\{y_u\}_{u \in V(G_i)}$. We now consider the following rounding scheme, which always returns a feasible independent set $S$ of the whole graph. For each $G_i$, we (independently) sample an independent set $S_i$ of $G_i$ according to the distribution $p_i$. Each vertex $u \in V \setminus B(j)$ belongs to exactly on graph $G_i$ and is included in the final independent set $S$ if $u \in S_i$. A vertex $u \in B(j)$ might belong to two different graphs $G_i, G_{i+1}$, and so, it is included in the final independent set $S$ only if $u \in S_i \cap S_{i+1}$. The algorithm then returns $S$.

Before we analyze the algorithm, we note that standard arguments that use the tree decomposition of the graph show that the above rounding scheme is constructive (i.e. polynomial-time; this fact is not needed for the proof of integrality of SA for stable instances of planar \MIS{}, but it will be used when designing algorithms for weakly stable instances).
\begin{theorem}\label{thm:planar_rounding}
The above randomized rounding scheme always returns a feasible independent set $S$, such that for every vertex $u \in V$,
\begin{enumerate}
    \item $\Pr[u \in S] \geq \frac{k-1}{k} \cdot y_u + \frac{1}{k} \cdot y_u^2$,
    \item $\Pr[u \notin S] \leq \left(1 + \frac{1}{k} \right) \cdot (1 - y_u)$.
\end{enumerate}
\end{theorem}
\begin{proof}
First, it is easy to see that $S$ is always a feasible independent set. We now compute the corresponding probabilities. Since the marginal probability of $p_i$ on a vertex $u \in G_i$ is $y_u$, we get that, for any fixed $j$, for every vertex $u \in V \setminus B(j)$, we have $\Pr[u \in S] = y_u$, and for every vertex $u \in B(j)$, we have $\Pr[u \in S] \geq y_u^2$. Since $j$ is picked uniformly at random, each vertex $u \in V$ belongs to $B(j)$ with probability exactly equal to $\frac{1}{k}$. Thus, we conclude that for every vertex $u \in V$, we have
\begin{equation*}
    \Pr[u \in S] \geq \frac{k-1}{k} \cdot y_u + \frac{1}{k} \cdot y_u^2 \geq \frac{k - 1}{k} \cdot y_u,
\end{equation*}
and
\begin{equation*}
    \Pr[u \notin S] \leq 1 - \left(\frac{k-1}{k} \cdot y_u + \frac{1}{k} \cdot y_u^2 \right) = 1 - y_u + \frac{y_u}{k} \cdot (1 - y_u) \leq \left(1 + \frac{1}{k} \right) \cdot (1 - y_u).
\end{equation*}
\end{proof}
The above theorem implies that the rounding scheme is a $\left(\frac{k}{k-1}, \frac{k+1}{k} \right)$-rounding. It is easy now to prove the following theorem.
\begin{theorem}\label{thm:planar-stable-SA}
For every $\varepsilon > 0$, the SA relaxation of $\left(3\left\lceil \frac{2}{\varepsilon} \right\rceil + 5 \right) = O(1 / \varepsilon)$ rounds is integral for $(1 + \varepsilon)$-stable instances of \MIS{} on planar graphs.
\end{theorem}
\begin{proof}
The theorem is a direct consequence of Theorem~\ref{thm:MMV} and Theorem~\ref{thm:planar_rounding}. For any given $k \geq 2$, by Theorem~\ref{thm:planar_rounding}, the rounding scheme always returns a feasible independent set $S$ of $G$ that satisfies $\Pr[u \in S] \geq \frac{k-1}{k} \cdot y_u$ and $\Pr[u \notin S] \leq \left(1 + \frac{1}{k} \right) \cdot (1 - y_u)$ for every vertex $u \in V$. By Theorem~\ref{thm:MMV}, this means that $\{y_u\}_{u \in V}$ must be integral for $(1 + \frac{2}{k-1})$-stable instances. For any fixed $\varepsilon > 0$, by setting $k = \left\lceil \frac{2}{\varepsilon} \right\rceil + 1$, we get that $3\left\lceil \frac{2}{\varepsilon} \right\rceil + 5 = O(1 / \varepsilon)$ rounds of Sherali-Adams return an integral solution for $(1+\varepsilon)$-stable instances of \MIS{} on planar graphs.
\end{proof}

Again, by using the techniques of [MMV14] combined with Theorem~\ref{thm:planar_rounding}, we also get the following result.
\begin{theorem}
For every fixed $\varepsilon > 0$, there is a polynomial-time algorithm that, given a $(1 + \varepsilon, \mathcal{N})$-weakly-stable instance of planar \MIS{} with integer weights, finds a solution $S \in \mathcal{N}$.
\end{theorem}

\section{Stability and integrality gaps of convex relaxations}

In this section, we state a general theorem about the integrality gap of convex relaxations of maximization problems on stable instances. As already stated, [MMV14] was the first work that analyzed the performance of convex relaxations on stable instances, and gave sufficient conditions for a relaxation to be integral (see Theorem~\ref{thm:MMV}). Here, we show that, even if the conditions of Theorem~\ref{thm:MMV} are not satisfied, the integrality gap still significantly decreases as stability increases.
\begin{theorem}\label{thm:ig}
Consider a convex relaxation for \MIS{} that assigns a value $x_u \in [0,1]$ to every vertex $u$ of a graph $G=(V,E,w)$, such that its objective function is $\sum_{u \in V} w_u x_u$. Let $\alpha$ be its integrality gap, for some $\alpha > 1$. Then, the relaxation has integrality gap at most $\min\left\{\alpha, 1 + \frac{1}{\beta - 1} \right\}$ for $(\alpha \beta)$-stable instances, for any $\beta > 1$.
\end{theorem}
\begin{proof}
Let $G = (V,E,w)$ be an $(\alpha\beta)$-stable instance, let $I^*$ denote its (unique) optimal independent set and $\mathrm{OPT} = w(I^*)$ be its cost. We assume that $\beta$ is such that $1 + \frac{1}{\beta - 1} < \alpha$ (otherwise the statement is trivial), which holds for $\beta > 1 + \frac{1}{\alpha - 1}$.

Let $\mathrm{OPT}_{\mathrm{REL}}$ be the optimal value of the relaxation, and let's assume that $\frac{\mathrm{OPT}_{\mathrm{REL}}}{\mathrm{OPT}} > 1 + \frac{1}{\beta - 1}$. We now claim that $\sum_{u \in I^*} w_u x_u < (\beta - 1) \cdot \sum_{u \in V \setminus I^*} w_u x_u$. To see this, suppose that $\sum_{u \in I^*} w_u x_u \geq (\beta - 1) \cdot \sum_{u \in V \setminus I^*} w_u x_u$. We have $\mathrm{OPT}_{\mathrm{REL}} = \sum_{u \in I^*} w_u x_u + \sum_{u \in V \setminus I^*} w_u x_u \leq \left(1 + \frac{1}{\beta-1}\right) \cdot \sum_{u \in I^*} w_u x_u \leq \left(1 + \frac{1}{\beta-1}\right) \cdot \mathrm{OPT}$, which implies that $\frac{\mathrm{OPT}_{\mathrm{REL}}}{\mathrm{OPT}} \leq 1 + \frac{1}{\beta-1}$. This contradicts our assumption, and so we conclude that $\sum_{u \in I^*} w_u x_u < (\beta - 1) \cdot \sum_{u \in V \setminus I^*} w_u x_u$. This  implies that $\sum_{u \in V \setminus I^*} w_u x_u > \frac{\mathrm{OPT}_{\mathrm{REL}}}{\beta}$.

We now consider the induced graph $H = G[V \setminus I^*]$. Let $S \subseteq V \setminus I^*$ be an optimal independent set of $H$. We observe that the restriction of the fractional solution to the vertices of graph $H$ is a feasible solution for the corresponding relaxation for $H$. Thus, since the integrality gap is always at most $\alpha$, we have $w(S) \geq \frac{1}{\alpha} \cdot \sum_{u \in V \setminus I^*} w_u x_u$. Finally, we observe that $S$ is a feasible independent set of the graph $G$. The definition of stability now implies that $w(I^* \setminus S) > (\alpha\beta) \cdot w(S \setminus I^*)$, which gives $w(I^*) > (\alpha\beta) \cdot w(S)$. Combining the above inequalities, we conclude that $\mathrm{OPT} = w(I^*) > \mathrm{OPT}_{\mathrm{REL}}$. Thus, we get a contradiction.
\end{proof}

The above result is inherently non-constructive. Nevertheless, it suggests approximation estimation algorithms for stable instances of \MIS{}, such as the following.
\begin{corollary}[Bansal et al.~\cite{DBLP:conf/stoc/BansalGG15} + Theorem~\ref{thm:ig}]
For any fixed $\varepsilon > 0$, the Lovasz  $\theta$-function SDP relaxation has integrality gap at most $1 + \varepsilon$ on $\widetilde{O}\left(\frac{1}{\varepsilon} \cdot \frac{\Delta}{\log^{3/2} \Delta} \right)$-stable instances of \MIS{} of maximum degree $\Delta$, where the notation $\widetilde{O}$ hides some $\mathrm{poly}(\log \log \Delta)$ factors.
\end{corollary}

We note that the theorem naturally extends to many other maximization graph problems, and is particularly interesting for relaxations that require super-constant stability for the recovery of the optimal solution (e.g. the Max Cut SDP with $\ell_2^2$ triangle inequalities has integrality gap $1+ \varepsilon$ for $\left(\frac{2}{\varepsilon}\right)$-stable instances although the integrality gap drops to exactly 1 for $\Omega(\sqrt{\log n})$-stable instances).

In general, such a theorem is not expected to hold for minimization problems, but, in our case of study, \MIS{} gives rise to its complementary minimization problem, the minimum Vertex Cover problem, and it turns out that we can prove a very similar result for Vertex Cover as well. More precisely, we prove the following.
\begin{theorem}\label{thm:VC-estimation}
Suppose that there exists a convex relaxation for Independent Set whose objective function is $\sum_{u \in V} w_u x_u$ and its integrality gap, w.r.t.~Independent Set, is $\alpha$. Then, there exists a $\min \left\{2, 1+ \frac{1}{\beta-2} \right\}$-estimation approximation algorithm for $(\alpha\beta)$-stable instances of Vertex Cover, for any $\beta > 2$.
\end{theorem}

Before we prove the Theorem, we will need the following Lemma.
\begin{lemma}\label{lem:delete-points-inside}
Let $G = (V,E,w)$ be a $\gamma$-stable instance of \MIS{} whose optimal independent set is $I^*$. Let $v \in I^*$. Then, the instance $\widetilde{G} = G[V\setminus (\{v\} \cup N(v))]$ is also $\gamma$-stable, and its maximum independent set is $I^* \setminus \{v\}$.
\end{lemma}
\begin{proof}
It is easy to see that $I^* \setminus \{v\}$ is a maximum independent set of $\widetilde{G}$. We will now prove that the instance is $\gamma$-stable. Let's assume that there exists a perturbation $w'$ of $\widetilde{G}$ such that $I' \neq (I^* \setminus \{v\})$ is a maximum independent set of $\widetilde{G}$. This means that $w'(I') \geq w'(I^* \setminus {v})$. We now extend $w'$ to the whole vertex set $V$ by setting $w_u' = w_u$ for every $u \in \{v\} \cup N(v)$. It is easy to verify that $w'$ is now a $\gamma$-perturbation for $G$. Observe that $I' \cup \{v\}$ is a feasible independent set of $G$, and we now have $w'(I' \cup \{v\})  = w'(I') + w_v' \geq w'(I^* \setminus {v}) + w_v' = w'(I^*)$. Thus, we get a contradiction.
\end{proof}

\begin{proof}[Proof of Theorem~\ref{thm:VC-estimation}]
In this proof, we use a standard trick that is used for turning any good approximation algorithm for Maximum Independent Set to a good approximation algorithm for Minimum Vertex Cover. The trick is based on the fact that, if we solve the standard LP for Independent Set and look at the vertices that are half-integral, then in the induced graph on these vertices, the largest independent set is at most the size of the minimum vertex cover, and thus, any good approximate solution to Independent Set would directly translate to a good approximate solution to Vertex Cover.
	
Let $G=(V,E,w)$ be an $(\alpha\beta)$-stable instance of Vertex Cover and let $X^* \subseteq V$ be its (unique) optimal vertex cover, and $I^* = V \setminus X^*$ be its (unique) optimal independent set. We first solve the standard LP relaxation for \MIS{} (see Figure~\ref{fig:lp-mis}) and compute an optimal half-integral solution $x$. The solution $x$ naturally partitions the vertex set into three sets, $V_0 = \{u: x_u = 0\}$, $V_{1/2} = \{u: x_u = 1/2\}$ and $V_1 = \{u: x_u = 1\}$. It is well known (see~\cite{Nemhauser1975}) that $V_1 \subseteq I^*$ and $V_0 \cap I^* = \emptyset$. Thus, it is easy to see that $I^* = V_1 \cup I_{1/2}^*$, where $I_{1/2}^*$ is an optimal independent set of the induced graph $G[V_{1/2}]$ (similarly, $X^* = V_0 \cup (V_{1/2} \setminus I_{1/2}^*)$).
	
We now use the simple fact that $N(V_1) = V_0$. By iteratively applying Lemma~\ref{lem:delete-points-inside} for the vertices of $V_1$, we get that $G[V_{1/2}]$ is $(\alpha\beta)$-stable, and so it has a unique optimal independent set $I_{1/2}^*$. Let $X_{1/2}^* = V_{1/2} \setminus I_{1/2}^*$ be the unique optimal vertex cover of $G[V_{1/2}]$. It is easy to see that solution $\{x_u\}_{u \in V_{1/2}}$ (i.e. the solution that assigns value $1/2$ to every vertex) is an optimal fractional solution for $G[V_{1/2}]$. This implies that $w(I_{1/2}^*) \leq \frac{w(V_{1/2})}{2} \leq w(X_{1/2}^*)$.
	
Since $G[V_{1/2}]$ is $(\alpha\beta)$-stable, by Theorem~\ref{thm:ig} we know that the integrality gap of a convex relaxation relaxation for $G[V_{1/2}]$ is at most $\min\{\alpha, \beta / (\beta - 1)\}$. Let $A = \min\{\alpha, \beta / (\beta - 1)\}$, and let $\textrm{FRAC}$ be the optimal fractional cost of the relaxation for $G[V_{1/2}]$, w.r.t.~\MIS{}. Thus, we get that $w(I_{1/2}^*) \geq \frac{1}{A} \cdot \textrm{FRAC}$. From now on, we assume that $\beta > 2$, which implies that $1 \leq A < 2$. We now have
\begin{align*}
    w(V_{1/2}) - \textrm{FRAC} &\geq w(V_{1/2}) - A \cdot w(I_{1/2}^*) = w(V_{1/2}) - w(I_{1/2}^*) - (A - 1) \cdot w(I_{1/2}^*)\\
                               &\geq w(X_{1/2}^*) - (A - 1) \cdot w(X_{1/2}^*) = (2 - A) \cdot w(X_{1/2}^*).
\end{align*}
We conclude that $w(X_{1/2}^*) \leq \frac{1}{2 - A} \cdot (w(V_{1/2}) - \textrm{FRAC})$. Thus, for any $\beta > 2$,
\begin{align*}
    w(V_0) + (w(V_{1/2}) - \textrm{FRAC})  &\geq w(V_0) + (2 - A) w(X_{1/2}^*) \geq (2 - A) (w(V_0) + w(X_{1/2}^*)\\
                                           &= (2 - A) w(X^*).
\end{align*}
Since $\frac{1}{2 - A} \leq \frac{\beta - 1}{\beta - 2}$, we get that we have a $\left(1 + \frac{1}{\beta - 2} \right)$-estimation approximation algorithm for Vertex Cover on $(\alpha\beta)$-stable instances. We now combine this algorithm with any 2-approximation algorithm for Vertex Cover, and always return the minimum of the two algorithms. This concludes the proof.
\end{proof}

\begin{corollary}[Bansal et al.~\cite{DBLP:conf/stoc/BansalGG15} + Theorem~\ref{thm:VC-estimation}]
For every fixed $\varepsilon > 0$, there exists a $(1 + \varepsilon)$-estimation approximation algorithm for $\widetilde{O}\left(\frac{1}{\varepsilon} \cdot \frac{\Delta}{\log^{3/2} \Delta} \right)$-stable instances of Minimum Vertex Cover of maximum degree $\Delta$, where the notation $\widetilde{O}$ hides some $\mathrm{poly}(\log \log \Delta)$ factors.
\end{corollary}

\section[Stable Independent Set on general graphs]{Combinatorial algorithms for stable instances of Independent Set on general graphs}\label{sec:upper-bound-general}

In this section, we use our algorithm for $(k-1)$-stable instances of $k$-colorable graphs and the standard greedy algorithm as subroutines to solve $(\varepsilon \cdot n)$-stable instances on graphs of $n$ vertices, in time $n^{O(1/\varepsilon)}$. Thus, from now on we assume that $\varepsilon > 0$ is a fixed constant. Before presenting our algorithm, we will prove a few lemmas. First, we need the following standard fact about the chromatic number of a graph. For any graph $G$, let $\chi(G)$ be its chromatic number. We also denote the neighborhood of a vertex $u$ as $N(u) = \{v: (u,v) \in E\}$.
\begin{lemma}[Welsh-Powell coloring]\label{lemma:welsh-powell}
Let $G(V,E)$ be a graph, where $n = |V|$, and let $d_1 \geq d_2 \geq ... \geq d_n$ be the sequence of its degrees in decreasing order. Then, $\chi(G) \leq \max_i \min \{d_i + 1, i\}$.
\end{lemma}
\begin{proof}
The lemma is based on a simple observation. We consider the following greedy algorithm for coloring. Suppose that $u_1, ..., u_n$ are the vertices of the graph, with corresponding degrees $d_1 \geq d_2 \geq ... \geq d_n$. The colors are represented with the numbers $\{1, ..., n\}$. The greedy coloring algorithm colors one vertex at a time, starting from $u_1$ and concluding with $u_n$, and for each such vertex $u_i$, it picks the ``smallest" available color. It is easy to see that for each vertex $u_i$, the color that the algorithm picks is at most ``$i$". Since the algorithm picks the smallest available color, and since the vertex $u_i$ has $d_i$ neighbors, we observe that the color picked will also be at most ``$d_i + 1$". Thus, the color of vertex $u_i$ is at most $\min\{d_i + 1, i\}$. It is easy to see now that when the algorithm terminates, it will have used at most $\max_i \min\{d_i + 1, i\}$ colors, and thus $\chi(G) \leq \max_i \min\{d_1 + 1, i\}$.
\end{proof}

We can now prove the following useful fact.
\begin{lemma}\label{lemma:colors-degrees}
Let $G=(V,E)$ be a graph, with $n = |V|$. Then, for any natural number $k \geq 1$, one of the following two properties is true:
\begin{enumerate}
    \item $\chi(G) \leq \left\lceil \frac{n}{k} \right\rceil$, or
    \item there are at least $\left\lceil \frac{n}{k} \right\rceil + 1$ vertices in $G$ whose degree is at least $\left\lceil \frac{n}{k} \right\rceil$.
\end{enumerate}
\end{lemma}
\begin{proof}
Suppose that $\chi(G) > \left\lceil \frac{n}{k} \right\rceil$. Let $u_1, ..., u_n$ be the vertices of $G$, with corresponding degrees $d_1 \geq d_2 \geq ... \geq d_n$. It is easy to see that $\max_{1 \leq i \leq \left\lceil \frac{n}{k} \right\rceil} \min \{d_i + 1, i\} \leq \left\lceil \frac{n}{k} \right\rceil$. We now observe that if $d_{\left\lceil \frac{n}{k} \right\rceil + 1} < \left\lceil \frac{n}{k} \right\rceil$, then we would have $\max_{\left\lceil \frac{n}{k} \right\rceil + 1 \leq i \leq n} \min\{d_i + 1, i\} \leq \left\lceil \frac{n}{k} \right\rceil$, and thus, by Lemma~\ref{lemma:welsh-powell}, we would get that $\chi(G) \leq \left\lceil \frac{n}{k} \right\rceil$, which is a contradiction. We conclude that we must have $d_{\left\lceil \frac{n}{k} \right\rceil + 1} \geq \left\lceil \frac{n}{k} \right\rceil$, which, since the vertices are ordered in decreasing order of their degrees, implies that there are at least $\left\lceil \frac{n}{k} \right\rceil + 1$ vertices whose degree is at least $\left\lceil \frac{n}{k} \right\rceil$.
\end{proof}

We will also need the following lemma.
\begin{lemma}\label{lem:stable-delete-point-outside}
Let $G = (V,E, w)$ be $\gamma$-stable instance of \MIS{} whose optimal independent set is $I^*$. Then, $\widetilde{G} = G[V \setminus X]$ is $\gamma$-stable, for any set $X \subseteq V \setminus I^*$.
\end{lemma}
\begin{proof}
Fix a subset $X \subseteq V \setminus I^*$. It is easy to see that any independent set of $\widetilde{G} = G[V \setminus X]$ is an independent set of the original graph $G$. Let's assume that $\widetilde{G}$ is not $\gamma$-stable, i.e. there exists a $\gamma$-perturbation $w'$ such that $I' \neq I^*$ is a maximum independent set of $\widetilde{G}$. This means that $w'(I') \geq w'(I^*)$. By extending the perturbation $w'$ to the whole vertex set $V$ (simply by not perturbing the weights of the vertices of $X$), we get a valid $\gamma$-perturbation for the original graph $G$ such that $I'$ is at least as large as $I^*$. Thus, we get a contradiction.
\end{proof}

Since we will need the standard greedy algorithm as a subroutine, we explicitly state the algorithm here (see Algorithm~\ref{alg:greedy}).
\begin{algorithm}[h]
\begin{enumerate}
    \item Let $S := \emptyset$ and $X := V$.
    \item while $(X \neq \emptyset)$:\\
            \hspace*{20pt}Pick $u := \arg\max_{v \in X} \{w_v\}$.\\
            \hspace*{20pt}Set $S := S \cup \{u\}$ and $X := X \setminus (\{u\} \cup N(u))$.
    \item Return $S$.
\end{enumerate}
\caption{The greedy algorithm for \MIS{}}
\label{alg:greedy}
\end{algorithm}

Bilu~\cite{Bilu} proved the following theorem.
\begin{theorem}[\cite{Bilu}]\label{thm:bilu-greedy}
The Greedy algorithm (see Algorithm~\ref{alg:greedy}) solves $\Delta$-stable instances of \MIS{} on graphs of maximum degree $\Delta$.
\end{theorem}

We will now present an algorithm for $(n/k)$-stable instances of graphs with $n$ vertices, for any natural number $k \geq 1$, that runs in time $n^{O(k)}$. Let $G = (V,E,w)$ be a $(n/k)$-stable instance of \MIS{}, where $n = |V|$. The algorithm is defined recursively (see Algorithm~\ref{alg:unbounded-degree}).

\begin{algorithm}[h]
\noindent \texttt{Unbounded-Degree-Alg$(G,k)$}:
\begin{enumerate}
    \item If $k = 1$, run greedy algorithm (Algorithm~\ref{alg:greedy}) on $G$, report solution and exit.
    \item Solve standard LP relaxation for $G$ and obtain (fractional) solution $\{x_u\}_{u \in V}$.
    \item If $\{x_u\}_{u \in V}$ is integral, report solution and exit.
    \item Let $X = \{u \in V: \mathrm{deg}(u) \geq \left\lceil \frac{n}{k} \right\rceil\}$, and for each $u \in X$, let $G_u = G[V \setminus (\{u\} \cup N(u))]$.
    \item For each $u \in X$, run \texttt{Unbounded-Degree-Alg$(G_u,k-1)$} and obtain independent \\set $S_u$. Set $I_u = S_u \cup \{u\}$.
    \item Let $\widetilde{G} = G[V \setminus X]$ and run \texttt{Unbounded-Degree-Alg$(\widetilde{G},k-1)$} to obtain independent\\ set $\widetilde{I}$.
    \item Return the maximum independent set among $\{I_u\}_{u \in X}$ and $\widetilde{I}$.
\end{enumerate}
\caption{The algorithm for $(n/k)$-stable instances of \MIS{}}
\label{alg:unbounded-degree}
\end{algorithm}

\begin{theorem}
The above algorithm optimally solves $(n/k)$-stable instances of \MIS{} on graphs with $n$ vertices, in time $n^{O(k)}$.
\end{theorem}
\begin{proof}
We will prove the theorem using induction on $k$. Let $G = (V,E,w)$, $n = |V|$, be a $(n/k)$-stable instance whose optimal independent set is $I^*$. If $k = 1$, Theorem~\ref{thm:bilu-greedy} shows that the greedy algorithm computes the optimal solution (by setting $\Delta = n - 1$), and thus our algorithm is correct.

Let $k\geq 2$, and let's assume that the algorithm correctly solve $(N/k')$-stable instances of graphs with $N$ vertices, for any $1 \leq k' < k$. We will show that it also correctly solves $(N/k)$-stable instances. By Lemma~\ref{lemma:colors-degrees}, we know that either the chromatic number of $G$ is at most $\left\lceil \frac{n}{k} \right\rceil$, or there are at least $\left\lceil \frac{n}{k} \right\rceil + 1$ whose degree is at least $\left\lceil \frac{n}{k} \right\rceil$. If the chromatic number is at most $\left\lceil \frac{n}{k} \right\rceil$, then, by Theorem~\ref{thm:colorable-robust} we know that the standard LP relaxation is integral if $G$ is $(\left\lceil \frac{n}{k} \right\rceil - 1)$-stable. We have $\left\lceil \frac{n}{k} \right\rceil - 1 \leq \left\lfloor \frac{n}{k} \right\rfloor \leq n/k$. Thus, in this case, the LP will be integral and the algorithm will terminate at step (3), returning the optimal solution.

So, let's assume that the LP is not integral for $G = (V,E,w)$, which means that the chromatic number of the graph is strictly larger than $\left\lceil \frac{n}{k} \right\rceil$. This means that the set of vertices $X = \{u \in V: \mathrm{deg}(v) \geq \left\lceil \frac{n}{k} \right\rceil\}$ has size at least $|X| \geq \left\lceil \frac{n}{k} \right\rceil + 1$. Fix a vertex $u \in X$. If $u \in I^*$, then, by Lemma~\ref{lem:delete-points-inside}, we get that $G_u$ is $(n/k)$-stable, and moreover, $I^* = \{u\} \cup I_u^*$, where $I_u^*$ is the optimal independent set of $G_u$. Note that $G_u$ has at most $n - \left\lceil \frac{n}{k} \right\rceil - 1 \leq \left\lfloor\frac{(k-1)}{k} \cdot n \right\rfloor = n'$ vertices. It is easy to verify that $n/k \geq n'/(k-1)$, which implies that $G_u = (V_u, E_u, w)$ is a $\left(\frac{|V_u|}{k-1} \right)$-stable instance with $|V_u|$ vertices. Thus, by the inductive hypothesis, the algorithm will compute its optimal independent set $S_u \equiv I_u^*$.

There is only one case remaining, and this is the case where $X \cap I^* = \emptyset$. In this case, by Lemma~\ref{lem:stable-delete-point-outside}, we get that $\widetilde{G} = G[V\setminus X]$ is $(n/k)$-stable. There are at most $n - \left\lceil \frac{n}{k} \right\rceil - 1$ vertices in $\widetilde{G}$, and so, by a similar argument as above, the graph $\widetilde{G} = (\widetilde{V}, \widetilde{E}, w)$ is a $\left(\frac{|\widetilde{V}|}{k-1} \right)$-stable instance with $|\widetilde{V}|$ vertices, and so, by the inductive hypothesis, the algorithm will compute its optimal independent set.

It is clear now that, since the algorithm always picks the best possible independent set, then at step (7) it will return the optimal independent set of $G$. This concludes the induction and shows that our algorithm is correct.

Regarding the running time, it is quite easy to see that we have at most $k$ levels of recursion, and at any level, each subproblem gives rise to at most $n$ new subproblems. Thus, the total running time is bounded by $\poly(n) \cdot n^{k + 1} = n^{O(k)}$. Thus, the algorithm always runs in time $n^{O(k)}$.
\end{proof}

It is immediate now that, for any given $\varepsilon > 0$, we can set $k = \lceil 1/\varepsilon \rceil$, and run our algorithm in order to optimally solve $(\varepsilon n)$-stable instances of \MIS{} with $n$ vertices, in total time $n^{O(1/\varepsilon)}$.

\section{The greedy algorithms solves weakly stable instances of \MIS{}}

In this section, we observe that the greedy algorithm, analyzed by Bilu~\cite{Bilu}, solves $(\Delta, \mathcal{N})$-weakly-stable instances of \MIS{} on graphs of maximum degree $\Delta$. Let $G=(V,E,w)$ be a graph of maximum degree $\Delta$, and $N(u) = \{v: (u,v)\in E\}$.

\begin{theorem}\label{thm:certified-greedy}
Given a $(\Delta, \mathcal{N})$-weakly-stable instance of \MIS{} of a graph of maximum degree $\Delta$, the Greedy algorithm (see Algorithm~\ref{alg:greedy}) returns a solution $I \in \mathcal{N}$ in polynomial time.
\end{theorem}
\begin{proof}
Let $I$ be the solution returned by the greedy algorithm, and let $I^*$ be the unique optimal solution. We will prove that $w(I^* \setminus I) \leq \Delta \cdot w(I \setminus I^*)$. For that, we define the perturbation $w'$ that sets $w_u' = \Delta \cdot w_u$ for every $u \in I$ and $w_u' = w_u$ for every $u \in V \setminus I$. We will now prove that $I$ is optimal for $G' = (V, E, w')$. Let $I'$ be an optimal solution for $G'$. We look at the execution of the greedy algorithm, and let $u_1, ..., u_t$ be the vertices that the algorithm picks, in that order. Clearly, $u_1$ is a vertex of maximum weight. Since the degree is at most $\Delta$, this means that $\Delta \cdot w_{u_1} \geq w(N(u_1))$. Thus, if $N(u_1) \cap I' \neq \emptyset$, then we can always define a feasible independent set $(I' \setminus N(u_1)) \cup \{u_1\}$ whose cost (w.r.t.~$w'$) is at least as much as the cost of $I'$. Thus, we can always obtain an optimal independent set $I_1'$ for $G'$ that contains $u_1$. We now remove the vertices $\{u_1\} \cup N(u_1)$ from the graph and look at the induced graph $G'[V \setminus (\{u_1\} \cup N(u_1))]$. It is again easy to see that $u_2$ is a vertex of maximum weight in $G'[V \setminus (\{u_1\} \cup N(u_1))]$. So, we can use the same argument to show that there exists an optimal solution $I_2'$ that contains both $u_1$ and $u_2$. Applying this argument inductively, we conclude that there exists an optimal solution for $G'$ that contains all the vertices $\{u_1, ..., u_k\}$. In other words, $I$ is an optimal solution for $G'$. This means that $w'(I') \geq w(I^*)$, which implies that $w(I^* \setminus I) \leq \Delta \cdot w(I \setminus I^*)$. This concludes the proof.
\end{proof}

\chapter[LP's and perturbation-resilient clustering]{LP-based results for perturbation-resilient clustering}\label{chap:clustering}

In this chapter, we apply LP-based techniques to obtain a robust algorithm for 2-metric-perturbation-resilient instances of symmetric $k$-center, and also exhibit some lower bounds for the integrality of the $k$-median LP relaxation on perturbation-resilient instances. We note here that the robust algorithm for 2-metric-perturbation-resilient instances of symmetric $k$-center presented in Section~\ref{sec:robust-k-center-lp} was also independently obtained by Chekuri and Gupta~\cite{Chekuri-Approx18}; in fact, they extend the result for the case of asymmetric $k$-center as well.

\section[Metric-perturbation-resilient symmetric $k$-center]{A robust algorithm for $2$-metric-perturbation-resilient symmetric $k$-center}\label{sec:robust-k-center-lp}

The $k$-center problem is a very well studied clustering problem with many applications (see e.g.~\cite{DBLP:journals/tois/Can93, DBLP:journals/siamcomp/CharikarCFM04, DBLP:journals/jacm/ChuzhoyGHKKKN05, DYER1985285, DBLP:journals/jal/PanigrahyV98}). One classical application of $k$-center is the problem of placing $k$ fire stations in a city so as to minimize the maximum time for a fire truck to reach any location. For symmetric k-center, which will be our case of study in this section, there is a 2-approximation algorithm and, moreover, the problem is \NP-hard to approximate within a factor of $2 - \varepsilon$ (see~\cite{DBLP:journals/mor/HochbaumS85, HSU1979209}).

The problem was first studied in the context of stability and perturbation resilient by Awasthi et al.~\cite{DBLP:journals/ipl/AwasthiBS12}, in which they gave a non-robust algorithm for 3-perturbation-resilient instances of symmetric $k$-center. Balcan et al.~\cite{DBLP:conf/icalp/BalcanHW16} obtained improved results, and specifically they gave a non-robust algorithm for 2-perturbation-resilient instances of both symmetric and asymmetric $k$-center. They also proved that there are no algorithms for $(2 - \varepsilon)$-perturbation-resilient instances of symmetric $k$-center, unless $\NP = \RP$. A non-robust algorithm for 2-metric-perturbation-resilient instances of symmetric $k$-center was subsequently given in~\cite{DBLP:journals/corr/MakarychevM16, DBLP:conf/stoc/AngelidakisMM17}. In this section, we give a \emph{robust} algorithm for $2$-metric-perturbation-resilient instances of symmetric $k$-center, based on linear programming. We first define the problem. Throughout this chapter, we denote the distance of a point $u$ to a set $A$ as $d(u,A) = \min_{v \in A} d(u,v)$.

\begin{definition}[symmetric $k$-center]
Let $(\X,d)$ be a finite metric space, $|\X| = n$, and let $k \in [n]$. The goal is to select a set $C \subset \X$ of $k$ centers (i.e. $|C| = k$), so as to minimize the objective $\max_{u \in \X} d(u,C)$.
\end{definition}
We will always assume that the distance function is symmetric (i.e. $d(u,v) = d(v,u)$ for every $u,v\in \X$). We crucially use the following theorem, first proved by Balcan et al.~\cite{DBLP:conf/icalp/BalcanHW16}. We reprove the theorem here, confirming that the original proof works for metric perturbation resilience and not only perturbation resilience.

\begin{theorem}[\cite{DBLP:conf/icalp/BalcanHW16}]\label{thm:BHW16}
Let $(\X,d)$ be an $\alpha$-metric-perturbation-resilient instance of symmetric $k$-center with optimal value $R^*$, for some $\alpha > 1$. Let $C \subset \X$ be a set of $k$ centers such that $\max_{u \in \X} d(u,C) \leq \alpha \cdot R^*$. Then, the Voronoi partition induced by $C$ is the (unique) optimal clustering.
\end{theorem}
\begin{proof}
Let $C$ be a set of $k$ centers with $\max_{u \in \X} d(u,C) \leq \alpha \cdot R^*$. Let $c(u) = \arg \min_{c \in C} d(u, C)$, breaking ties arbitrarily. We define $l(u,v) = \alpha \cdot d(u,v)$ for all $u \in \X$ and $v \in \X \setminus \{c(u)\}$. For every $u \in \X$, we also set $l(u, c(u)) = \alpha \cdot \min \{d(u, c(u)), R^*\}$. Let $d'$ be the shortest-path metric induced by the length function $l$. We first show that $d'$ is an $\alpha$-metric perturbation. It is easy to see that for all $u, v \in \X$, we have $d'(u,v) \leq \alpha d(u,v)$. We also observe that $l(u, c(u)) \geq d(u, c(u))$, since $\alpha \cdot R^* \geq d(u, c(u))$. Thus, any path from $u$ to $v$ has length (w.r.t.~edge length $l$) at least the length based on distances $d$. This implies that $d'(u,v) \geq d(u,v)$, and so we conclude that $d'$ is indeed an $\alpha$-metric perturbation.

We will now prove that the optimal cost for the instance $(\X, d')$ is exactly $\alpha R^*$. To see this, let's assume that there is a set of centers $C'$ such that $\max_{u \in \X} d'(u, C') < \alpha R^*$. Let $u, v \in \X$ such that $d(u, v) \geq R^*$. By construction, it is easy to see that $d'(u,v) \geq \alpha R^*$. Thus, if $d'(u, v) < \alpha R^*$, then we must have $d(u, v) < R^*$. Since $\max_{u \in \X} d'(u, C') < \alpha R^*$, we must have $d(u, C') < R^*$ for every $u \in \X$, and so this contradicts the optimality of $R^*$. We conclude that the optimal cost of $(\X, d')$ is exactly $\alpha \cdot R^*$, since it cannot be less than that and $\max_{u \in \X} d'(u, C) \leq \alpha R^*$, by construction. This means that $C$ is an optimal set of centers for $(\X, d')$, whose induced Voronoi partition (under $d'$) gives the unique optimal clustering.

Finally, we show that the Voronoi partition induced by $C$ under $d$ is the same as the Voronoi partition under $d'$. This is easy to verify, as for any point $u \in \X$, its distance to $c(u)$ is at most $\alpha R^*$, and the only ``shortcuts" that are created by $d'$ involve pairwise distances of length at least $\alpha R^*$. Thus, we have $d'(u, c(u)) = d'(u, C)$. This concludes the proof.
\end{proof}

We now define the LP we will use. Let $B(u,r) = \{v \in \X: d(u,v) \leq r\}$. For every $R > 0$, let $\cP(R)$ be the polytope as defined in Figure~\ref{fig:kcenter-lp}. Observe that polytope is described by $\poly(n)$ variables and inequalities.
\begin{figure}
\begin{align*}
    \cP(R): & \quad \sum_{v \in \X} y_v \leq k, &&\\
            & \quad x_{uv} \leq y_v,  &&\forall (u,v) \in X \times \X, \\
            & \quad \sum_{v \in B(u,R)} x_{uv} \geq 1,  &&\forall u \in X, \\
            & \quad \sum_{v \in \X \setminus B(u,R)} x_{uv} = 0, &&\forall u \in X, \\
            & \quad 0 \leq x_{uv} \leq 1 &&\forall (u,v) \in \X \times \X,\\
            & \quad 0 \leq y_v \leq 1, && \forall v \in \X.
\end{align*}
\caption{The symmetric $k$-center polytope.}
\label{fig:kcenter-lp}
\end{figure}
In the intended integral solution, the variable $y_v$ denotes whether $v$ is selected as a center or not, and $x_{uv}$ is 1 if $u$ is ``served" by center $v$. Let $(x,y) \in \cP(R)$, given that $\cP(R) \neq \emptyset$. Let $\mathrm{supp}_u(x,y) = \{v \in \X: x_{uv} > 0\}$ denote the set of centers which $u$ is (fractionally) connected to; we call this set the support of $u$. When the solution $(x,y)$ is clear from the context, we will write $\mathrm{supp}_u$ instead of $\mathrm{supp}_u(x,y)$. We now prove the following lemma.

\begin{lemma}\label{lemma:kcenter-support}
Let $(\X,d,k)$ be a $2$-metric-perturbation-resilient instance of symmetric $k$-center with optimal clustering $\{C_1, ..., C_k\}$ and optimal value $R^*$. Let $R \leq R^*$. If $\cP(R) \neq \emptyset$, then for every $(x,y) \in \cP(R)$, for every $i \in [k]$ and $u \in C_i$, we have $\mathrm{supp}_u(x,y) \subseteq C_i$.
\end{lemma}
\begin{proof}
Let's assume that $(x,y) \in \cP(R) \neq \emptyset$. We consider the following standard greedy algorithm:
\begin{enumerate}
\setlength{\itemsep}{0pt}
    \item Let $A := \X$ and $j := 0$.
    \item while $(A \neq \emptyset)$:\\
            \hspace*{13pt}- let $j := j + 1$.\\
            \hspace*{13pt}- pick any point $u \in A$.\\
            \hspace*{13pt}- set $B_j := B(u,2R) \cap A$ and $b_j := u$\\
            \hspace*{13pt}- let $A := A \setminus B_j$.
    \item Return clustering $\{B_1, ..., B_j\}$.
\end{enumerate}
We will prove that the above algorithm is a 2-approximation algorithm for symmetric $k$-center. It is trivial to prove that the algorithm always terminates, since $u \in B(u,2R) \cap A$ for every $u \in A$. We will now prove that $j \leq k$. For that, we will prove that every center $b_i$ that the algorithm picks is ``fully paid for" by the fractional solution, and thus, at most $k$ such centers are opened. To see this, we first observe that for every point $u \in \X$, $\mathrm{supp}_u \subseteq B(u,R) \subseteq B(u,2R)$. This means that if at some iteration $i$ we have two different points $u, u' \in A$ such that $v \in \mathrm{supp}_u \cap \mathrm{supp}_{u'}$, then, if the algorithm sets $b_i = u$, we have $u' \in B_i$, as $d(u,u')\leq d(u,v) + d(v,u') \leq R + R = 2R$. Using induction, this implies that at the end of iteration $i$, the remaining points in $A \setminus B_i$ have disjoint support from the centers $b_1, ..., b_{i-1}$ that the algorithm has already opened. So, whenever the algorithm opens a center $b_i$, we know that $\sum_{v \in \mathrm{supp}_{b_i}} x_{uv} \geq 1$, which implies that $\sum_{v \in \mathrm{supp}_{b_i}} y_v \geq 1$, and, by the previous observation, none of these $y$-variables have paid for any other already open center. Thus, the variables $\{y_v\}_{v \in \mathrm{supp}_{b_i}}$ can fully pay for the center $b_i$. By induction, this holds for every center that the algorithm opens, and since $\sum_{v \in \X} y_v \leq k$, we get that the algorithm opens at most $k$ centers. Thus, we get a feasible solution of cost at $\max_{i \in[k]}\max_{u \in B_i} d(u, b_i) \leq 2R \leq 2R^*$.

By Theorem~\ref{thm:BHW16}, we now get that the clustering returned by the previous algorithm must in fact be the optimal clustering $\{C_1, ..., C_k\}$. Now, fix a point $u\in C_i$. The above algorithm always works, regardless of how the elements of A are selected at each iteration. Thus, if the algorithm selects $u$ in the first iteration, then $B_1 = B(u,2R)$ must be equal to the optimal cluster $C_i$, and, as noted in the previous paragraph, $\mathrm{supp}_u \subseteq B(u,2R)$. Thus, $\mathrm{supp}_u \subseteq C_i$, for every $i \in [k]$ and $u \in C_i$.
\end{proof}

We are now ready to state the main theorem of this section.
\begin{theorem}
Let $(\X,d, k)$ be a $2$-metric-perturbation-resilient instance of symmetric $k$-center with optimal clustering $\{C_1, ..., C_k\}$ and optimal value $R^*$. Let $\cP(R)$ be the polytope, as defined in Figure~\ref{fig:kcenter-lp}. Then, for every $R < R^*$, $\cP(R) = \emptyset$.
\end{theorem}
\begin{proof}
Let's assume that for some $R < R^*$, $\cP(R) \neq \emptyset$, and let $(x,y) \in \cP(R)$. By Lemma~\ref{lemma:kcenter-support}, we know that for every $i \in [k]$ and $u \in C_i$, we have $\mathrm{supp}_u \subseteq C_i$. This means that, for any $u\in C_i$, we have $1 \leq \sum_{v \in \mathrm{supp}_u} x_{uv} \leq \sum_{v \in \mathrm{supp}_u} y_v \leq \sum_{v \in C_i} y_v$. This, combined with the fact that $\sum_{v \in \X} y_v \leq k$, implies that we must have $\sum_{v \in C_i} y_v = 1$, for every $i \in [k]$. This further implies that for every $u,u' \in C_i$, $\mathrm{supp}_u = \mathrm{supp}_{u'}$. Let $\mathrm{supp}^{(i)} =  \mathrm{supp}_u$, for any $u \in C_i$. By the previous discussion, this is well defined.

We now pick any point $q_i \in \mathrm{supp}^{(i)}$, for each $i \in [k]$. From the observations of the previous paragraph, it is easy to see now that we have $\max_{u \in C_i} d(u, q_i) \leq R < R^*$. Thus, we get a feasible (integral) solution (with centers $\{q_1, ..., q_k\}$) with cost at most $R < R^*$, which is a contradiction. Thus, we must have $\cP(R) = \emptyset$ for every $R < R^*$.
\end{proof}

\begin{corollary}
There exists an efficient robust algorithm for $2$-metric-perturbation-resilient instances of symmetric $k$-center.
\end{corollary}
\begin{proof}
The previous theorem suggests a very simple robust algorithm. We check all possible values of $R$ (there are at most $O(n^2)$ such values), and let $\bar{R}$ be the smallest value such that $\cP(\bar{R}) \neq \emptyset$. We also use any 2-approximation algorithm (e.g.~the greedy algorithm presented above) and let $R'$ be the cost of the clustering returned by the algorithm. If $\bar{R} = R'$, we return the clustering computed by the algorithm, otherwise we report that the instance is not stable. The correctness follows from the previous discussion.
\end{proof}


\section{The $k$-median LP relaxation}

The $k$-median problem is yet another fundamental location problem in combinatorial optimization, that has received much of attention throughout the years (see e.g.~\cite{DBLP:journals/jcss/CharikarGTS02, DBLP:journals/jacm/JainV01, DBLP:conf/stoc/JainMS02, DBLP:journals/siamcomp/AryaGKMMP04}. The current best algorithm is due to Byrka et al.~\cite{DBLP:journals/talg/ByrkaPRST17} and gives a $(2.675 + \varepsilon)$-approximation, building upon the recent breakthrough of Li and Svennson~\cite{DBLP:journals/siamcomp/LiS16}. On the negative side, Jain et al.~\cite{DBLP:conf/stoc/JainMS02} proved that the $k$-median problem is hard to approximate within a factor $1 + 2/e \approx 1.736$.  Moreover, the natural LP relaxation of $k$-median, which we will introduce in this section, is known to have an integrality gap of at least 2. The best upper bound is by Archer et al.~\cite{DBLP:conf/esa/ArcherRS03} who showed that the integrality gap is at most 3 by giving an exponential-time rounding algorithm. In contrast, the best polynomial-time LP-rounding algorithm achieves an approximation ratio of 3.25~\cite{DBLP:conf/icalp/CharikarL12}.

Again, in the setting of perturbation resilience, the problem was first studied by Awasthi et al.~\cite{DBLP:journals/ipl/AwasthiBS12}, in which they gave a non-robust algorithm for 3-perturbation-resilient instances of $k$-median with no Steiner points (proper definitions are given below) and a non-robust algorithm for $(2 + \sqrt{3})$-perturbation-resilient instances in general metrics (with Steiner points). Balcan and Liang~\cite{DBLP:journals/siamcomp/BalcanL16} gave a non-robust algorithm for $(1 + \sqrt{2})$-perturbation-resilient instances $k$-center with no Steiner points. Finally, a non-robust algorithm for 2-metric-perturbation-resilient instances of $k$-median with no Steiner points was given in~\cite{DBLP:journals/corr/MakarychevM16, DBLP:conf/stoc/AngelidakisMM17}. In this section, we initiate the study of the $k$-median LP for perturbation-resilient instances, and we present some constructions of perturbation-resilient instances of $k$-median for which the standard LP relaxation is not integral. We first define the problem.
\begin{definition}[$k$-median]
Let $\{\X, \F, d\}$ such that $\X \cup \F$ is a finite metric space, $|\X| = n$, and let $k \in [n]$. Our goal is to select a set $C \subseteq \F$ of $k$ centers so as to minimize the function $\sum_{u \in \X} d(u, C)$.
\end{definition}

We will consider two classes of instances. The first is when $\X \cap \F = \emptyset$, and we will say that these are instances with Steiner points, and the second is when $\X = \F$, and we will say that these are instances with no Steiner points. We remind the reader that the algorithm for 2-metric-perturbation-resilient instances of~\cite{DBLP:conf/stoc/AngelidakisMM17} works for instances with no Steiner points; a slight modification of it can work for 3-metric-perturbation-resilient instances with Steiner points.

\subsection{Instances with Steiner points}

In this section, we prove that for every $\varepsilon > 0$, there exist $(2-\varepsilon)$-perturbation-resilient instances of $k$-median with Steiner points for which the standard LP relaxation is fractional. In particular, we consider the instance shown in Figure~\ref{fig:steiner}.

\begin{figure}[h!]
    \begin{center}
    \scalebox{0.8}{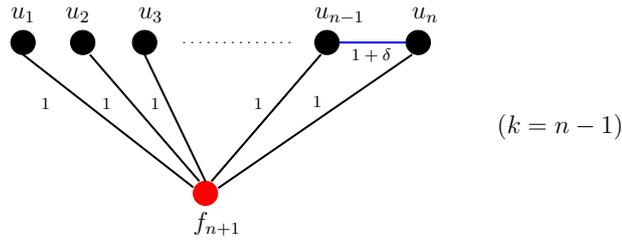}
    \caption{An integrality gap perturbation-resilient instance for $k$-median with Steiner points.}
    \label{fig:steiner}
    \end{center}
\end{figure}

The set of points (clients) to be clustered is $\X = \{u_1, ..., u_n\}$, and the set of facilities is $\F = \{f_1, ..., f_n, f_{n+1}\}$, where $f_i \equiv u_i$, for $1 \leq i \leq n$ (in case we want to have $\X \cap \F = \emptyset$, we can simply consider facilities at an $\varepsilon'$-distance from each point, for sufficiently small $\varepsilon' = \varepsilon'(n) > 0$). The black edges all have length 1 (i.e. $d(u_i, f_{n + 1}) = 1$ for every $i \in [n]$) and the blue edge has length $d(u_{n - 1}, n) = 1 + \delta$, where $\delta = \frac{n}{(n-1)^2}$. Let $d: (\X \cup \F) \times (\X \cup \F) \to \mathbb{R}_{\geq 0}$ be the shortest-path metric induced by the above graph.

We now consider the $k$-median objective, with $k = n-1$, and the standard LP relaxation, given in Figure~\ref{fig:kmedianLP-Steiner}. In this relaxation, the variables $\{z_f\}_{f \in \F}$ denote which facilities are open and the variables $x(u,f)$ denote which facilities the point $u$ is (fractionally) connected to.
\begin{figure}
\begin{align*}
    \min:          & \quad \sum_{u \in \X,f \in \F} d(u,f) x(u,f) &&\\
    \textrm{s.t.:} & \quad \sum_{f \in \F} x(u,f) = 1, &&\forall u \in \mathcal{X}, \\
                   & \quad x(u,f) \leq z_f, &&\forall u \in X,f\in \F,\\
                   & \quad \sum_{f \in \F} z_f \leq k, \\
                   & \quad x(u,f) \geq 0, &&\forall u \in \X, f \in \F,\\
                   & \quad z_f \geq 0, &&\forall f \in \F.
\end{align*}
\caption{The standard $k$-median LP relaxation for instances with Steiner points.}
\label{fig:kmedianLP-Steiner}
\end{figure}

\paragraph{Optimal integral solution.} We first show that the clustering $C_i = \{u_i\}, 1 \leq i \leq n-2$, $C_{n-1} = \{u_{n-1}, u_n\}$ is the unique optimal clustering, with corresponding centers $F = \{f_1, ..., f_{n-1}\}$. The cost of this clustering is $1 + \delta$. Observe that if we open facility $f_{n+1}$ or if we open both $f_{n-1}$ and $f_n$, then the cost will be at least $2$, which is strictly larger than $1 + \delta$ for every $n \geq 3$. Thus, $F$ is indeed an optimal selection of centers for the unique optimal clustering $\{C_1, ..., C_{n-1}\}$ whose cost, as noted, is $OPT = 1 + \delta$.

We will now show that the instance is $\gamma$-perturbation-resilient, for $\gamma = \frac{2 - \delta}{1 + \delta}$. In order for the optimal clustering to change after some perturbation, $u_{n-1}$ or $u_n$ must join some other cluster, or they must be in separate clusters with 1 point each:
\begin{itemize}
    \item $u_{n-1}$ moves together with $u_i$, for some $1 \leq i \leq n - 2$: The cost of any such solution is at least $2$, and the cost of the original optimal clustering is now at most $\gamma (1 + \delta)$. We have $\gamma (1 + \delta) = 2 - \delta < 2$. Thus, the original optimal clustering remains strictly better.
    \item $u_n$ moves together with $u_i$, for some $1 \leq i \leq n - 2$: This case is identical to the above.
    \item $u_{n-1}$ and $u_n$ are in separate clusters by themselves, and $u_i, u_j$ move together in some cluster, with $1 \leq i < j \leq n-2$. The cost of this solution is $2$, while the cost of the original optimal clustering is, again, at most $\gamma(1 + \delta) = 2 - \delta < 2$, and, so, it remains strictly better.
\end{itemize}
Thus, the instance is $\gamma$-perturbation-resilient for $\gamma = \frac{2-\delta}{1 + \delta}$.

\paragraph{Fractional solution.} Consider the fractional solution $z_{f_i} = x(u_i, f_i) = \frac{n-2}{n-1}$, for every $i \in [n]$, $z_{f_{n+1}} = \frac{1}{n-1}$ and $x(u_i, f_{n + 1}) = \frac{1}{n-1}$ for every $i \in [n]$. Note that this is indeed a feasible solution. The cost for each point $u_i$ is $\frac{1}{n-1}$ and, so, the optimal fractional cost is at most
\begin{equation*}
    OPT_{LP} \leq \frac{n}{n-1} = 1 + \frac{1}{n-1} < 1 + \frac{n}{n-1} \cdot \frac{1}{n-1} = 1 + \delta = OPT.
\end{equation*}

Thus, we have proved that $OPT_{LP} < OPT$, and so the integrality gap of the LP is strictly larger than 1. We are ready to formally state our result.
\begin{theorem}
For every $\varepsilon > 0$, there exist $(2 - \varepsilon)$-perturbation-resilient instances for which the standard LP relaxation for $k$-median with Steiner points is not integral.
\end{theorem}
\begin{proof}
By the previous analysis, we know that for every $n \geq 3$, there exist $\left(\frac{2 - \delta}{1 + \delta} \right)$-perturbation-resilient instances of $k$-median with $n$ points (where $\delta = n / (n-1)^2$) for which the LP has integrality gap strictly larger than 1. Thus, for any given $\varepsilon > 0$, we want to have $\frac{2 - \delta}{1 + \delta} \geq 2 - \varepsilon$, which is equivalent to $\delta \leq \frac{\varepsilon}{3 - \varepsilon}$. Setting $n\geq 12/\varepsilon$, the previous inequality is satisfied.
\end{proof}

\subsection{Instances with no Steiner points}

In this section, we prove that for every $\varepsilon > 0$, there exist $(\phi-\varepsilon)$-perturbation-resilient instances of $k$-median with no Steiner points for which the standard LP relaxation is fractional, where $\phi = \frac{1 + \sqrt{5}}{2} \approx 1.618$ is the golden ratio. In particular, we consider the instance shown in Figure~\ref{fig:no_steiner}.

\begin{figure}[h!]
    \begin{center}
    \scalebox{0.8}{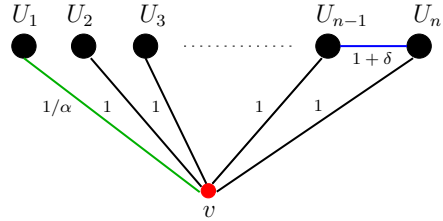}
    \caption{An integrality gap perturbation-resilient instance for $k$-median with no Steiner points.}
    \label{fig:no_steiner}
    \end{center}
\end{figure}

Each vertex $U_i$ is a ``super-vertex", i.e.~it contains $n$ different points at pairwise distances 0 (in order to have a proper metric, we can set the points in pairwise distances $\varepsilon' = \varepsilon'(n) > 0$ for sufficiently small $\varepsilon'$). The center $v$ of the star is a single vertex. Thus, the set of points to be clustered is $\X = (\bigcup_i U_i) \cup \{v\}$, and we are allowed to open a center (facility) at any point $u \in \X$. Note that $|\X| = n^2 + 1$. The black edges have length 1 (i.e.~$d(U_i, v) = 1$, $i \in \{2, ..., n\}$), the blue edge has length $1 + \delta$ where $\delta = \frac{2}{\alpha n}$ (i.e.~$d(U_{n - 1}, U_n) = 1 + \delta$), and the green edge has length $1 / \alpha$, where $\alpha \in (1.5, 2)$ is a constant to be specified later (i.e.~$d(U_1, v) = 1 / \alpha$). Let $d: \X \times \X \to \R_{\geq 0}$ be the shortest-path metric on the above graph.

We now consider the $k$-median objective, with $k = n-1$, and the standard LP relaxation, given in Figure~\ref{fig:kmedianLP-noSteiner}.
\begin{figure}
\begin{align*}
    \min:         & \quad \sum_{u,c \in \X} d(u,c) x(u,c) \\
    \textrm{s.t.:}& \quad \sum_{c \in \X} x(u,c) = 1, &&\forall u \in \mathcal{X}, \\
                  & \quad x(u,c) \leq z_c, && \forall u, c \in X,\\
                  & \quad \sum_{c \in \X} z_c \leq k, \\
                  & \quad x(u,c) \geq 0, && \forall u, c \in \X, \\
                  & \quad z_c \geq 0, && \forall c \in \X.
\end{align*}
\caption{The standard $k$-median LP relaxation for instances with no Steiner points.}
\label{fig:kmedianLP-noSteiner}
\end{figure}

\paragraph{Optimal integral solution.} We first show that the clustering $\C = \{C_1, ..., C_{n-1}\}$ is the unique optimal clustering, where:
\begin{itemize}
    \item $C_1 = U_1 \cup \{v\}$, with any point of $U_1$ serving as a center.
    \item $C_i = U_i$, for $ 2 \leq i \leq n -2$, with any point of the cluster serving as a center.
    \item $C_{n-1} = U_{n-1} \cup U_n$, with any point of the cluster serving as a center.
\end{itemize}
The cost of this clustering is $n(1 + \delta) + 1/\alpha = n\left(1 + \frac{2}{\alpha n}\right) + 1/\alpha = n + 3/\alpha$. In order to prove that the above clustering is indeed the unique optimal clustering, we do some case analysis. Let's assume that there exists an optimal solution that opens a center in $v$. Then, it is easy to see that the best clustering we can get is the following: one cluster is $U_1 \cup U_j \cup \{v\}$, for any $j \in \{2, ..., n\}$, and the remaining clusters are exactly the sets $U_i$, $i \in [n] \setminus \{1,j\}$. Its cost is $n + n/\alpha > n + 3/\alpha$ for every $n > 3$.

Let's assume now that $v$ is not a center. It is easy to see that opening two centers in the same location can never lead to an improved solution (for sufficiently small $\varepsilon'$). Thus, the only other case we have to consider is when we have centers both in $U_{n-1}$ and $U_n$. Then, it is again easy to see that the best clustering we can get is the following: one cluster is $U_1 \cup U_j \cup \{v\}$, for any $j \in \{2, ..., n-2\}$, with a center in $U_1$ (since we have assumed that $v$ cannot be a center), and the remaining clusters are exactly the sets $U_i$, $i \in [n] \setminus \{1,j\}$. Its cost is $n(1 + 1/\alpha) + 1 / \alpha$ which, similar to the previous paragraph, is strictly larger than $n + 3/\alpha$ for every $n > 3$. Thus, the clustering $\C$ is indeed the optimal clustering and its cost is $OPT = n + 3 /\alpha$.

We will now show that the instance is $\gamma$-stable, for $\gamma = \frac{\alpha n}{\alpha n + 4} \cdot \alpha < \alpha$. We first observe that there is no way that any set $U_i$ will be ``split" into two clusters. Since $\gamma < \alpha$, it is also immediate that $v$ will always be clustered together with $U_1$, unless we open centers both in $v$ and in some point of $U_1$. Let's assume that this is indeed the case. Then it is easy to see that the cost of the clustering is at least $2n$, and the cost of clustering $\C$ is at most $\gamma \cdot OPT$. We will now show that $\gamma \cdot OPT < \alpha n$. We have
\begin{equation*}
\begin{split}
    \gamma \cdot OPT &= \frac{\alpha n}{\alpha n + 4} \cdot \alpha (n + 3/\alpha) = \frac{\alpha n}{\alpha n + 4} \cdot \alpha n + \frac{3 \alpha n}{\alpha n + 4} = \alpha n \left(\frac{\alpha n}{\alpha n + 4} + \frac{3}{\alpha n + 4} \right) \\
                     &= \alpha n \cdot \frac{\alpha n + 3}{\alpha n + 4} < \alpha n.
\end{split}
\end{equation*}
Since $\alpha \in (1.5,2)$ we immediately get that $\gamma \cdot OPT < \alpha n < 2n$, and, thus, the original optimal clustering remains strictly better in this case.

So from now on, we assume that $v$ and $U_1$ belong to the same cluster. If the center of such a cluster is $v$, then from the analysis above, the cost of any such clustering is at least $n/\alpha + n = n(1  + 1/\alpha)$. Since $\gamma \cdot OPT < \alpha n$, it is sufficient to require that $\alpha \in (1.5, 2)$ is picked such that $\alpha \leq 1 + 1/\alpha$. Let's assume that this is indeed satisfied. Then, it is clear that the original optimal clustering remains strictly better in this case as well. Thus, we can now assume that $v$ and $U_1$ belong to the same cluster and the center of this cluster is some point in $U_1$. Suppose we open a center both in $U_{n-1}$ and $U_n$, and so the sets $U_{n-1}$ and $U_n$ are separated. From previous analysis, the cost is at least $n(1 + 1/\alpha) + 1 / \alpha$, which is strictly worse than the case considered in the beginning of the paragraph. So, the only remaining case is when, wlog, $U_n$ is in the same cluster as $U_j$, for some $j \in \{1, ..., n-2\}$. The cost of such a clustering is at least $n(1 + 1/\alpha) + 1/\alpha$, and, so, again, this falls into the previous analysis.

So, we conclude that the instance is indeed $\gamma$-stable, for $\gamma = \frac{\alpha n}{\alpha n + 4} \cdot \alpha $, given that $\alpha \in (1.5, 2)$ and $\alpha \leq 1 + 1/\alpha$. Solving this last inequality gives $\frac{1 - \sqrt{5}}{2} \leq \alpha \leq \frac{1 + \sqrt{5}}{2}$. Since we want the largest possible value of $\gamma$, we set $\alpha = \frac{1 + \sqrt{5}}{2} = \phi$, the golden ratio (note that $\phi \approx 1.618 \in (1.5, 2)$). Thus, the above instance is $\gamma$-stable for $\gamma = \frac{\phi n}{\phi n + 4} \cdot \phi$, and its optimal value is $OPT = n + 3/\phi$.

\paragraph{Fractional solution.} We arbitrarily pick some point $c_i \in U_i$, for each $i \in [n]$ and set $z_{c_i} = \frac{n-2}{n-1}$. We also set $z_v = \frac{1}{n-1}$. We then set $x(u, c_i) = z_{c_i}$ for each $i \in [n]$ and $u \in U_i$, and also set $x(u, v) = z_v$ for all $u \in \X$. Finally, we set $x(v, c_1) = z_{c_1}$. Again, it is easy to see that this is a feasible solution. The cost for each point $u \in U_1$ is $\frac{1}{\phi(n-1)}$, the cost for each point $u \in U_i$ for any $i \in \{2, ..., n\}$ is $\frac{1}{n-1}$ and the cost for $v$ is $\frac{n-2}{\phi (n-1)}$. Thus, the optimal fractional value is at most
\begin{equation*}
    OPT_{LP} \leq \frac{n}{\phi(n-1)} + \frac{n(n-1)}{n-1} + \frac{n-2}{\phi (n-1)} = n + \frac{2}{\phi} < n + \frac{3}{\phi} = OPT.
\end{equation*}
Thus, for $n > 3$, the integrality gap of the LP is strictly larger than 1. We are ready to formally state our result.
\begin{theorem}
For every $\varepsilon > 0$, there exist $(\phi - \varepsilon)$-perturbation-resilient instances for which the standard LP relaxation for $k$-median with no Steiner points is not integral.
\end{theorem}
\begin{proof}
By the previous analysis, we know that for every $n \geq 4$, there exist $\left(\frac{\phi n}{\phi n + 4 } \cdot \phi \right)$-per\-tur\-ba\-tion-resilient instances of $k$-median with $n^2 + 1$ points for which the LP has integrality gap strictly larger than 1. Thus, for any given $\varepsilon > 0$, it is easy to see that by setting $n \geq 4 / \varepsilon$, we have $\left(\frac{\phi n}{\phi n + 4 } \cdot \phi \right) \geq \phi - \varepsilon$.
\end{proof}


\chapter{Stability and the Traveling Salesman problem} \label{chap:tsp}

In this chapter, we prove that the standard ``subtour-elimination" LP has integrality gap exactly 1 for $1.8$-stable instances of the (symmetric) Traveling Salesman problem (which we denote as TSP from now on). For completeness, we first define the problem.
\begin{definition}[symmetric TSP]
Let $G = (V, E, w)$ be a complete graph with $n$ vertices, where $w: V \times V \to \mathbb{R}_{\geq 0}$ is a metric. The goal is to compute an ordering (tour) of the vertices $\pi: [n] \to V$ that contains all vertices so as to minimize the total length of the tour $\sum_{i = 1}^n w(\pi_i, \pi_{i + 1})$ (where we set $\pi_{n + 1} = \pi_1$).
\end{definition}

TSP is one of the most famous problems and has been studied for decades. However, its exact approximability is still unknown. Christofides' classic algorithm from 1976~\cite{Christofides-TSP} is a $1.5$-approximation algorihtm. Improving this guarantee is a notorious open question in approximation algorithms. It is conjectured that the LP given in Figure~\ref{fig:tsp-lp} has an integrality gap of 4/3, but the current best upper bound on the integrality gap only matches Christofides' performance.

\begin{figure}[h]
\begin{align*}
    \min: & \quad \sum_{e \in E} w_e x_e \\
    \textrm{s.t.:} & \quad x(\delta(u)) = 2, &&\forall u\in V\\
                   & \quad x(\delta(S)) \geq 2, && \forall S \subset V, \;\;S \notin \{\emptyset, V\} \\
                   & \quad x_e \in [0,1], &&\forall e \in E,
\end{align*}
\caption{The subtour-elimination LP for TSP.}
\label{fig:tsp-lp}
\end{figure}

In the stability framework, the only work that has studied TSP so far that we are aware of is the work of Mihal{\'{a}}k et al.~\cite{DBLP:conf/sofsem/MihalakSSW11}, in which they prove that a simple greedy algorithm solves 1.8-stable instances. In this chapter, we will show that the integrality gap of the LP of Figure~\ref{fig:tsp-lp} is exactly 1 for $1.8$-stable instances. This directly implies a robust analog of the algorithm of~\cite{DBLP:conf/sofsem/MihalakSSW11}. Before proving our result, we introduce the notation used in the LP of Figure~\ref{fig:tsp-lp}. For a subset $S \subseteq V$ of vertices, we denote as $\delta(S)$ the set of edges with exactly one endpoint in $S$, and $x(E') = \sum_{e \in E'} x_e$, for any $E' \subset E$.

We now prove a few lemmas that we need. The following lemma was first proved in~\cite{DBLP:conf/sofsem/MihalakSSW11}.

\begin{lemma}[\cite{DBLP:conf/sofsem/MihalakSSW11}]\label{lemma:bounded_ratio}
Let $G = (V,E,w)$ be a $\gamma$-stable instance of TSP, and let $O = (e_1, ..., e_n)$ be the unique optimal tour. If there exists $i \in[n]$ such that $\frac{w(e_i)}{w(e_{i+1})} > q$ (where $e_{n+1} \equiv e_1$), then $\gamma < \frac{(q+1)^2}{q^2 + 1}$.
\end{lemma}
\begin{proof}
Let's assume that there exist consecutive edges $e_i$, $e_{i+1}$, such that $w(e_i) > q \cdot w(e_{i+1})$. This means that there must exist 3 consecutive edges $e_t$, $e_{t+1}$, $e_{t+2}$ in the optimal tour such that $w(e_t) > q \cdot w(e_{t+1})$ and $w(e_{t+1}) \leq q \cdot w(e_{t+2})$. Let $e_t = (u_1,u_2)$, $e_{t+1} = (u_2,u_3)$ and $e_{t+2} = (u_3,u_4)$.

We consider the tour $C = (O \setminus \{e_t, e_{t+2}\}) \cup \{(u_1, u_3), (u_2, u_4)\}$. Since the instance is $\gamma$-stable, we must have $\frac{w(u_1,u_3) + w(u_2,u_4)}{w(e_t) + w(e_{t+2})} > \gamma$. We have
\begin{equation*}
\begin{split}
    \frac{w(u_1,u_3) + w(u_2,u_4)}{w(e_t) + w(e_{t+2})} &\leq \frac{w(e_t) + 2w(e_{t+1}) + w(e_{t+2})}{w(e_t) + w(e_{t+2})} = 1 + \frac{2w(e_{t+1})}{w(e_t) + w(e_{t+2})} \\
                                                        &< 1 + \frac{2w(e_{t+1})}{q \cdot w(e_{t+1}) + w(e_{t+1})/q} = 1 + \frac{2q}{q^2 + 1}.
\end{split}
\end{equation*}
Thus, we conclude that $\gamma < \frac{(q+1)^2}{q^2+1}$.
\end{proof}
The above lemma implies that for $\gamma = 1.8$, we have $\frac{w(e_i)}{w(e_{i+1})} \leq 2$ for every $i \in [n]$. We now prove a useful lemma that is a consequence of the previous one. Throughout the rest of the text, we use the following notation. The optimal tour is denoted by $O$ and is simply the order in which edges (or vertices) are visited in the tour, and for each vertex $u \in V$, the two edges of $O$ that are adjacent to $u$ are denoted as $e_{u,1}$ and $e_{u,2}$.

\begin{lemma}\label{lemma:bound_non_edges}
Let $G = (V,E,w)$ be a $1.8$-stable instance of TSP and let $u$ and $v$ be two non-adjacent (w.r.t.~the optimal tour) vertices. Then, we have $w(u,v) > \frac{1}{2} \cdot \left(w(e_{u,i}) + w(e_{v,j})\right)$, for every $i,j \in \{1,2\}$.
\end{lemma}
\begin{proof}
To simplify notation, we set $e = (u,v)$. W.l.o.g.~we assume that the edges $e_{u,2}$ and $e_{v,2}$ are non-adjacent, since otherwise $e_{u,1}-e_{v,1}$ and $e_{u,2}-e_{v,2}$ are adjacent and thus the total number of vertices of the graph is exactly 6. We now distinguish between two cases:\\

\noindent \textbf{Case 1:} the edges $e_{u,1}$ and $e_{v,1}$ are non-adjacent (see Figure \ref{fig:case1}). We consider the tour $C = O \setminus \{e_{u,1}, e_{v,1}, e_{v,2}\} \cup \{e, x, y\}$. We have
\begin{equation*}
\begin{split}
    1.8 &< \frac{w(e) + w(x) + w(y)}{w(e_{u,1}) + w(e_{v,1}) + w(e_{v,2})} \leq \frac{w(e_{u,1}) + 2w(e) + w(e_{v,1}) + w(e_{v,2})}{w(e_{u,1}) + w(e_{v,1}) + w(e_{v,2})}\\
        &= 1 + \frac{2w(e)}{w(e_{u,1}) + w(e_{v,1}) + w(e_{v,2})}.
\end{split}
\end{equation*}
This gives $w(e) > \frac{4}{5} \cdot \frac{w(e_{u,1}) + w(e_{v,1}) + w(e_{v,2})}{2}$. Using Lemma~\ref{lemma:bounded_ratio}, we get that $w(e_{v,2}) \geq w(e_{v,1})/2$. Thus, we conclude that $w(e) > \frac{4}{5} \cdot \left(\frac{w(e_{u,1})}{2} + \frac{3w(e_{v,1})}{4} \right)$. Note that using the fact that $w(e_{v,1}) \geq w(e_{v,2})/2$, we also get that $w(e) > \frac{4}{5} \cdot \left(\frac{w(e_{u,1})}{2} + \frac{3w(e_{v,2})}{4} \right)$.

By symmetry of the above argument, we also get that $w(e) > \frac{4}{5} \cdot \left(\frac{3w(e_{u,1})}{4} + \frac{w(e_{v,1})}{2} \right)$ and $w(e) > \frac{4}{5} \cdot \left(\frac{3w(e_{u,1})}{4} + \frac{w(e_{v,2})}{2} \right)$  Adding the corresponding inequalities, we conclude that $w(e) > \frac{1}{2} \cdot \left( w(e_{u,1}) + w(e_{v,1}) \right)$ and, similarly, $w(e) > \frac{1}{2} \cdot \left( w(e_{u,1}) + w(e_{v,2}) \right)$. Using symmetry again, we conclude that for every $i,j \in\{1,2\}$, $w(e) > \frac{1}{2} \cdot \left( w(e_{u,i}) + w(e_{v,j}) \right)$.
\begin{figure}[ht!]
\centering
\scalebox{1.4}{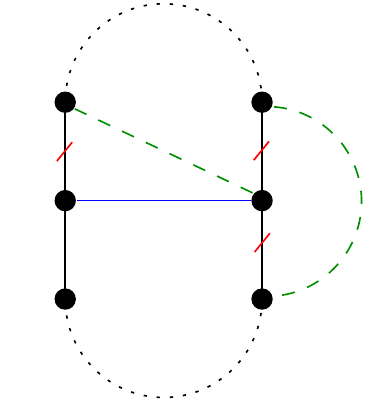}
\caption{Case 1: edges $e_{u,1}$ and $e_{v,1}$ are not adjacent.}
\label{fig:case1}
\end{figure}

\noindent \textbf{Case 2:} $e_{u,1}$ and $e_{v,1}$ are adjacent (see Figure \ref{fig:case2}). W.l.o.g.~we assume that $w(e_{v,1}) \leq w(e_{u,1})$. We consider the tour $C = (O \setminus \{e_{u,1}, e_{v,2}\}) \cup \{e, x\}$. We have
\begin{equation*}
    1.8 < \frac{w(e) + w(x)}{w(e_{u,1}) + w(e_{v,2})} \leq \frac{w(e) + w(e_{v,1}) + w(e_{v,2})}{w(e_{u,1}) + w(e_{v,2})},
\end{equation*}
which implies that $w(e) > 1.8 w(e_{u,1}) + 1.8 w(e_{v,2}) - w(e_{v,1}) - w(e_{v,2})$. This further implies that
\begin{equation*}
    w(e) > \frac{4}{5} \left( w(e_{u,1}) + w(e_{v,2}) \right) \geq  \frac{4}{5} \left( w(e_{u,1}) + \frac{w(e_{v,1})}{2} \right) \geq \frac{4}{5} \left( \frac{3w(e_{u,1})}{4} + \frac{3w(e_{v,1})}{4} \right).
\end{equation*}
Thus, we conclude that $w(e) > \frac{4}{5} \left( \frac{w(e_{u,1})}{2} + \frac{3w(e_{v,1})}{4} \right)$. The above inequalities also immediately imply that $w(e) > \frac{4}{5} \left( \frac{3w(e_{u,1})}{4} + \frac{w(e_{v,1})}{2} \right)$. Adding the two inequalities, we again get that $w(e) > \frac{1}{2} \cdot \left( w(e_{u,1}) + w(e_{v,1}) \right)$. It is also immediate to get that $w(e) > \frac{1}{2} \cdot \left( w(e_{u,1}) + w(e_{v,2}) \right)$.

The remaining inequalities (i.e. that $w(e) > \frac{1}{2} \cdot \left( w(e_{u,2}) + w(e_{v,j}) \right)$ for $j \in \{1,2\}$) fall into Case 1, and so an identical proof as in Case 1 gives the desired inequalities.
\begin{figure}[ht!]
\centering
\scalebox{1.4}{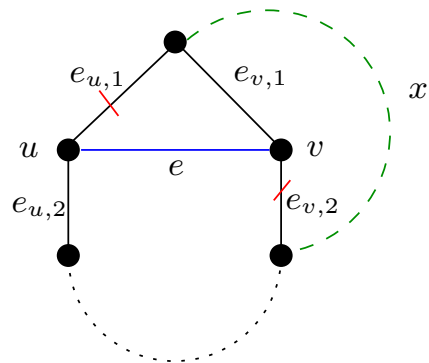}
\caption{Case 2: edges $e_{u,1}$ and $e_{v,1}$ are adjacent.}
\label{fig:case2}
\end{figure}
\end{proof}

We are now ready to prove the main theorem.
\begin{theorem}
The LP of Figure~\ref{fig:tsp-lp} has integrality gap 1 for 1.8-stable instances of TSP.
\end{theorem}
\begin{proof}
Let's assume that the integrality gap is strictly larger than 1, and let $G = (V,E,w)$ be a 1.8-stable instance such that $OPT_{LP} < OPT$. Let $x$ be an optimal fractional solution, and let $O = (u_1, u_2, ..., u_n, u_1)$ be the unique optimal tour, where $u_i \in V$. Let $\zeta(u) = \{e \in E \setminus O: e \textrm{ adjacent to }u \textrm{ and }x_e > 0\}$ (see Figure \ref{fig:frac}).
\begin{figure}[ht!]
\centering
\scalebox{1.4}{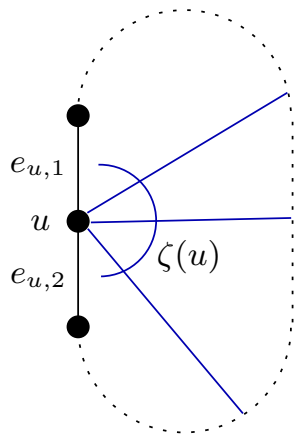}
\caption{The fractional support of the LP solution.}
\label{fig:frac}
\end{figure}

We have $x(\zeta(u)) = \sum_{e \in \zeta(u)} x_e = 2 - x_{e_{u,1}} - x_{e_{u,2}}$. Since the fractional optimal is strictly better than the integral optimal, we have $\sum_{e \in O} w_e x_e + \sum_{e \notin O} w_e x_e < \sum_{e \in O} w_e$ which implies that $\sum_{e \notin O} w_e x_e < \sum_{e \in O} w_e (1 - x_e)$.

Let $e = (u,v) \notin O$ such that $x_e > 0$. Its contribution to the sum is $w_e x_e$. Lemma \ref{lemma:bound_non_edges} implies that $w_e x_e > y_{e,u,1} \cdot \frac{w(e_{u,1})}{2} + y_{e,u,2} \cdot \frac{w(e_{u,2})}{2} + y_{e,v,1} \cdot \frac{w(e_{v,1})}{2} + y_{e,v,2} \cdot \frac{w(e_{v,2})}{2}$, for any non-negative numbers $y_{e,u,1}, y_{e,u,2}, y_{e,v,1}, y_{e,v,2}$ such that $y_{e,u,1} + y_{e,u,2} = y_{e,v,1} + y_{e,v,2} = x_e$. Thus, we can write
\begin{equation*}
    \sum_{e \notin O} w_e x_e > \sum_{u \in V} \left(\frac{w(e_{u,1})}{2}  \sum_{e \in \zeta(u)} y_{e,u,1} + \frac{w(e_{u,2})}{2}  \sum_{e \in \zeta(u)} y_{e,u,2} \right).
\end{equation*}
Since $y_{e,u,1} + y_{e,u,2} = x_e$, we get that
\begin{equation*}
    \sum_{e \in \zeta(u)}y_{e,u,1} + \sum_{e \in \zeta(u)}y_{e,u,2} = \sum_{e \in \zeta(u)} x_e = (1 - x_{e_{u,1}}) + (1 - x_{e_{u,2}}).
\end{equation*}
This implies that there exists a choice of values $y_{e,u,1}, y_{e,u,2}$ such that $\sum_{e \in \zeta(u)}y_{e,u,1} = 1 - x_{e_{u,1}}$ and $\sum_{e \in \zeta(u)}y_{e,u,2} = 1 - x_{e_{u,2}}$. Plugging this to the inequality above, we get that
\begin{equation*}
    \sum_{e \notin O} w_e x_e > \sum_{u \in V} \left(\frac{w(e_{u,1})}{2}  (1 - x_{e_{u,1}}) + \frac{w(e_{u,2})}{2}  (1 - x_{e_{u,2}}) \right).
\end{equation*}
Observe that the edges appearing in the above sum are exactly the edges of the optimal tour. Moreover, since each edge has 2 endpoints, each edge of the optimal tour appears twice in the above sum. Thus, we conclude that
\begin{equation*}
    \sum_{e \notin O} w_e x_e > \sum_{e \in O} w_e  (1 - x_e),
\end{equation*}
which is a contradiction. Thus, the integrality gap of the LP is exactly 1.
\end{proof}

The above theorem immediately suggests a robust algorithm for $1.8$-stable instances of TSP: we run the greedy algorithm of Mihal{\'{a}}k et al.~\cite{DBLP:conf/sofsem/MihalakSSW11}, and then, by solving the LP, we verify whether the solution that we got from the greedy algorithm is optimal or not.

We also make the following interesting observation. If one looks closely at the proof of the above theorem, the exponentially many ``subtour elimination" constraints are not really used in the above proof. This means that even if we drop them, the LP would still have integrality gap exactly 1 on 1.8-stable instances. In other words, the much simpler ``cycle-cover" LP is sufficient to obtain a robust algorithm for 1.8-stable instances of TSP.

\chapter{Open problems from Part II}\label{chapter:open-problems-stability}

In this concluding chapter of Part II, we will state a few open problems related to Bilu-Linial stability that we believe are of interest.
\begin{enumerate}
    \item Since the Edge Multiway Cut problem is one of the very well studied problems, it would be nice to pinpoint the performance of the CKR relaxation on $\gamma$-stable instances. In Chapter~\ref{chap:multiway-cut}, we have shown that, as $k$ grows, the CKR LP is integral for $2$-stable instances and non-integral for $(4/3 - \varepsilon)$-stable instances, so it would be interesting to close this gap.
    \item In Chapter~\ref{chap:independent-set}, we presented several combinatorial and LP-based algorithms for stable instances of Independent Set. However, as the current state-of-the-art algorithms for Independent Set on bounded-degree graphs use SDPs (and corresponding hierarchies), it is only natural to ask what is the performance of SDPs on stable instances. In particular, does the Lov{\'{a}}sz Theta function SDP~\cite{Lovasz:2006:SCG:2263335.2269451} solve $o(\Delta)$-stable instances of Independent Set on graphs of maximum degree $\Delta$?
    \item It would be interesting to further explore the power of convex relaxations for perturbation-resilient instances of clustering problems such as $k$-median and $k$-means. Our preliminary results in Chapter~\ref{chap:clustering} about symmetric $k$-center (which were obtained independently and also extended to the asymmetric case of $k$-center in~\cite{Chekuri-Approx18}) reinforce our belief that convex relaxations, in most cases, perform at least as good as combinatorial algorithms. It would be interesting to explore the power of the standard LP relaxation for $k$-median, and either obtain lower bound constructions for $(2 - \varepsilon)$-perturbation-resilient instances of $k$-median (thus excluding the possibility of improved algorithms by using the LP), or obtain upper bounds for $\gamma$-perturbation-resilient instances, for $\gamma < 2$ (the latter would be an improvement over the combinatorial algorithm that solves $2$-perturbation-resilient instances of $k$-median~\cite{DBLP:conf/stoc/AngelidakisMM17}).
    \item As already implied, an interesting fact about stable instances is that, in most cases, convex relaxations seem to perform at least as good as the best combinatorial algorithms. A very useful property of algorithms based on the integrality of convex relaxations on stable instances is that the algorithms are robust (i.e.~they never err). It would be interesting to find a problem for which there is a gap between the performance of non-robust and robust algorithms. More precisely, it would be interesting to demonstrate a problem for which there is a lower bound of $\gamma$ for robust algorithms that solve $\gamma$-stable instances and for which there is a non-robust algorithm that solves $\gamma'$-stable instances, for $\gamma' < \gamma$.
    \item Since the notion of stability of Bilu and Linial is quite general and versatile, many more problems could, potentially, be studied in this framework, such as Sparsest Cut, Balanced Cut, Minimum $k$-Cut etc.

\end{enumerate}

\clearpage
\phantomsection
\addtocontents{toc}{\vspace*{0.5cm}}
\addcontentsline{toc}{chapter}{Bibliography}
\bibliography{references}
\bibliographystyle{plain}

\end{document}